\documentclass[a4paper,11pt,twoside=false]{scrbook}

\usepackage{hyperref}
\usepackage{graphicx}
\usepackage{amssymb}
\usepackage{slashed}
\usepackage{amsmath, amscd, amssymb, amsthm}
\usepackage[mathscr]{eucal}
\usepackage{esint}
\usepackage{enumitem}
\usepackage[utf8]{inputenc}	%
\usepackage{setspace} %
\usepackage[letterspace=20]{microtype}	%

\usepackage[off]{auto-pst-pdf} %
\usepackage{epsfig}
\usepackage{pst-grad} %
\usepackage{pst-plot} %
\usepackage[mathscr]{eucal}

\usepackage{xcolor}
\definecolor{darkblue}{RGB}{0,91,163} %
\definecolor{lightgreen}{RGB}{169,245,169} %
\definecolor{lightgreen2}{RGB}{152,228,152} %
\definecolor{lightred}{RGB}{250,88,88} %

\newcommand{\SetFigFont}[3]{}

\newtheorem{Def}{Definition}[section]
\newtheorem{Thm}[Def]{Theorem}
\newtheorem{Prp}[Def]{Proposition}
\newtheorem{Lemma}[Def]{Lemma}
\newtheorem{Remark}[Def]{Remark}
\newtheorem{Notation}[Def]{Notation}
\newtheorem{Corollary}[Def]{Corollary}
\newtheorem{Example}[Def]{Example}
\newtheorem{Assumption}[Def]{Assumption}

\newcommand{\beq}{\begin{equation}}
\newcommand{\eeq}{\end{equation}}
\newcommand{\Proof}{\begin{proof}}
\newcommand{\QED}{\end{proof} \noindent}
\newcommand{\QEDrem}{\ \hfill $\Diamond$}

\newcommand{\la}{\langle}
\newcommand{\ra}{\rangle}

\newcommand{\Sl}{\mbox{$\prec \!\!$ \nolinebreak}}
\newcommand{\Sr}{\mbox{\nolinebreak $\succ$}}

\newcommand{\C}{\mathbb{C}}
\newcommand{\R}{\mathbb{R}}
\newcommand{\1}{\mbox{\textrm 1 \hspace{-1.05 em} 1}}
\newcommand{\Z}{\mathbb{Z}}
\newcommand{\N}{\mathbb{N}}

\newcommand{\F}{\mathscr{F}} %

\newcommand{\Pdd}{\mbox{$\partial$ \hspace{-1.2 em} $/$}}
\newcommand{\slsh}{\mbox{ \hspace{-1.13 em} $/$}}
\newcommand{\Aslsh}{\mbox{ $\!\!A$ \hspace{-1.2 em} $/$}}

\newcommand{\na}{n_{\mathrm{a}}}

\newcommand{\scrU}{{\mathscr{U}}}
\newcommand{\scrA}{{\mathscr{A}}}
\newcommand{\D}{{\mathscr{D}}}

\DeclareMathOperator{\re}{Re}
\DeclareMathOperator{\im}{Im}
\DeclareMathOperator{\Tr}{Tr}
\DeclareMathOperator{\tr}{tr}

\DeclareMathOperator{\supp}{supp}

\renewcommand{\O}{{\mathscr{O}}}
\renewcommand{\L}{{\mathcal{L}}}
\newcommand{\Sact}{{\mathcal{S}}}
\newcommand{\T}{{\mathcal{T}}}
\newcommand\B{{\mathscr{B}}}
\newcommand\calB{{\mathcal{B}}}
\newcommand{\U}{\text{\textrm{U}}}
\newcommand{\SU}{\text{\textrm{SU}}}
\newcommand{\J}{\mathfrak{J}}

\newcommand{\Jdiff}{\mathfrak{J}^\text{\textrm{\tiny{diff}}}}
\newcommand{\Jtest}{\mathfrak{J}^\text{\textrm{\tiny{test}}}}

\newcommand{\Jlin}{\mathfrak{J}^\text{\textrm{\tiny{lin}}}}

\newcommand{\Jfield}{\mathfrak{J}^\text{\textrm{\tiny{field}}}}
\newcommand{\Gdiff}{\Gamma^\text{\textrm{\tiny{diff}}}}
\newcommand{\Gtest}{\Gamma^\text{\textrm{\tiny{test}}}}
\newcommand{\Ctest}{C^\text{\textrm{\tiny{test}}}}
\renewcommand{\u}{\mathfrak{u}}
\renewcommand{\v}{\mathfrak{v}}
\newcommand{\w}{\mathfrak{w}}

\newcommand{\UMNS}{U_\text{\textrm{\tiny{MNS}}}}

\renewcommand{\H}{\mathscr{H}}
\newcommand{\Lin}{\text{\textrm{L}}}
\newcommand{\itemD}{\item[{\raisebox{0.125em}{\tiny $\blacktriangleright$}}]}

\newcommand*\mcupinn[2]{\vcenter{\hbox{$\mathsurround=0pt
  \ifx\displaystyle#1\textstyle\else#1\fi\bigcup$}}}

\newcommand{\scrM}{\mycal M}
\newcommand{\scrN}{\mycal N}
\newcommand{\x}{{\textit{\myfont x}}}
\newcommand{\y}{{\textit{\myfont y}}}
\newcommand*{\myfont}{\fontfamily{ppl}\selectfont}		%
\DeclareFontFamily{OT1}{rsfso}{}
\DeclareFontShape{OT1}{rsfso}{m}{n}{ <-7> rsfso5 <7-10> rsfso7 <10-> rsfso10}{}
\DeclareMathAlphabet{\mycal}{OT1}{rsfso}{m}{n}
\newcommand{\wnabla}{\widetilde{\nabla}}

\newcommand\Later[1]{\marginpar {\flushleft\sffamily\footnotesize {\texttt{\textcolor{lightgreen2}{Later}}}: \textcolor{lightgreen}{#1}}}

\newcommand\Chd[1]{}
\newcommand\chd[1]{#1}

\renewcommand\Later[1]{}	%
\newcommand\Remt[1]{\textsc{#1}}	%
\newcommand\Thmt[1]{\textsc{\textbf{#1}}}	%

\usepackage{geometry}

\setkomafont{disposition}{\normalfont} %
\usepackage{titlesec}
\titleformat{\chapter}[display]
{\normalfont\Large}{\filcenter{\chaptertitlename}\ {\Large \thechapter}}{2pt}{\filcenter\Huge}
\setkomafont{subsubsection}{\large}

\begin{document}

\begin{titlepage}
\newgeometry{bottom=1in,left=1in,right=1in}
\begin{center}
\textls{\Huge Dynamics of Causal Fermion Systems\\[.5em]%
{\LARGE Field Equations and Correction Terms\\[.8em]%
for a New Unified Physical Theory}
}\\[3.1em]

{\Large Dynamik kausaler Fermionensysteme}\\[.3em]
{\em Feldgleichungen und Korrekturterme\\[.3em]
für eine neue vereinheitlichte Theorie}\\

\vspace{1.3cm}

\includegraphics[scale=.17]{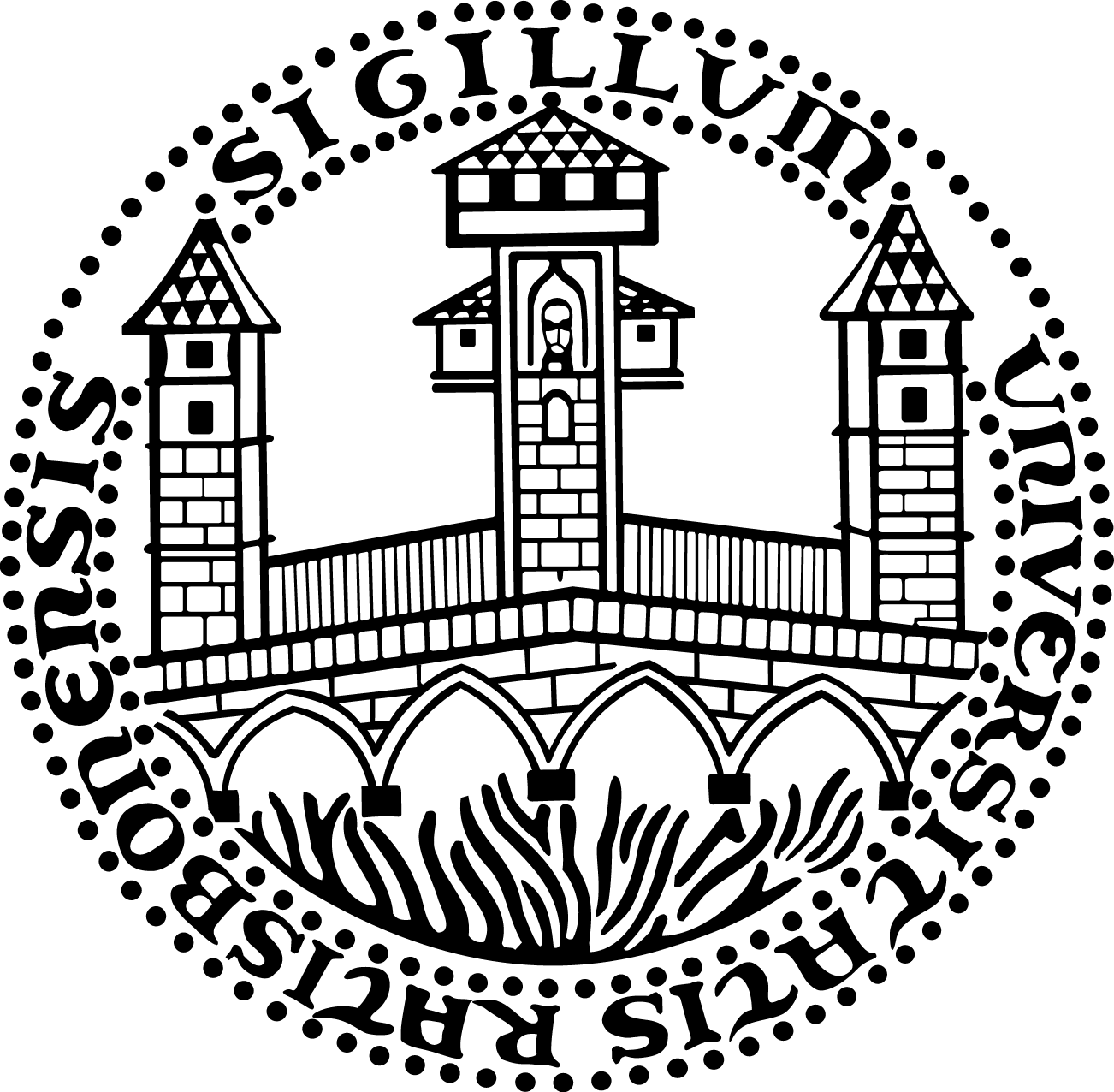}
\vspace{1.5cm}

\textls{\setstretch{1.2}
{\Large Dissertation}\\
zur Erlangung des Doktorgrades\\
der Naturwissenschaften (Dr.\:rer.\:nat)\\
der Fakultät für Mathematik\\
der Universität Regensburg\\[2.5em]%

vorgelegt von\\
Johannes Kleiner\\
aus Kempten (Allgäu)\\
im Jahr 2017\\ 
}
\end{center}
\newpage \thispagestyle{empty} \	%
\end{titlepage}

\newpage \thispagestyle{empty} \
\ \\ \vfill
\noindent Promotionsgesuch eingereicht am 19.07.2017.\\[2em]
Die Arbeit wurde angeleitet von Prof. Dr. Felix Finster.\\[2em]
\begin{tabular}{@{}lll}
Prüfungsausschuss: & Vorsitzender: & Prof. Dr. Denis-Charles Cisinski \\[.4em]
 & 1.\! Gutachter: & Prof. Dr. Felix Finster \\[.4em] 
 & 2.\! Gutachter: & Prof. Dr. Peter Pickl  (LMU München) \\[.4em] 
 & weiterer Prüfer: & Prof. Dr. Helmut Abels \\ 
\end{tabular} 
\vspace{1cm}

\newpage \thispagestyle{empty} \ %
\newpage \thispagestyle{empty} %

\newpage \thispagestyle{empty}

\vspace*{10cm}
\textit{I dedicate this thesis
to humanity in the hope that it will continue its bright path to
establish a world which is just and good to all.\Chd{Ich würde das lieber so lassen. Inhaltlich ist meiner Ansicht nach ganz klar, dass sich ``it'' auf das letzte Substantiv bezieht. Wäre das okay?}%
} 

\newpage \thispagestyle{empty} \ %

\newpage \thispagestyle{empty}

\vspace*{.5cm}
\begin{center}
{\Large Abstract}
\end{center}
 \begin{quote}\small
The theory of causal fermion systems is a new physical theory which aims to describe a fundamental level of physical reality.
Its mathematical core is the causal action principle. In this thesis, we develop a
formalism which connects the causal action principle to a suitable notion of fields
on space-time. We derive field equations from the causal action principle and find that the dynamics induced by the field equations conserve a symplectic form
which gives rise to an Hamiltonian time evolution
if the causal fermion system admits a notion of `time'.
In this way, we establish the dynamics of causal fermion systems.
\\
Remarkably, the causal action principle implies that there are correction terms to the field equations, which we subsequently derive and study. In particular, we prove that there is a stochastic and a non-linear correction term and investigate how they relate to the Hamiltonian time evolution. 
Furthermore, we give theorems which generalize the connection between symmetries and conservation laws in Noether's theorems to the theory of causal fermion systems.
The appearance of the particular correction terms is reminiscent of dynamical collapse models in quantum theory.
\end{quote}

\vspace{.5cm} 
 
\begin{center}
{\Large Kurzzusammenfassung}
\end{center}
\begin{quote}\small
Die Theorie der kausalen Fermionensysteme ist eine kürzlich entwickelte phy\-si\-ka\-lische Theorie, welche die fundamentalen Theorien
in der Physik vereinheitlicht und somit einen neuen Vorschlag für die Beschreibung der grundlegenden Strukturen der physikalischen Realität darstellt.
Die maßgebliche mathematische Struktur in dieser Theorie bildet das sogenannte kausale Wirkungsprinzip. Gegenstand dieser Arbeit ist es,
einen Formalismus zu entwickeln, welcher aufbauend auf diesem abstrakten Prinzip eine Beschreibung liefert, \chd{die Parallelen zu den üblichen Theorien
in der Physik hat: Die Dynamik von kausalen Fermionensystemen.}
\Chd{Ist dieser Satz so in Ordnung?}%

Ergebnis dieser Forschungen ist der sogenannte Jet-Formalismus von kausalen Fermionensystemen, welcher als verallgemeinerte physikalische Felder $1$-Jets
zu Grunde legt. Basierend auf den Euler-Lagrange Gleichungen des kausalen Wirkungsprinzips leiten wir Feldgleichungen ab und studieren deren Lösungen.
Es zeigt sich, dass eine symplektische Form konstruiert werden kann, welche unter der von den Feldgleichungen induzierten Dynamik erhalten ist.
Dies führt zur Definition einer Hamiltonschen Zeitentwicklung für kausale Fermionensysteme.

Bemerkenswerterweise führt das kausale Wirkungsprinzip zu Korrekturtermen für die Feldgleichungen, die wir im weiteren Verlauf der Arbeit untersuchen.
Unsere Theoreme etablieren einen stochastischen und einen nicht-linearen Korrekturterm, sowie deren Beziehung zu markoskopischen Gleichungen und zur
eingangs erwähnten symplektischen Form. Das Auftreten dieser speziellen Korrekturterme deutet auf Parallelen zu dynamischen Kollaps-Theorien in der nicht-relativistischen Quantenmechanik hin.
Die Arbeit enthält außerdem Theoreme, welche den durch die Noetherschen Theoreme gegebenen Zusammenhang zwischen Symmetrien und Erhaltungsgrößen
auf die Theorie der kausalen Fermionensysteme übertragen.
\end{quote}

\newpage \thispagestyle{empty} \ 
\newpage \thispagestyle{empty} 	%
\begin{center}
{\Large Acknowledgements}
\end{center}

First of all, I want to thank the one person whom I could talk to about the details of my research, my supervisor Felix Finster.
I am grateful for the inspiration, conversations, ways of thinking and knowledge which you gave me,
and deeply appreciate the freedom you provided for me to study other fields and visit other institutions. 

Next, I want to thank my personal friends and my family who enriched my life so considerably during my studies and during my PhD.
For providing for all my needs and comforts, both material and mental, I want to thank Lucia, Teresa and Meinrad.
For being the most supportive partner, I want to thank Barbara.  
For spurring my mind with the most inspiring conversations about the foundations of reality, I want to thank Robin and Clemens.
For the most enlightening late hour coffee house physics conversations and the most fascinating dance moves, I want to thank Jan-Hendrik.
For being the most reliable and honest friends over many years, and for being 
hidden safe harbours in my life, I want to thank Martin, Sven and Clemens.

For getting entangled in spiralling scientific and philosophic debates, I want to thank the many friends which I have \chd{made during} the Rethinking workshops, and in particular
my co-organizers Robin, Fede, Jan-Hendrik and Franz.
\Chd{``former'' raus}%
For putting up with my ever critical attitude, and for much enlightenment, I want to thank the members of my reading circle in Regensburg.
For co-organizing and being a source of countless ideas, I want to thank Hermann Josef, Sven and Barbara.
For being the most fun company and for keeping my bounds with home, I want to thank my Allgäu friends, in particular Thommy, as well as Soldo and
Elfriede.

Finally, my thanks go to the many friends I cannot mention here, in particular from my studies in Regensburg, Heidelberg and Freiburg.
In endless conversations, you have given me orientation in life and goals to strive for. Thanks so much for that.
They go to my former supervisors, who have invested so much of their time and effort into furthering my education, to my former and current colleagues, who I have shared with many jolly and amusing moments and to the many wonderful people who I have met in other scientific institutions, in particular in Oxford and Marseille.
All of you have given me confidence and trust in science, and have shown me the warm-hearted part of this joint endeavour of ours.

Concerning institutions, most of all I want to thank the Studienstiftung des Deutschen Volkes for providing trust in my capabilities and visions by financing my PhD and
many journeys.
Next, I want to thank the Department of Mathematics at the University of Regensburg for providing my daily needs and a very lively mathematical atmosphere, and in particular the 
DFG Graduate School GRK 1692 for financing journeys and for enabling one of the most dearest activities of my PhD, the organization of the Rethinking Foundations of Physics Workshops, by providing financial means and infrastructure. Thank you so much for that! 
Finally, I want to thank the Center of Mathematical Sciences and Applications of Harvard University, the Department of Computer Science of Oxford University and the 
Centre de Physique Théorique of Aix-Marseille Université for hospitality and support.

 \thispagestyle{empty}

\newpage \thispagestyle{empty} \  %

\newpage
\setcounter{page}{1}
\pagenumbering{roman}
\tableofcontents

\chapter{Introduction}\label{DissIntroduction}
\pagenumbering{arabic}
\setcounter{page}{1}

The theory of causal fermion systems is a new physical theory which aims to describe a fundamental level of physical reality.
\Chd{Seitenzahlen des Inhaltsverzeichnisses auf ``roman'' geändert.}%
Since it gives quantum mechanics, general relativity and quantum field theory as limiting cases (\cite{cfs,qft}),
it constitutes a unification of the currently accepted fundamental physical theories. 

The mathematical structure of the theory of causal fermion systems differs substantially from the mathematical structure of 
quantum mechanics, quantum field theory and general relativity.
The basic object is a \textit{causal fermion system}, defined as 
a triple $(\H,\F,\rho)$, where $\H$ is a separable complex Hilbert space, $\F \subset \Lin(\H)$ is the set of all those self-adjoint linear operators on $\H$ which (counting multiplicities) have at most~$n$ positive and at most~$n$ negative eigenvalues, where $n$ is a parameter of the theory, and $\rho$ is a Borel measure on $\F$.
\Chd{``regular'' gelöscht.}%
Its basic principle is the \emph{causal action principle}, which consists of minimizing the
action
\beq\label{ICausalAction}
\Sact(\rho) = \int_\F \int_\F \L(x,y)\: d\rho(x) \: d\rho(y)
\eeq
under variations of the measure $\rho$, taking into account additional constraints, for example $\rho(\F)$ to be constant.
Here, $\L(x,y) := \frac{1}{4n} \sum _{i,j=1}^{2n}
\big( | \lambda_i^{xy} | - | \lambda_j^{xy} | \big)^2$ is the Lagrangian of the theory, where $\lambda_i^{xy}$ denote the 
eigenvalues of the operator product $x y$ for $x,y \in \F$.

The goal of this thesis is to connect this mathematical structure to a description which is more similar to the currently used physical theories:
A formulation in terms of fields on space-time, where in the context of causal fermion systems, for physical reasons, space-time $M$ is defined as
\[ M := \supp \rho\: . \]

To achieve this goal of constructing an effective description of the theory in terms of fields on space-time, we consider variations of a measure $\rho$ which are infinitesimally described by a function $b: \F \rightarrow \R$ and  by a vector field
$v$ on $\F$. Combined, these two form a \emph{jet} $\v := (b,v)$. It turns out that the restriction of jets $\v$ to $M \subset \F$
yields a fruitful and satisfying notion of fields on space-time in this context.

Based on jets, we develop the so-called \emph{jet-formalism} of causal fermion systems. Starting from the Euler-Lagrange equations of the causal action principle, we derive
field equations for jets on space-time and show that those field-equations give rise to a symplectic form. This symplectic form is expressed in terms of 
so-called \emph{surface layer integrals} which generalize surface integrals to causal fermion systems. 
If space-time $M$ contains a suitable notion of Cauchy surfaces (and hence time), the surface layer integrals can be associated to Cauchy surfaces at different times.
We prove that this yields a symplectic form which is conserved with respect to the time-evolution induced by the field equations. We call this a \emph{Hamiltonian time evolution}.

Importantly, the causal action principle implies that there are correction terms to the field equations, which we subsequently derive and study. In particular, we prove that there is a stochastic and a non-linear correction term and investigate how they relate to the Hamiltonian time evolution. In this way, we establish the \emph{dynamics of causal fermion systems}.
 An important first step in our investigations is the construction of theorems which generalize Noether's theorems to the setting of causal fermion systems.\medskip

Throughout this thesis, we refer to two different settings. The first is the \emph{setting of causal fermion systems} which we have described above (second paragraph of the introduction).
A second setting is the \emph{general setting}, where we assume that $\F$ is a smooth, finite dimensional manifold, $\rho$ is a Borel measure on $\F$ and $\L: \F \times \F \rightarrow \R^+_0$ is a general function.
\Chd{``generalized'' $\rightarrow$ ``general'' hier und im Folgenden}%
\Chd{``regular'' gelöscht}%
The name of this setting comes about because it is more general than the setting of causal fermion systems (and also easier accessible).  Nevertheless, the physically interesting case is the setting of causal fermion systems.

When working in the general setting, we use different regularity assumptions for the Lagrangian $\L$, depending on what is necessary to make all constructions well-defined. To make this easily apparent, we introduce the names\\[-2em]
\begin{itemize}
\itemsep-.5em
\item[-] \emph{lower semi-continuous setting}, if we work in the general setting and $\L$ is assumed to be lower semi-continuous,
\item[-] \emph{Lipschitz-continuous setting}, if we work in the general setting and $\L$ is assumed to be Lipschitz-continuous, and
\item[-] \emph{smooth setting}, if we work in the general setting and $\L$ is assumed to be smooth.\\[-2em]
\end{itemize}
Furthermore, we define the
\\[-2em]
\begin{itemize}
\itemsep-.5em
\item[-] \emph{compact setting} as the Lipschitz-continuous setting with the additional assumption of $\F$ being compact.
\\[-2em]
\end{itemize}

For simplicity and clarity, in this introduction, we mostly outline our results and methods in the general setting.
\Chd{Folgendes weggelassen: ``In all parts of this thesis, we first establish the results in the general setting. The connection of the results to the setting of causal fermion systems is explained in a second step.''}%
As usual, we aim for a concise presentation and hence omit technical details, referring to later parts of this thesis whenever necessary.\medskip

This dissertation and its introduction are structured as follows.  In a first step towards finding a connection to fields on space-time, we prove generalizations of Noether's theorems to causal fermion systems. An outline of this work is given in Section~\ref{INoether}, all details can be found in Chapter~\ref{DissNoether}. 
Chapter~\ref{DissJet} contains the jet-formalism of causal fermion systems, the derivation of field equations and the construction of the above-mentioned symplectic form and of the Hamiltonian time evolution. The results are outlined in Section~\ref{JIntr}. 
Chapter~\ref{DissStoch} is concerned with the derivation of correction terms to the field equations as well as the investigation of its implications.
This is outlined in Section~\ref{Indfield}.
We conclude this introduction in Section~\ref{IModelQT} by explaining the connection of the results of this thesis to foundations of quantum theory. This connection has been one of the main driving forces behind these investigations.

Chapter~\ref{DissIntroCFS} is devoted to give an easily accessible introduction to the theory of causal fermion systems, focussing on the 
basic concepts and the general physical picture behind. The full mathematical setup of the setting of causal fermion systems is introduced in Section~\ref{Nseccfsbasic}.
\medskip

We remark that Chapter~\ref{DissIntroCFS} has been published as~\cite{dice2014}, Chapter~\ref{DissNoether} as~\cite{noether} and Chapter~\ref{DissJet} as~\cite{jet}. Even though these chapters relate to one another, they are self-contained, allowing the reader to leap forward to any one of them at any point. 
Being the most recent development, Chapter~\ref{DissStoch} has not yet been published. It is not self-contained, but a reading
of Sections~\ref{Jsecnoncompact},~\ref{JsecEL} as well as~\ref{JWeakElLSC} to~\ref{JCondS1} contains all necessary prerequisites.

\section{Noether-Like Theorems}\label{INoether}
In modern physics, the connection between symmetries and conservation laws
is of central importance. For continuous symmetries, this connection is made mathematically precise
by Noether's theorem~\cite{noetheroriginal}.

In a first part of this thesis, we explore symmetries and the resulting conservation laws
in the framework of causal fermion systems. The difficulty is that a priori, it is not clear at all how 
the definitions of symmetries in contemporary physics (cf. Section~\ref{Nsecclass}) can be generalized to causal fermion systems,
nor which form conservation laws could take. We now present the answers to these questions in the compact setting.

To formalize the notion of symmetries, we consider mappings~\eqref{NPhidef},
\[
\Phi : (-\tau_{\max}, \tau_{\max}) \times M \rightarrow \F
\qquad \text{with} \qquad \Phi(0,.) = \1 \:,
\] 
denoted as $\Phi_\tau(x)$, which are called {\em{variations}} of~$M$ in~$\F$.

It turns out that three definitions of symmetries can be considered.
A variation~$\Phi$ is a {\em{symmetry
of the Lagrangian}} (Definition~\ref{Ndefsymmlagr}) if 
\[
\L \big( x, \Phi_\tau(y) \big) = \L \big( \Phi_{-\tau}(x), y \big)
\qquad \text{for all~$\tau \in (-\tau_{\max}, \tau_{\max})$
and~$x, y \in M$\:.}
\]
It is a \emph{symmetry of the universal measure} (Definition~\ref{Ndefsymmrho}) if
\[
(\Phi_\tau)_* \rho = \rho
\qquad \text{for all~$\tau \in (-\tau_{\max}, \tau_{\max})$\:,}
\]
where $(\Phi_\tau)_* \rho$ is the push-forward measure,
and a \emph{generalized integrated symmetry} (Definition~\ref{Ndefgis}) if
\[
\int_M d\rho(x) \int_\Omega d\rho(y) \Big( \L\big( \Phi_\tau(x), y\big) - \L(x,y) \Big)  = 0 \:.
\]
Generalized integrated symmetries unite symmetries of the Lagrangian and symmetries of the universal measure, cf. Section~\ref{Nsecsymmgis}.
All of these notions of symmetries give rise to conservation laws (Theorems~\ref{Nthmsymmlag},~\ref{Nthmsymmum},~\ref{Nthmsymmgis}).\medskip
\begin{itemize}[leftmargin=5.3em]
\item[\textbf{Theorem.}]\em{
Let~$\rho$ be a measure which minimizes the causal variational principle and $\Phi_\tau$ a 
symmetry of the Lagrangian which is differentiable as in Definition~\ref{Ndefsymm}. Then for any compact subset~$\Omega \subset M$,
we have
\beq \label{IConserved}
\frac{d}{d\tau} \int_\Omega d\rho(x) \int_{M \setminus \Omega} d\rho(y)\:
\Big( \L \big( \Phi_\tau(x),y \big) - \L \big( \Phi_{-\tau}(x), y \big) \Big) \Big|_{\tau=0} = 0 \:.
\eeq}
\end{itemize}

\begin{itemize}[leftmargin=5.3em]
\item[\textbf{Theorem.}]\em{
Let~$\rho$ be a minimizing measure and $\Phi_\tau$ a symmetry of the universal measure which is differentiable as specified in Definition~\ref{Ndefsymm}.
For any compact subset~$\Omega \subset M$,
\beq \label{IConservedUniv}
\frac{d}{d\tau} \int_\Omega d\rho(x) \int_{M \setminus \Omega} d\rho(y)\: \Big( 
\L \big( \Phi_\tau(x), y \big)  - \L \big( x, \Phi_\tau(y) \big) \Big) \Big|_{\tau=0} = 0 \:.
\eeq
}
\end{itemize}

\begin{itemize}[leftmargin=5.3em]
\item[\textbf{Theorem.}]\em{
Let~$\rho$ be a minimizing measure and $\Phi_\tau$ a generalized integrated symmetry which is differentiable as specified in Definition~\ref{Ndefsymm}.
Then for any compact subset~$\Omega \subset M$,~\eqref{IConservedUniv} holds.
}
\end{itemize}

All of these conservation laws are examples of {\em{surface layer integrals}}, as introduced in Section~\ref{Nsecsli}.
The structure of surface layer integrals can be understood most easily 
in the special situation that the Lagrangian is of short range
in the sense that~$\L(x,y)$ vanishes unless~$x$ and~$y$ are close together.
In this situation, we only get a contribution to the double integrals~\eqref{IConserved} and~\eqref{IConservedUniv}
if both~$x$ and~$y$ are close to the boundary~$\partial \Omega$ of $\Omega$, as indicated by the
dark grey region in Figure~\ref{Ifignoether1}. Therefore, surface layer integrals can be understood as an adaptation
of surface integrals to the setting of causal variational principles
(for a more detailed explanation see Section~\ref{Nsecsli}). Figure~\ref{Ifignoether1} illustrates
the similarity between surface integrals (left) and surface layer integrals (right).

\begin{figure}\centering
\psscalebox{1.0 1.0} %
{
\begin{pspicture}(0,-1.511712)(10.629875,1.511712)
\definecolor{colour0}{rgb}{0.8,0.8,0.8}
\definecolor{colour1}{rgb}{0.6,0.6,0.6}
\pspolygon[linecolor=black, linewidth=0.002, fillstyle=solid,fillcolor=colour0](6.4146066,0.82162136)(6.739051,0.7238436)(6.98794,0.68384355)(7.312384,0.66162133)(7.54794,0.67939913)(7.912384,0.7593991)(8.299051,0.8705102)(8.676828,0.94162136)(9.010162,0.9549547)(9.312385,0.9371769)(9.690162,0.8571769)(10.036829,0.7371769)(10.365718,0.608288)(10.614607,0.42162135)(10.614607,-0.37837866)(6.4146066,-0.37837866)
\pspolygon[linecolor=black, linewidth=0.002, fillstyle=solid,fillcolor=colour1](6.4146066,1.2216214)(6.579051,1.1616213)(6.770162,1.1127324)(6.921273,1.0905102)(7.103495,1.0816213)(7.339051,1.0549546)(7.530162,1.0638436)(7.721273,1.0993991)(7.8857174,1.1393992)(8.10794,1.2060658)(8.299051,1.2549547)(8.512384,1.3038436)(8.694607,1.3260658)(8.890162,1.3305103)(9.081273,1.3393991)(9.379051,1.3216213)(9.659051,1.2593992)(9.9746065,1.1705103)(10.26794,1.0460658)(10.459051,0.94384354)(10.614607,0.82162136)(10.610162,0.028288014)(10.414606,0.1660658)(10.22794,0.26828802)(10.010162,0.37051025)(9.663495,0.47273245)(9.356829,0.53051025)(9.054606,0.548288)(8.814607,0.54384357)(8.58794,0.5171769)(8.387939,0.48162135)(8.22794,0.44162133)(7.90794,0.34828803)(7.6946063,0.29939914)(7.485718,0.26828802)(7.272384,0.26828802)(7.02794,0.28162134)(6.82794,0.3171769)(6.676829,0.35273245)(6.543495,0.38828802)(6.4146066,0.42162135)
\pspolygon[linecolor=black, linewidth=0.002, fillstyle=solid,fillcolor=colour0](0.014606438,0.82162136)(0.3390509,0.7238436)(0.5879398,0.68384355)(0.9123842,0.66162133)(1.1479398,0.67939913)(1.5123842,0.7593991)(1.8990508,0.8705102)(2.2768288,0.94162136)(2.610162,0.9549547)(2.9123843,0.9371769)(3.290162,0.8571769)(3.6368287,0.7371769)(3.9657176,0.608288)(4.2146063,0.42162135)(4.2146063,-0.37837866)(0.014606438,-0.37837866)
\psbezier[linecolor=black, linewidth=0.04](6.4057174,0.8260658)(7.6346064,0.45939913)(7.8634953,0.8349547)(8.636828,0.92828804)(9.410162,1.0216213)(10.165717,0.7927325)(10.614607,0.42162135)
\psbezier[linecolor=black, linewidth=0.04](0.005717549,0.8260658)(1.2346064,0.45939913)(1.4634954,0.8349547)(2.2368286,0.92828804)(3.0101619,1.0216213)(3.7657175,0.7927325)(4.2146063,0.42162135)
\rput[bl](2.0101619,0.050510235){$\Omega$}
\rput[bl](8.759051,0.0016213481){\normalsize{$\Omega$}}
\psline[linecolor=black, linewidth=0.04, arrowsize=0.09300000000000001cm 1.0,arrowlength=1.7,arrowinset=0.3]{->}(1.9434953,0.85495466)(1.8057176,1.6193991)
\rput[bl](2.0946064,1.1705103){$\nu$}
\psbezier[linecolor=black, linewidth=0.02](6.4146066,0.42384356)(7.6434956,0.057176903)(7.872384,0.43273246)(8.645718,0.52606577)(9.419051,0.61939913)(10.174606,0.39051023)(10.623495,0.019399125)
\psbezier[linecolor=black, linewidth=0.02](6.410162,1.2193991)(7.639051,0.8527325)(7.86794,1.228288)(8.6412735,1.3216213)(9.414606,1.4149547)(10.170162,1.1860658)(10.619051,0.8149547)
\rput[bl](8.499051,0.9993991){\normalsize{$y$}}
\rput[bl](7.8657174,0.49273247){\normalsize{$x$}}
\psdots[linecolor=black, dotsize=0.06](8.170162,0.65273243)
\psdots[linecolor=black, dotsize=0.06](8.796828,1.1327325)
\rput[bl](3.6146064,0.888288){$\scrN$}
\rput[bl](1.1146064,-1.4117119){$\displaystyle \int_\scrN \cdots\, d\mu_\scrN$}
\rput[bl](5.7146063,-1.511712){$\displaystyle \int_\Omega d\rho(x) \int_{M \setminus \Omega} d\rho(y)\: \cdots\:\L(x,y)$}
\end{pspicture}
}
\caption{A surface integral and a corresponding surface layer integral.}
\label{Ifignoether1}
\end{figure}
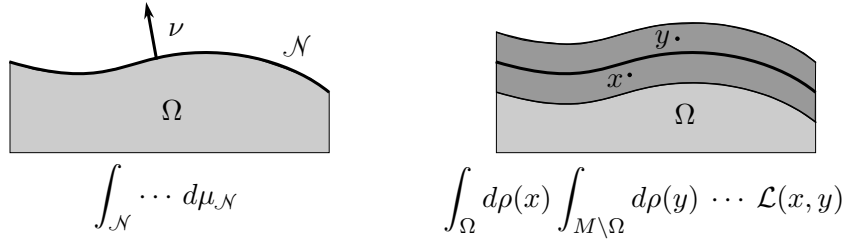

In Section~\ref{Nseccfs}, we give the definitions of the above symmetries in the setting of causal fermion systems and derive the corresponding conservation laws.
For brevity we do not repeat this here. Building on those generalizations, in Section~\ref{Nsecexcurrent}, we show that 
a transformation
\[
\Phi_\tau \,:\, \R \times \F \rightarrow \F\:,\qquad \Phi(\tau, x) = \scrU_\tau \,x\, \scrU_\tau^{-1} \:.
\]
with
\[
\scrU_\tau := \exp(i \tau \scrA) \: ,
\]
where~$\scrA$ is a bounded symmetric operator on~$\H$, yields a symmetry of the Lagrangian (cf. Lemma~\ref{NPhiSymmLag}).
We proof that in a limit of the theory of causal fermion systems, evaluation of the conservation law corresponding to~\eqref{IConserved} for this symmetry
yields current conservation for the Dirac equation (Theorem~\ref{Ncorcurrent}).
\begin{itemize}[leftmargin=5.3em]
\item[\textbf{Theorem.}]\em{
Let~$(\H, \F, \rho^\varepsilon)$ be local minimizers of the causal action which
describe the Minkowski vacuum~\eqref{NPvac}.
Considering the limiting procedure explained in Figure~\ref{Nfignoether2} and taking the
continuum limit, the conservation laws of Theorem~\ref{Nthmcurrent} go over to
a linear combination of the probability integrals in every generation.
}
\end{itemize}

Finally, in Section~\ref{NsecexEM}, we establish a connection to energy-momentum conservation. To this end, we give a suitable definition of Killing symmetries
(Definition~\ref{Ndefkilling}) and derive a corresponding conservation law (Theorem~\ref{NThmKillingCons}). Theorem~\ref{NthmEMcons} establishes that
the conservation law associated to Killing symmetries of causal fermion systems indeed yields energy-momentum conservation.
Thus the conservation laws of charge and energy-momentum can be viewed as
special cases of more general conservation laws which are intrinsic to causal fermion systems.

We remark that our conservation laws also apply to ``quantum space-times'' (\cite{lqg}) which cannot
be approximated by a Lorentzian manifold.  

\section{Hamiltonian Formulation and Linearized Field Equations}\label{JIntr}

In Chapter~\ref{DissJet}, we develop a formalism to define the dynamics of the theory of causal fermion systems,
where, as mentioned above, the term `dynamics' loosely refers to objects which propagate on space-time. The difficulty here
is that the theory of causal fermion systems, as introduced above, has a structure which is completely distinct from any formulation in terms of
dynamics. Hence, even which objects to consider is a priori unclear, let alone how to construct field equations and conserved quantities.
We now explain our constructions and results in the general setting.

Let~$\rho$ be a minimizer of the causal variational principle (for mathematical details
see Section~\ref{Jsecnoncompact} below).
The key step towards describing the causal variational principle in terms of a Hamiltonian time evolution
is to consider variations of $\rho$ described by
a diffeomorphism $F: \F \rightarrow \F$ and a weight function $f: \F \rightarrow \R_0^+$.
More precisely, we consider families $(F_\tau)_{\tau \in \R}$ and $(f_\tau)_{\tau \in \R}$
of such diffeomorphisms and weight functions and
form a corresponding family~$(\rho_\tau)_{\tau \in \R}$ of measures by
\begin{align}\label{JIntrRhoTau}
\rho_\tau = (F_\tau)_* \big( f_\tau \, \rho \big) \, .
\end{align}
Here $(F_\tau)_*\mu$ denotes the push-forward of the measure $\mu$.
Thus the measure $\rho_\tau$ is obtained from $\rho$ by first multiplying with the weight function $f_\tau$ and then ``transporting'' the resulting measure with the diffeomorphism $F_\tau$ on $\F$.

Infinitesimal versions of the variation~\eqref{JIntrRhoTau} consist of a scalar part corresponding to the
$\tau$-derivative of $f_\tau$ and a vectorial part corresponding to the $\tau$-derivative of $F_\tau$.
Thus variations of the form~\eqref{JIntrRhoTau} can be described infinitesimally by 
a pair~$(b,v)$ of a real-valued function and a vector field.
It turns out that in the context of causal fermion systems, the pairs~$(b,v)$, if restricted to $M$, are suitable candidates for 
physical fields on space-time. Put differently, pairs~$(b,v)$ can be viewed as generalized physical fields.
We remark at this point that the definition of {\em{space-time}}~$M$ as $M := \supp \rho \subset \F $
indeed generalizes the usual notion of space-time (being Minkowski space or a Lorentzian manifold)
in the setting of causal fermion systems (for details see~\cite[Section~1.2]{cfs} or~\cite[Sections~4 and~5]{lqg}).

In order to have a short name which cannot be confused with common notions in physics,
we refer to the pairs~$(b,v)$ as {\em{jets}}, being elements of the
corresponding {\em{jet space}}\footnote{The connection to jets in differential geometry
(see for example~\cite{saunders}) is obtained by considering real-valued functions on~$\F$.
Then their one-jets are elements in~$C^\infty(\F) \oplus \Gamma(\F, T^*\F)$. Identifying the
cotangent space with the tangent space gives our jet space~$\J$.}
\[ \J := \big\{ \v = (b,v) \text{ with } b : \F \rightarrow \R \text{ and } v \in \Gamma(\F) \} \]
(see~\eqref{JJdef} and~\eqref{JJDiffLip}).

We are interested in variations of the form~\eqref{JIntrRhoTau} which are minimizers of the causal action also
for~$\tau \neq 0$. Such ``families of minimizers'' are of interest because 
in the theory of causal fermion systems, they
correspond to variations which satisfy the physical equations.
The requirement of~$\rho_\tau$ being a minimizer for all~$\tau$ gives rise to
conditions for the jet~$\v = \partial_\tau \rho_\tau|_{\tau=0} \in \J$ describing the infinitesimal
variation. In view of the similarities and the correspondence to classical field theory
(as worked out in~\cite[\S1.4.1 and Chapters~3-5]{cfs} and~\cite{perturb}),
we refer to these conditions as the \textit{linearized field equations}.
In order to specify the linearized field equations, we introduce the function
\beq\label{Iell}
\ell(x) = \int_\F \L(x,y)\: d\rho(y) - \frac{\nu}{2} 
\eeq
which is used in the Euler-Lagrange equations of causal variational principles (Section~\ref{JsecEL}).
The parameter~$\nu \geq 0$ is the Lagrange multiplier corresponding to the volume constraint.
We define the differentiable one-jets $\Jdiff \subset \J$ to consists of all those jets $\u = (a,u)$ for which $\ell$ is differentiable in the direction of $u$ on $M$,
\[
\Jdiff := \{ \u= (a,u) \in \J \, \big| \, D^+_{u} \ell|_M = D^-_{u} \ell|_M \} \: ,
\]
where $D^+_u$ denotes the right directional semi-derivative~\eqref{JDvL}
and $D^-_u$ denotes the left directional semi-derivative~\eqref{JLeftSemiDeriv}.
Also, we define the {\em{derivative}}~$\nabla_{\v}$ \emph{in the direction of a one-jet}~$\v = (b, v)$ as
as a combination of multiplication and differentiation,
\[ \nabla_{\v} \eta(x)= b(x) \,\eta (x) + D_v \eta(x) \]
(where~$D_v$ is the usual directional derivative of functions on~$\F)$, and similarly for semi-derivatives $ \nabla^+_{\v}$ and $ \nabla^-_{\v}$
(cf.~\eqref{JDjetSemi}). The notation~$\nabla_{1, \v}$ and~$\nabla_{2, \v}$ denotes partial derivatives acting on the
first and second argument of~$\L(x,y)$, respectively.

For differentiable one-jets, the Euler-Lagrange equations of the causal variational principle imply the so-called
\emph{weak Euler-Lagrange equations},
\beq\label{IWElDiff}
\nabla_{\u} \ell(x) = 0 \qquad \text{for all~$x \in M$ and~$\u \in \Jdiff$}
\eeq
(cf.~\ref{JWElDiff}), which are the starting point of the derivation of the linearized field equations.

Using the above definitions, the \emph{linearized field equations} can be written as
\begin{align}\label{JIntreqLinFieldEq}
\nabla_{\u} \bigg( \int_M \big( \nabla_{1, \v} + \nabla_{2, \v} \big) \L(x,y)\: d\rho(y) - \nabla_\v \:\frac{\nu}{2} \bigg)
= 0
\end{align}
for all~$\u \in \Jtest$ and~$x\in M$. (For details see Lemma~\ref{Jlemmalin} and~\eqref{Jeqlinlip}.)

Here, the {\em{test jets}}~$\, \Jtest \subset \Jdiff$ are defined as a subspace of the differentiable one-jets used to test the 
requirement of minimality in a weak sense. This so-called {\em{weak evaluation of the Euler-Lagrange
equations}} is an important mathematical and physical concept because
by choosing~$\Jtest$ appropriately, one can restrict attention to the part of the information contained
in the Euler-Lagrange equations which is relevant for the application in mind.
In order to illustrate how this works, we give a typical example:
For the description of macroscopic physics, one would like to
disregard effects which come into play only on the Planck scale. To this end,
one chooses~$\Jtest$ as a space of jets which vary only on the macroscopic scale,
so that ``fluctuations on the Planck scale are filtered out''.

The above results hold in the general setting. Next, we restrict to the smooth setting, i.e. we assume that $\L$ is smooth, cf. Section~\ref{JSecSmooth}. The assumptions in this section ensure that $\Jdiff = \J$. We consider the set $\calB$ of all measures of the form~\eqref{JIntrRhoTau} which satisfy the weak Euler-Lagrange 
equations~\eqref{IWElDiff}
and assume that this set is a smooth Fr{\'e}chet manifold. 
Then infinitesimal variations of~\eqref{JIntrRhoTau} which are solutions of~\eqref{JIntreqLinFieldEq} are 
vectors of the tangent space~$T_\rho \calB$.
In this setting, we show that the structure of the causal variational principle
gives rise to a {\em{symplectic form}} on $\calB$ (Section~\ref{JSecSympForm}).
Namely, for any~$\u, \v \in T_\rho \calB$ and~$x,y \in M$, let
\[ \sigma_{\u, \v}(x,y) := \nabla_{1,\u} \nabla_{2,\v} \L(x,y) - \nabla_{1,\v} \nabla_{2,\u} \L(x,y) \:. \]
Given a compact subset~$\Omega \subset \F$, we define the bilinear form
\beq \label{JIntrOSI}
\sigma_\Omega \::\: T_\rho \calB \times T_\rho \calB \rightarrow \R\:,
\qquad \sigma_\Omega(\u, \v) = \int_\Omega d\rho(x) \int_{M \setminus \Omega} d\rho(y)\:
\sigma_{\u, \v}(x,y) \:.
\eeq
Thus here we again make use of the structure of a surface layer integral.
The following theorem holds (Theorem~\ref{JthmOSI}).
\begin{itemize}[leftmargin=5.3em]
\item[\textbf{Theorem.}] \em{
For any compact subset~$\Omega \subset \F$, the surface layer integral~\eqref{JIntrOSI}
vanishes for all~$\u, \v \in T_\rho \calB$.}
\end{itemize}

This theorem has the following connection to conservation laws.
Let us assume that~$M$ admits a sensible notion of ``spatial hypersurfaces'' and
that the jets~$\u, \v \in T_\rho \calB$ have suitable decay properties at spatial infinity.
Then, given two such hypersurfaces $N_1$ and $N_2$, one can chose a sequence~$\Omega_n \subset M$ of compact sets
which form an exhaustion of the space-time strip between $N_1$ and $N_2$
(see Figure~\ref{Jfigjet1}~(a) and~(b)). Let us denote this space-time strip by $\Omega$.
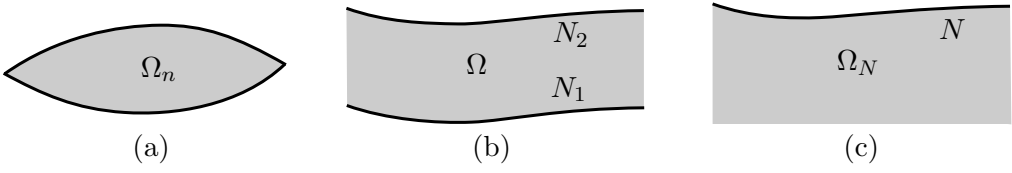
\begin{figure}\centering
{
\begin{pspicture}(0,-1.0682992)(16.27295,1.0682992)
\definecolor{colour0}{rgb}{0.8,0.8,0.8}
\pspolygon[linecolor=colour0, linewidth=0.02, fillstyle=solid,fillcolor=colour0](9.348506,-0.5327436)(9.348506,1.0361453)(9.63295,0.92947865)(9.926284,0.8805898)(10.339617,0.84058976)(10.837395,0.85836756)(11.326283,0.9072564)(11.668506,0.93836755)(11.988505,0.96058977)(12.184061,0.94281197)(12.801839,0.9961453)(13.241839,1.0050342)(13.259617,-0.537188)(12.948505,-0.50163245)(12.539617,-0.50163245)(12.11295,-0.48385468)(11.748506,-0.497188)(11.459617,-0.47496578)(11.246284,-0.457188)(10.9529505,-0.46163246)(10.717395,-0.4349658)(10.406283,-0.42163244)(9.997395,-0.39941022)(13.259617,-0.5327436)
\pspolygon[linecolor=colour0, linewidth=0.02, fillstyle=solid,fillcolor=colour0](4.5396166,-0.27941024)(4.530728,0.9872564)(4.815172,0.8805898)(5.1085057,0.83170086)(5.530728,0.7961453)(5.9573946,0.7872564)(6.4596167,0.8139231)(6.819617,0.84947866)(7.0773945,0.8761453)(7.3662834,0.8939231)(7.9840612,0.94725645)(8.424061,0.95614535)(8.424061,-0.31052133)(8.03295,-0.33274356)(7.624061,-0.3460769)(7.3040614,-0.38163245)(7.0373945,-0.40385467)(6.779617,-0.4438547)(6.468506,-0.47496578)(6.1885056,-0.50163245)(5.8773947,-0.51052135)(5.55295,-0.497188)(5.179617,-0.4482991)(4.8462834,-0.377188)
\pspolygon[linecolor=colour0, linewidth=0.02, fillstyle=solid,fillcolor=colour0](0.03295012,0.12947866)(0.19739456,0.24503422)(0.35739458,0.3250342)(0.57517236,0.43614534)(0.9485057,0.57836753)(1.3885057,0.6939231)(1.8151723,0.74725646)(2.121839,0.7605898)(2.43295,0.742812)(2.7973945,0.6539231)(3.121839,0.542812)(3.441839,0.3917009)(3.681839,0.25836754)(3.5973945,0.13836755)(3.39295,0.013923102)(3.130728,-0.123854674)(2.8551724,-0.24385467)(2.4773946,-0.31941023)(2.161839,-0.36829913)(1.8151723,-0.38163245)(1.5173945,-0.37274358)(1.081839,-0.30607688)(0.62850565,-0.16385467)(0.281839,-0.0038546752)
\rput[bl](1.8151723,0.025034213){$\Omega_n$}
\rput[bl](6.0973945,0.12281199){\normalsize{$\Omega$}}
\psbezier[linecolor=black, linewidth=0.04](0.02406123,0.12281199)(0.7173699,0.6167855)(1.5001423,0.76261884)(2.0707278,0.7672564358181418)(2.6413136,0.77189404)(3.0712237,0.6218133)(3.7240613,0.242812)
\psbezier[linecolor=black, linewidth=0.04](0.009616784,0.12947866)(0.52181435,-0.15099226)(0.95347565,-0.34627005)(1.5962834,-0.38941023084852644)(2.2390912,-0.4325504)(3.124557,-0.3104089)(3.709617,0.24947865)
\psbezier[linecolor=black, linewidth=0.04](4.525172,-0.297188)(5.046259,-0.4843256)(5.7001424,-0.5373812)(6.1696167,-0.518299119737415)(6.6390915,-0.49921706)(7.315668,-0.33707556)(8.438506,-0.3282991)
\psbezier[linecolor=black, linewidth=0.04](4.520728,0.9917009)(5.0418143,0.8045633)(5.5401425,0.7915077)(5.991839,0.783923102484807)(6.443536,0.7763385)(7.3112235,0.95181334)(8.434061,0.96058977)
\psbezier[linecolor=black, linewidth=0.04](9.338506,1.0494787)(9.859592,0.8623411)(10.375698,0.85372996)(10.751839,0.8683675469292492)(11.12798,0.88300514)(12.129002,1.0095911)(13.251839,1.0183675)
\rput[bl](10.97295,0.10947866){\normalsize{$\Omega_N$}}
\rput[bl](1.7040613,-1.0682992){(a)}
\rput[bl](6.1751723,-1.0638547){(b)}
\rput[bl](11.059617,-1.0638547){(c)}
\rput[bl](7.2085056,-0.23718801){\normalsize{$N_1$}}
\rput[bl](7.2351723,0.482812){\normalsize{$N_2$}}
\rput[bl](12.319616,0.5850342){\normalsize{$N$}}
\end{pspicture}
}
\caption{Choices of space-time regions.}
\label{Jfigjet1}
\end{figure}%
Considering the surface layer integrals~\eqref{JIntrOSI} for~$\Omega_n$ and passing to the limit,
the above theorem implies that also the surface layer integral corresponding to~$\Omega$ vanishes,
so that the contributions from $N_1$ and $N_2$ coincide. (This is made precise in Section~\ref{JSecSympForm}.)
Let us define a surface layer integral corresponding to any spatial hypersurface $N$ as
\begin{align}\label{JIntrOSIN}
 \sigma_{N}(\u, \v) := \int_{\Omega_N} d\rho(x) \int_{M \setminus {\Omega_N}} d\rho(y) \:\sigma_{\u, \v}(x,y)
\end{align}
where $\Omega_N$ is a set with~$\partial \Omega_N = N$ as shown in Figure~\ref{Jfigjet1}~(c).
Then the above theorem implies that this quantity is well-defined and that
$ \sigma_{N}(\u, \v)  =  \sigma_{{N'}}(\u, \v)$
for any two such hypersurfaces $N$ and $N'$. In other words, this surface layer integral it is independent of
the particular choice of hypersurface $N$.

If the hypersurfaces~$N$ are similar to a foliation of Cauchy surfaces $N_t$ and if the parameter $t$ of the foliation can be interpreted as time, the last theorem implies that the bilinear form~$\sigma_{{N_t}}$
is preserved under the time evolution. This is what we call \textit{Hamiltonian time evolution}.

To avoid misunderstandings, we point out that, in contrast to classical field theory, in our setting
the time evolution is {\em{not defined infinitesimally}} by a Hamiltonian or a Hamiltonian vector field
(this is obvious from the fact that causal fermion systems allow for the description of discrete space-times,
where a continuous time evolution makes no sense).
Instead, the time evolution should be thought of as a mapping from the jets
in a surface layer around~$N_1$ to the jets in a surface layer around~$N_2$.
This mapping is a {\em{symplectomorphism}} with regard
to~$\sigma_{{N_1}}$ and~$\sigma_{{N_2}}$, respectively.
For clarity, we also note that it is essential here that~$\F$ and~$M$ are non-compact because otherwise, Theorem~\ref{JthmOSI} would immediately imply
that~$\sigma_{_N}(\u, \v) \equiv 0$. For non-compact $M$, however,
Theorem~\ref{JthmOSI} only applies to the
difference~$\sigma_{{N_1}}\! - \sigma_{{N_2}}$, thereby giving a conservation law
for a non-trivial surface layer integral.

The independence of~\eqref{JIntrOSIN} from $N$ allows us to define a bilinear form $\sigma$ on $\calB$ by
\[ \sigma \::\: T_\rho \calB \times T_\rho \calB \rightarrow \R\:,
\qquad \sigma(\u, \v) := \sigma_{N}(\u, \v) \]
for an arbitrary choice of hypersurface~$N$. This bilinear form turns out to be closed
(see Lemma~\ref{Jlemmaclosed}), thus defining a {\em{presymplectic form}} on $\calB$. 
Finally, by restricting~$\sigma$ to a suitable subspace of~$T_\rho \calB$,
we arrange that~$\sigma$ is non-degenerate, giving a {\em{symplectic form}} (see end of Section~\ref{JSecSympForm}).

In the theory of causal fermion systems, the Lagrangian is not a smooth, but
merely a Lipschitz-continuous function (see Section~\ref{Nseccfsbasic},~\cite[Section~1.1]{cfs} and~\cite{support}).
In order to cover this situation, in Section~\ref{JlowerSemiCont} we treat the more general 
{\em{lower semi-continuous setting}} where we have a
lower semi-continuous Lagrangian. Our main motivation for this generalization compared to Lipschitz-continuity is
that many simple examples are easier to state if we allow for discontinuities of the Lagrangian.
After defining jet spaces as infinitesimal versions of families of solutions similar as described above
(see Sections~\ref{JWeakElLSC} and~\ref{JSmoothVar}) and establishing the necessary conditions for the linearized field equations to be well-defined (see Definition~\ref{Jdeflin}),
we again establish a conservation law for the bilinear
form $\sigma_\Omega(\u,\v)$ (see Theorem~\ref{JthmOSIlip}). 

Our results apply directly to the setting of causal fermion systems if one restricts the set of operators $\F$ to operators of maximal rank, i.e.
with $n$ positive and $n$ negative eigenvalues. This is explained in detail in Section~\ref{Jseccfs}.

In Section~\ref{Jseclattice}, we illustrate our constructions in an example
which is simple enough for an explicit analysis but nevertheless
captures some features of a physical field theory.
We choose $\L(x,y)$  in such a way that the minimizing measure is
supported on a two-dimensional lattice. This reflects a general feature
of the theory of causal fermion systems that space-time ``discretizes itself''
on the Planck scale, thus avoiding the ultraviolet divergences of quantum field theory
(see~\cite[Section~4]{rrev}).

\section{Stochastic and Non-Linear Correction Terms}\label{Indfield}

Chapter~\ref{DissStoch} is devoted to finding correction terms to the linearized field equations.
The \emph{motivation} to look for correction terms in the theory of causal fermion systems is the following. The linearized field equations~\eqref{JIntreqLinFieldEq}
are a consequence of the weak Euler-Lagrange equations~\eqref{IWElDiff}. But since the Lagrangian $\L$ in the theory of causal fermion systems
(and hence the function $\ell$ as defined in~\eqref{Iell}) are not differentiable but only Lipschitz-continuous, instead of~\eqref{IWElDiff}, we only have
\beq \label{IELweaklip}
\nabla^+_{\u} \ell(x) \geq 0 \qquad \text{for all~$x \in M$ and~$\u \in \J|_M$}\:,
\eeq
where the semi-derivative in the direction of a jet $\u = (a,u) \in \J$ is defined as (cf.~\eqref{JELweaklip})
\begin{align}\label{IDjetSemi}
\nabla^+_{\u} \ell(x) := a(x)\, \ell(x) + \big(D^+_u \ell \big)(x) \:.
\end{align}
Due to the inequality in~\eqref{IELweaklip}, one can expect additional terms to appear in~\eqref{JIntreqLinFieldEq}.

Whereas the this motivation to look for correction terms is simple, it turns out that many ways to derive field \emph{equations} which include a correction term turn out
unsuccessful. The successful way turns out to be deeply rooted in the structure of the theory of causal fermion systems. Put in very simple terms,
one uses the definition of $\ell$ as in~\eqref{Iell}, evaluated for a family of measures as~\eqref{JIntrRhoTau},
\[
\ell_\tau(x) := \int_\F \L(x,y)\: d\rho_\tau(y) - \frac{\nu}{2} \: .
\]
Composing this definition
with the flow $\Phi$ of the vectorial component $w$ of jet $\w \in \J$ in a suitable way, and taking semi-derivatives with respect to the parameters of the
family~\eqref{JIntrRhoTau} and of the flow $\Phi$ gives the~\emph{non-differentiable linearized field equations} (Proposition~\ref{SmodFieldEqA}),
\beq \label{INonDiffFEq}
\frac{1}{2} \: \big( \nabla^+_{\w} - \nabla^+_{- \w} \big)  \: \int_M d\rho(y) \: 
\big( \nabla^+_{1,\v} + \nabla^+_{2,\v} \big) \: \L \big( x, y \big)  - \nabla^+_\w \nabla^+_\v \: \frac{\nu}{2}\: = \chi_{\w,\v}(x)
\eeq
for \emph{all} $\w \in \J$ and $x \in M$, where the term $\chi_{\w,\v}(x)$ is defined as
\beq\label{IStochTerm}
\chi_{\w,\v}(x) =   \frac{1}{2} \: \frac{d}{d s}^{\!+}_{|_0}  \frac{d}{d \tau}^{\!+}_{|_0}  \: 
f_\tau(x) \: \Big( \ell_\tau \big (F_\tau(\Phi_s (x)) \big)   -  \ell_\tau \big( F_\tau(\Phi_{-s}(x))\Big) \: ,
\eeq
and where the subscripts indicate that the semi-derivatives are evaluated at $s=\tau=0$.
Thus $\chi_{\w,\v}(x)$ is the desired correction term!

The explanations in Section~\ref{SNonDiffFieldEq} and Example~\ref{SExStoch}
show that $\chi_{\w,\v}(x)$ is a term which has a varying sign as $x$ changes. Since this sign
depends on the microscopic structure of space-time $M$ and the point-wise behaviour of $\ell$,
$\chi_{\w,\v}(x)$ can be interpreted as a \emph{stochastic term}.
Note that if $\int_M d\rho(y) \big( \nabla_{1,\v} + \nabla_{2,\v} \big) \L \big( x, y \big)$ exists and is differentiable, the left hand side of~\eqref{INonDiffFEq} is equal
to the left hand side of the linearized field equations~\eqref{JIntreqLinFieldEq}. Hence it indeed constitutes a generalization thereof.

The appearance of a correction term brings up the question of why a corresponding correction term is not apparent in contemporary experiments.
We address this question in Section~\ref{StochTermVan} by defining the following physically motivated assumption (Definition~\ref{SmacSymm0}).
\begin{itemize}[leftmargin=5.9em]
\item[\textbf{Definition.}]\em{
`Symmetric derivatives vanish macroscopically' 
if for every minimizer~$\rho$ of the causal variational principle, every vector field $w \in \Gamma(T\F)$ and
every macroscopic region $\tilde \Omega \in \mathscr M$,
\beq \label{ISymmMVanish}
\frac{1}{2} \,  \int_{\tilde \Omega} \big( D^+_w  - D^+_{-w} \big) \, \ell(x)  \, d\rho(x) = 0 \:. 
\eeq
}
\end{itemize}
Here, we assume that the specification of which subsets of $M$ are macroscopic is part of the data of any application of the theory. (But cf.
Remark~\ref{SHamiltInt}.) Formally, we assume that this data is given as a set $\mathscr M$ of subsets of $\F$, cf. Definition~\ref{SDefMac}.
The assumption of symmetric derivatives to vanish macroscopically is reasonable because the integrand has a varying sign, and hence the different contributions at different space-time points, if summed over a macroscopic region, might cancel. 

Proposition~\ref{SstochTermVanishes} shows that if symmetric derivatives vanish macroscopically, the stochastic term vanishes macroscopically as well,
\beq\label{IVanishM}
\int_{\tilde \Omega} \chi_{\w,\v}(x) \, d\rho(x) = 0 \: .
\eeq
(The same result would hold if the right hand side of~\eqref{ISymmMVanish} and~\eqref{IVanishM} were replaced by~$\varepsilon$.)
Hence we can justify the above assumption a posteriori as an explanation of why correction terms to field equations are currently not observed.
We point out, however, that the assumption could very well be wrong -- investigations on minimizers of the causal action principle could prove
that~\eqref{ISymmMVanish} does not vanish, cf. Chapter~\ref{DissOutlook}. In this way, the theory of causal fermion systems could allow to deduce what
counts as a `macroscopic region' from first principles.

After studying the connection to the differentiable case of Chapter~\ref{DissJet} in detail in Section~\ref{SConDiff}, we investigate the symplectic form~\eqref{JIntrOSIN}
in the present setting. In contrast to the theorem mentioned in Section~\ref{JIntr}, it does not vanish, but additional terms appear, as shown by Theorem~\ref{SthmOSIlipNonDiff}.
\begin{itemize}[leftmargin=5.3em]
\item[\textbf{Theorem.}]\em{
Let $\w$ and $\v$ be solutions of the non-differentiable linearized field equations~\eqref{INonDiffFEq}.
Then for any compact~$\Omega \in \Sigma(\F)$, the symplectic form~\eqref{JIntrOSI}
satisfies
\[
\sigma_{ \Omega}(\w, \v) = \int_{ \Omega}  \widetilde \chi_{\w,\v}(x) \: d\rho(x)  - \int_{ \Omega} \wnabla_{[\w,\v]} \, \ell(x) \: d\rho(x) \: ,
\]
where
\[
\widetilde \chi_{\w,\v}(x) = \frac{1}{2} \big( \chi_{\w,\v}(x) -  \chi_{\w,-\v}(x) - \chi_{\v,\w}(x) + \chi_{\v,-\w}(x) \big) 
\]
and
\[
\wnabla_{[\w,\v]} = \frac{1}{2} \, \big( \, \nabla^+_{[\w,\v]} - \nabla^+_{-[\w,\v]} \, \big) \: .
\]
}
\end{itemize}
This theorem says that the Hamiltonian time evolution is broken on the microscopic level
by the stochastic term and by the non-differentiability of $\ell$.
However, if symmetric derivatives vanish macroscopically, the Hamiltonian time evolution is conserved macroscopically, as shown by Proposition~\ref{SthmOSIlipNonDiffMac} (cf. Remark~\ref{SHamiltInt}).
\begin{itemize}[leftmargin=5.3em]
\item[\textbf{Theorem.}]\em{
If symmetric derivatives vanish macroscopically, for any solutions $\w$ and $\v$ of the non-differentiable linearized field equations and any compact macroscopic region $\tilde \Omega$,
\[
\sigma_{\tilde \Omega}(\w, \v) = 0 \: .
\]
}
\end{itemize}

Finally, in Section~\ref{SO2}, we derive quadratic corrections to the linearized field equations. Using the notations
\[
\nabla^{+}_{\v,\v} := \ddot f_0 + 2 \dot f_0 D^+_v + D^+_v D^+_v \: 
\]
for any family~\eqref{JIntrRhoTau} with generator $\v = (\dot f_0 , v) \in \J$ (where the dot indicates a $\tau$-derivative),
as well as
\[
\wnabla_\w := \frac{1}{2} \big( \, \nabla^+_{\w} - \nabla^+_{-\w} \, \big) \:,
\]
we can thus specify the main theorem of Chapter~\ref{DissStoch}:
\begin{itemize}[leftmargin=5.3em]
\item[\textbf{Theorem.}] \textbf{(Full non-differentiable field equations to second order)}\\
\em{
For every family~\eqref{JIntrRhoTau} of minimizers with generator $\v$, any $x \in M$ and any $\w \in \J$, we have
\begin{align}\begin{split}\label{IFull}
&\wnabla_\w  \int_M d\rho(y) \: \big( \nabla^+_{1,\v} + \nabla^+_{2,\v} \big) \: \L \big( x, y \big)  - \wnabla_\w \nabla^+_\v \: \frac{\nu}{2}\: \\
& + \wnabla_\w \int_M d\rho(y) \: \Big( \nabla^+_{1,\v,\v} + \nabla^+_{2,\v,\v} +  2 \, \nabla^+_{1,\v} \nabla^+_{2,\v}  \Big) \: 
\L(x,y) - \wnabla_\w \:  \nabla^+_{\v,\v} \: \frac{\nu}{2} \\
& \quad = \chi_{\w,\v}(x) + \chi^{(2)}_{\w,\v}(x) \: ,
\end{split}\end{align}
where~$\chi_{\w,\v}(x)$ is as in~\eqref{IStochTerm} and
\[
\chi^{(2)}_{\w,\v}(x) =   \frac{1}{2} \: \frac{d}{d s}^{\!+}_{|_0}  \frac{d^{\, 2}}{d \tau^2}^{\!+}_{|_0}  \: 
f_\tau(x) \: \Big( \ell_\tau \big (F_\tau(\Phi_s (x)) \big)   -  \ell_\tau \big( F_\tau(\Phi_{-s}(x))\Big) \: .
\]
}
\end{itemize}

We conclude Chapter~\ref{DissStoch} in Section~\ref{SNoether} by adapting the conservation law~\eqref{IConserved} to the regularity assumptions in Chapter~\ref{DissStoch} 
(Proposition~\ref{SNoetherNonDiff}). It turns out that the conservation law is broken by the appearance of a term with varying sign on the right hand side of~\eqref{IConserved}.
However, if symmetric derivatives vanish macroscopically, the additional term vanishes (Proposition~\ref{SConservLawMacroscopic}), showing that under this assumption,~\eqref{IConserved} (and hence also current conservation
and conservation of energy-momentum) hold macroscopically  but not microscopically. This result can be interpreted as giving further support to the assumption
of symmetric derivatives to vanish macroscopically.

The results of Chapter~\ref{DissStoch} can be applied to the setting causal fermion systems as explained in Section~\ref{Jseccfs}.

We expect that the correction terms in the full non-differentiable field equations to second order yield modifications of the field equations in the continuum limit, thus possibly opening the doors to experimental predictions.

\section{Connection to Foundations of Quantum Theory}\label{IModelQT}

As mentioned above, in the so-called continuum limit~\cite{cfs,qft}, the Euler-Lagrange equations of the causal action principle give rise to the fundamental equations of quantum theory, general relativity and quantum field theory.
The analysis of the jet-formalism in the continuum limit is still open (cf.~Chapter~\ref{DissOutlook}). However, we have strong reasons for the following
conjecture.
\begin{itemize}[leftmargin=5.3em]
\item[\em{Conjecture.}]\em{
In the continuum limit, the linearized field equations~\eqref{JIntreqLinFieldEq} yield the fundamental equations of contemporary physics, in particular the Dirac equation.
}
\end{itemize}

In this section we evaluate the consequences of this conjecture with respect to quantum theory.
Since the Dirac equation reduces to the Pauli equation or the Schr\"odinger equation in the non-relativistic limit (cf.~\cite{bjorken, peskin+schroeder}),
the conjecture thus implies that the full non-differentiable field equations to second order~\eqref{IFull}
give rise to correction terms for the Dirac, Pauli and Schrödinger equation.

It is well-known in the foundations of quantum theory that the linearity of the Schr\"odinger equation
conflicts with the von Neumann collapse postulate if one assumes that measurement apparati are composed of objects (``atoms'')
which themselves follow the laws of quantum theory. 
This problem, referred to as the \textit{measurement problem} (cf.~\cite{peres}), can be remedied in several ways, leading to modifications of the original quantum mechanics as formulated by von Neumann~\cite{vneumann}. Since the theory of causal fermion systems is a candidate for a unified physical theory, this raises the question
of what its implications are on the measurement problem.

Equation~\eqref{IFull} shows that to second order, the corrections for the field equations from the theory of causal fermion systems consist of a stochastic term and a quadratic term.
Thus, based on the above conjecture we expect that if one evaluates~\eqref{IFull} in the continuum limit, the Dirac equation (and hence also the Pauli and Schrödinger equation) arise equipped with an additional stochastic and an additional quadratic term.

This is similar to modifications of the Schrödinger equation which are referred to as spontaneous localization or dynamical collapse models
(see~\cite{pearle0, ghirardi, bassi,TumulkaQFTCollapse} or~\cite[Chapter~8]{joos} for a small sample of the vast literature on this topic). In those models,
one adds a specific quadratic and a specific stochastic term to the Schrödinger equation in order to break the linearity in a way which is compatible with the experimental observations to date.
Thus, according to the above conjecture, the theory of causal fermion systems seems to be an effective dynamical collapse theory, where the term ``effective'' points to the fact that the theory of causal fermion systems only appears to be such a model in the continuum \chd{limit. The} fundamental dynamics is yet different from the modified Schrödinger/Pauli/Dirac equation.

Clearly, it remains open at this point which form the correction terms in~\eqref{IFull} take if evaluated in the continuum limit, and how exactly they compare with the above-mentioned models. But this analysis is to be carried out in the near future. An answer will insofar be interesting as the theory of causal fermion systems and the jet-formalism are completely covariant, i.e. they do not depend on any choice of coordinates on $\F$ (cf. also Section~\ref{CsecPrinciples}). Hence, if one avoids to break general covariance when taking the limit,
the correction terms for the Dirac equation necessarily are covariant as well.

Concerning the conservation laws of Chapter~\ref{DissNoether} and the generalization to the non-differentiable setting in Section~\ref{SNoether}, we remark the following.
Suppose that space-time $M$ is a globally hyperbolic manifold so that we can sensibly talk about ``times''.
Assume that the wave function undergoes a collapse at some time~$t_c$.
It is a reasonable assumption that the continuum limit should still be a good
description at some earlier time~$t_0<t_c$ and some later time~$t_1>t_c$.
In this situation the conservation law of Theorem~\ref{SNoetherNonDiff}
states that the current integrals at times~$t_0$ and~$t_1$ do not coincide if~$t_0$ and~$t_1$
are close to each other. The sign-varying contributions from the symmetric directional derivatives
$\frac{1}{2} \big( D^+_w - D^+_{-} \big) \ell(x)$, integrated over the time-strip from~$t_0$ to~$t_1$
destroy current conservation as well as conservation of energy-momentum.

But, if Definition~\ref{SmacSymm0} (`symmetric derivatives vanish macroscopically') holds,
Proposition~\ref{SConservLawMacroscopic} implies that if the time-strip from~$t_0$ to~$t_1$
is a macroscopic region, the collapse mechanism necessarily preserves the normalization of the wave function.
Thus, in contrast to some continuous dynamical localization models, given Assumption~\ref{SmacSymm0},
in our approach it does not seem to be necessary to rescale the wave function so as to arrange its proper normalization,
and similarly for the conservation of energy-momentum in the collapse process. \medskip

Clearly, even so our results suggest a particular resolution of the measurement problem based on the theory
of causal fermion systems, there remain many open questions about the relation of the theory to foundations of quantum mechanics.
(A summary of how entanglement and non-locality arise in this context is given in Section~\ref{CFoundations}.)
For an overview over the most relevant problems in this context, see the Chapter~\ref{DissOutlook}.
We plan to carry out further investigations of those and related matters in the near future.

\chapter[Causal Fermion Systems as a Candidate for a Unified Physical Theory]{Causal Fermion Systems as a Candidate for a Unified Physical Theory}\label{DissIntroCFS}

This chapter represents an introduction to causal fermion systems which is intended to explain
the basic concepts and the general physical picture behind the theory in an easily accessible way. 
In order to achieve this goal, we avoid technical details
if they are not crucial for understanding and explain even simple notions
which are otherwise assumed to be known by the reader of this thesis.
A thorough introduction to the mathematical setup of causal fermion systems which is relevant for Chapters~\ref{DissNoether},~\ref{DissJet} and~\ref{DissStoch}
can be found in Section~\ref{Nseccfsbasic} and in~\cite[Chapter 1]{cfs}.

This chapter is organized as follows. In Section~\ref{CTheory} we define the basic objects of the theory.
In Section~\ref{CIntro-Minkvac} we proceed
by explaining how those objects appear naturally in the familiar physical situation of Dirac particles
in Minkowski space. In Section~\ref{Csecinherent} it is shown how the objects of quantum
mechanics are encoded in a causal fermion system.
Section~\ref{Csecminkvac} explains for the example of the Minkowski vacuum how a
causal fermion system encodes causal structure.
In Section~\ref{Csecgenst} we exemplify how to describe other physical situations or more general space-times.
In Section~\ref{CContinuum} we outline a limiting case in which the causal fermion system
can be described by a second-quantized Dirac field coupled to classical gauge fields and gravity.
Section~\ref{CFoundations} adds some remarks on the resulting perspective on foundations of quantum mechanics
concerning non-locality and entanglement (for remarks about the measurement problem, see Section~\ref{IModelQT}).
In Section~\ref{CClarifyingRemarks} we conclude
with a few clarifying remarks.
This chapter has been published in a slighly modified form as~\cite{dice2014}.

\section{The Theory} \label{CTheory}
The general structure of the theory of causal fermion systems can be understood in analogy to
general relativity. In general relativity, our universe is described by a four-dimensional
space-time (Lorentzian manifold) together with particles and fields.
However, not every configuration of Lorentzian metric, particles and fields is
considered to be ``physical'' in the sense that it could be realized in nature.
Namely, for the configuration to be physically realizable, the Einstein equations
must hold.
Moreover, the particles must satisfy the equations of motion, and the additional fields
must obey the field equations (like Maxwell's equations).
This means that in general relativity, there are two conceptual parts: on the one hand one has
mathematical objects describing possible configurations, and on the other hand there is
a principle which singles out the physical configurations. \\[-0.75em]

The theory of causal fermion systems has the same conceptual structure
consisting of mathematical objects and a principle which singles out the physical configurations.
We first introduce the mathematical objects:
\begin{Def}  {\em{ \textbf{(Causal fermion system)}\\[-1.5em] \label{CCFS}
\begin{itemize}[leftmargin=2em]
\itemsep0em
\itemD Let $(\H,\langle.|.\rangle_\H)$ be a separable complex Hilbert space.
\itemD Given a parameter~$n \in \N$ (the {\em{spin dimension}}), let $\F \subset \Lin(\H)$ be the set of all self-adjoint operators on $\H$ of finite rank, which (counting multiplicities) have at most~$n$ positive and
at most~$n$ negative eigenvalues.
\itemD Let $\rho$ be a positive measure on $\F$ (the {\em{universal measure}}).
\end{itemize}
Then $(\H,\F,\rho)$ is a \textit{causal fermion system}. }}
\end{Def} \noindent
Here separable means that the Hilbert space has an at most countable orthonormal basis.
Mapping the basis vectors to each other, one sees that any two Hilbert spaces are isomorphic,
provided that their dimensions coincide. Therefore, the structure~$(\H, \F)$
is completely determined by the parameters~$n \in \N$ and~$f:= \dim \H \in \N \cup \{\infty\}$.
Apart from these parameters, the only object specifying a causal fermion system is the
universal measure~$\rho$.

It will be outlined below that this definition indeed generalizes mathematical structures used in contemporary physics. The picture is that one causal fermion system describes a space-time together with all structures and objects therein (including the metric, particles and fields). \\[-0.75em]

Next, we state the principle which singles out the physical configurations.
Similar to the Lagrangian formulation of contemporary physics, we work with a
variational principle, referred to as the {\em{causal action principle}}. It states that
a causal fermion system which can be realized in nature should be
a minimizer of the so-called causal action.
In order to formulate the causal action principle, we assume that the
Hilbert space $(\H,\langle.|.\rangle_\H)$ and the spin dimension $n$ have been chosen.
Let $\F$ be as in Definition~\ref{CCFS} above. 
Then for any~$x, y \in \F$, the product~$x y$ is an operator
of rank at most~$2n$. We denote its non-trivial eigenvalues (counting algebraic multiplicities)
by~$\lambda^{xy}_1, \ldots, \lambda^{xy}_{2n} \in \C$.
We introduce the {\em{spectral weight}}~$| \,.\, |$ of an operator as the sum of the absolute values
of its eigenvalues. In particular, the spectral weight of the operator
products~$xy$ and~$(xy)^2$ is defined by
\[ |xy| = \sum_{i=1}^{2n} \big|\lambda^{xy}_i \big|
\qquad \text{and} \qquad \left| (xy)^2 \right| = \sum_{i=1}^{2n} \big| \lambda^{xy}_i \big|^2 \:. \]
Next, the {\em{Lagrangian}} $\L : \F \times \F \rightarrow \R_0^+$ is defined by
\begin{align}
\label{CLagr}
\L(x,y) := \big|(xy)^2 \big| - \frac{1}{2n} |xy|^2 = \frac{1}{4n} \sum _{i,j=1}^{2n}
\Big( \big| \lambda_i^{xy} \big| - \big| \lambda_j^{xy} \big| \Big)^2 \:.
\end{align}
The particular form of this Lagrangian is the result of research carried
out over several years (see Section~\ref{Cseccap}).

\begin{Def}  {\em{ \label{CCVP} \textbf{(Causal action principle)}
The {\em{causal action}}~$\Sact$ is obtained by integrating the Lagrangian with respect to the universal measure,
\[ \Sact(\rho) = \iint_{\F \times\F} \L(x,y)\: d\rho(x) \: d\rho(y) \:. \]
The {\em{causal action principle}} is to minimize~$\Sact$ under variations of the universal measure,
taking into account the following constraints:
\begin{align}
&\text{\em{volume constraint:}} & \rho(\F) = \text{const} \quad\;\; & \label{Cvolconstraint} \\
&\text{\em{trace constraint:}} &  \int_\F \tr(x)\: d\rho(x) = \text{const}& \label{Ctrconstraint} \\
&\text{\em{boundedness constraint:}} &   \T := \iint_{\F \times \F} |xy|^2\: d\rho(x)\, d\rho(y) &\leq C \:,
\end{align}
where~$C$ is a given constant (and~$\tr$ denotes the trace of a linear operator on~$\H$). }}
\end{Def}

In mathematical terms, the measure~$\rho$ is varied within the class of
positive regular Borel measures on~$\F$, where on~$\F$ one takes the topology
induced by the $\sup$-norm on~$\Lin(\H)$ (for basic definitions see for example~\cite[Chapters~2 and~5]{rudin}
or~\cite[Chapter~X]{halmosmt}).
The volume and trace constraints are needed in order to avoid trivial minimizers and are important for the analysis of the corresponding Euler-Lagrange equations because they give rise to Lagrange multiplier terms.
The boundedness constraint is needed in order to ensure the existence of minimizers.
In most applications, it does not give rise to a Lagrange multiplier term.
Therefore, it does not seem to have any physical consequences. \\[-0.75em]

This concludes the outline of the mathematical definition of the theory.
In order to obtain a physical theory, we need to give the mathematical objects a physical
interpretation. It is one of the main objectives of the next sections to do so by explaining 
how the above mathematical objects relate to the common notions in physics.
The conclusion will be that
causal fermion systems are indeed a candidate for a fundamental physical theory.

\section{Example: Dirac Wave Functions in Minkowski Space}\label{CIntro-Minkvac}
As a first step towards explaining how causal fermion systems relate to contemporary physics, we
now explain how the familiar physical situation of Dirac particles in Minkowski space can be
described by a causal fermion system.

Let~$\scrM$ be Minkowski space and $\mu$ the natural volume measure thereon, i.e.~$d\mu = d^4\x$ if $\x=(\x^0,\x^1,\x^2,\x^3)$ is an inertial frame (we use the signature convention $(+,-,-,-)$).
We consider a finite number of~$f$ Dirac particles described by one-particle wave
functions~$\psi_1, \ldots, \psi_f$ which are solutions of the Dirac equation,
\begin{align}
\label{CDiracEq}
\big( i \gamma ^j \partial_j - m \big)\, \psi_k = 0, \qquad k=1,\ldots, f\:,
\end{align}
where~$m$ is the rest mass, and~$\gamma^j$ are Dirac matrices in the Dirac representation.
For simplicity, we assume that the wave functions~$\psi_1, \ldots, \psi_f$ are continuous.

Before going on, we remark that this description of the $f$-particle system
by $f$ one-particle wave functions departs from the usual Fock space description.
The connection to Fock spaces will be explained later in this chapter (see
Section~\ref{CFoundations}).
For the moment, it is preferable to work with the one-particle wave functions.
We also remark that the assumption of considering a finite number of continuous wave functions
merely is a technical simplification for our presentation. All constructions can be extended
to an infinite number of possibly discontinuous wave functions
(for details see~\cite[Section~4]{finite} or~\cite[Chapter~1]{cfs}).

The wave functions $\psi_k$ span a vector space which we denote by $\H$,
\beq \label{CHspan}
\H := \textrm{span}(\psi_1, \ldots,\psi_f) \:.
\eeq
On~$\H$ we consider the usual scalar product on solutions of the Dirac equation
\begin{align}
\label{CScalProd}
\la \psi | \phi \ra_\H := 2 \pi \int_{t=\textrm{const}} (\overline \psi \gamma^0  \phi) (t,\vec x) \:d^3 x
\end{align}
(here~$\overline{\psi} = \psi^\dagger \gamma^0$ is the adjoint spinor, where the dagger denotes complex conjugation and transposition). If one evaluates~\eqref{CScalProd} for~$\phi=\psi$,
the integrand can be written as~$(\overline{\psi}\gamma^0\psi)(t,\vec{x}) = (\psi^\dagger \psi)(t,\vec{x})$,
having the interpretation as the probability density of the Dirac particle corresponding to~$\psi$
to be at the position~$\vec{x}$. In view of the conservation of probability (being a consequence
of current conservation), the integral in~\eqref{CScalProd} is time independent.
Since the probability density is positive, the inner product~\eqref{CScalProd} is indeed positive definite.
We thus obtain an $f$-dimensional Hilbert space~$(\H, \la .|. \ra_\H)$.

For any~$\x \in \scrM$, we now introduce the sesquilinear form
\[ b_\x :  \H \times \H \rightarrow \C \:,\qquad b_\x(\psi, \phi) = -(\overline \psi \phi) (\x) \:, \]
which maps two solutions of the Dirac equation to their inner product at $\x$.
The sesquilinear form $b_\x$ can be represented by a self-adjoint operator $F(\x)$ on $\H$,
which is uniquely defined by the relations
\[ \la \psi | F(\x) \phi \ra_\H =b_\x(\psi,\phi) \qquad \text{for all~$\psi, \phi \in \H$}\:. \]
More concretely, in the basis~$(\psi_k)_{k = 1, \ldots,f}$ of~$\H$, the last relation can be written as
\begin{align} \label{CFdef}
\la \psi_i | F(\x) \psi_j \ra_\H = - \big(\overline{\psi_i} \psi_j \big)(\x) \:.
\end{align}
If the basis is orthonormal, the calculation
\[ F(\x) \,\psi_j = \sum_{i=1}^f \la \psi_i | F(\x) \psi_j \ra_\H\; \psi_i
= - \sum_{i=1}^f \big(\overline{\psi_i} \psi_j \big)(\x)\; \psi_i \]
(where we used the completeness relation~$\phi = \sum_i \la \psi_i | \phi \ra\, \psi_i$),
shows that the operator~$F(\x)$ has the matrix representation
\[ \big(F(\x) \big)^i_j = - \big(\overline{\psi_i} \psi_j \big)(\x) \:. \]
In physical terms, the matrix element~$-(\overline{\psi_i} \psi_j)(\x)$ gives information on the correlation of the
wave functions~$\psi_i$ and~$\psi_j$ at the space-time point~$\x$.
Therefore, we refer to~$F(\x)$ as the {\em{local correlation operator}} at~$\x$.

Let us analyze the properties of $F(\x)$. First of all, the calculation
\[ \la F(\x) \,\psi \,|\, \phi \ra_\H = \overline{ \la \phi \,|\, F(\x) \,\psi \,\ra_\H}
= -\overline{(\overline \phi \psi) (\x)} = -(\overline \psi \phi) (\x) = \la \psi \,|\, F(\x) \,\phi \ra_\H \]
shows that the operator~$F(\x)$ is self-adjoint
(where we denoted complex conjugation by a bar).
Furthermore, since the pointwise inner product $(\overline \psi \phi)(\x)$ has signature $(2,2)$,
we know that~$b_\x$ has signature $(p,q)$ with $p,q \leq 2$.
As a consequence, the operator $F(\x)$ has at most two positive and at most two negative eigenvalues
(counting multiplicities). It follows immediately, that~$F(\x) \in \F$ if the spin dimension
in Definition~\ref{CCVP} is chosen as~$n=2$.

Constructing the operator $F(\x) \in \F$ for every space-time point $\x \in M$, we
obtain the mapping
\begin{align*}
F: \: & \scrM \rightarrow \F \:,\qquad \x \mapsto F(\x) \:.
\end{align*}
This allows us to introduce a measure $\rho$ on $\F$ as follows. For any~$\Omega \subset \F$,
one takes the pre-image $F^{-1}(\Omega) \subset \scrM$ and computes its space-time volume,
\[ \rho(\Omega) := \mu \big( F^{-1}(\Omega) \big) \:. \]
This gives rise to the so-called {\em{push-forward measure}} which in mathematics
is denoted by~$\rho = F_\ast \mu$
(see for example~\cite[Section~3.6]{bogachev}; we remark for the mathematically oriented reader that the $\sigma$-algebra of
$\rho$-measurable sets is defined as all sets~$\Omega \subset \F$ whose pre-image~$F^{-1}(\Omega)$
is $\mu$-measurable).

Putting the above structures together, we obtain a causal fermion system~$(\H, \F, \rho)$
of spin dimension two. Thus we have succeeded in constructing a causal fermion system
starting from a system of Dirac wave functions in Minkowski space.
But it is not obvious how much of the information on the physical system is encoded in the
causal fermion system. In other words, taking the causal fermion system~$(\H, \F, \rho)$
as the starting point, the question is which structures of the original system can be recovered.
For example, is the Minkowski metric still determined?
\label{CIntro-Question1}
Is it possible to reconstruct the Dirac wave functions? Precise answers to these questions will be given
in Section~\ref{Csecinherent} below. In preparation, we now give a few hints.

We first explain what the points of Minkowski space correspond to in our
causal fermion system. Recall that to every space-time point~$\x \in \scrM$ we associated
a linear operator~$F(\x) \in \F$. Hence the space-time points correspond to the subset~$F(\scrM) \subset \F$.
This subset can also be characterized as the set where the measure~$\rho$ is non-zero.
In mathematical terms, this is captured in the notion of the {\em{support}} of the universal measure,
defined as the set of all the points of~$\F$ such that every open neighborhood of this point
has a non-zero measure. Then (for details see~\cite[Chapter~1]{cfs})
\begin{align}
\label{Csupp-rho}
\supp{\rho} = \overline{ F(\scrM) }\:,
\end{align}
where the bar denotes the closure.
In all situations of physical interest, the mapping~$F$ will be injective and its image closed
(see again~\cite[Chapter~1]{cfs}). Provided that this is the case,
identifying~$\x \in \scrM$ with the corresponding operator~$F(\x) \in \F$
makes it possible to identify Minkowski space with the support of~$\rho$ as a topological space.
Under suitable smoothness and non-degeneracy assumptions, one can identify~$\scrM$
with~$\supp \rho$ even as a differentiable manifold.
We make this identification manifest by using the letter~$x$ for the operator~$F(\x)$.
In order to avoid confusion, we use two different fonts in this chapter, making it possible for the
reader to distinguish a point~$\x \in \scrM$ of Minkowski space from
the corresponding point~$x \in M := \supp \rho$.
Once the reader has become familiar with our concepts, the different fonts will be unnecessary.

This consideration shows that the topological and differentiable structures of our space-time
are encoded in the causal fermion system.
Clearly, Minkowski space also has metric and causal structures, which we have not yet
addressed. The general idea for recovering these structures is to take operators~$x,y \in \supp \rho$
and to analyze the eigenvalues of the operator product~$x y$.
The eigenvalues of such operator products contain plenty of information, inducing
relations and structures between the space-time points.
This will be explained more concretely in the next section.

\section{Inherent Structures}\label{Csecinherent} 
Let~$(\H, \F, \rho)$ be a causal fermion system of spin dimension $n$ (see Definition~\ref{CCFS}).
We now introduce additional objects which will turn out to generalize familiar notions in physics.
All of these structures are inherent in the sense that
we only use information already encoded in the causal fermion system.

Motivated by the consideration above (see the paragraph before~\eqref{Csupp-rho}),
{\em{space-time}}~$M$ is defined as the support of the universal measure,
\[ M := \text{supp}\, \rho \subset \F \:. \]
On~$M$ we introduce the following notion of causality.
Recall that for~$x, y \in M$, the product $x y$ is an operator of rank at most $2n$.
We again denote its non-trivial eigenvalues (counting algebraic multiplicities)
by~$\lambda^{xy}_1, \ldots, \lambda^{xy}_{2n} \in \C$.

\begin{Def}  {\em{ \textbf{(Causality)} \label{CKS}
The space-time points~$x$ and $y$ are defined to be\\[-1.5em]
\begin{itemize}[leftmargin=2em, itemsep=0.1em] \itemsep0em
\itemD \textit{spacelike} separated if all the~$\lambda^{xy}_j$ have the same absolute value.
\itemD \textit{timelike} separated if the~$\lambda_{i}^{xy}$ do not all have the same absolute value
and are all real.
\itemD \textit{lightlike} separated if the~$\lambda_{i}^{xy}$ do not all  have the same absolute value
and are not all real.
\end{itemize} }}
\end{Def}
This definition is compatible with the causal action in the following sense.
If the points~$x$ and~$y$ are spacelike separated, then all the~$\lambda^{xy}_j$ have the same absolute value,
so that the Lagrangian~$\L(x,y)$ vanishes according to~\eqref{CLagr}.
In a more physical language, this means that no interaction takes place 
between regions with spacelike separation
(this does not exclude nonlocal correlations and entanglement, as will be
discussed in Section~\ref{CFoundations}).
In this way, our setting incorporates a general version of the principle of causality.

The next step is to introduce wave functions. 
The construction is guided by the usual structure of
a Dirac wave function~$\psi$, which to every space-time point~$\x$
associates a spinor~$\psi(\x)$. The latter is a vector in the corresponding spinor space~$S_\x \scrM \simeq \C^4$,
which is endowed with the inner product~$\overline{\psi} \phi$ of signature~$(2,2)$.
In the setting of causal fermion systems, for a space-time point~$x \in M$ we define the
{\em{spin space}}~$S_x \subset \H$ as the image of the operator~$x$,
\[ S_x := x(\H) \:. \]
It is a subspace of~$\H$ of dimension at most~$2n$.
On~$S_x$ we introduce the inner product
\beq \label{CSprod}
\Sl .|. \Sr_x : S_x \times S_x \rightarrow \C \:,\qquad
\Sl u | v \Sr_x := - \la u | x v \ra _ \H \:,
\eeq
referred to as the {\em{spin scalar product}}.
Since~$x$ has at most~$n$ positive and at most~$n$ negative eigenvalues,
the spin scalar product is an indefinite inner product of signature $(p,q)$ with $p,q \leq n$.
A {\em{wave function}}~$\psi$ is defined as a function
which to every~$x \in M$ associates a vector of the corresponding spin space,
\[ \psi \::\: M \rightarrow \H \qquad \text{with} \qquad \psi(x) \in S_x \quad \text{for all~$x \in M$}\:. \]

Clearly, it is not sufficient to define wave functions abstractly, but we need to specify
those wave functions which are realized in the physical system.
Using a familiar physical language, we need to declare which one-particle states are
occupied (for the connection to multi-particle Fock states see Section~\ref{CFoundations}).
To this end, to every vector~$u \in \H$ of the Hilbert space we associate a wave function~$\psi^u$
by projecting the vector~$u$ to the spin spaces, i.e.
\beq \label{Cpsiudef}
\psi^u \::\: M \rightarrow \H\:,\qquad \psi^u(x) := \pi_x u \in S_x \:,
\eeq
where~$\pi_x$ is the orthogonal projection in~$\H$ on the subspace~$x(\H) \subset \H$.
We refer to~$\psi^u$ as the {\em{physical wave function}} corresponding to the vector~$u \in \H$.

Finally, we define the {\em{kernel of the fermionic projector}}~$P(x,y)$
for any~$x, y \in M$ by
\beq \label{CPxydef}
P(x,y) = \pi_x \,y|_{S_y} \::\: S_y \rightarrow S_x
\eeq
(where~$|_{S_y}$ denotes the restriction to the subspace~$S_y \subset \H$).
This object is useful for analyzing the relations and structures between space-time points.
In particular, the kernel of the fermionic projector encodes the causal structure
and makes it possible to compute the eigenvalues~$\lambda^{xy}_1, \ldots, \lambda^{xy}_{2n}$
which appear in the Lagrangian~\eqref{CLagr}. In order to see how this comes about,
we first define the {\em{closed chain}} as the product
\beq \label{CAxydef}
A_{xy} = P(x,y)\, P(y,x) \::\: S_x \rightarrow S_x\:.
\eeq
Computing powers of the closed chain and using that~$y \pi_y = y$
(because the image and kernel of self-adjoint operators are orthogonal), we obtain
\[ A_{xy} = (\pi_x y)(\pi_y x)|_{S_x} = \pi_x\, yx|_{S_x} \qquad \text{and thus} \qquad
(A_{xy})^p = \pi_x\, (yx)^p|_{S_x} \:. \]
Taking the trace, we obtain for all~$p \in \N$,
\begin{align*}
\Tr_{S_x} \big( (A_{xy})^p \big) &= \Tr_{S_x} \big(\pi_x\, (yx)^p|_{S_x} \big)
= \tr \big(\pi_x\, (yx)^p|_{S_x} \big) \\
&= \tr \big((yx)^p \pi_x \big)
= \tr \big((yx)^p \big) = \tr \big( (xy)^p \big)
\end{align*}
(where~$\tr$ again denotes the trace of a linear operator on~$\H$).
Since the coefficients of the characteristic polynomial of an operator
can be expressed in terms of traces of powers of the corresponding matrix,
we conclude that the
eigenvalues of the closed chain coincide with the non-trivial
eigenvalues~$\lambda^{xy}_1, \ldots, \lambda^{xy}_{2n}$ of the operator~$xy$ in
Definition~\ref{CCVP}.
In this way, one can recover the~$\lambda^{xy}_1, \ldots, \lambda^{xy}_{2n}$ as the eigenvalues of 
a $(2n \times 2n)$-matrix.
In particular, the kernel of the fermionic operator encodes the causal structure of~$M$.

The kernel of the fermionic projector is the starting point for constructions which
unveil the geometric structures of a causal fermion system.
More specifically, this kernel gives rise to a spin connection and corresponding curvature.
Moreover, one can introduce tangent spaces endowed with a Lorentzian metric
together with a corresponding metric connection and curvature.
For brevity, we cannot enter these topics here. Instead we refer the interested
reader to~\cite{lqg, topology}, where also questions concerning the topology
of causal fermion systems are treated.
The important point to keep in mind is that all these constructions are tailored
in order to understand the meaning of information contained in the causal
fermion system. No additional input is required.
The system is completely determined by the causal fermion system~$(\H, \F, \rho)$.
In particular, when varying the universal measure in the causal action principle,
one also varies all the derived structures mentioned above.

\section{The Minkowski Vacuum} \label{Csecminkvac} 
In order to illustrate the above inherent structures, we now return to the example of Dirac particles in Minkowski
space introduced in Section~\ref{CIntro-Minkvac}.
In this example, the Hilbert space~$\H$ is spanned by solutions of the Dirac equation.
Thus a vector~$u \in \H$ is a Dirac wave function, which at a point~$\x \in \scrM$ of
Minkowski space takes values in the corresponding spinor space, $u(\x) \in S_\x\scrM$.
On the other hand, in the previous section we introduced the corresponding physical wave function~$\psi^u$,
which at a point~$x = F(\x) \in M \subset \F$ takes values in the corresponding spin space,
$\psi^u(x) \in S_x$. We now show that these objects can be identified.
Indeed, for any~$u, v \in S_x \subset \H$,
\begin{align*}
\Sl \psi^u(x) \,|\, \psi^v(x) \Sr_x &\overset{\eqref{CSprod}}{=} -\la \pi_x u \,|\, x\, \pi_x v \ra_\H
= -\la u \,|\, x\, v \ra_\H = -\la u \,|\, F(\x)\, v \ra_\H
\overset{\eqref{CFdef}}{=} \overline{u(\x)} v(\x)\:.
\end{align*}
This shows that the inner products on~$S_\x\scrM$ and~$S_x$ are compatible.
It implies that, after choosing suitable bases, one can indeed identify~$S_\x\scrM$ with~$S_x$
(for details see~\cite[Section~1.2]{cfs} or~\cite[Section~4]{lqg}).
This identification implies that~$\psi^u(x) = u(\x)$ for all~$u \in \H$ and~$x \in M$
respectively~$\x \in \scrM$.

Next, it is instructive to bring the kernel of the fermionic projector~\eqref{CPxydef} into a more tractable form.
To this end, we choose an orthonormal basis~$u_1, \ldots, u_f$ of~$\H$. Then
for any~$\phi \in S_y$,
\begin{align*}
P(x,y)\, \phi &= \pi_x \,y\, \phi \overset{(\star)}{=} \sum_{\ell=1}^f \big(\pi_x u_\ell\big) \, \la u_\ell |  \,y\, \phi \ra_\H \\
&\!\!\overset{\eqref{CSprod}}{=} -\sum_{\ell=1}^f  \big(\pi_x u_\ell\big) \,  \Sl \pi_y u_\ell | \phi \Sr_y
\overset{\eqref{Cpsiudef}}{=}  -\sum_{\ell=1}^f  \psi^{u_\ell}(x) \,  \Sl \psi^{u_\ell}(y) \,|\, \phi \Sr_y \:,
\end{align*}
where in~$(\star)$ we used the completeness relation. Using the above identifications
of spinors and their inner products, we can write this formula in the shorter form
\beq \label{CPxyuseful}
P(x,y) = -\sum_\ell u_\ell(\x)\, \overline{u_\ell(\y)}\:.
\eeq
This shows that the kernel of the fermionic projector is composed of all the physical wave functions
of the system.

In order to work in a more concrete example, we next consider the {\em{Minkowski vacuum}}. 
To this end, we want to implement the concept of the Dirac
sea which in non-technical terms states that 
in the vacuum all the negative-energy states of the Dirac equation should be occupied
(see Section~\ref{Csec73} for further explanations of this point).
In order to implement this concept, one needs to consider an infinite number of
physical wave functions. This can be achieved simply by letting~$\H$ in Definition~\ref{CCFS}
be an infinite-dimensional Hilbert space.
However, a difficulty arises in the construction of the local correlation operators,
because the Dirac wave functions (being square-integrable functions) are in general not defined pointwise,
so that the right side of~\eqref{CFdef} is ill-defined.
In order to resolve this problem, one needs to introduce an ultraviolet regularization.
For conceptual clarity, we postpone the explanation of the ultraviolet regularization to Section~\ref{CContinuum}
and now merely mention that
an ultraviolet regularization amounts to modifying the Dirac wave functions
on a microscopic scale~$\varepsilon$, which can be thought of as the Planck scale.
In order to avoid the technical issues involved in the regularization,
we here simply use the formula~\eqref{CPxyuseful},
but now sum over all negative-energy solutions of the Dirac equation.
This sum can be rewritten as an integral over the lower mass shell
(see again~\cite[Section~1.2]{cfs}),
\beq \label{CPxyvac}
P(x,y) = \int \frac{d^4k}{(2 \pi)^4}\:(k_j \gamma^j+m)\: \delta(k^2-m^2)\: \Theta(-k_0)\: e^{-ik(\x-\y)} \:.
\eeq
In this formula, the necessity for an ultraviolet regularization is apparent
in the fact that the Fourier integral is not defined pointwise, but only in the distributional sense.
More precisely, the distribution~$P(x,y)$ is singular if the vector~$\xi := \y- \x$ is lightlike,
but it is a smooth function otherwise
(as can be verified for example by explicit computation).
As a consequence, a typical ultraviolet regularization will affect the behavior of~$P(x, y)$ only
in a small neighborhood of the light cone of the form~$\big| |\xi^0| - |\vec{\xi}| \big| \lesssim \varepsilon$.
With this in mind, for the following argument we may disregard the ultraviolet regularization
simply by restricting attention to the region outside this neighborhood.

The representation~\eqref{CPxyvac} allows us to understand the relation
between the Definition~\ref{CKS} and the usual notion of causality in Minkowski space:
Since the expression~\eqref{CPxyvac} is Lorentz invariant and is composed of a vector
and a scalar component, the function~$P(x, y)$ can be written as
\[ P(x, y) = \alpha\, \xi_j \gamma^j + \beta\:\1 \]
with two complex-valued functions~$\alpha$ and~$\beta$ (where again~$\xi =\y-\x$).
Taking the conjugate with respect to the spin scalar product, we see that
\[ P(y, x) = P(x, y)^* = \overline{\alpha}\, \xi_j \gamma^j + \overline{\beta}\:\1 \:. \]
As a consequence,
\[ A_{xy} = P(x,y)\, P(y,x) = a\, \xi_j \gamma^j + b\, \1 \]
with two real-valued functions~$a$ and $b$ given by
\[ a = \alpha \overline{\beta} + \beta \overline{\alpha} \:,\qquad
b = |\alpha|^2 \,\xi^2 + |\beta|^2 \:. \]
Applying the formula~$(A_{xy} - b \1)^2 = a^2\:\xi^2\,\1$,
the roots of the characteristic polynomial of~$A_{xy}$ are computed by
\[ b \pm \sqrt{a^2\: \xi^2} \:. \]
Thus if the vector~$\xi$ is timelike, the term~$\xi^2$ is positive, so that
the~$\lambda_j$ are all real. By explicit computation one sees that the coefficients~$a$
and~$b$ are non-zero (see~\cite[Section~\S1.2.5]{cfs}), implying that the
eigenvalues~$\lambda_j$ do not all have the same absolute value. Conversely, if the vector~$\xi$
is spacelike, then the term~$\xi^2$ is negative. Thus the~$\lambda_j$ form a complex conjugate
pair, implying that they all have the same absolute value.
We conclude that the notions of spacelike and timelike as defined for causal fermion systems in
Definition~\ref{CKS} indeed agree with the usual notions in Minkowski space.
We remark that this simple argument cannot be used for lightlike directions because
in this case the distribution~$P(x,y)$ is singular, making it necessary to consider an
ultraviolet regularization (the reader interested in the technical details is referred to~\cite{reg}).

To summarize, we have seen that the inherent structures of a causal fermion system
give back the usual causal structure if one considers the Dirac sea vacuum in Minkowski space.
Indeed, a more detailed analysis reveals that the additional inherent structures
mentioned at the end of Section~\ref{Csecinherent} also give back the
geometric structures of Minkowski space (like the metric and the
connection).

\section{Description of More General Space-Times} \label{Csecgenst}
The constructions explained above also apply to more general physical situations.
First, one can consider systems involving {\em{particles}}
and {\em{anti-particles}} by occupying additional states and removing
states from the Dirac sea, respectively.
Moreover, our construction also apply in {\em{curved space-time}} (see~\cite{finite, infinite})
or in the presence of an {\em{external potential}} (see~\cite{cfs}). In all these situations, the
resulting causal fermion systems again encode all the information on the physical system (see~\cite{lqg, cfs}).

The framework of causal fermion systems also allows to describe generalized space-times,
sometimes referred as \emph{quantum space-times}.
We now illustrate this concept in the simple example of a {\em{space-time lattice}}. Thus we replace Minkowski
space by a four-dimensional lattice~$\scrM := (\varepsilon \Z)^4$ of lattice spacing~$\varepsilon$.
Likewise, the volume measure~$d^4\x$ is replaced by a counting measure~$\mu$
(thus~$\mu(\Omega)$ is equal to the number of lattice points contained in~$\Omega$).
Restricting the Dirac spinors of Minkowski space to the lattice, one gets a
spinor space~$S_\x \scrM$ at every point~$\x \in \scrM$. Dirac wave functions~$\psi_1, \ldots, \psi_f$
can again be introduced as mappings which to every~$\x \in \scrM$ associate a
vector in the corresponding spinor space. These Dirac wave functions can be chosen
for example as solutions of a discretized version of the Dirac equation. Again choosing~$\H$
as the span of the wave functions~\eqref{CHspan} and choosing a suitable scalar product~$\la .|. \ra_\H$,
one defines the local correlation operators again by~\eqref{CFdef}.
Introducing the universal measure as the push-forward of the counting measure~$\mu$,
we obtain a causal fermion system~$(\H, \F, \rho)$ of spin dimension two.
The only difference to the causal fermion system in Minkowski space as constructed
in Section~\ref{CIntro-Minkvac} is that now the universal measure is not a continuous
but a discrete measure. 

When describing the Dirac sea vacuum on the lattice, the lattice spacing gives rise to
a natural ultraviolet regularization on the scale~$\varepsilon$.
For example, one may consider all plane-wave solutions~$\psi(x) \sim e^{i kx}$ of the Dirac
equation whose four-momenta lie in the first Brillouin zone, i.e.~$-\pi < \varepsilon \,k_j \leq \pi$ for all~$j=0,\ldots, 3$.
Then one introduces~$\H$ as the Hilbert space generated by
all these plane-wave solutions restricted to the lattice.

Other examples of discrete or singular space-times are described in~\cite{topology}.

\section{The Continuum Limit} \label{CContinuum} 
In the previous sections we saw that a causal fermion system has inherent structures which
generalize corresponding notions in quantum theory and relativity.
The next task is to analyze the
dynamics of these objects as described by the causal action principle.
To this end, one considers the Euler-Lagrange (EL) equations corresponding to the
causal action. These equations have a mathematical structure
which is quite different from conventional physical equations  (see~\cite{lagrange}). Therefore, the main
difficulty is to reexpress the EL equations in terms of the inherent structures so as to
make them comparable with the equations of contemporary physics.
This can indeed be accomplished in the so-called continuum limit.
Since the mathematical methods needed for the analysis of the continuum limit
go beyond the scope of this introduction (for details see~\cite{cfs, PFP}),
here we can only explain the general concept and discuss the obtained results.

We outlined in Sections~\ref{CIntro-Minkvac} and~\ref{Csecminkvac} how
to describe the Minkowski vacuum by a causal fermion system.
Recall that the construction required an {\em{ultraviolet regularization}} on a microscopic scale~$\varepsilon$.
Such a regularization can be performed in many different ways.
The simplest method is to smooth out the wave functions on the microscopic scale by
convolution with a test function. Another method is to introduce a cutoff in momentum space
on the scale~$\varepsilon^{-1}$. Alternatively, one can regularize by putting the system
on a four-dimensional lattice with lattice spacing~$\varepsilon$ (for example as explained
in Section~\ref{Csecgenst} above).
It is important to note that each regularization gives rise to a different causal fermion system,
describing a physical space-time with a different microstructure.
Thus in the context of causal fermion systems, the regularization has a physical
significance. The freedom in regularizing reflects our lack of knowledge on the
microstructure of physical space-time.
When analyzing the EL equations corresponding to the causal action, it is not obvious why the effective
macroscopic equations should be independent of the regularization details.
Therefore, it is necessary to consider a sufficiently large class of regularizations,
and one needs to analyze carefully how the results depend on the regularization.
This detailed analysis, referred to as the {\em{method of variable regularization}}
(for more explanations see~\cite[\S4.1]{PFP}),
reveals that for a large class of regularizations,
the structure of the effective macroscopic equations is indeed independent of the regularization
(for details see~\cite[Chapters~3-5]{cfs}).

The {\em{continuum limit}} is a method for evaluating the EL equations corresponding
to the causal action in the limit~$\varepsilon \searrow 0$ when the ultraviolet regularization
is removed. The effective equations obtained in this limit can be evaluated conveniently in
a formalism in which the unknown microscopic structure of space-time (as described by the regularization)
enters only in terms of a finite (typically small) number of so-called {\em{regularization parameters}}.

It turns out that the causal fermion system describing the Minkowski vacuum
satisfies the EL equations in the continuum limit (for any choice of the regularization parameters).
If one considers instead a system involving additional particles and anti-particles,
it turns out the EL equations in the continuum limit no longer hold.
In order to again satisfy these equations, we need to introduce an interaction.
In mathematical terms, this means that the universal measure~$\rho$ must be modified.
Expressed in terms of the inherent structures of a causal fermion system, all the physical wave functions
$\psi^{u_k}(x)$ must be changed collectively.
The analysis shows that this collective behavior of all physical wave functions (including the states
of the Dirac sea) can be described by inserting a potential~$\B$ into the
Dirac equation~\eqref{CDiracEq},
\beq \label{CCont:DiracEq}
\big( i \Pdd + \B - m \big) \,u_k(\x) = 0, \qquad k=1,\ldots, f
\eeq
(where as usual~$\Pdd = \gamma^j \partial_j$).
Moreover, the EL equations in the continuum limit are satisfied if and only if the potential~$\B$
satisfies field equations. Before specifying these field equations,
we point out that in the above procedure, the
potential~$\B$ merely is a convenient device in order to describe the collective
behavior of all physical wave functions. It should not be considered as a fundamental
object of the theory.
We also note that, in order to describe variations of the physical wave functions,
the potential in~\eqref{CCont:DiracEq} can be chosen arbitrarily. Each choice of~$\B$
describes a different variation of the physical wave functions.
The EL equations in the continuum limit single out the physically admissible potentials
as being those which satisfy the field equations.

In~\cite{cfs} the continuum limit is worked out in several steps beginning from simple systems
and ending with a system realizing the fermion configuration of the standard model.
For each of these systems, the continuum limit gives rise to effective equations for
second-quantized fermion fields
coupled to classical bosonic gauge fields (for the connection to second-quantized bosonic
fields see Section~\ref{CsecQFT} below).
To explain the structure of the obtained results, it is preferable to first describe the
system modelling the leptons as analyzed in~\cite[Chapter~4]{cfs}.
The input to this model is the configuration of the leptons in the standard model
without interaction. Thus the fermionic projector of the vacuum is assumed to be
composed of three generations of Dirac particles of masses~$m_1, m_2, m_3>0$
(describing~$e$, $\mu$, $\tau$)
as well as three generations of Dirac particles of masses~$\tilde{m}_1, \tilde{m}_2,
\tilde{m}_3 \geq 0$ (describing the corresponding neutrinos).
Furthermore, we assume that the regularization of the neutrinos breaks the chiral
symmetry (implying that we only see their left-handed components).
We point out that the definition of the model does not involve any assumptions on the interaction.

The detailed analysis in~\cite[Chapter~4]{cfs} reveals that the effective interaction in the continuum limit
has the following structure.
The fermions satisfy the Dirac equation coupled to a left-handed $\SU(2)$-gauge
potential~$A_L=\big( A_L^{ij} \big)_{i,j=1,2}$,
\[ \left[ i \Pdd + \begin{pmatrix} \Aslsh_L^{11} & \Aslsh_L^{12}\, \UMNS^* \\[0.2em]
\Aslsh_L^{21}\, \UMNS & -\Aslsh_L^{11} \end{pmatrix} \chi_L
- m Y \right] \!\psi = 0 \:, \]
where we used a block matrix notation (in which the matrix entries are
$3 \times 3$-matrices). Here~$mY$ is a diagonal matrix composed of the fermion masses,
\beq \label{CmY}
mY = \text{diag} (\tilde{m}_1, \tilde{m}_2, \tilde{m}_3,\: m_1, m_2, m_3)\:,
\eeq
and~$\UMNS$ is a unitary $3 \times 3$-matrix (taking the role of the
Maki-Nakagawa-Sakata matrix in the standard model).
The gauge potentials~$A_L$ satisfy a classical Yang-Mills-type equation, coupled
to the fermions. More precisely, writing the isospin dependence of the gauge potentials according
to~$A_L = \sum_{\alpha=1}^3 A_L^\alpha \sigma^\alpha$ in terms of Pauli matrices,
we obtain the field equations
\beq \label{Cl:YM}
\partial^k \partial_l (A^\alpha_L)^l - \Box (A^\alpha_L)^k - M_\alpha^2\, (A^\alpha_L)^k = c_\alpha\,
\overline{\psi} \big( \chi_L \gamma^k \, \sigma^\alpha \big) \psi\:,
\eeq
valid for~$\alpha=1,2,3$ (for notational simplicity, we wrote the Dirac current for one Dirac
particle; for a second-quantized Dirac field, this current is to be replaced by the expectation value
of the corresponding fermionic field operators). Here~$M_\alpha$ are the bosonic masses and~$c_\alpha$
the corresponding coupling constants.
The masses and coupling constants of the two off-diagonal components are
equal, i.e.\ $M_1=M_2$ and~$c_1 = c_2$,
but they may be different from the mass and coupling constant
of the diagonal component~$\alpha=3$. Generally speaking, the mass ratios~$M_1/m_1$, $M_3/m_1$ as
well as the coupling constants~$c_1$, $c_3$ depend on the regularization. For a given regularization,
they are computable.

Finally, the model involves a gravitational field described by the Einstein equations
\beq \label{Cl:Einstein}
R_{jk} - \frac{1}{2}\:R\: g_{jk} + \Lambda\, g_{jk} = \kappa\, T_{jk} \:,
\eeq
where~$R_{jk}$ denotes the Ricci tensor, $R$ is scalar curvature, and~$T_{jk}$
is the energy-momentum tensor of the Dirac field. Moreover, $\kappa$ and~$\Lambda$ denote the
gravitational and the cosmological constants, respectively.
We find that the gravitational constant scales like~$\kappa \sim \delta^2$, where~$\delta \geq \varepsilon$ is
the length scale on which the chiral symmetry is broken.

In~\cite[Chapter~5]{cfs} a system is analyzed which realizes the
configuration of the leptons and quarks in the standard model.
The result is that the field equation~\eqref{Cl:YM} is replaced by
field equations for the electroweak and strong interactions after spontaneous
symmetry breaking (the dynamics of the corresponding Higgs field has not yet been analyzed).
Furthermore, the system again involves gravity~\eqref{Cl:Einstein}.

A few clarifying remarks are in order. First, the above field equations come with
corrections which for brevity we cannot discuss here (see~\cite[Sections~3.8, 4.4 and~4.6]{cfs}).
Next, it is worth noting that, although
the states of the Dirac sea are explicitly taken into account in our analysis, they do not
enter the field equations. More specifically, in a perturbative treatment,
the divergences of the Feynman diagram
describing the vacuum polarization drop out of the EL equations of the causal action.
Similarly, the naive ``infinite negative energy density'' of the
sea drops out of the Einstein equations, making it unnecessary to subtract any counter terms.
We finally remark that the only free parameters of the theory are the masses in~\eqref{CmY}
as well as the parameter~$\delta$ which determines the gravitational constant.
The coupling constants, the bosonic masses and the mixing matrices
are functions of the regularization parameters
which are unknown due to our present lack of knowledge on the microscopic structure of space-time.
The regularization parameters cannot be chosen arbitrarily because they must satisfy certain
relations. But except for these constraints, the regularization parameters are currently treated
as free empirical parameters.

To summarize, the dynamics in the continuum limit is described by Dirac spinors
coupled to classical gauge fields and gravity. The effective continuum theory is manifestly covariant
under general coordinate transformations.
The only limitation of the continuum limit is that the bosonic fields are merely classical.
However, as will be briefly mentioned in Section~\ref{CsecQFT}, a detailed analysis
which goes beyond the continuum limit gives rise even to second-quantized bosonic fields.
Based on these results, the theory of causal fermion systems
seems to be a promising candidate for a unified physical theory.

\section{Entanglement and Nonlocality}\label{CFoundations}

For general remarks on the connection of the theory of causal fermion systems to foundations of quantum theory,
we refer to Section~\ref{IModelQT}. Here, based on the previous explanations, we add some remarks on the connection to
{\em{nonlocality}} and {\em{entanglement}}. Both are experimentally tested features of quantum mechanics
which need to be explained by any fundamental theory which gives quantum theory as a limiting case.

To understand the role of nonlocality, one should keep in mind that in a causal fermion system,
a fermion is described by a physical wave function~$\psi^u(x)$ as defined in~\eqref{Cpsiudef}.
As in standard quantum mechanics, these wave functions are nonlocal objects
spread out in space-time, giving rise to the usual nonlocal correlations for
one-particle measurements.

In order to describe entanglement, one needs to work with multi-particle wave functions.
The simplest method to obtain the connection to those is to choose
an orthonormal basis~$u_1, \ldots, u_f$ of~$\H$ and to form the $f$-particle Hartree-Fock state
\beq \label{Cantisymm}
\Psi := \psi^{u_1} \wedge \cdots \wedge \psi^{u_f} \:.
\eeq
Clearly, the choice of the orthonormal basis is unique only up to the unitary transformations
\[ u_i \rightarrow \tilde{u}_i = \sum_{j=1}^f U_{ij} \,u_j \quad \text{with} \quad U \in \U(f)\:. \]
Due to the anti-symmetrization, this transformation changes the corresponding Hartree-Fock state
only by an irrelevant phase factor,
\[ \psi^{\tilde{u}_1} \wedge \cdots \wedge \psi^{\tilde{u}_f} = \det U \;
\psi^{u_1} \wedge \cdots \wedge \psi^{u_f} \:. \]
Thus the configuration of the physical wave functions can be described by a
fermionic multi-particle wave function.

The shortcoming of the above construction is that the Hartree-Fock state~\eqref{Cantisymm}
does not allow for the description of entanglement.
But entanglement arises naturally if the effect of {\em{microscopic mixing}} is taken into
account, as we now briefly outline.
Microscopic mixing is based on the observation that the causal action of
a Dirac sea configuration is smaller if the physical wave functions have fluctuations
on the microscopic scale. To be more precise, one constructs a universal measure~$\rho$
which consists of~$L$ components, i.e.\ $\rho = \rho_1 + \cdots +\rho_L$.
This also gives rise to a decomposition of the corresponding space-time, i.e.\
$M= M_1 \cup \cdots \cup M_L$ with~$M_\ell := \supp \rho_\ell$.
Now one considers variations of the measures~$\rho_\ell$ obtained by modifying the
phases of the physical wave functions in the sub-space-times~$M_\ell$.
Minimizing the causal action under such variations, one sees that 
the kernel of the fermionic projector~$P(x,y)$ becomes very small if~$x$ and~$y$
are in different sub-space-times. This effect can be understood similar to
a dephasing of the physical wave functions in different sub-space-times.

The resulting space-time~$M$ has a structure which cannot be understood classically.
One way of visualizing~$M$ is that it consists of
different global space-times~$M_\ell$ which are interconnected by relations between them.
An alternative intuitive picture is to regard~$M$ as a single space-time which is
``fine-grained'' on the microscopic scale by the sub-space-times $M_\ell$.
For the physical wave functions, the above dephasing effect means that
every physical wave function~$\psi^u(x)$ has {\em{fluctuations}} on the microscopic scale.
Moreover, comparing~$\psi^u(x)$ and~$\psi^u(y)$ for~$x$ and~$y$ in the same sub-space-time,
one finds {\em{nonlocal correlations}} on the macroscopic scale.
A detailed analysis shows that taking averages over the sub-space-times
gives rise to an effective description of the interaction in terms of multi-particle wave functions and
Fock spaces (see~\cite[Sections~5, 6 and~8]{qft}).
In particular, this gives agreement with the usual description of entanglement.

To summarize, entanglement arises naturally in the framework of causal
fermion systems when taking into account the effect of microscopic mixing.
The reader who wants to understand the concept of microscopic mixing
on a deeper quantitative level is referred to~\cite{qft}. 
Further remarks about the connection of microscopic mixing to the conservation laws which are
established in Chapter~\ref{DissNoether} are given in Section~\ref{Nremmicro}.

\section{Clarifying Remarks}\label{CClarifyingRemarks} 
This section aims to address some of the questions which might have come to the mind
of the reader.
\addtocontents{toc}{\protect\setcounter{tocdepth}{-1}}	%

\subsection{Where does the name ``causal fermion system'' come from?} \label{Ccfsname}
The term {\em{``causality''}} in the name causal fermion system refers to the fact that there are causal
relations among the space-time points (see Definition~\ref{CKS}). The causal action is ``causal''
because it vanishes for space-time points with spacelike separation.
In this way, the notion of causality is intimately connected with the framework of causal fermion systems.
The term ``fermion'' refers to the fact that a causal fermion system encodes
physical wave functions~$\psi^u(x)$ (see~\eqref{Cpsiudef}) which are interpreted as
{\em{fermionic}} wave functions (like Dirac waves).
This interpretation as fermionic wave functions is justified because,
rewriting the configuration of the physical wave functions in the
Fock space formalism, one obtains
a totally anti-symmetric multi-particle state (see~\eqref{Cantisymm}).
{\em{Bosonic}} fields appear in the causal fermion systems merely as a device to describe the collective
behavior of the fermions (see~\eqref{CCont:DiracEq}).

\subsection{Why this form of the causal action principle?} \label{Cseccap}
The first attempts to formulate a variational principle in space-time
in terms of fermionic wave functions can be found in the unpublished preprint~\cite{endlich}.
The variational principle proposed in~\cite[Section~3.5]{PFP} coincides with
the causal action principle, except that it is formulated in the setting of discrete space-times
and that the constraints~\eqref{Cvolconstraint} and~\eqref{Ctrconstraint} are missing.
The general structure of the Lagrangian~\eqref{CLagr} can be understood from the requirements that
it should be non-negative and that it should vanish for spacelike separation.
The detailed form of the Lagrangian~\eqref{CLagr} is determined uniquely by demanding that the
Dirac sea vacuum should be a stable minimizer of the variational principle
(as is made precise by the notion of ``state stability''; see~\cite[Section~5.6]{PFP}).
The necessity and significance of the constraints~\eqref{Cvolconstraint} and~\eqref{Ctrconstraint} became clear
when analyzing the existence theory~\cite{discrete, continuum} and deriving the EL equations~\cite{lagrange}.
It should also be noted that the so-called identity constraint considered in~\cite{continuum}
has turned out to be a too strong condition which is not compatible with the so-called spatial
normalization of the fermionic projector as discussed in~\cite[Section~2.2]{norm} and established
by Chapter~\ref{DissNoether} (cf. Remark~\ref{Nremnorm}).

\subsection{Why the name ``continuum limit''?}
Causal fermion systems were first analyzed in the more restrictive formulation of discrete space-times
(see~\cite[Section~3.3]{PFP}). In this setting, the continuum limit as introduced in~\cite[Chapter~4]{PFP}
arises as the limit when the discretization scale~$\varepsilon$ tends to zero, meaning that
the discrete space-time goes over to a space-time continuum.
The more general notion of causal fermion systems given here allows for the description
of both continuous and discrete space-times. Then the parameter~$\varepsilon$ should be regarded as
a regularization length, but space-time could very well be continuous on this scale.
In this more general context, the notion ``continuum limit'' merely means that
we take the limit~$\varepsilon \searrow 0$ in which space-time $M:= \supp \rho$ goes over to the
{\em{usual}} space-time continuum~$\scrM$ (i.e.\ Minkowski space or a Lorentzian manifold).

\subsection{Connection to the notion of the Dirac sea} \label{Csec73}
The concept of the Dirac sea was introduced by Dirac in order to remedy the problem of the
negative-energy solutions of the Dirac equation.
Dirac's original conception was that in vacuum all negative-energy states are occupied.
Due to the Pauli exclusion principle, additional particles must occupy states of positive energy.
This concept led to the prediction of anti-particles, which are described as ``holes'' in the sea.

If taken literally, the concept of the Dirac sea leads to problems such as an infinite negative
energy density or an infinite charge density. This is the main reason why in modern quantum field theory,
the concept of the Dirac sea is no longer apparent. It
has been replaced by Wick ordering and the reinterpretation of creation and annihilation operators
corresponding to the negative-energy states.
Therefore, it is a common view that the Dirac sea is merely a historical relic which is no longer
needed.

In the theory of causal fermion system, Dirac's original concept is revived.
Namely, when constructing a causal fermion system starting from a classical space-time
the states of the Dirac sea need to be taken into account
(cf.~\eqref{CPxyuseful} and~\eqref{CPxyvac} in the Minkowski vacuum).
This can be understood as follows. It is a general concept behind causal fermion systems
that all structures in space-time should be encoded in the physical wave functions.
This concept only works if there are ``sufficiently many'' physical wave functions.
More specifically, this is the case if the causal fermion system is composed of
a regularized Dirac sea configuration, possibly with additional particles and/or anti-particles.

In contrast to the problems in the naive Dirac sea picture, in the description with
causal fermion systems the ensemble of the sea states does {\em{not}} give rise to an infinite negative energy
density or an infinite charge density. Namely, due to the specific form of the causal action principle,
the sea states drop out of the Euler-Lagrange equations in the continuum limit.

\subsection{Connection to quantum field theory}\label{CsecQFT}
The continuum limit gives an effective description of the interaction on the level of
second-quantized fermionic fields coupled to classical bosonic fields.
A full quantum field theory, in which also the bosonic fields are quantized,
arises if the effect of microscopic mixing is taken into account.
We refer the reader to Section~\ref{CFoundations} as well as to~\cite{qft}.
The detailed analysis of the resulting Feynman diagrams, renormalization
and a comparison with standard quantum field theory is work in progress.

\subsection{Which physical principles are incorporated in a causal fermion system?}\label{CsecPrinciples}
Causal fermion systems evolved from an attempt to combine several physical principles
in a coherent mathematical framework. As a result, these principles appear in the
framework in a specific way:
\begin{itemize}[leftmargin=1.3em, itemsep=0.2em]
\itemD The {\textbf{principle of causality}} is built into a causal fermion system in a specific way,
as explained in Section~\ref{Ccfsname} above.
\itemD The {\textbf{Pauli exclusion principle}} is incorporated in a causal fermion system,
as can be seen in various ways. 
One formulation of the Pauli exclusion principle states that every fermionic one-particle state
can be occupied by at most one particle. In this formulation, the Pauli exclusion principle
is respected because every wave function can either be represented in the form~$\psi^u$
(the state is occupied) with~$u \in \H$ or it cannot be represented as a physical wave function
(the state is not occupied). But it is impossible to describe higher occupation numbers.
When working with multi-particle wave functions, the Pauli exclusion principle becomes apparent
in the total anti-symmetrization of the wave function (see~\eqref{Cantisymm}).
\itemD A {\textbf{local gauge principle}} becomes apparent once we choose
basis representations of the spin spaces and write the wave functions in components.
Denoting the signature of~$(S_x, \Sl .|. \Sr_x)$ by~$(p(x),q(x))$, we choose
a pseudo-orthonormal basis~$(\mathfrak{e}_\alpha(x))_{\alpha=1,\ldots, p+q}$ of~$S_x$.
Then a wave function~$\psi$ can be represented as
\[ \psi(x) = \sum_{\alpha=1}^{p+q} \psi^\alpha(x)\: \mathfrak{e}_\alpha(x) \]
with component functions~$\psi^1, \ldots, \psi^{p+q}$.
The freedom in choosing the basis~$(\mathfrak{e}_\alpha)$ is described by the
group~$\U(p,q)$ of unitary transformations with respect to an inner product of signature~$(p,q)$.
This gives rise to the transformations
\[ \mathfrak{e}_\alpha(x) \rightarrow \sum_{\beta=1}^{p+q} U^{-1}(x)^\beta_\alpha\;
\mathfrak{e}_\beta(x) \qquad \text{and} \qquad
\psi^\alpha(x) \rightarrow  \sum_{\beta=1}^{p+q} U(x)^\alpha_\beta\: \psi^\beta(x) \]
with $U \in \U(p,q)$.
As the basis~$(\mathfrak{e}_\alpha)$ can be chosen independently at each space-time point,
one obtains {\em{local gauge transformations}} of the wave functions,
where the gauge group is determined to be the isometry group of the spin scalar product.
The causal action is
{\em{gauge invariant}} in the sense that it does not depend on the choice of spinor bases.
\itemD The {\textbf{equivalence principle}} is incorporated in the following general way.
Space-time $M:= \supp \rho$ together with the universal measure~$\rho$ form a topological
measure space, being a more general structure than a Lorentzian manifold.
Therefore, when describing~$M$ by local coordinates, the freedom in choosing such
coordinates generalizes the freedom in choosing general reference frames in a space-time manifold.
Therefore, the equivalence principle of general relativity is respected. The causal action is {\em{generally
covariant}} in the sense that it does not depend on the choice of coordinates.
\end{itemize}

\subsection{Philosophical remarks}
Since causal fermion systems are a candidate for a unified physical theory,
one may take a consistent realist point of view and assume that our universe is a causal fermion
system. Here by ``realist point of view'' we mean that one assumes that there is a
reality independent of human observation and that one can describe this reality in a mathematical
language. ``Consistent'' means that this point of view does not lead to contradictions or
inconsistencies. Finally, by ``universe is a causal fermion system'' we mean that
the fundamental entities of our universe are the causal fermion system~$(\H,\F,\rho)$ as well as its inherent
structures.

This position could be investigated from a philosophical point of view.
We find the following points interesting:
\begin{itemize}[leftmargin=1.3em]
\itemD Space-time is a set of operators. The relations between space-time points
are all encoded in properties of products of these operators.
No additional structures need to be specified.
\itemD Similar to the picture in dynamical collapse theories, the basic object to
describe a fermion is the physical wave function~$\psi^u(x)$.
The particle character, however, comes about merely as a consequence of the dynamics
as described by the causal action principle.
\itemD The structures of space-time and matter are described in terms of a single
object: the universal measure. In particular, it is no longer possible to separate
space-time from the matter content therein.
This seems to go a step further than relativity: In relativity, space and time
do not exists separately, but are combined to space-time.
In the approach of causal fermion systems, space-time does not exist without
the matter content (including the Dirac sea). 
Space-time and the matter content are combined in one object.
\end{itemize}
A further investigation of these and related points might offer new perspectives on
questions in philosophy of physics.

\addtocontents{toc}{\protect\setcounter{tocdepth}{3}}		%

\chapter{Noether-Like Theorems}\label{DissNoether}

In this chapter, we explore symmetries and the resulting conservation laws
in the framework of causal fermion systems.
We prove that there are indeed conservation laws, which however have a
structure which is quite different from that of the classical Noether theorem.
These conservation laws are so general that they apply to ``quantum space-times'' which cannot
be approximated by a Lorentzian manifold.  We prove that in the proper limiting case,
our conservation laws simplify to charge conservation and the conservation of energy and momentum
in Minkowski space.

In order to make this chapter easily accessible and self-contained, we develop our concepts step by step.
Section~\ref{Nsecprelim} provides the necessary background:
After a brief review of the classical Noether theorem (Section~\ref{Nsecclass}), we
introduce causal variational principles in the compact setting (a mathematical simplification
of the setting of causal fermion systems) in Section~\ref{Nsecintroc} and define the concept of
surface layer integrals in Section~\ref{Nsecsli}.

After these preparations, in Section~\ref{Nseccompact} we prove conservation laws
for causal variational principles in the compact setting.
We distinguish two different kinds of symmetries: symmetries of the
Lagrangian (see Definition~\ref{Ndefsymmrho} and Theorem~\ref{Nthmsymmum})
and symmetries of the universal measure (see
Definition~\ref{Ndefsymmlagr} and Theorem~\ref{Nthmsymmlag}).
These symmetries and the corresponding conservation laws can be combined
in so-called {\em{generalized integrated symmetries}} (see Definition~\ref{Ndefgis}
and Theorem~\ref{Nthmsymmgis}).

In Section~\ref{Nseccfs} we generalize the previous results to the setting
of causal fermion systems. After a brief introduction to the mathematical setup
(Section~\ref{Nseccfsbasic}), we derive corresponding Noether-like theorems
(see Theorem~\ref{Nthmsymmgis2}, Corollary~\ref{Ncorsymmlag} and Corollary~\ref{Ncorsymmum}
in Section~\ref{Nseccfsnoether}).
In the following Sections~\ref{Nsecexcurrent} and~\ref{NsecexEM}, we work out 
examples which give the correspondence to current conservation (Theorem~\ref{Nthmcurrentmink})
and to the conservation of energy-momentum (Corollary~\ref{NcorEMcons}) in a limiting case.
In Section~\ref{Nsecremark}, the mathematical assumptions and the physical picture is discussed and
clarified by a few remarks. 
In Section~\ref{Nsecexrho} it is explained why the conservation laws corresponding to
symmetries of the universal measure are trivially satisfied in Minkowski space and do not capture any
interesting dynamical information.
Finally, in Section~\ref{Nremmicro}, we explain the relation to
the mechanism of microscopic mixing of the wave functions (as introduced in~\cite[Section~3]{qft} and briefly
explained in Section~\ref{CFoundations}). This chapter has been published with minor modifications as~\cite{noether}.

\section{Preliminaries} \label{Nsecprelim}
\subsection{The Classical Noether Theorem} \label{Nsecclass}
We now briefly review Noether's theorem~\cite{noetheroriginal} in the form most suitable for our
purposes (similar formulations are found in~\cite[Section~13.7]{goldstein} or~\cite[Chapter~III]{barutbook}).
For simplicity, we begin in four-dimensional Minkowski space~$\scrM$.
In the Lagrangian formulation of classical field theory, one seeks for 
critical points of an action of the form
\[ \Sact = \int_{\scrM} \L \big( \psi(x), \psi_{,j}(x), x\big)\: d^4x \]
(where~$\psi$ is for example a scalar, tensor or spinor field, and~$\psi_{,j} \equiv \partial_j \psi$
denotes the partial derivative).
The critical field configurations satisfy the Euler-Lagrange (EL) equations
\beq \label{NELclass}
\frac{\partial \L}{\partial \psi} - \frac{\partial}{\partial x^j} \left( \frac{\partial \L}{\partial \psi_{,j}} \right) = 0 \:.
\eeq
Symmetries are formulated in terms of variations of the
field and the space-time coordinates. More precisely, for given~$\tau_{\max}>0$
we consider smooth families~$(\psi_\tau)$ and~$(x_\tau)$ parametrized
by~$\tau \in (-\tau_{\max}, \tau_{\max})$ with~$\psi_\tau|_{\tau=0}=\psi$ and~$x_\tau|_{\tau=0}=x$.
We assume that these variations describe a {\em{symmetry of the action}},
meaning that for every compact space-time region~$\Omega \subset \scrM$ and
every field configuration~$\psi$ the equation
\beq \label{Nsymm}
\int_\Omega \L\big( \psi(x), \psi_{,j}(x), x\big)\: d^4x
= \int_{\Omega'} \L\big( \psi_\tau(y), (\psi_\tau)_{,j}(y), y \big)\: d^4y
\eeq
holds for all~$\tau \in (-\tau_{\max}, \tau_{\max})$,
where~$\Omega' = \{x_\tau \,|\, x \in \Omega\}$ is the transformed region.
The corresponding {\em{Noether current}}~$J$ is defined by
\[ J^k = \frac{\partial \L}{\partial \psi_{,k}} \,\delta \psi
+ \L\: \delta x^k - \frac{\partial \L}{\partial \psi_{,k}}\: \partial_j \psi\: \delta x^j \:, \]
where~$\delta x$ and~$\delta \psi$ are the first variations
\[ \delta x := \frac{d}{d\tau}\, x_\tau |_{\tau=0} \qquad \text{and} \qquad
\delta \psi(x) := \frac{d}{d\tau}\, \psi_\tau(x_\tau) |_{\tau=0} \:. \]
Noether's theorem states that if~$\psi$ satisfies the EL equations, then
the Noether current is divergence-free,
\[ \partial_k J^k = 0 \:. \]
Using the Gau{\ss} divergence theorem, one may integrate this equation to obtain
a corresponding conserved quantity. To this end, one chooses a space-time region~$\Omega$
whose boundary~$\partial \Omega$ consists of two space-like hypersurfaces~$\scrN_1$
and~$\scrN_2$. Then\Chd{$d\mu_{\scrN_1}(x) \rightarrow d\mu_{\scrN_1}$}%
\beq \label{Nconserve}
\int_{\scrN_1} J^k \nu_k\: d\mu_{\scrN_1} = \int_{\scrN_2} J^k \nu_k\: d\mu_{\scrN_2} \:,
\eeq
where~$\nu$ is the future-directed normal, and~$d\mu_{\scrN_{1\!/\!2}}$ is the induced volume measure
(if~$\Omega$ is unbounded, one needs to assume suitable decay of~$J^k$ at infinity).

We now mention two well-known applications of Noether's theorem which will be most relevant here.
The first application is to consider the Lagrangian of a quantum mechanical wave function~$\psi$
(like the Schr\"odinger, Klein-Gordon or Dirac Lagrangian) and to consider global phase transformations of
the wave function,
\beq \label{Nglobalphase}
\psi_\tau(x) = e^{i \tau} \psi(x) \:, \qquad x_\tau = x \:.
\eeq
Then the symmetry condition~\eqref{Nsymm} is satisfied because the
Lagrangian depends only on the modulus of~$\psi$. 
The corresponding Noether current is the probability current,
giving rise to {\em{current conservation}}.
We remark that, if the quantum mechanical wave function is coupled to an electromagnetic field,
then this current coincides, up to a multiplicative constant, with the electromagnetic current of the
particle. Therefore, the conservation law can also be interpreted as the {\em{conservation of electric
charge}}. The second application is to consider translations in space-time, i.e.
\[ \psi_\tau(x) = \psi(x) \:, \qquad x_\tau = x + \tau v \]
with a fixed vector~$v \in \scrM$. In this case, the symmetry condition~\eqref{Nsymm}
is satisfied if we assume that~$\L=\L(\phi, \phi_{,j})$ does not depend explicitly on~$x$.
After a suitable symmetrization procedure (see~\cite[\S32 and~\S94]{landau2} or the
systematic treatment in~\cite{forger+roemer}),
the corresponding Noether current can be written as
\[ J^k = T^{kj} v_j \:, \]
where~$T_{jk}$ is the energy-momentum tensor.
Noether's theorem yields the {\em{conservation of energy and momentum}}.

Noether's theorem also applies in curved space-time. In this case, the Lagrangian
involves the Lorentzian metric. As a consequence, the symmetry condition~\eqref{Nsymm}
implies that the metric must be invariant under the variation~$x_\tau$.
This is made precise by the notion of a {\em{Killing field}}~$K$, being a
vector field which satisfies the Killing equation
\[ \nabla_i K_j = - \nabla_j K_i \]
(see for example~\cite[Section~2.6]{hawking+ellis} or~\cite[Section~1.9]{straumann}).
If space-time admits a Killing field~$K$, the corresponding Noether current is
most conveniently constructed as follows. 
As a consequence of the Einstein equations, the energy-momentum tensor is divergence-free,
\[ \nabla_j T^{jk} = 0 \:. \]
This by itself does not give rise to conserved quantities because the Gau{\ss} divergence theorem
only applies to vector fields, but not to tensor fields.
However, a direct computation shows that contracting the energy-momentum tensor
with the Killing field,
\[ J^k := T^{kj} K_j \:, \]
gives rise to a divergence-free vector field (see~\cite[Section~3.2]{hawking+ellis}
or~\cite[Section~2.4]{straumann}). Hence integration again gives a conservation law of the form~\eqref{Nconserve}.

\subsection{Causal Variational Principles in the Compact Setting} \label{Nsecintroc}
We now introduce the setting of causal variational principles in the compact case,
slightly generalizing the presentation in~\cite[Section~1.2]{support}.
Let~$\F$ be a smooth {\em{compact}} manifold and~$\L \in C^{0,1}(\F \times \F, \R^+_0)$
a non-negative Lipschitz-continuous function which is symmetric, i.e.
\beq \label{NsymmL}
\L(x,y) = \L(y,x) \qquad \text{for all~$x,y \in \F$}\:.
\eeq
The {\em{causal variational principle}} is to minimize the action~$\Sact$ defined by
\beq \label{NSdefcompact}
\Sact(\rho) = \iint_{\F \times \F} \L(x,y)\: d\rho(x)\: d\rho(y)
\eeq
under variations of~$\rho$ in the class of (positive) normalized regular Borel measures.
The existence of minimizers follows immediately from abstract compactness arguments
(see~\cite[Section~1.2]{continuum}).

In what follows, we let~$\rho$ be a given minimizing measure, referred to as the {\em{universal measure}}.
The resulting EL equations are derived in~\cite[Section~3.1]{support}. For the sake of self-consistency,
we now state them and repeat the proof.
\begin{Lemma} {\Thmt{(Euler-Lagrange equations)}} \label{NlemmaEL}
Let~$\rho$ be a minimizing measure of the causal variational principle~\eqref{NSdefcompact}.
Then the function~$\ell \in C^{0,1}(\F)$ defined by
\beq
\ell(x) = \int_\F \L(x,y)\: d\rho(y) \label{Nldef}
\eeq
is minimal on the support of~$\rho$,
\beq \label{NEL1}
\ell|_{\supp \rho} \,\equiv\, \inf_\F \ell \:.
\eeq
\end{Lemma}
\noindent We remark that in Chapters~\ref{DissJet} and~\ref{DissStoch},
we add a constant $-\frac{\nu}{2}$ to the right hand side of~\eqref{Nldef} which we choose
such that $ \inf_\F \ell = 0$ (compare e.g.~\eqref{Jelldef} and~\eqref{JELstrong}). However,
in the present context, this additional property of $\ell$ is not necessary.
\Proof Carrying out one of the integrals, one sees that
\beq \label{NSl}
\Sact(\rho) = \iint_{\F \times \F} \L(x,y)\, d\rho(x)\: d\rho(y) = \int_\F \ell\: d\rho \:.
\eeq
Since~$\ell$ is continuous and~$\F$ is compact, there clearly is~$y \in \F$ with
\[ \ell(y) = \inf_\F \ell \:. \]
We consider for~$\tau \in [0,1]$ the family of normalized regular Borel measures
\[ \tilde{\rho}_\tau = (1-\tau)\, \rho + \tau \, \delta_y \:, \]
where~$\delta_y$ denotes the Dirac measure supported at~$y$. Applying this formula in~\eqref{NSdefcompact}
and differentiating, we obtain for the first variation
\[ \delta \Sact := \lim_{t \searrow 0} \frac{\Sact \big(\tilde{\rho}_\tau \big) - \Sact\big(\tilde{\rho}_0 \big)}{\tau}
= -2 \Sact(\rho) + 2 \ell(y)\:. \]
Since~$\rho$ is a minimizer, $\delta \Sact$ is non-negative. Hence
\[ \inf_\F \ell = \ell(y) \:\geq\: \Sact(\rho) \overset{\eqref{NSl}}{=} \int_\F \ell\: d\rho \:. \]
It follows that~$\ell$ is constant on the support of~$\rho$, giving the result.
\QED

As explained in detail in Chapter~\ref{DissIntroCFS}, the physical picture is that the universal measure gives rise to a space-time and also induces all
the objects therein. In the compact setting considered here, one only obtains space-time
endowed with a causal structure in the following way.
Space-time is defined as the support of the universal measure,
\[ \text{\em{space-time}} \qquad M:= \supp \rho \:. \]
For a space-time point~$x \in M$,
we define the open {\em{light cone}} ${\mathcal{I}}(x)$ and the
closed light cone~${\mathcal{J}}(x)$ by
\[ {\mathcal{I}}(x) = \{ y \in M \:|\: \L(x,y) > 0 \} \qquad \text{and} \qquad
 {\mathcal{J}}(x) = \overline{{\mathcal{I}}(x)}\:. \]
This makes it possible to define a {\em{causal structure}} on space-time
by saying that two space-time points~$x, y \in M$ are
{\em{timelike}} separated if $\L(x,y)>0$ and {\em{spacelike}} separated if~$\L(x,y)=0$.
We remark that, in the setting of causal fermion systems, these notions
indeed agree with the usual notion of causality in Minkowski space or on a globally
hyperbolic manifold (cf. Section~\ref{Csecminkvac} and~\cite{cfs}).

\subsection{The Concept of Surface Layer Integrals} \label{Nsecsli}
It is not at all obvious how the classical Noether theorem should be generalized
to causal variational principles. First, the mathematical structure of the
EL equations~\eqref{NEL1} is completely different from that
of the classical EL equations~\eqref{NELclass}.
Moreover, for writing down surface integrals as in~\eqref{Nconserve}
one needs structures like the Lorentzian metric as well as the normal to a hypersurface
and the induced volume measure thereon. All these structures are
not directly available in the setting of causal variational principles.
Therefore, it is a priori not clear how conservation laws should be stated.

The first task is to introduce an analog of the surface integral in~\eqref{Nconserve}.
The only objects to our disposal are the Lagrangian~$\L(x,y)$ and the universal measure~$\rho$.
We make the assumption that the Lagrangian is of {\em{short range}} in the following sense.
We let~$d \in C^0(M \times M, \R^+_0)$ be a distance function on~$M$
(since~$M$ is compact, any two such distance functions are equivalent). The assumption
of short range means that~$\L$ vanishes on distances larger than~$\delta$, i.e.
\beq \label{Nshortrange}
d(x,y) > \delta \quad \Longrightarrow \quad \L(x,y) = 0
\eeq
Then a double integral of the form
\beq \label{Nintdouble}
\int_\Omega \bigg( \int_{M \setminus \Omega} \cdots\: \L(x,y)\: d\rho(y) \bigg)\, d\rho(x)
\eeq
only involves pairs~$(x,y)$ of distance at most~$\delta$,
where~$x$ is in~$\Omega$ and~$y$ is in the complement~$M \setminus \Omega$.
Thus the integral only involves points in a layer around the boundary of~$\Omega$
of width~$\delta$, i.e.
\[ x, y \in B_\delta \big(\partial \Omega \big) \:. \]
Therefore, a double integral of the form~\eqref{Nintdouble} can be regarded as an approximation
of a surface integral on the length scale~$\delta$, as shown in Figure~\ref{Nfignoether1}.
\begin{figure}\centering
\psscalebox{1.0 1.0} %
{
\begin{pspicture}(0,-1.511712)(10.629875,1.511712)
\definecolor{colour0}{rgb}{0.8,0.8,0.8}
\definecolor{colour1}{rgb}{0.6,0.6,0.6}
\pspolygon[linecolor=black, linewidth=0.002, fillstyle=solid,fillcolor=colour0](6.4146066,0.82162136)(6.739051,0.7238436)(6.98794,0.68384355)(7.312384,0.66162133)(7.54794,0.67939913)(7.912384,0.7593991)(8.299051,0.8705102)(8.676828,0.94162136)(9.010162,0.9549547)(9.312385,0.9371769)(9.690162,0.8571769)(10.036829,0.7371769)(10.365718,0.608288)(10.614607,0.42162135)(10.614607,-0.37837866)(6.4146066,-0.37837866)
\pspolygon[linecolor=black, linewidth=0.002, fillstyle=solid,fillcolor=colour1](6.4146066,1.2216214)(6.579051,1.1616213)(6.770162,1.1127324)(6.921273,1.0905102)(7.103495,1.0816213)(7.339051,1.0549546)(7.530162,1.0638436)(7.721273,1.0993991)(7.8857174,1.1393992)(8.10794,1.2060658)(8.299051,1.2549547)(8.512384,1.3038436)(8.694607,1.3260658)(8.890162,1.3305103)(9.081273,1.3393991)(9.379051,1.3216213)(9.659051,1.2593992)(9.9746065,1.1705103)(10.26794,1.0460658)(10.459051,0.94384354)(10.614607,0.82162136)(10.610162,0.028288014)(10.414606,0.1660658)(10.22794,0.26828802)(10.010162,0.37051025)(9.663495,0.47273245)(9.356829,0.53051025)(9.054606,0.548288)(8.814607,0.54384357)(8.58794,0.5171769)(8.387939,0.48162135)(8.22794,0.44162133)(7.90794,0.34828803)(7.6946063,0.29939914)(7.485718,0.26828802)(7.272384,0.26828802)(7.02794,0.28162134)(6.82794,0.3171769)(6.676829,0.35273245)(6.543495,0.38828802)(6.4146066,0.42162135)
\pspolygon[linecolor=black, linewidth=0.002, fillstyle=solid,fillcolor=colour0](0.014606438,0.82162136)(0.3390509,0.7238436)(0.5879398,0.68384355)(0.9123842,0.66162133)(1.1479398,0.67939913)(1.5123842,0.7593991)(1.8990508,0.8705102)(2.2768288,0.94162136)(2.610162,0.9549547)(2.9123843,0.9371769)(3.290162,0.8571769)(3.6368287,0.7371769)(3.9657176,0.608288)(4.2146063,0.42162135)(4.2146063,-0.37837866)(0.014606438,-0.37837866)
\psbezier[linecolor=black, linewidth=0.04](6.4057174,0.8260658)(7.6346064,0.45939913)(7.8634953,0.8349547)(8.636828,0.92828804)(9.410162,1.0216213)(10.165717,0.7927325)(10.614607,0.42162135)
\psbezier[linecolor=black, linewidth=0.04](0.005717549,0.8260658)(1.2346064,0.45939913)(1.4634954,0.8349547)(2.2368286,0.92828804)(3.0101619,1.0216213)(3.7657175,0.7927325)(4.2146063,0.42162135)
\rput[bl](2.0101619,0.050510235){$\Omega$}
\rput[bl](8.759051,0.0016213481){\normalsize{$\Omega$}}
\psline[linecolor=black, linewidth=0.04, arrowsize=0.09300000000000001cm 1.0,arrowlength=1.7,arrowinset=0.3]{->}(1.9434953,0.85495466)(1.8057176,1.6193991)
\rput[bl](2.0946064,1.1705103){$\nu$}
\psbezier[linecolor=black, linewidth=0.02](6.4146066,0.42384356)(7.6434956,0.057176903)(7.872384,0.43273246)(8.645718,0.52606577)(9.419051,0.61939913)(10.174606,0.39051023)(10.623495,0.019399125)
\psbezier[linecolor=black, linewidth=0.02](6.410162,1.2193991)(7.639051,0.8527325)(7.86794,1.228288)(8.6412735,1.3216213)(9.414606,1.4149547)(10.170162,1.1860658)(10.619051,0.8149547)
\rput[bl](8.499051,0.9993991){\normalsize{$y$}}
\rput[bl](7.8657174,0.49273247){\normalsize{$x$}}
\psdots[linecolor=black, dotsize=0.06](8.170162,0.65273243)
\psdots[linecolor=black, dotsize=0.06](8.796828,1.1327325)
\psline[linecolor=black, linewidth=0.02](6.1146064,1.2216214)(6.103495,0.82162136)
\rput[bl](5.736829,0.8993991){\normalsize{$\delta$}}
\rput[bl](3.6146064,0.888288){$\scrN$}
\rput[bl](1.1146064,-1.4117119){$\displaystyle \int_\scrN \cdots\, d\mu_\scrN$}
\rput[bl](5.7146063,-1.511712){$\displaystyle \int_\Omega d\rho(x) \int_{M \setminus \Omega} d\rho(y)\: \cdots\:\L(x,y)$}
\psline[linecolor=black, linewidth=0.02](6.0146065,1.2216214)(6.2146063,1.2216214)
\psline[linecolor=black, linewidth=0.02](6.0146065,0.82162136)(6.2146063,0.82162136)
\end{pspicture}
}
\caption{A surface integral and a corresponding surface layer integral.}
\label{Nfignoether1}
\end{figure}
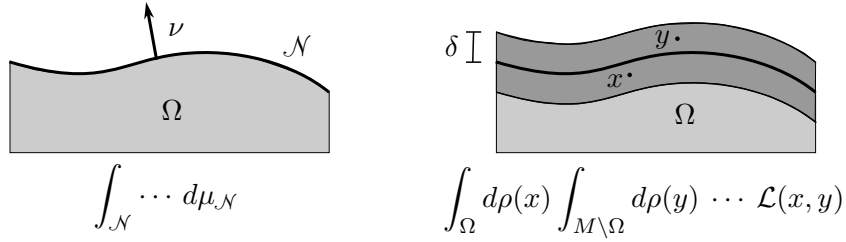
We refer to integrals of the form~\eqref{Nintdouble} as {\em{surface layer integrals}}.
In the setting of causal variational principles, they take the role of surface integrals
in Lorentzian geometry.
Our strategy is to find expressions for the integrand ``\ldots'' in~\eqref{Nintdouble} such that the
surface layer integral vanishes. Choosing~$\Omega$ as a space-time region such that~$\partial \Omega$
has two connected components~$\scrN_1$ and~$\scrN_2$, one then obtains
a conservation law similar to~\eqref{Nconserve}, with the surface integrals replaced by
corresponding surface layer integrals.

We remark for clarity that the correspondence between surface integrals and surface layer integrals
could be made mathematically precise by taking the limit~$\delta \searrow 0$.
However, this would make it necessary to consider a family of Lagrangians~$\L_\delta$
together with corresponding minimizers~$\rho_\delta$.
This seems an interesting technical problem for the future.
For our purposes, it suffices to identify the surface layer integrals~\eqref{Nintdouble}
as the objects which replace the usual surface integrals.

We finally remark that, in the physical setting of causal fermion systems, the
condition of short range~\eqref{Nshortrange} will be replaced by the weaker requirement
that the main contribution to the double integral~\eqref{Nintdouble} comes from pairs
of points~$(x,y)$ whose distance is at most~$\delta$. This will be explained in 
detail in Section~\ref{Nseccurcor}, where will also identify the length scale~$\delta$
with the Compton scale (see the paragraph after after~\eqref{NUSymm20}). 

\section{Noether-Like Theorems in the Compact Setting} \label{Nseccompact}
We now derive Noether-like theorems in the compact setting.
We consider two different symmetries: symmetries of the Lagrangian (Theorem~\ref{Nthmsymmlag})
and symmetries of the universal measure (Theorem~\ref{Nthmsymmum}).
In Section~\ref{Nsecsymmgis}, these symmetries will be combined in
the notion of generalized integrated symmetries (Theorem~\ref{Nthmsymmgis}).

\subsection{Symmetries of the Lagrangian} \label{Nsecsymmlag}
The assumption~\eqref{Nsymm} can be understood as a symmetry condition
for the Lagrangian. We now want to impose a similar symmetry condition for the
Lagrangian~$\L(x,y)$ of a causal variational principle. The most obvious method would be to
consider a one-parameter group of diffeomorphisms~$\Phi_\tau$,
\beq \label{Ntry}
\Phi : \R \times \F \rightarrow \F \qquad \text{with} \qquad
\Phi_\tau \Phi_{\tau'} = \Phi_{\tau+\tau'}
\eeq
and to impose that~$\L$ be invariant under these diffeomorphisms in the sense that
\beq \label{NsymmF1}
\L(x,y) = \L \big( \Phi_\tau(x), \Phi_\tau(y) \big) \qquad \text{for all~$\tau \in \R$
and~$x, y \in \F$\:.}
\eeq
However, this condition is unnecessarily strong for two reasons. First, it suffices to consider families which
are defined locally for~$\tau \in (-\tau_{\max}, \tau_{\max})$. Second, the mapping~$\Phi$ does not
need to be defined on all of~$\F$. Instead, it is more appropriate to impose the symmetry condition
only on space-time~$M \subset \F$.
This leads us to  consider instead of~\eqref{Ntry} a mapping
\beq \label{NPhidef}
\Phi : (-\tau_{\max}, \tau_{\max}) \times M \rightarrow \F
\qquad \text{with} \qquad \Phi(0,.) = \1 \:.
\eeq
We also write~$\Phi_\tau(x) \equiv \Phi(\tau,x)$ and refer to~$\Phi_\tau$
as a {\textbf{variation}} of~$M$ in~$\F$.
Next, we need to specify what we mean by ``smoothness'' of this variation.
This is a subtle point because in view of the results in~\cite{support},
the universal measure does not need to be smooth (in the sense that it cannot in general be written as
a smooth function times the Lebesgue measure), and therefore the function~$\ell$
will in general only be Lipschitz continuous.
Our Noether-like theorems require only that the function~$\ell$ be differentiable in the
direction of the variations:\medskip
\begin{Def} \label{Ndefsymm} A variation~$\Phi_\tau$ of the form~\eqref{NPhidef} is
{\textbf{continuously differentiable}} if the composition
\[ \ell \circ \Phi \::\: (-\tau_{\max}, \tau_{\max}) \times M \rightarrow \R \]
is continuous and if its partial derivative~$\partial_\tau (\ell \circ \Phi)$ exists
and is continuous.
\end{Def}
The next question is how to adapt the symmetry condition~\eqref{NsymmF1} to the mapping~$\Phi$
defined only on~$(-\tau_{\max}, \tau_{\max}) \times M$.
This is not obvious because setting~$\tilde{x} = \Phi_\tau(x)$ and using the group property,
the condition~\eqref{NsymmF1} can be written equivalently as
\beq \label{NsymmF2}
 \L \big( \Phi_{-\tau}(\tilde{x}), y \big) = \L \big( \tilde{x},\Phi_\tau(y) \big) \qquad \text{for all~$\tau \in \R$
 and~$\tilde{x}, y \in \F$\:.}
\eeq
But if we restrict attention to pairs~$x,y \in M$, the equations in~\eqref{NsymmF1} and~\eqref{NsymmF2}
are different. It turns out that the correct procedure is to work with the expression in~\eqref{NsymmF2}.

\begin{Def} \label{Ndefsymmlagr} A variation~$\Phi_\tau$ of the form~\eqref{NPhidef} is a {\textbf{symmetry
of the Lagrangian}} if
\beq \label{Nsymmlagr}
\L \big( x, \Phi_\tau(y) \big) = \L \big( \Phi_{-\tau}(x), y \big)
\qquad \text{for all~$\tau \in (-\tau_{\max}, \tau_{\max})$
and~$x, y \in M$\:.}
\eeq
\end{Def}

We now state our first Noether-like theorem.
\begin{Thm} \label{Nthmsymmlag} Let~$\rho$ be a minimizing measure and $\Phi_\tau$ a 
continuously differentiable symmetry of the Lagrangian. Then for any compact subset~$\Omega \subset M$,
we have
\beq
\label{NconservationLagrEq}
\frac{d}{d\tau} \int_\Omega d\rho(x) \int_{M \setminus \Omega} d\rho(y)\:
\Big( \L \big( \Phi_\tau(x),y \big) - \L \big( \Phi_{-\tau}(x), y \big) \Big) \Big|_{\tau=0} = 0 \:.
\eeq
\end{Thm}
Before coming to the proof, we explain the connection to surface layer integrals.
To this end, let us assume that~$\Phi_\tau$ and the Lagrangian are differentiable in the sense that the derivatives
\beq \label{Ndifferentiable}
\frac{d}{d\tau} \Phi_\tau(x)\big|_{\tau=0} =: u(x) \qquad \text{and} \qquad
\frac{d}{d\tau} \L\big( \Phi_\tau(x),y \big) \big|_{\tau=0}
\eeq
exist for all~$x, y \in M$ and are continuous on~$M$ respectively~$M \times M$.
Then one may exchange differentiation and integration in~\eqref{NconservationLagrEq}
and apply the chain rule to obtain
\[ \int_\Omega d\rho(x) \int_{M \setminus \Omega} d\rho(y)\: D_{u(x)} \L(x,y) = 0 \:, \]
where~$D_{u(x)}$ is the derivative in the direction of the vector field~$u(x)$.
This expression is a surface layer integral as in~\eqref{Nintdouble}.
In general, the derivatives in~\eqref{Ndifferentiable} need {\em{not}} exist, because
we merely imposed the weaker differentiability assumption of Definition~\ref{Ndefsymm}.
In this case, the statement of the theorem implies that the derivative 
of the integral in~\eqref{NconservationLagrEq} exists and vanishes.

\Proof[Proof of Theorem~\ref{Nthmsymmlag}.] We multiply~\eqref{Nsymmlagr} by a bounded measurable function~$f$ on~$M$
and integrate. This gives
\begin{align*}
0 &= \iint_{M \times M} f(x)\, f(y)\: \Big(
\L \big( x, \Phi_\tau(y) \big) - \L \big( \Phi_{-\tau}(x), y \big) \Big)\, d\rho(x)\, d\rho(y) \\
&=\iint_{M \times M} f(x)\, f(y)\: \Big(
\L \big( \Phi_\tau(x),y \big) - \L \big( \Phi_{-\tau}(x), y \big) \big) \Big) \, d\rho(x)\, d\rho(y)\:,
\end{align*}
where in the last step we used the symmetry of the Lagrangian~\eqref{NsymmL}
and the symmetry of the integrand in~$x$ and~$y$.
We replace~$f(y)$ by~$1 - (1-f(y))$, multiply out and use the definition of~$\ell$, \eqref{Nldef}.
We thus obtain
\begin{align*}
0 &=\int_M f(x)\: 
\Big(  \ell \big( \Phi_\tau(x)\big) - \ell\big( \Phi_{-\tau}(x)) \Big) \,d\rho(x) \\
&\quad -\iint_{M \times M} f(x)\, \big(1-f(y) \big)\:
\Big( \L \big( \Phi_\tau(x),y \big) - \L \big( \Phi_{-\tau}(x), y \big) \Big) \, d\rho(x)\, d\rho(y)\:.
\end{align*}
Choosing~$f$ as the characteristic function of~$\Omega$, we obtain the identity
\beq \begin{split}
\int_\Omega &d\rho(x) \: \int_{M \setminus \Omega} d\rho(y)\:
\Big( \L \big( \Phi_\tau(x), y \big) - \L \big(\Phi_{-\tau}(x),  y \big) \Big) \\
&= \int_\Omega \Big( \ell \big( \Phi_\tau(x) \big)  - \ell \big( \Phi_{-\tau}(x) \big) \Big)\: d\rho(x) \:.
\end{split} \label{Nform}
\eeq
Using that~$\ell(\Phi_\tau(x))$ is continuously differentiable (see Definition~\ref{Ndefsymm})
and that~$\Omega$ is compact, we conclude that
the right side of this equation is differentiable at~$\tau=0$. Moreover,
we are allowed to exchange the $\tau$-differentiation with integration.
The EL equations~\eqref{NEL1} imply that
\beq \label{Nlfirst}
\frac{d}{d\tau} \ell \big( \Phi_{\tau}(x) \big) \Big|_{\tau=0} = 0 =
\frac{d}{d\tau} \ell \big( \Phi_{-\tau}(x) \big) \Big|_{\tau=0} \:.
\eeq
Hence the right side of~\eqref{Nform} is differentiable at~$\tau=0$,
and the derivative vanishes. This gives the result.
\QED

\subsection{Symmetries of the Universal Measure} \label{Nsecsymmum}
We now prove a conservation law for a different type of symmetry.

\begin{Def} \label{Ndefsymmrho} A variation~$\Phi_\tau$
of the form~\eqref{NPhidef} is a {\textbf{symmetry of the universal measure}} if
\beq \label{Nsymmrho}
(\Phi_\tau)_* \rho = \rho
\qquad \text{for all~$\tau \in (-\tau_{\max}, \tau_{\max})$\:.}
\eeq
\end{Def} \noindent
Here~$(\Phi_\tau)_* \rho$ is the push-forward measure
(defined by~$((\Phi_\tau)_* \rho)(\Omega) := \rho(\Phi_\tau^{-1}(\Omega))$).

\begin{Thm} \label{Nthmsymmum} Let~$\rho$ be a minimizing measure and $\Phi_\tau$  be a
continuously differentiable symmetry of the universal measure. Then for any compact subset~$\Omega \subset M$,
\[ \frac{d}{d\tau} \int_\Omega d\rho(x) \int_{M \setminus \Omega} d\rho(y)\: \Big( 
\L \big( \Phi_\tau(x), y \big)  - \L \big( x, \Phi_\tau(y) \big) \Big) \Big|_{\tau=0} = 0 \:. \]
\end{Thm}
\Proof We again let~$f$ be a bounded measurable function on~$M$. Then, by symmetry
in~$x$ and~$y$,
\[ \iint_{M \times M} f(x)\,f(y)\: \Big( \L \big( \Phi_\tau(x), y \big) -\L \big( x, \Phi_\tau(y) \big) \Big)
\: d\rho(x)\, d\rho(y) = 0 \:. \]
We replace~$f(y)$ by~$1 - (1-f(y))$ and multiply out. The double integrals which do not involve~$f(y)$
can be simplified as follows,
\begin{align*}
&\iint_{M \times M} f(x)\, \L \big( \Phi_\tau(x), y \big) \: d\rho(x)\, d\rho(y) = \int_M f(x)\:\ell \big( \Phi_\tau(x) \big) 
\: d\rho(x) \\
&\iint_{M \times M} f(x)\, \L \big( x, \Phi_\tau(y) \big) \: d\rho(x)\, d\rho(y) =
\iint_{\F \times \F} f(x)\, \L \big( x, \Phi_\tau(y) \big) \: d\rho(x)\, d\rho(y) \\
&=\iint_{\F \times \F} f(x)\, \L(x,y) \: d\rho(x)\, d \big((\Phi_\tau)_*\rho \big)(y) 
\overset{(\star)}{=} \iint_{\F \times \F} f(x)\, \L(x,y) \: d\rho(x)\, d \rho (y) \\
&= \iint_{M \times M} f(x)\, \L(x,y) \: d\rho(x)\, d \rho (y)
= \int_M f(x)\, \ell(x)\, d\rho(x) \:,
\end{align*}
where in~($\star$) we used the symmetry assumption~\eqref{Nsymmrho}.
We thus obtain
\begin{align*}
0 &= -\iint_{M \times M} f(x)\,\big(1-f(y) \big)\: \Big( \L \big( \Phi_\tau(x), y \big) 
- \L \big( x, \Phi_\tau(y) \big) \Big)\: d\rho(x)\, d\rho(y) \\
&\quad + \int_M f(x)\, \Big( \ell \big( \Phi_\tau(x) \big) - \ell(x) \Big) \: d\rho(x)\:.
\end{align*}
Choosing~$f$ as the characteristic function of~$\Omega$ gives
\[ \int_\Omega d\rho(x) \int_{M \setminus \Omega} d\rho(y)\: \Big( \L \big( \Phi_\tau(x), y \big) -
\L \big( x, \Phi_\tau(y) \big) \Big)
= \int_\Omega \Big(  \ell \big( \Phi_\tau(x) \big) - \ell(x) \Big) \: d\rho(x) \:. \]
Now the $\tau$-derivative can be computed just as in the proof of Theorem~\ref{Nthmsymmlag}.
\QED

\subsection{Generalized Integrated Symmetries} \label{Nsecsymmgis}
We now combine the symmetries of the previous sections
in the notion of ``generalized integrated symmetries.''
Our method is based on the following simple but useful identity.
\begin{Prp} \label{Nprpuseful} Let~$\Phi_\tau$ be a variation of the form~\eqref{NPhidef}. Then
\begin{align}
\int_M d\rho&(x) \int_\Omega d\rho(y) \Big( \L\big( \Phi_\tau(x), y\big) - \L(x,y) \Big) \label{Nid1} \\
=\:& \int_\Omega \Big( \ell \big( \Phi_\tau(x)\big) - \ell(x) \Big)\, d\rho(x) \label{Nid2} \\
&- \int_\Omega d\rho(x) \int_{M \setminus \Omega} d\rho(y) \:\Big( \L\big( \Phi_\tau(x), y\big) -
\L\big( x, \Phi_\tau(y) \big) \Big) \:. \label{Nid3}
\end{align}
\end{Prp}
\Proof We rewrite the integration domains as follows,
\begin{align}
\int_M & d\rho(x) \int_\Omega d\rho(y) \Big( \L\big( \Phi_\tau(x), y\big) - \L(x,y) \Big) \notag \\
&= \int_\Omega d\rho(x) \int_\Omega d\rho(y) \Big( \L\big( \Phi_\tau(x), y\big) - \L(x,y) \Big) \notag \\
&\quad+ \int_{M \setminus \Omega} d\rho(x) \int_\Omega d\rho(y) \Big( \L\big( \Phi_\tau(x), y\big) - \L(x,y) \Big) \notag \\
&= \int_\Omega d\rho(x) \int_M d\rho(y) \Big( \L\big( \Phi_\tau(x), y\big) - \L(x,y) \Big) \label{Nint1} \\
&\quad- \int_\Omega d\rho(x) \int_{M \setminus \Omega} d\rho(y) \Big( \L\big( \Phi_\tau(x), y\big) - \L(x,y) \Big) \\
&\quad+ \int_{M \setminus \Omega} d\rho(x) \int_\Omega d\rho(y) \Big( \L\big( \Phi_\tau(x), y\big) - \L(x,y) \Big) \:.
\label{Nint3}
\end{align}
In~\eqref{Nint1} we can carry out the $y$-integration using~\eqref{Nldef}.
In~\eqref{Nint3} we exchange the integrals and use that the Lagrangian is symmetric in its
two arguments~\eqref{NsymmL}. This gives the result.
\QED

Note that the term~\eqref{Nid3} is a surface layer integral. The term~\eqref{Nid2}, on the
other hand, only involves~$\ell$, and therefore its first variation vanishes in view
of the EL equations~\eqref{NEL1}. We thus obtain a conservation law, provided that
the term~\eqref{Nid1} vanishes. This motivates the following definition.

\begin{Def} \label{Ndefgis}
A variation~$\Phi_\tau$ of the form~\eqref{NPhidef} is a {\textbf{generalized integrated symmetry}} in
the space-time region~$\Omega \subset M$ if
\beq \label{Nsymmgis}
\int_M d\rho(x) \int_\Omega d\rho(y) \Big( \L\big( \Phi_\tau(x), y\big) - \L(x,y) \Big)  = 0 \:.
\eeq
\end{Def}
This notion of symmetry indeed generalizes our previous notions of symmetry
(see Definitions~\ref{Ndefsymmlagr} and~\ref{Ndefsymmrho}) in the sense that
symmetries of the Lagrangian and of the universal measure imply that~\eqref{Nsymmgis}
holds for first variations. Namely, if~$\Phi_\tau$ is a symmetry of the universal measure,
we can use~\eqref{Nsymmrho} to obtain
\beq \begin{split}
&\int_M d\rho(x) \int_\Omega d\rho(y) \Big( \L\big( \Phi_\tau(x), y\big) - \L(x,y) \Big) \\
&= \int_\F d \big( (\Phi_{\tau})_* \rho \big)(x) \int_\Omega d\rho(y) \: \L(x,y)
- \int_\F d\rho(x) \int_\Omega d\rho(y) \:\L(x,y) = 0 \:. \label{Ncalc1}
\end{split}
\eeq
Likewise, if~$\Phi_\tau$ is a symmetry of the Lagrangian, we can apply~\eqref{Nsymmlagr}. This gives
the identity
\beq \begin{split}
&\int_M d\rho(x) \int_\Omega d\rho(y) \Big( \L\big( \Phi_\tau(x), y\big) - \L(x,y) \Big) \\
&= \int_M d\rho(x) \int_\Omega d\rho(y) \Big( \L\big( x, \Phi_{-\tau}(y) \big) - \L(x,y) \Big) \\
&= \int_\Omega \Big( \ell \big( \Phi_{-\tau}(y) \big) - \ell(y) \Big)\: d\rho(y) \:, \label{Ncalc2}
\end{split}
\eeq
whose first variation vanishes in view of~\eqref{Nlfirst}.

Combining Definition~\ref{Ndefgis} with Proposition~\ref{Nprpuseful} immediately gives the following result.
\begin{Thm} \label{Nthmsymmgis}
Let~$\rho$ be a minimizing measure and $\Phi_\tau$ a 
continuously differentiable generalized integrated symmetry (see Definition~\ref{Ndefgis}).
Then for any compact subset~$\Omega \subset M$,
\[
\frac{d}{d\tau} \int_\Omega d\rho(x) \int_{M \setminus \Omega} d\rho(y) \:\Big( \L\big( \Phi_\tau(x), y\big) -
\L\big( x, \Phi_\tau(y) \big) \Big) \Big|_{\tau=0} = 0 \:.
\]
\end{Thm} \noindent
In view of~\eqref{Ncalc1} and~\eqref{Ncalc2},
the previous conservation laws of Theorems~\ref{Nthmsymmlag} and~\ref{Nthmsymmum}
are immediate corollaries of this theorem.

\section{The Setting of Causal Fermion Systems} \label{Nseccfs}
We now turn attention to the setting of causal fermion systems.
After a short review of the mathematical framework and the Euler-Lagrange equations
(Section~\ref{Nseccfsbasic}), we prove Noether-like theorems (Section~\ref{Nseccfsnoether}).
The reader interested in a more detailed introduction to causal fermion systems
is referred to Chapter~\ref{DissIntroCFS} and to the introductory chapter in~\cite{cfs}.

\subsection{Basic Definitions and the Euler-Lagrange Equations} \label{Nseccfsbasic}
\begin{Def} \label{Ndefparticle} {\textbf{(causal fermion system)}} {\em{
Given a separable complex Hilbert space $\H$ with scalar product~$\la .|. \ra_\H$
and a parameter~$n \in \N$ (the {\em{``spin dimension''}}), we let~$\F \subset \Lin(\H)$ be the set of all
self-adjoint operators on~$\H$ of finite rank, which (counting multiplicities) have
at most~$n$ positive and at most~$n$ negative eigenvalues. On~$\F$ we are given
a positive measure~$\rho$ (defined on a $\sigma$-algebra of subsets of~$\F$), the so-called
{\em{universal measure}}. We refer to~$(\H, \F, \rho)$ as a {\em{causal fermion system}}.
}}
\end{Def}

We next introduce the causal action principle. For any~$x, y \in \F$, the product~$x y$ is an operator
of rank at most~$2n$. We denote its non-trivial eigenvalues (counting algebraic multiplicities)
by~$\lambda^{xy}_1, \ldots, \lambda^{xy}_{2n} \in \C$.
We introduce the {\em{spectral weight}}~$| \,.\, |$ of an operator as the sum of the absolute values
of its eigenvalues. In particular, the spectral weight of the operator
products~$xy$ and~$(xy)^2$ is defined by
\[ |xy| = \sum_{i=1}^{2n} \big| \lambda^{xy}_i \big|
\qquad \text{and} \qquad \big| (xy)^2 \big| = \sum_{i=1}^{2n} \big| \lambda^{xy}_i \big|^2 \:. \]
We introduce the Lagrangian and the action by
\begin{align}
\text{\em{Lagrangian:}} && \L(x,y) &= \big| (xy)^2 \big| - \frac{1}{2n}\: |xy|^2 \label{NLagrange} \\
\text{\em{action:}} && \Sact(\rho) &= \iint_{\F \times \F} \L(x,y)\: d\rho(x)\, d\rho(y) \:. \label{NSdef}
\end{align}
The {\em{causal action principle}} is to minimize~$\Sact$ by varying the universal measure
under the following constraints:
\begin{align}
\text{\em{volume constraint:}} && \rho(\F) = \text{const} > 0 \quad\;\; & \label{Nvolconstraint} \\
\text{\em{trace constraint:}} && \int_\F \tr(x)\: d\rho(x) = \text{const} \neq 0 & \label{Ntrconstraint} \\
\text{\em{boundedness constraint:}} && \T(\rho) := \iint_{\F \times \F} |xy|^2\: d\rho(x)\, d\rho(y) &\leq C \:, \label{NTdef}
\end{align}
where~$C$ is a given parameter (and~$\tr$ denotes the trace of linear operators on~$\H$).

\subsubsection{The finite-dimensional setting}
If~$\H$ is {\em{finite-dimensional}} and~$\rho$ has {\em{finite total volume}}, the existence of
minimizers is proven in~\cite{continuum}, and the corresponding EL equations
are derived in~\cite{lagrange}. We now recall a few of these results.
Under the above assumptions, on~$\F$ one considers the topology induced by the
operator norm
\beq \label{Nsupnorm}
\|A\| := \sup \big\{ \|A u \|_\H \text{ with } \| u \|_\H = 1 \big\} \:.
\eeq
In this topology, the Lagrangian as well as the integrands in~\eqref{Ntrconstraint}
and~\eqref{NTdef} are continuous. We vary~$\rho$ within the class of bounded Borel measures of~$\F$.
The existence of minimizers of the action~\eqref{NSdef} under the constraints~\eqref{Nvolconstraint}--\eqref{NTdef} is proven in~\cite[Theorem~2.1]{continuum}.
For our purposes, the resulting EL equations are most conveniently
stated as follows (for a heuristic derivation see the introduction in~\cite{lagrange}).

\begin{Thm} \label{Nthm3} Assume that~$\rho$ is a minimizer of the causal
action principle for~$C$ so large that
\beq
C > \inf \big\{ \T(\mu) \:|\: \text{$\mu$ satisfies~\eqref{Nvolconstraint}  and~\eqref{Ntrconstraint}} \big\} \:. \label{NCbound}
\eeq
Moreover, assume that one of the following two technical assumptions hold:
\begin{enumerate}
\item[\textrm{(i)}] The boundedness constraint is satisfied with a strict inequality,
\beq
\T(\rho) < C \:. \label{NCbound2}
\eeq
\item[\textrm{(ii)}] The minimizer is regular in the sense of~\cite[Definition~3.12]{lagrange}.
\end{enumerate}
Then for a suitable choice of Lagrange multipliers~$\lambda, \kappa \in \R$,
the measure~$\rho$ is supported on the intersection of the level sets
\beq \label{Nhyper}
\Phi_1(x) = -4 \Sact(\rho)  \qquad \text{and} \qquad
\Phi_2(x) = 2 \Sact(\rho) \:,
\eeq
where
\beq
\Phi_1(x) := -\lambda \tr(x) \:,\qquad
\Phi_2(x) := 2 \int_\F \L_\kappa(x,y) \,d\rho(y) \label{NPhi2def}
\eeq
and
\beq\label{NLagrangeKappa}
\L_\kappa(x,y) := \L(x,y) + \kappa \, |xy|^2 \:.
\eeq
Moreover, the function
\[ \Phi(x) := \Phi_1 + \Phi_2 \]
is minimal on the support of~$\rho$, i.e.
\beq \label{NPhiminimal}
\Phi|_{\supp \rho} = \inf_\F \Phi \:.
\eeq
\end{Thm}
\Proof We first apply~\cite[Theorem~1.3]{lagrange} to the causal variational principle
with trace constraint in the case~$\T(\rho)<C$. This yields that~$\rho$ is supported on the
intersection of the level sets~\eqref{Nhyper}. Moreover, this theorem implies that~$\Phi|_{\supp \rho}
= -2 \Sact(\rho)$.
The minimality~\eqref{NPhiminimal} is proven in~\cite[Theorem~3.13]{lagrange}, noting
that the regularity condition of~\cite[Definition~3.12]{lagrange} is automatically satisfied
if the trace constraint is considered and if~\eqref{NCbound2} holds.
\QED
We remark for clarity that the inequality~\eqref{NCbound} can always be arranged
by choosing~$C$ sufficiently large.
The assumptions~(i) or~(ii) are needed in order for the Lagrange multiplier method
to be applicable. The basic difficulty comes about because the set of positive Borel
measures is not a vector space, but only a convex set.
Moreover, one must make sure that the constraints describe locally
a Banach submanifold. We refer the reader interested in the technical details to the
paper~\cite{lagrange}. In what follows, we take the assumptions~(i) or~(ii) for granted.

For the derivation of our conservation laws, we only need a weaker version of the
EL equations~\eqref{NPhiminimal}. Namely, it suffices to assume that the function~$\Phi$ is constant on the
support of~$\rho$,
\beq \label{NPhiminimala}
\Phi|_{\supp \rho} = \text{const} \:,
\eeq
and that the support of~$\rho$ is a {\em{local}} minimum in the sense that every~$x \in \supp \rho$ has a
neighborhood~$U(x) \subset \F$ such that
\beq \label{NPhiminimalb}
\Phi(x) = \inf_{U(x)} \Phi \:.
\eeq
We subsume~\eqref{NPhiminimala} and~\eqref{NPhiminimalb} by saying that~$\rho$ is a
{\em{local minimizer}} of the causal action principle.
Working with local minimizers is also preferable because
the regularized Dirac sea configurations to be considered in the examples
of Sections~\ref{Nseccurcor} and~\ref{NCorrDiracEM} are known to satisfy~\eqref{NPhiminimala}
and~\eqref{NPhiminimalb} in the continuum limit, but
but they are not global minimizers of the causal action principle
(for a detailed discussion of this point in the connection to microscopic mixing
and second-quantized bosonic fields we refer to~\cite[\S1.5.3]{cfs}).

\subsubsection{The infinite-dimensional setting}
We next consider the case that~$\H$ is infinite-dimensional or the total volume~$\rho(\F)$ is infinite.
First, a scaling argument shows that in the case~$\rho(\F)=\infty$ and~$\dim \H<\infty$, the
action is infinite for all measures satisfying the constraints, so that the variational principle is not sensible.
Similarly, if~$\rho(\F)<\infty$ and~$\dim \H=\infty$, the infimum of the action is zero,
but this infimum is not attained (for details see~\cite[Exercise~1.3]{cfs}). Therefore, the only interesting case is the
{\em{infinite-dimensional}} setting
when~$\rho(\F)=\infty$ and~$\dim \H=\infty$. In this setting, the causal action principle
makes mathematical sense if the volume constraint~\eqref{Nvolconstraint}
is implemented by demanding that the variations~$(\rho(\tau))_{\tau \in (-\tau_{\max}, \tau_{\max})}$
should for all~$\tau, \tau' \in (-\tau_{\max}, \tau_{\max})$ satisfy the conditions
\[ \big| \rho(\tau) - \rho(\tau') \big|(\F) < \infty \qquad \text{and} \qquad
\big( \rho(\tau) - \rho(\tau') \big) (\F) = 0 \]
(where~$|.|$ denotes the total variation of a measure; see~\cite[\S28]{halmosmt}).
But the existence of minimizers has not yet been proven.
Nevertheless, the EL equations are well-defined in the following sense:
\begin{Def} \label{Ndeflocmin}
Let~$(\rho, \H, \F)$ be a causal fermion system (possibly with~$\dim \H=\infty$
and~$\rho(\F)=\infty$). The measure~$\rho$ is a {\textbf{local minimizer}} of the causal action principle
if the integral in~\eqref{NPhi2def} is finite for all~$x \in \F$ and if the EL
equations~\eqref{NPhiminimala} and~\eqref{NPhiminimalb} hold for a suitable parameter~$\lambda \in \R$.
\end{Def} \noindent
Such local minimizers arise naturally when analyzing the continuum limit of
causal fermion systems (see~\cite{cfs}). Also, the physical examples 
in Sections~\ref{Nsecexcurrent} and~\ref{NsecexEM} will be formulated for local minimizers in the
infinite-dimensional setting. Finally, the above notion of local minimizers is
of relevance in view of future extensions of the existence  theory to the infinite-dimensional setting.

Let~$\rho$ be a local minimizer of the causal action principle.
We again define {\em{space-time}} by~$M=\supp \rho$; it is a closed but in
general non-compact subset of~$\F\subset \Lin(\H)$.
We again define the function~$\ell$ by
\beq \label{Nelldef}
\ell(x) = \int_M \L_\kappa(x,y) \,d\rho(y)
\eeq
and for notational convenience set~$\nu = \lambda/2$.
By assumption, this function is well-defined and finite for all~$x \in \F$.
Moreover, the EL equations~\eqref{NPhiminimala} and~\eqref{NPhiminimalb} imply that
\beq \label{NELgen}
\begin{split}
\ell&(x) - \nu \,\tr(x) \qquad \text{is constant on~$M$} \\
\ell(x) - \nu \,\tr(x) &= \inf_{y \in U(x)} \big( \ell(y) - \nu \,\tr(y) \big) \qquad \text{for all~$x \in M$}
\end{split}
\eeq
(where~$U(x) \subset \F$ is again a neighborhood of~$x$).
However, the function~$\ell$ need not be integrable. In particular, the action~\eqref{NSdef}
may be infinite.

These EL equations imply analogs of the relations~\eqref{Nhyper} and~\eqref{NPhi2def}.
Namely, evaluating the identity
\[ \frac{d}{dt} \big( \ell(tx) - \nu\, \tr(tx) \big) \big|_{t=1} = 0 \]
and using that the Lagrangian~\eqref{NLagrange} is homogeneous of degree two, one finds that on~$M$,
\[ 2 \ell(x) - \nu\, \tr(x) = 0 \:. \]
Combining this relation with~\eqref{NELgen}, one concludes that
on~$M$, the two terms in~\eqref{NELgen} are separately constant, i.e.
\beq \label{Nseprel}
\ell(x)  = -\inf_{y \in \F} \big( \ell(y) - \nu \,\tr(y) \big) = \frac{\nu}{2}\: \tr(x) \qquad \text{for all~$x \in M$}\:.
\eeq
These identities are very useful because they show that on~$M$,
both summands in~\eqref{NELgen} are separately constant. Moreover,
these relations make it possible to compute the Lagrange multiplier~$\nu$.

\subsection{Noether-Like Theorems} \label{Nseccfsnoether}
Let~$(\H, \F, \rho)$ be a causal fermion system, where~$\rho$ is a local minimizer
of the causal action (see Definition~\ref{Ndeflocmin}).
We do not want to assume that~$\H$ is finite-dimensional nor that the total volume of~$\rho$
is finite. But we shall assume that~$\rho$ is {\textbf{locally finite}} in the sense that~$\rho(K)< \infty$
for every compact subset~$K \subset \F$.

We again consider variations~$\Phi_\tau$ of~$M$ in~$\F$ described by a mapping~$\Phi$
of the form~\eqref{NPhidef},\\[-.6em]
\beq \label{NPhidef2}
\Phi : (-\tau_{\max}, \tau_{\max}) \times M \rightarrow \F
\qquad \text{with} \qquad \Phi(0,.) = \1 \:.
\eeq
Similar to Definition~\ref{Ndefsymm}, the regularity of
the variation is defined by composing~$\Phi$ with an operator mapping to the real numbers.
However, we now compose both with~$\ell$ and with the trace operation.

\begin{Def} \label{Ndefsymm2} A variation~$\Phi_\tau$ of the form~\eqref{NPhidef2} is
is {\textbf{continuous}} if the compositions
\[ \ell \circ \Phi ,\;\tr \circ \Phi : (-\tau_{\max}, \tau_{\max}) \times M \rightarrow \R \]
are continuous. If in addition their partial derivative~$\partial_\tau(\ell \circ \Phi)$
and~$\partial_\tau(\tr \circ \Phi)$ exist and are continuous on~$(-\tau_{\max}, \tau_{\max})
\times M \rightarrow \R$, then the variation is said to be {\textbf{continuously differentiable}}.
\end{Def}

We now generalize Proposition~\ref{Nprpuseful} to the setting of causal fermion systems.
\begin{Prp} \label{Nprpuseful2} Let~$\Phi_\tau$ be a continuous variation of the form~\eqref{NPhidef2}. Then
for any compact subset~$\Omega \subset M$,
\begin{align}
\int_M d\rho&(x) \int_\Omega d\rho(y) \Big( \L_\kappa\big( \Phi_\tau(x), y\big) - \L_\kappa(x,y) \Big) \label{Nid1n} \\
=\:& \int_\Omega \Big( \ell \big( \Phi_\tau(x)\big) - \ell(x) \Big)\: d\rho(x) \label{Nid2n} \\
&- \int_\Omega d\rho(x) \int_{M \setminus \Omega} d\rho(y) \:\Big( \L_\kappa\big( \Phi_\tau(x), y\big) -
\L_\kappa\big( x, \Phi_\tau(y) \big) \Big) \:. \label{Nid3n}
\end{align}
\end{Prp}
\Proof The subtle point is that~$M$ is in general non-compact, so that some of the integrals may diverge.
Therefore, we need to carefully consider the different integrals one after each other:
From~\eqref{Nseprel} we know that the functions~$\ell$ and~$\tr(x)$ are both constant on~$M$.
Moreover, the functions~$\ell \circ \Phi$ and~$\tr \circ \Phi$ are continuous on~$(-\tau_{\max}, \tau_{\max})
\times M$. As a consequence, it follows that for any {\em{compact}} subset~$\Omega \subset M$
and any~$\delta<\tau_{\max}$, the restriction
\[ \ell \circ \Phi \big|_{[-\delta, \delta] \times \Omega} \::\: 
[-\delta, \delta] \times \Omega \rightarrow \R \]
is a bounded function. Using that the Lagrangian is non-negative, this implies that
for any~$\tau \in (-\delta, \delta)$, the double integrals of the form
\[ \int_\Omega d\rho(x) \int_U d\rho(y) \:\L_\kappa\big( \Phi_\tau(x), y \big) \]
are well-defined and finite for any measurable subset~$U \subset M$. Moreover, one
may exchange the orders of integration using Tonelli's theorem
(i.e.\ the version of Fubini's theorem for non-negative integrands).
In particular, we conclude that the following integrals in~\eqref{Nid1n} and~\eqref{Nid3n} are
well-defined and finite,
\[ \int_M d\rho(x) \int_\Omega d\rho(y) \:\L_\kappa(x,y) \qquad \text{and} \qquad
\int_\Omega d\rho(x) \int_{M \setminus \Omega} d\rho(y) \:\L_\kappa\big( \Phi_\tau(x), y\big) \:. \]
For the integral in~\eqref{Nid2n}, we can argue similarly: We saw above that
the functions~$\ell \circ \Phi$ and~$\tr \circ \Phi$ are bounded on~$\{0\} \times M$ and continuous
on~$(-\tau_{\max}, \tau_{\max}) \times M$. Therefore, they are bounded on~$[-\delta, \delta] \times \Omega$,
implying that the integral in~\eqref{Nid2n} is well-defined and finite.

It remains to consider the two integrals
\beq \label{Nintdiverge}
\int_M d\rho(x) \int_\Omega d\rho(y) \:\L_\kappa\big( \Phi_\tau(x), y\big) \quad \text{and} \quad
\int_\Omega d\rho(x) \int_{M \setminus \Omega} d\rho(y) \:\L_\kappa\big( x, \Phi_\tau(y) \big) \:.
\eeq
These integrals could diverge. But since the integrand is non-negative, Tonelli's theorem
nevertheless allows us to exchange the two integrals.
Then the integrands of the two integrals coincide. The integration ranges
coincide up to the compact set~$\Omega \times \Omega$.
Therefore, the first integral in~\eqref{Nintdiverge} diverges if and only if the second integral diverges.
If this is the case, the left and the right side of the equation~\eqref{Nid1n}--\eqref{Nid3n}
both take the value~$+\infty$, so that the statement of the proposition holds.
In the case that the integrals in~\eqref{Nintdiverge} are both finite,
we can repeat the computation in the proof of Proposition~\ref{Nprpuseful}
and apply~\eqref{Nelldef} to obtain the result.
\QED

\begin{Def} \label{Ndefgis2}
The variation~$\Phi_\tau$ is a {\textbf{generalized integrated symmetry}} in
the space-time region~$\Omega \subset M$ if the following two identities hold:
\begin{gather}
\int_M d\rho(x) \int_\Omega d\rho(y) \Big( \L_\kappa\big( \Phi_\tau(x), y\big) - \L_\kappa(x,y) \Big) = 0 \label{NLagint} \\
\int_\Omega \Big( \tr \big( \Phi_\tau(x)\big) - \tr(x) \Big)\: d\rho(x) = 0 \:. \label{Ntraceint}
\end{gather}
\end{Def}

Combining this definition with Proposition~\ref{Nprpuseful2} 
and the EL equations~\eqref{NELgen} immediately gives the following result:
\begin{Thm} \label{Nthmsymmgis2}
Let~$\rho$ be a local minimizer of the causal action (see Definition~\ref{Ndeflocmin})
and $\Phi_\tau$ a continuously differentiable generalized integrated symmetry (see Definitions~\ref{Ndefsymm2}
and~\ref{Ndefgis2}). Then for any compact subset~$\Omega \subset M$, the following
surface layer integral vanishes,
\beq \label{Nconserve5}
\frac{d}{d\tau} \int_\Omega d\rho(x) \int_{M \setminus \Omega} d\rho(y) \:\Big( \L_\kappa\big( \Phi_\tau(x), y\big) -
\L_\kappa\big( x, \Phi_\tau(y) \big) \Big) \Big|_{\tau=0} = 0 \:.
\eeq
\end{Thm} \noindent
In order to explain the necessity of the condition~\eqref{Ntraceint},
we point out that, although the functions~$\ell(x)$ and~$\tr(x)$
are both constant on~$M$ (see~\eqref{Nseprel}), this does not imply that transversal derivatives of
these functions vanish. Only for their specific linear combination in~\eqref{Nseprel} the
derivative vanishes on~$M$.
We also note that the condition for the trace~\eqref{Ntraceint}, which did not appear in the compact setting,
can always be satisfied by rescaling the variation according to
\[ \Phi_\tau(x) \rightarrow \Phi_\tau(x)\:\frac{\tr(x)}{\tr\big( \Phi_\tau(x) \big)} \]
(note that by continuity, the trace in the denominator is non-zero for sufficiently small~$\tau$).
However, when doing so, the remaining condition~\eqref{NLagint} as well as the regularity conditions
of Definition~\ref{Ndefsymm2} might become more involved.
This is the reason why we prefer to write two separate conditions~\eqref{NLagint} and~\eqref{Ntraceint}.

The above results give rise to corollaries
which extend Theorems~\ref{Nthmsymmlag} and~\ref{Nthmsymmum} to the setting of
causal fermion systems.
\begin{Def} \label{Ndefsymmnc} A variation~$\Phi_\tau$ of the form~\eqref{NPhidef2}
is a {\textbf{symmetry of the Lagrangian}} if
\beq \label{NsymmLagnc}
\L_\kappa \big( x, \Phi_\tau(y) \big) = \L_\kappa \big( \Phi_{-\tau}(x), y \big) \qquad
\text{for all~$\tau \in (-\tau_{\max}, \tau_{\max})$ and all~$x,y \in M$}\:.
\eeq
It is a {\textbf{symmetry of the universal measure}} if
\[ (\Phi_\tau)_* \rho = \rho \qquad
\text{for all~$\tau \in (-\tau_{\max}, \tau_{\max})$}\:. \]
Moreover, it {\textbf{preserves the trace}} if
\[ \tr \big( \Phi_\tau(x)\big) = \tr(x) \qquad \text{for all~$\tau \in (-\tau_{\max}, \tau_{\max})$
and all~$x \in M$}\:. \]
\end{Def}

\begin{Corollary} \label{Ncorsymmlag}
Let~$\rho$ be a local minimizer of the causal action (see Definition~\ref{Ndeflocmin})
and $\Phi_\tau$ a continuously differentiable variation. Assume that~$\Phi_\tau$ is a
symmetry of the Lagrangian and preserves the trace.
Then for any compact subset~$\Omega \subset M$, the conservation law~\eqref{Nconserve5} holds.
\end{Corollary}

\begin{Corollary} \label{Ncorsymmum}
Let~$\rho$ be a local minimizer of the causal action (see Definition~\ref{Ndeflocmin})
and $\Phi_\tau$ a continuously differentiable variation. Assume that~$\Phi_\tau$ 
is a symmetry of the universal measure and preserves the trace.
Then for any compact subset~$\Omega \subset M$, the conservation law~\eqref{Nconserve5} holds.
\end{Corollary} \noindent
These corollaries follow immediately by calculations similar to~\eqref{Ncalc1} and~\eqref{Ncalc2}.

\section{Example: Current Conservation} \label{Nsecexcurrent}
This section is devoted to the important example of current conservation,
also referred to as charge conservation.
For Dirac particles, the electric charge is (up to a multiplicative constant)
given as the integral over the probability density. Therefore, charge conservation
also corresponds to the conservation of the probability integral in quantum mechanics.
In the context of the classical Noether theorem, charge conservation is a consequence
of an internal symmetry of the system, which can be described by a phase transformation~\eqref{Nglobalphase}
of the wave function and is often referred to as global gauge symmetry.
As we shall see in Section~\ref{Nsecgenchargecons}, causal fermion systems also have such
an internal symmetry, giving rise to a general class of conservation laws (see Theorem~\ref{Nthmcurrent}).
In Section~\ref{Nseccurcor}, these conservation laws are evaluated for Dirac spinors
in Minkowski space, giving a correspondence to the conservation of the Dirac current
(see Theorem~\ref{Nthmcurrentmink} and Corollary~\ref{Ncorcurrent}).
In Section~\ref{Nsecremark}, we conclude with a few clarifying remarks.

\subsection{A General Conservation Law} \label{Nsecgenchargecons}
Let~$\scrA$ be a bounded symmetric operator on~$\H$ and
\beq \label{NUtaudef}
\scrU_\tau := \exp(i \tau \scrA)
\eeq
be the corresponding one-parameter family of unitary transformations.
We introduce the mapping\Chd{$\Phi_\tau \rightarrow \Phi$}%
\beq \label{Nvarunit}
\Phi :\, \R \times \F \rightarrow \F\:,\qquad \Phi(\tau, x) = \scrU_\tau \,x\, \scrU_\tau^{-1} \:.
\eeq
Restricting this mapping to~$(-\tau_{\max}, \tau_{\max}) \times M$, we obtain
a variation~$(\Phi_\tau)_{\tau \in (-\tau_{\max}, \tau_{\max})}$ of the form~\eqref{NPhidef2}.

\begin{Lemma} \label{NPhiSymmLag} The variation~$\Phi_\tau$ given by~\eqref{Nvarunit} is a symmetry of the
Lagrangian and preserves the trace (see Definition~\ref{Ndefsymmnc}).
\end{Lemma}
\Proof Since~$\Phi_\tau(x)$ is unitarily equivalent to~$x$, they obviously have the same trace.
In order to prove~\eqref{NsymmLagnc}, we first recall that the Lagrangian~$\L_\kappa(x,y)$
is defined in terms of the spectrum of the operator product~$xy$ (see~\eqref{NLagrange}).
The calculation
\[ x \: \Phi_\tau(y) = x \; \scrU_\tau \,y\, \scrU_\tau^{-1}
= \scrU \,\big( \scrU_\tau^{-1} \,x \,\scrU_\tau \; y \big) \,\scrU_\tau^{-1}
= \scrU \,\big( \Phi_{-\tau}(x) \: y \big) \, \scrU_\tau^{-1} \]
shows that the operators~$x \,\Phi_\tau(y)$ and~$\Phi_{-\tau}(x) \,y$ are unitarily equivalent
and therefore isospectral. This concludes the proof.
\QED

It remains to verify whether the variation~$\Phi_\tau$ is continuously differentiable
in the sense of Definition~\ref{Ndefsymm2}. For the trace, this is obvious because~$\Phi_\tau$
leaves the trace invariant, so that~$\tr \,\circ\, \Phi_\tau(\tau, x) = \tr(x)$, which clearly depends
continuously on~$x$ (in the topology induced by the $\sup$-norm~\eqref{Nsupnorm}).
For~$\ell \circ \phi$, we cannot in general expect differentiability because the Lagrangian~$\L_\kappa$
is only Lipschitz continuous in general. Therefore, we must include the differentiability of~$\ell \circ \phi$
as an assumption in the follo\-wing theo\-rem.

\begin{Thm} \label{Nthmcurrent}
Given a bounded symmetric operator~$\scrA$ on~$\H$, we let~$\Phi_\tau$ be the
variation~\eqref{Nvarunit}. Assume that the mapping~$\ell \circ \Phi \::\:
(-\tau_{\max}, \tau_{\max}) \times M \rightarrow \R$
is continuously differentiable in the sense that it is continuous and that~$\partial_\tau (\ell \circ \Phi)$
exists and is also continuous on~$(-\tau_{\max}, \tau_{\max}) \times M$.
Then for any compact subset~$\Omega \subset M$, the conservation law~\eqref{Nconserve5} holds.
\end{Thm}

\subsection{Correspondence to Dirac Current Conservation} \label{Nseccurcor}
The aim of this section is to relate the conservation law of Theorem~\ref{Nthmcurrent} to
the usual current conservation in relativistic quantum mechanics in Minkowski space.

To this end, we consider causal fermion systems~$(\F, \H, \rho^\varepsilon)$
describing the regularized Dirac sea vacuum in Minkowski space~$(\scrM, \la .,. \ra)$.
We briefly recall the construction (for the necessary preliminaries see~\cite[Section~2]{cfsrev}, \cite{cfs},
\cite[Section~4]{lqg} or the introduction in Chapter~\ref{DissIntroCFS}.
As in~\cite[Chapter~3]{cfs} we consider three generations of Dirac particles
of masses~$m_1$, $m_2$ and~$m_3$ (corresponding to the three generations of elementary particles
in the standard model; three generations are necessary in order to obtain well-posed equations
in the continuum limit). Denoting the generations by an index~$\beta$, we consider the
Dirac equations
\beq \label{NDireqns}
(i \Pdd - m_\beta) \,\psi_\beta = 0 \qquad (\beta=1,2,3)\:.
\eeq
On solutions~$\psi = (\psi_\beta)_{\beta=1,2,3}$, we consider the scalar product
\[ ( \psi | \phi) := 2 \pi \sum_{\beta=1}^3 \int_{\R^3} (\overline{\psi}_\beta \gamma^0 \phi_\beta)(t, \vec{x})\: d^3x \:. \]
The Dirac equation has solutions on the upper and lower mass shell, which have positive
respectively negative energy.
In order to avoid potential confusion with other notions of energy, we here prefer the notion of
solutions of positive and negative {\em{frequency}}.
We choose~$\H$ as the subspace spanned by all solutions of negative frequency,
together with the scalar product~$\la .|. \ra_\H := ( .|.)|_{\H \times \H}$.
We now introduce an ultraviolet regularization
(for details see~\cite[Section~2]{cfsrev}) and denote the regularized quantities by
a superscript~$\varepsilon$. Now the local correlation operators are defined by
\[ \la \psi^\varepsilon \,|\, F^\varepsilon(x)\, \phi^\varepsilon \ra_\H
= - \sum_{\alpha, \beta=1}^3 \overline{\psi_\alpha^\varepsilon(x)} \phi_\beta^\varepsilon(x)
\qquad \text{for all~$\psi,\phi \in \H$}\:. \]
Next, the universal measure is defined as the push-forward of the Lebesgue measure~$d\mu = d^4x$,
\[ \rho^\varepsilon := (F_\varepsilon)_*(\mu) \:. \]
Then~$(\H, \F, \rho^\varepsilon)$ is a causal fermion system of spin dimension two.
As shown in~\cite[Chapter~1]{cfs}, the kernel of the fermionic projector~$P(x,y)$
converges as~$\varepsilon \searrow 0$ to the distribution
\beq \label{NPvac}
P(x,y) = \sum_{\beta=1}^3 \int \frac{d^4k}{(2 \pi)^4}\: (\slashed{k}+m_\beta)\:
\delta \big(k^2-m_\beta^2 \big)\: e^{-ik(x-y)}
\eeq
(this configuration is also referred to as three generations in a single sector;
see~\cite[Chapter~3]{cfs}). We remark that our ansatz can be generalized by
introducing so-called weight factors (see~\cite{reg} and Remark~\ref{Nremweights} below).

We want to apply Theorem~\ref{Nthmcurrent}.
Since in this theorem, the set~$\Omega$ must be compact,
we choose it as a lens-shaped region whose boundary is composed of two
space-like hypersurfaces (see the left of Figure~\ref{Nfignoether2}).
\begin{figure}\centering
\psscalebox{1.0 1.0} %
{
\begin{pspicture}(0,-0.9569027)(10.911616,0.9569027)
\definecolor{colour0}{rgb}{0.8,0.8,0.8}
\psframe[linecolor=colour0, linewidth=0.02, fillstyle=solid,fillcolor=colour0, dimen=outer](9.869394,0.56531954)(5.789394,-0.8213471)
\pspolygon[linecolor=colour0, linewidth=0.02, fillstyle=solid,fillcolor=colour0](0.033838496,-0.34579158)(0.19828294,-0.23023602)(0.35828295,-0.15023603)(0.5760607,-0.03912491)(0.94939405,0.10309731)(1.389394,0.21865287)(1.8160607,0.27198622)(2.1227274,0.28531954)(2.4338386,0.26754177)(2.7982829,0.17865287)(3.1227274,0.067541756)(3.4427273,-0.083569355)(3.6827273,-0.21690269)(3.598283,-0.33690268)(3.3938384,-0.46134713)(3.1316164,-0.5991249)(2.8560607,-0.7191249)(2.478283,-0.7946805)(2.1627274,-0.84356934)(1.8160607,-0.85690266)(1.5182829,-0.8480138)(1.0827274,-0.78134716)(0.62939405,-0.63912493)(0.2827274,-0.4791249)
\rput[bl](1.4738384,-0.41468045){$\Omega$}
\rput[bl](7.458283,-0.27690268){\normalsize{$\Omega$}}
\psline[linecolor=black, linewidth=0.04, arrowsize=0.09300000000000001cm 1.0,arrowlength=1.7,arrowinset=0.3]{->}(1.3138385,0.18531954)(1.2116163,0.6919862)
\rput[bl](1.4160607,0.45643064){$\nu$}
\rput[bl](10.027172,0.41865286){\normalsize{$t=t_1$}}
\rput[bl](10.031616,-0.9569027){\normalsize{$t=t_0$}}
\psbezier[linecolor=black, linewidth=0.04](0.011616275,-0.3569027)(0.704925,0.13707085)(1.4876974,0.28290415)(2.0582829,0.28754175)(2.6288686,0.29217938)(3.0587788,0.14209865)(3.7116163,-0.23690268)
\psbezier[linecolor=black, linewidth=0.04](0.014949608,-0.34134713)(0.5271472,-0.62181807)(0.9588085,-0.8170959)(1.6016163,-0.86023605)(2.244424,-0.90337616)(3.12989,-0.7812347)(3.7149496,-0.22134712)
\psline[linecolor=black, linewidth=0.04, arrowsize=0.09300000000000001cm 1.0,arrowlength=1.7,arrowinset=0.3]{->}(2.4071717,-0.81468046)(2.3049495,-0.3080138)
\psline[linecolor=black, linewidth=0.04](5.789394,0.57420844)(9.864949,0.57420844)
\psline[linecolor=black, linewidth=0.04](5.789394,-0.8213471)(9.869394,-0.8257916)
\psline[linecolor=black, linewidth=0.04, arrowsize=0.09300000000000001cm 1.0,arrowlength=1.7,arrowinset=0.3]{->}(6.669394,0.5764306)(6.6738386,1.069764)
\psline[linecolor=black, linewidth=0.04, arrowsize=0.09300000000000001cm 1.0,arrowlength=1.7,arrowinset=0.3]{->}(8.660505,-0.82356936)(8.664949,-0.33023602)
\rput[bl](2.5138385,-0.60579157){$\nu$}
\rput[bl](8.864949,-0.6413471){$\nu$}
\rput[bl](6.8649497,0.7630973){$\nu$}
\end{pspicture}
}
\caption{Choice of the space-time region~$\Omega \subset \scrM$.}
\label{Nfignoether2}
\end{figure}
Considering a sequence
of compact sets~$\Omega_n$ which exhaust the region~$\Omega$
between two Cauchy surfaces at times~$t=t_0$
and~$t=t_1$, the surface layer integral~\eqref{Nconserve5} reduces to the difference
of surface layers integrals at times~$t \approx t_0$ and~$t \approx t_1$.
The detailed analysis (which will be carried out below) gives the following result:
\begin{Thm} {\Thmt{(Current conservation)}} \label{Nthmcurrentmink}
Let~$(\H, \F, \rho^\varepsilon)$ be local minimizers of the causal action which
describe the Minkowski vacuum~\eqref{NPvac}.
Considering the limiting procedure explained in Figure~\ref{Nfignoether2} and taking the
continuum limit, the conservation laws of Theorem~\ref{Nthmcurrent} go over to
a linear combination of the probability integrals in every generation.
More precisely, there are non-negative constants~$c_\beta$ such that
for all~$u \in \H$ for which~$\psi^u$ is a negative-frequency solution of the Dirac equation,
the surface layer integral~\eqref{Nconserve5} goes over the equation
\beq \label{Nthmcurr1}
\sum_{\beta=1}^3 m_\beta\, c_\beta \int_{t=t_0} \Sl \psi_\beta^u(x) | \gamma^0 \psi_\beta^u(x) \Sr\:d^3x
= \sum_{\beta=1}^3 m_\beta\, c_\beta \int_{t=t_1} \Sl \psi_\beta^u(x) | \gamma^0 \psi_\beta^u(x) \Sr\:d^3x \:.
\eeq
\end{Thm} \noindent
The constants~$c_\beta$ depend on properties of the distribution~$\hat{Q}$ in the continuum limit,
as will be specified in Definition~\ref{Ndef611} and~\eqref{NUSymm20} below.

Before coming to the proof, we explain the statement and significance of this theorem.
We first note that the restriction to negative-frequency solutions is needed
because the description of positive-frequency solutions involves the 
so-called mechanism of microscopic mixing which for brevity we cannot address in this
paper (see however the remarks in Section~\ref{Nremmicro} below).
Next, we point out that the theorem implies the statement that the function~$\ell \circ \Phi$ in
Theorem~\ref{Nthmcurrent} is continuously differentiable in the continuum limit.
However, this does not necessarily mean that this differentiability statement
holds for any local minimizer~$(\H, \F, \rho^\varepsilon)$ with regularization.
This rather delicate technical point will be discussed in Remark~\ref{Nremdiscuss} below.

Considering Cauchy hyperplanes in~\eqref{Nthmcurr1} is indeed no
restriction because the theorem can be extended immediately to
general Cauchy surfaces:
\begin{Corollary} {\Thmt{(Current conservation on Cauchy surfaces)}}  \label{Ncorcurrent}
Let~$\scrN_0, \scrN_1$ be two Cauchy surfaces in Minkowski space, where~$\scrN_1$
lies to the future of~$\scrN_0$.
Then, under the assumptions of Theorem~\ref{Nthmcurrentmink},
the conservation law of Theorem~\ref{Nthmcurrent} goes over to the conservation law
for the current integrals
\beq \label{NconsCauchy}
\sum_{\beta=1}^3 m_\beta\, c_\beta \int_{\scrN_0} \Sl \psi_\beta^u | \slashed{\nu} \psi_\beta^u \Sr\:d\mu_{\scrN_0}
= \sum_{\beta=1}^3 m_\beta\, c_\beta \int_{\scrN_1} \Sl \psi_\beta^u | \slashed{\nu} \psi_\beta^u \Sr\:d\mu_{\scrN_1} \:,
\eeq
where~$\nu$ denotes the future-directed normal.
\end{Corollary}
\Proof We choose~$\Omega$ as the space-time region between the two Cauchy surfaces.
Using that the integrand in~\eqref{Nconserve5} is anti-symmetric in its arguments~$x$ and~$y$,
the integration range can be rewritten as
\begin{align}
\int_\Omega &d\rho(x) \int_{M \setminus \Omega} d\rho(y) \:\Big( \L_\kappa\big( \Phi_\tau(x), y\big) -
\L_\kappa\big( x, \Phi_\tau(y) \big) \Big) \nonumber \\
&= \int_{J^\wedge(\scrN_1)} d\rho(x) \int_{J^\vee(\scrN_1)} d\rho(y)
\:\Big( \L_\kappa\big( \Phi_\tau(x), y\big) - \L_\kappa\big( x, \Phi_\tau(y) \big) \Big) \label{Nsl1} \\
&\quad\: -\int_{J^\wedge(\scrN_0)} d\rho(x) \int_{J^\vee(\scrN_0)} d\rho(y)
\:\Big( \L_\kappa\big( \Phi_\tau(x), y\big) - \L_\kappa\big( x, \Phi_\tau(y) \big) \Big) \:, \nonumber %
\end{align}
where~$J^\wedge$ and~$J^\vee$ denote the causal past and causal future, respectively.
For ease in notation, we refer to the integrals in~\eqref{Nsl1} as a
{\em{surface layer integral over}}~$\scrN_1$.
Thus the surface layer integral in~\eqref{Nconserve5} is the difference
of two surface layer integrals over the Cauchy surfaces~$\scrN_0$ and~$\scrN_1$.

In order to compute for example the surface layer integral over~$\scrN_0$,
one chooses~$\Omega$ as the region between the
Cauchy surface~$\scrN_0$ and the Cauchy surface~$t=t_0$
(for sufficiently small~$t_0$; in case that these Cauchy surfaces intersect for every~$t_0$, one
modifies~$\scrN_0$ near the asymptotic end without affecting our results).
Applying the conservation law of Theorem~\ref{Nthmcurrent}
to this new region~$\Omega$, one concludes that the the surface layer integral over~$\scrN_0$
coincides with the surface layer integral at time~$t \approx t_0$. The latter surface layer
integral, on the other hand, was computed in Theorem~\ref{Nthmcurrentmink}
to go over to the sum of the probability integrals in~\eqref{Nthmcurr1}.
Finally, the usual current conservation for the Dirac dynamics shows that the
the integrals in~\eqref{Nthmcurr1} coincide with the surface integral over~$\scrN_0$
in~\eqref{NconsCauchy}. This concludes the proof.
\QED
Using similar arguments, Theorem~\ref{Nthmcurrentmink} can also be extended
to interacting systems (see Remark~\ref{Nreminteract} below). \\[-0.7em]

The remainder of this section is devoted to the proof of Theorem~\ref{Nthmcurrentmink}.
We first rewrite the causal action principle in terms of the kernel of the fermionic projector
(for details see~\cite[\S1.1]{cfs}). The kernel of the fermionic projector~$P(x,y)$ is defined by
\beq \label{NPxydef}
P(x,y) = \pi_x \,y|_{S_y} \::\: S_y \rightarrow S_x \:.
\eeq
The {\em{closed chain}} is defined as the product
\[ A_{xy} = P(x,y)\, P(y,x) \::\: S_x \rightarrow S_x\:. \]
The nontrivial eigenvalues~$\lambda^{xy}_1, \ldots, \lambda^{xy}$ of the operator~$xy$
coincide with the eigenvalues of the closed chain.
Moreover, it is useful to express~$P(x,y)$ in terms of the wave evaluation operator defined by
\beq
\label{Nweo}
\Psi(x) \::\: \H \rightarrow S_x\:, \qquad u \mapsto \psi^u(x) = \pi_x u \:. 
\eeq
Namely,
\[ x = - \Psi(x)^* \,\Psi(x)  \qquad \text{and} \qquad P(x,y) = -\Psi(x)\, \Psi(y)^*\:. \]

Our task is to compute the term~$\L_\kappa(\Phi_\tau(x), y)$ in~\eqref{Nconserve5} for~$x, y \in M$.
The detailed computations in ~\cite[\S3.6.1]{cfs} show that the
fermionic projector of the Minkowski vacuum satisfies the
EL equations in the continuum limit for~$\kappa=0$
(in our setting, this result means that the measures~$\rho^\varepsilon$
are local minimizers in the sense of Definition~\ref{Ndeflocmin} in the limiting case~$\varepsilon \searrow 0$).
Therefore, we may set~$\kappa$ to zero.
Thus our task is to compute the term~$\L(\Phi_\tau(x), y)$.
In preparation, we compute~$P(\Phi_\tau(x), y)$.
To this end, we first note that
\beq \label{Nisospec}
\begin{split}
\Phi_\tau(x)\, y &= \scrU_\tau \,x\, \scrU_\tau^{-1}\; y
= \scrU_\tau \,\Psi(x)^* \,\Psi(x)\, \scrU_\tau^{-1}\; \Psi(y)^* \,\Psi(y) \\
&\simeq \Psi(x)\, \scrU_\tau^{-1}\; \Psi(y)^* \,\Psi(y)\, \scrU_\tau \,\Psi(x)^* \:,
\end{split}
\eeq
where in the last line we cyclically commuted the operators
and $\simeq$ means that the operators are isospectral (up to irrelevant zeros in the spectrum).
Therefore, introducing the notations
\begin{gather}
\Psi_\tau(x) = \Psi(x)\, \scrU_\tau^{-1} \::\: \H \rightarrow S_x \\
P\big( \Phi_\tau(x), y \big) = -\Psi_\tau(x)\, \Psi(y)^* \:, \qquad
P\big( y, \Phi_\tau(x) \big) = -\Psi(y)\, \Psi_\tau(x)^* \label{NPmodform}
\end{gather}
one sees that the operator product~$\Phi_\tau(x)\, y$ is isospectral to the modified closed chain
\beq \label{Nchainform}
P\big( \Phi_\tau(x), y \big)\: P\big( y, \Phi_\tau(x) \big) \:.
\eeq
Considering the Lagrangian as a function of this modified closed chain, the variation
is described in a form suitable for computations.

For clarity, we explain in which sense the kernel of the fermionic projector as given by~\eqref{NPmodform}
agrees with the abstract definition~\eqref{NPxydef},
\beq \label{NPtauabstract}
P\big( \Phi_\tau(x), y \big) = \pi_{\Phi_\tau(x)} y \:.
\eeq
It is a subtle point that the point~$\Phi_\tau(x) \in \F$ depends on~$\tau$, so that space-time
itself changes. However, when identifying the spin space~$S_{\Phi_\tau(x)}$ with a corresponding
spinor space in Minkowski space, the base point~$x \in \scrM$ should be kept fixed.
Therefore, the spin space~$S_{\Phi_\tau(x)}$ is to be identified with the spinor space~$S_x \scrM$.
For each~$\tau$, this can be accomplished as explained above.
This identification made, the kernel~\eqref{NPmodform} indeed agrees with~\eqref{NPtauabstract}.
The reason why we do not give the details of this construction is that
the computation~\eqref{Nisospec} already shows that the Lagrangian
can be computed with the closed chain~\eqref{Nchainform}, and this is all we need for what follows.

We now choose~$\scrA =\pi_{\la u \ra}$ as the projection on the one-dimensional subspace
generated by a vector~$u \in \H$ and let~$\pi_{\la u \ra^\perp}$ be the projection on
the orthogonal complement of~$u$. Then
\begin{align*}
\Psi_\tau(x) &= \Psi(x) \; \big( \pi_{\la u \ra^\perp} + e^{-i \tau}\, \pi_{\la u \ra} \big) \\
P\big(\Phi_\tau(x),y \big) &= -\Psi(x) \; \big( \pi_{\la u \ra^\perp} + e^{-i \tau}\, \pi_{\la u \ra} \big)\: \Psi(y)^* \\
&= P(x,y) + (1- e^{-i \tau}) \;\Psi(x) \:\pi_{\la u \ra}\: \Psi(y)^* \:.
\end{align*}
Normalizing~$u$ such that~$\la u|u \ra_\H=1$, the last equation can be written
in the form that for any~$\chi \in S_y$,
\[ P\big(\Phi_\tau(x),y \big) \, \chi
= P(x,y)\, \chi + (1- e^{-i \tau}) \;\psi^u(x)\; \Sl \psi^u(y) \,|\, \chi \Sr_y \:. \]

We now compute the first order variation.
\beq \label{NdelP} \begin{split}
\frac{d}{d\tau} P\big(\Phi_\tau(x),y \big) \big|_{\tau=0} \:\chi &=
i \psi^u(x)\; \Sl \psi^u(y)\,|\, \chi \Sr_y =: \delta P(x,y)\, \chi \\
\frac{d}{d\tau} P\big(y, \Phi_\tau(x)\big) \big|_{\tau=0} &= \big(\delta P(x,y) \big)^*
\end{split}
\eeq
The variation of the Lagrangian can be written as (cf.~\cite[Section~5.2]{PFP}
or~\cite[Section~1.4]{cfs})
\begin{align} \label{NLTrQ} \begin{split}
\delta \L(x,y) \;&\!:=  \frac{d}{d\tau}\: \L\big( \Phi_\tau(x),y \big)\big|_{\tau=0} \\
&= \Tr_{S_y} \big( Q(y,x)\, \delta P(x,y) \big) + \Tr_{S_x} \big( Q(x,y)\, \delta P(x,y)^* \big) \\
&= i \, \Sl \psi^u(y) \,|\, Q(y,x) \,\psi^u(x) \Sr_y
-i \,\Sl \psi^u(x) \,|\, Q(x,y) \,\psi^u(y) \,\Sr_x \:, \end{split}
\end{align}
where in the last line we used~\eqref{NdelP}, and~$Q(x,y)$ is a distributional kernel
to be specified below.
Using that the kernel~$Q(x,y)$ is symmetric in the sense that
\[ Q(x,y)^* = Q(y,x)\:, \]
we can write the variation of the Lagrangian in the compact form
\[ \delta \L(x,y) = -2 \im \big( \Sl \psi^u(y) \,|\, Q(y,x) \,\psi^u(x) \Sr_y \big) \:. \]
Using this identity, the surface layer integral in~\eqref{Nconserve5} can be written as
\[ \int_\Omega d^4x \int_{\scrM \setminus \Omega} d^4y \;\im \big( \Sl \psi^u(y) \,|\, Q(y,x) \,\psi^u(x) \Sr_y \big)
= 0 \:. \]

Taking the liming procedure as shown in Figure~\ref{Nfignoether2}, it suffices to
consider a surface layer integral at a fixed time~$t_0$,
which for convenience we choose equal to zero. Thus our task is to compute the double integral
\beq \label{Ninttask}
J := \int_{t \geq 0} d^4x \int_{t < 0} d^4y \:\im \big( \Sl \psi^u(y) \,|\, Q(y,x) \,\psi^u(x) \Sr \big) \:.
\eeq
Here we omitted the subscript~$y$ at the spin scalar product because in Minkowski space
all spinor spaces can be naturally identified.

In order to explain our method for computing the integrals in~\eqref{Ninttask},
we first state a simple lemma where integrals of this type are computed.
As will be explained below, this lemma cannot be applied to our problem
for technical reasons, but it nevertheless clarifies the structure of our results.
\begin{Lemma} \label{Nlemmagen}
Let~$f : \scrM \times \scrM \rightarrow \R$ be an integrable function
with the following properties:
\begin{itemize}
\item[\textrm{(a)}] $f$ is anti-symmetric, i.e.\ $f(x,y) = -f(y,x)$.
\item[\textrm{(b)}] $f$ is homogeneous in the sense that it depends only on the difference vector~$y-x$.
\item[\textrm{(c)}] The following integral is finite,
\beq \label{Nfinint}
\int_{\scrM} \big|x^0\: f(x,0) \big|\: d^4x < \infty \:.
\eeq
\end{itemize}
Then
\beq \label{Ndint}
\int_{-\infty}^0 dt \int_0^\infty dt' \int_{\R^3} d^3y \:f \big( (t,\vec{x}), (t', \vec{y}) \big) = \frac{i}{2}\:
\frac{\partial}{\partial k^0} \hat{f}(k) \Big|_{k=0} \:,
\eeq
where~$\hat{f}$ is the Fourier transform, i.e.\
\beq \label{NfFourier}
f(x,y) = \int \frac{d^4k}{(2 \pi)^4} \:\hat{f}(k)\: e^{-ik (x-y)} \:.
\eeq
\end{Lemma}
\Proof Substituting~\eqref{NfFourier} into the left side of~\eqref{Ndint}, we can carry out the spatial integral to obtain
\beq \label{Ndoubleint}
\int_{-\infty}^0 dt \int_0^\infty dt' \int_{\R^3} d^3y \:f \big( (t,\vec{x}), (t', \vec{y}) \big) =
\int_{-\infty}^0 dt \int_0^\infty dt' \,g(t-t') \:,
\eeq
where
\beq \label{NgFourier}
g(\tau) =  \int_{\R^3} f\big((\tau, \vec{x}), (0, \vec{y}) \big)\: d^3y
= \int_{-\infty}^\infty \frac{d\omega}{2 \pi} \:\hat{f}\big( (\omega, \vec{0} ) \big)\: e^{-i\omega \tau} \:.
\eeq
We now transform variables in the inner integral in~\eqref{Ndoubleint},
\[ \int_0^\infty g(t-t')\: dt' = \int_{-\infty}^t g(\tau)\: d\tau = \int_{-\infty}^0 g(\tau)\: \Theta(t-\tau) \: d\tau \:. \]
Using~\eqref{Nfinint} and~\eqref{NgFourier}, we know that
\[ \iint_{\R^- \times \R^-} \big| g(\tau)\: \Theta(t-\tau) \big| \: dt\, d\tau
= \int_{-\infty}^0 \big| \tau\: g(\tau)\big| \: d\tau
\leq \int_{\scrM} \big|x^0\: f(x,0) \big|\: d^4x < \infty \:. \]
Hence in~\eqref{Ndoubleint} we may switch the order of integration according to Fubini's theorem to obtain
\begin{align*}
\int_{-\infty}^0 &dt \int_0^\infty dt' \,g(t-t')
= \int_{-\infty}^0 d\tau \,g(\tau) \int_{-\infty}^0 dt  \: \Theta(t-\tau) \\
&= \int_{-\infty}^0 d\tau \,g(\tau) \int_\tau^0 dt  = - \int_{-\infty}^0 d\tau \,\tau\, g(\tau) 
=  -\frac{1}{2} \int_{-\infty}^\infty d\tau \,\tau\, g(\tau) \:,
\end{align*}
where in the last step we used the anti-symmetry of~$g$.
Now  we insert~\eqref{NgFourier} and apply Plancherel's theorem,
\begin{align*}
\int_{-\infty}^\infty &dt \int_0^\infty dt' \,g(t-t')
=-\frac{i}{2} \int_{-\infty}^\infty d\tau \int_{-\infty}^\infty \frac{d\omega}{2 \pi} \:\hat{f}\big( (\omega, \vec{0}) \big)\:
\frac{\partial}{\partial \omega} e^{-i\omega \tau} \\
&=\frac{i}{2} \int_{-\infty}^\infty d\tau \int_{-\infty}^\infty \frac{d\omega}{2 \pi} \:
\Big(\frac{\partial}{\partial \omega} \hat{f}\big( (\omega, \vec{0}) \big) \Big)\:
e^{-i\omega \tau} = \frac{i}{2} \:\frac{\partial}{\partial \omega} \hat{f}\big( (\omega, \vec{0}) \big) \Big|_{\omega=0}\:.
\end{align*}
This concludes the proof.
\QED
In order to apply this lemma to our problem, we would have to show that the integrand
in~\eqref{Ninttask} satisfies the condition~\eqref{Nfinint}. As we shall now explain,
this condition will indeed {\em{not}} be satisfied, making it necessary to modify the method.

Let us specify the kernel~$Q(x,y)$. To this end, we make use of the fact
that the fermionic projector of the vacuum should correspond to a stable
minimizer of the causal action. This is made mathematically precise
in the so-called state stability analysis carried out in~\cite[Section~5.6]{PFP},
\cite{vacstab} and~\cite{reg}. The detailed analysis of the continuum
limit in~\cite[Chapter~3]{cfs} shows that in order to obtain well-defined field equations
in the continuum limit, the number of generations must be equal to three.
Therefore, we now consider an unregularized fermionic projector of the vacuum
involving a sum of three Dirac seas~\eqref{NPvac}.
The corresponding kernel~$Q(x,y)$ obtained in the continuum limit depends only on
the difference vector~$y-x$ and can thus be written as the Fourier transform
of a distribution~$\hat{Q}(k)$,
\[ Q(x,y) = \int \frac{d^4k}{(2 \pi)^4} \:\hat{Q}(k)\: e^{-ik (x-y)} \:. \]
The state stability analysis in~\cite[Section~5.6]{PFP} implies that the Fourier
transform~$\hat{Q}$ has the form as specified in the next definition
(cf.~\cite[Definition~5.6.2]{PFP}).

\begin{Def} \label{Ndef611}\index{state stability}
The fermionic projector of the vacuum~\eqref{NPvac} is called {\textbf{state stable}}
if the corresponding operator $\hat{Q}(k)$ is well-defined inside the
lower mass cone
\[ \mathcal{C}^\land := \{ k \in \R^4 \,|\, k^i k_i >0 \text{ and } k^0<0 \} \]
and can be written as
\beq \label{N62f}
\hat{Q} (k) = a\:\frac{k\slsh}{|k|} + b
\eeq
with continuous real functions $a$ and $b$ on $\mathcal{C}^\land$ having
the following properties:
\begin{itemize}
\item[\textrm{(i)}] $a$ and $b$ are Lorentz invariant,
\[ a = a(k^2)\:,\qquad b = b(k^2) \:. \]
\item[\textrm{(ii)}] $a$ is non-negative.
\item[\textrm{(iii)}] The function $a+b$ is minimal on the mass shells,
\beq \label{Nabmin}
(a+b)(m^2_\beta) = \inf_{q \in {\mathcal{C}}^\land} (a+b)(q^2) \quad\mbox{for~$\beta=1,2,3$}\:.
\eeq
\end{itemize}
\end{Def}
We point out that, according to this definition, the function~$\hat{Q}(k)$ does {\em{not}} need to be {\em{smooth}},
but only continuous. In particular, Lemma~\ref{Nlemmagen} cannot be
applied, because the derivative in~\eqref{Ndint} is ill-defined.
If~$\hat{Q}(k)$ were smooth, its Fourier transform~$Q(x,y)$ would decay rapidly as~$(y-x)^2 \rightarrow \pm \infty$.
In this case, $Q(x,y)$ would be of short range as explained in Section~\ref{Nsecsli},
except that~\eqref{Nshortrange} would have to be replaced by
the statement that~$\L(x,y)$ is very small if~$|(y-x)^2|>\delta$
(and~$\L(x,y)$ could indeed be made arbitrarily small by increasing~$\delta$).
The fact that~$\hat{Q}(k)$ does not need to be differentiable implies that~$Q(x,y)$
does not need to decay rapidly, also implying that the condition~\eqref{Nfinint} may be violated.

In fact, this non-smoothness in momentum space
will be of importance in the following computation. Moreover, our results will depend
only on the behavior~$\hat{Q}(k)$ in a neighborhood of the mass shells~$k^2=m_\beta^2$.
Therefore, the crucial role will be played by the regularity of~$\hat{Q}$ on the mass shells.
In order to keep the setting as simple as possible, we shall assume that the functions~$a$
and~$b$ in~\eqref{N62f} are {\em{semi-differentiable}} on the mass shells, meaning that the
left and right derivatives exist. For the resulting semi-derivatives of~$\hat{Q}$ we use the notation
\beq \label{Nsemidiff}
\begin{split}
\partial_\omega^+ \hat Q( -\omega_{\beta, \vec k}, \vec k) &= \lim_{h \searrow 0}\:
\frac{1}{h}\, \Big( \hat Q( -\omega_{\beta, \vec k}+h, \vec k) - \hat Q( -\omega_{\beta, \vec k}, \vec k) \Big)  \\
\partial_\omega^- \hat Q( -\omega_{\beta, \vec k}, \vec k) &= \lim_{h \nearrow 0}\:
\frac{1}{h}\, \Big( \hat Q( -\omega_{\beta, \vec k}+h, \vec k) - \hat Q( -\omega_{\beta, \vec k}, \vec k) \Big) \:,
\end{split} \eeq
where~$\omega_{\beta, \vec k}$ is given by the dispersion relation
\beq \label{Ndisperse}
\omega_{\beta, \vec k} = \sqrt{m_\beta^2 + |\vec{k}|^2}\:.
\eeq
The parameters~$c_\beta$ in Theorem~\ref{Nthmcurrentmink} are given by
\beq  \label{NUSymm20}
c_\beta := \partial^+_\omega  a(m_\beta^2) + \partial^+ _\omega b(m_\beta^2) + \partial^-_\omega  a(m_\beta^2) + \partial^-_\omega b(m_\beta^2)
\eeq
As explained above, even though the function $a+b$ is minimal at $m_\beta^2$,
it is in general not differentiable at this value. But the minimality implies that~$c_\beta \geq 0$.

The discontinuity of the derivatives of~$\hat{Q}$ on the mass shells implies that~$Q(x,y)$
will {\em{not}} decay rapidly as~$(y-x)^2 \rightarrow \pm \infty$. Instead, we obtain contributions which
decay only polynomially and oscillate on the {\em{Compton scale}}
(this oscillatory behavior comes about similar as explained for the Fourier transforms of the mass shells
in detail in~\cite[\S1.2.5]{cfs}).
Due to these oscillations on the Compton scale, the integrals in~\eqref{Ninttask}
are indeed well-defined, and the dominant contribution to the integrals will come from a
layer of width~$\sim m^{-1}$ around the hyperplane~$\{t=0\}$.
Therefore, although~$\L(x,y)$ does not decay rapidly,
the concept of the surface layer integral as introduced in Section~\ref{Nsecsli} remains valid,
and the parameter~$\delta$ shown in Figure~\ref{Nfignoether1} can be identified
with the Compton scale~$\sim m_\alpha^{-1}$ of the Dirac particles. Thus the width of the
surface layer is a small but
macroscopic length scale. In particular, the surface layer integrals cannot be
identified with or considered as a generalization of the surface integrals of the classical Noether theorem.
However, in most situations of interest, when the surface is almost flat on the
Compton scale, the surface layer integral can be well-approximated by a corresponding surface integral.
Theorem~\ref{Nthmcurrentmink} shows that in the limiting case that the surface is a hyperplane,
the surface layer integral indeed goes over to a surface integral.

The just-mentioned oscillatory behavior of the integrand in~\eqref{Ninttask} implies that
the integrals will in general not exist in the Lebesgue sense. But they do exist in the sense
of an improper Riemann integral. For computational purposes, this is implemented
most conveniently by inserting convergence-generating factors.
We begin with the simplest possible choice of a convergence-generating factor~$e^{-\eta |t|}$.
Thus instead of~\eqref{Ninttask} we consider the integral
\beq
J = \lim_{\eta \searrow 0} \int_{t \geq 0} d^4x \int_{t < 0} d^4y \:e^{-\eta x^0 + \eta y^0}\: 
\im \big( \Sl \psi^u(y) \,|\, Q(y,x) \,\psi^u(x) \Sr \big) \label{NCurreg1} \:.
\eeq

We now introduce a convenient representation for~$\hat{\psi}^u(k)$. Since the wave function~$\psi^u$ is
a linear combination of solutions of the Dirac equation corresponding to the masses~$m_\beta$
(with~$\beta=1,2,3$), its Fourier transform is supported on the mass shells~$k^2=m^2_\beta$.
Moreover, since in the Dirac sea vacuum all physical wave functions have negative frequency,
we can write~$\hat{\psi}^u(k) = (\hat{\psi}^u_\beta(k))_{\beta=1,2,3}$ as
\beq \label{NUSymm6}
\hat\psi^u_\beta(k) = 2 \pi\, \chi_\beta(\vec{k}) \: \delta \big( k^0 + \omega_{\beta,\vec k} \big)
\eeq
(with~$\omega_{\beta,\vec k}$ as in~\eqref{Ndisperse}).
The Dirac equations~\eqref{NDireqns} reduce to the algebraic equations
\beq \label{NDiralg}
(\slashed{k}_\beta - m_\beta) \chi_\beta(\vec{k}) = 0 \qquad \text{where} \qquad
k_\beta := \big( -\omega_{\beta,\vec{k}}, \vec{k} \big)\:.
\eeq
The representation~\eqref{NUSymm6} has the convenient feature that the wave function at time~$t$
is given by
\[ \psi^u_\beta(t,\vec{x}) = \int \frac{d^4k}{(2 \pi)^4}\: \hat\psi^u_\beta(k)\: e^{- i k x}
= e^{i \omega_{\beta,\vec k} t} \int \frac{d^3k}{(2 \pi)^3}\: \chi_\beta(\vec k)\: e^{i \vec{k} \vec{x}} \:, \]
showing that~$\chi_\beta(\vec{k})$ simply is the {\em{spatial}} Fourier transform of
the Dirac wave function at time zero.

\begin{Lemma} \label{Nlemmaoffdiagonal}
The integral~\eqref{NCurreg1}, can be written as
\begin{gather}
J = \sum_{\alpha, \beta=1}^3 J_{\alpha, \beta} \:, \label{NJsum} \\
\intertext{where the~$J_{\alpha, \beta}$ are given by}
J_{\alpha, \beta} = \lim_{\eta \searrow 0} \;\im \int \frac{d^4k}{(2 \pi)^4} \, \Sl  \chi_\alpha(\vec k ) \: \frac{ i}{k^0 + \omega_{\alpha,\vec k} + i \eta}  \:|\: \hat Q(k) \:  \chi_\beta(\vec k ) \: \frac{- i}{k^0 + \omega_{\beta,\vec k} - i \eta} \Sr \:. 
\label{NJabDef}
\end{gather}
\end{Lemma}
\Proof We first rewrite~\eqref{NCurreg1} as
\[ J = \lim_{\eta \searrow 0}\; \im \int d^4x \, \int d^4y \; \Sl \Theta(x^0)\,e^{-\eta x^0} \,  \psi^u(x) \,|\, Q(x,y)\, 
\Theta(-y^0) \, e^{\eta y^0} \,\psi^u(y) \Sr \:. \]
Since~$Q$ depends only on the difference vector~$y-x$, the~$y$-integration can be
regarded as a convolution in position space. We now rewrite this convolution as
a multiplication in momentum space. Setting
\[ \hat \psi_\eta^\pm (k) := \int \Theta_\eta (\pm y^0)\, \psi^u(y) \, e^{ i k y} \: d^4y \:, \]
where we introduced the ``regularized Heaviside function''
\[ \Theta_\eta(x) = \Theta(x)\, e^{-\eta x}\:, \]
we obtain
\[ J = \lim_{\eta \searrow 0}\;\im \int_{M} d^4x \,  \, \Sl \Theta_\eta(x^0) \,  \psi^u(x) \,|\, {\mathcal{F}}^{-1} \big( \hat{Q} \, \hat\psi_\eta^- \big)(x) \Sr \:, \]
where ${\mathcal{F}}^{-1}$ denotes the inverse Fourier transformation. Plancherel's theorem yields
\beq \label{NUSymm2}
J = \lim_{\eta \searrow 0}\;\im \int \frac{d^4k}{(2 \pi)^4} \, \Sl \hat \psi_\eta^+(k)  \,|\, \hat Q(k)
\,\hat \psi_\eta^-(k) \Sr \:.
\eeq

We next compute~$\hat \psi_\eta^\pm (k)$. Since multiplication in position space corresponds to
convolution in momentum space, we know that
\beq \label{Nhatppm}
\hat \psi_\eta^\pm (k) =  \int \frac{d\omega}{2 \pi} \:  \hat \Theta_\eta(\pm \omega) \: \hat \psi^u \big( k-(\omega,\vec 0) \big) \:.
\eeq
Here the Fourier transformation of the regularized Heaviside function is computed by
\beq \label{NhatT}
\hat \Theta_\eta (\omega) = \int_{-\infty}^\infty \Theta_\eta(t) \: e^{i  \omega t}\: dt = \frac{i}{\omega + i \, \eta} \:.
\eeq

Using~\eqref{NhatT} and~\eqref{NUSymm6} in~\eqref{Nhatppm}, we obtain
\[ \hat\psi_\eta^\pm ( k) = \bigg(\chi_\beta(\vec k ) \: \frac{i}{ \pm(k^0 + \omega_{\beta,\vec k}) + i \eta}
\bigg)_{\beta=1,2,3} \:. \]
Using these formulas in~\eqref{NUSymm2} gives the result.
\QED

The next lemma shows that the summands for~$\alpha \neq \beta$ drop out of~\eqref{NJsum}.
\begin{Lemma} The currents~\eqref{NJabDef} satisfy the relation
\[ \sum_{\alpha \neq \beta} J_{\alpha, \beta} = 0 \:. \]
\end{Lemma}
\Proof
In the case~$\alpha \neq \beta$, we know that~$\omega_{\alpha,\vec k} \neq \omega_{\beta,\vec k}$,
so that in~\eqref{NJabDef} there are two single poles at~$k^0 = -\omega_{\alpha, \vec{k}}-i\eta$
and~$k^0 = -\omega_{\beta, \vec{k}}+i\eta$. This makes it possible to take the limit~$\eta \rightarrow 0$
using the formula
\[ \lim_{\eta \searrow 0} \frac{1}{x \pm i \eta} = \frac{\text{PP}}{x} \mp i \pi \, \delta(x) \]
(where PP denotes the principal value). We thus obtain
\begin{align*}
J_{\alpha, \beta} \;=& -\im \int_{M} \frac{d^4k}{(2 \pi)^4} \, 
\frac{\text{PP}}{k^0 + \omega_{\alpha,\vec k}}  \:\frac{\text{PP}}{k^0 + \omega_{\beta,\vec k}} \;
\Sl  \chi_\alpha(\vec k ) \:|\: \hat Q(k) \:  \chi_\beta(\vec k ) \Sr \\
&- \im \int_{M} \frac{d^4k}{(2 \pi)^4} \, \Sl  \chi_\alpha(\vec k ) \: \big(- i \pi\, \delta(k^0 + \omega_{\alpha,\vec k}) \big) \:|\:  \hat Q(k) \: \chi_\beta(\vec k ) \: \frac{\text{PP}}{k^0 + \omega_{\beta,\vec k}} \Sr \\
&- \im \int_{M} \frac{d^4k}{(2 \pi)^4} \, \Sl  \chi_\alpha(\vec k ) \: \frac{\text{PP}}{k^0 + \omega_{\alpha,\vec k}}  \:|\: \hat Q(k) \:  \chi_\beta(\vec k ) \: \big(i \pi\, \delta(k^0 + \omega_{\beta,\vec k}) \big) \Sr \:.
\end{align*}
Carrying out the $k^0$-integration in the last two lines gives
\begin{align}
J_{\alpha, \beta} \;=& -\im \int_{M} \frac{d^4k}{(2 \pi)^4} \, 
\frac{\text{PP}}{k^0 + \omega_{\alpha,\vec k}}  \:\frac{\text{PP}}{k^0 + \omega_{\beta,\vec k}} \;
\Sl  \chi_\alpha(\vec k ) \:|\: \hat Q(k) \:  \chi_\beta(\vec k ) \Sr \notag \\
&+\pi \re \int_{M} \frac{d^3k}{(2 \pi)^4} \, \Sl  \chi_\alpha(\vec k ) \:  |\:  \hat
Q \big( -\omega_{\alpha,\vec k}, \vec{k} \big) \: \chi_\beta(\vec k ) \: \frac{\text{PP}}{-\omega_{\alpha,\vec k} + \omega_{\beta,\vec k}} \Sr \notag \\
&+\pi \re \int_{M} \frac{d^3k}{(2 \pi)^4} \, \Sl  \chi_\alpha(\vec k ) \: \frac{\text{PP}}{-\omega_{\beta,\vec k} + \omega_{\alpha,\vec k}}  \:|\: \hat Q\big( -\omega_{\beta,\vec k}, \vec{k} \big) \:  \chi_\beta(\vec k ) \Sr \notag \\
=& -\im \int_{M} \frac{d^4k}{(2 \pi)^4} \, 
\frac{\text{PP}}{k^0 + \omega_{\alpha,\vec k}}  \:\frac{\text{PP}}{k^0 + \omega_{\beta,\vec k}} \;
\Sl  \chi_\alpha(\vec k ) \:|\: \hat Q(k) \:  \chi_\beta(\vec k ) \Sr \label{Naneqb1} \\
&+\pi \re \int_{M} \frac{d^3k}{(2 \pi)^4} \:\frac{\text{PP}}{\omega_{\alpha,\vec k} - \omega_{\beta,\vec k}} \notag \\
&\qquad \qquad \quad \times \Sl  \chi_\alpha(\vec k ) \:  |\:  \Big( \hat{Q} \big( -\omega_{\beta,\vec k}, \vec{k} \big) - \hat{Q} \big( -\omega_{\alpha,\vec k}, \vec{k} \big) \Big) \chi_\beta(\vec k ) \: \Sr \:. \label{Naneqb2}
\end{align}
Obviously, the contribution~\eqref{Naneqb1} is anti-symmetric when exchanging~$\alpha$ and~$\beta$.
In the contribution~\eqref{Naneqb2}, on the other hand,
we can use the Dirac equation~\eqref{NDiralg} together with~\eqref{N62f} to rewrite
the spin scalar product as
\[ \Sl  \chi_\alpha(\vec k ) \:  |\:  \big( (a+b)(m_\beta^2) - (a+b)(m_\alpha^2) \big) \chi_\beta(\vec k ) \: \Sr \:, \]
and this vanishes by~\eqref{Nabmin}. This gives the result.
\QED

Using this lemma, our conserved integral~\eqref{NJsum} simplifies to
\beq \label{NJsum2}
J = \sum_{\beta=1}^3 J_{\beta, \beta} \:.
\eeq
We now compute~$J_{\beta, \beta}$. First,
\begin{align}
& J_{\beta, \beta}
= \lim_{\eta \searrow 0} \;\im \int_{M} \frac{d^4k}{(2 \pi)^4} \, \Sl  \chi_\beta(\vec k ) \: \frac{ i}{k^0 + \omega_{\beta,\vec k} + i \eta}  \:|\: \hat Q(k) \:  \chi_\beta(\vec k ) \: \frac{- i}{k^0 + \omega_{\beta,\vec k} - i \eta}   \Sr \nonumber \\
&\:= \lim_{\eta \searrow 0} \int_{M} \frac{d^4k}{(2 \pi)^4} \, \Sl  \chi_\beta(\vec k ) \:|\: \hat Q(k) \: \frac{1}{2i} \bigg( \frac{-1}{( k^0 + \omega_{\beta,\vec k} - i \eta)^2} - \frac{-1}{( k^0 + \omega_{\beta,\vec k} + i \eta)^2} \bigg)
\chi_\beta(\vec k ) \: \Sr \nonumber  \\
&\,\overset{(\star)}{=} -\lim_{\eta \searrow 0} \int \frac{d^3k}{(2 \pi)^2} \int_{-\infty}^\infty \frac{dq}{2 \pi} \, \Sl 
\chi_\beta(\vec k ) \: |\: \hat Q 
\big(q - \omega_{\beta,\vec k}, \vec k \big) \: \frac{1}{2i} \bigg( \frac{1}{( q + i \eta)^2} - \frac{1}{( q - i \eta)^2} \bigg)
\chi_\beta(\vec k ) \: \Sr \nonumber  \\
&\:= - 2 \lim_{\eta \searrow 0} \int \frac{d^3k}{(2 \pi)^3} \, \Sl  \chi_\beta(\vec k ) \: |\: \int_{-\infty}^\infty \frac{dq}{2 \pi} \:\bigg( \hat Q\big(q - \omega_{\beta,\vec k}, \vec k \big) \;  \frac{q \, \eta}{( q^2 +\eta^2)^2} \bigg)  \chi_\beta(\vec k ) \: \Sr \:,  \label{NCurrentCons1}
\end{align}
where in~$(\star)$ we introduced the variable~$q=k^0 + \omega_{\beta,\vec k}$.
We now use~\eqref{Nsemidiff} to expand~$\hat{Q}$ for small~$q$ according to
\begin{align*}
&\hat{Q}\big( q- \omega_{\beta,\vec k} , \vec k\big) \\
&=  \, \hat Q\big(- \omega_{\beta,\vec k} , \vec k \big)
+ q  \, \Theta(q) \:\partial^+_\omega \hat Q\big( - \omega_{\beta,\vec k} , \vec k \big)
+ q \, \Theta(-q)  \: \partial^-_\omega\hat Q\big( - \omega_{\beta,\vec k} , \vec k \big) + o(q)
\end{align*}
(where $o(q)$ is the usual remainder term).
Substituting this Taylor expansion into~\eqref{NCurrentCons1}, the constant term
of the expansion drops out because the integrand is odd.
For the left and right derivatives, the integral can be carried out explicitly using that
\beq \label{Nexint}
\int_0^\infty \frac{q^2 \,\eta}{(q^2+\eta^2)^2} \, dq = \frac{\pi}{4}
= \int_{- \infty}^0 \frac{q^2 \,\eta}{(q^2+\eta^2)^2} \, dq \:.
\eeq
Thus, disregarding the remainder term, we obtain
\beq \label{NUSymm10}
J_{\beta, \beta}= - \frac{1}{4} \int \frac{d^3k}{(2 \pi)^3} \, \Sl  \chi_\beta(\vec k ) \: |\: 
\Big( (\partial^+_\omega + \partial^-_\omega) \hat Q\big( - \omega_{\beta,\vec k} , \vec k \,\big) \Big)
 \chi_\beta(\vec k ) \: \Sr \:.
\eeq
This formula corresponds to the result of Lemma~\ref{Nlemmagen} 
in our setting where~$\hat{Q}(k)$ is not differentiable on the mass shells.

It remains to analyze the remainder term. Naively, the integrated remainder term is of the order~$\eta$ and
should thus vanish in the limit~$\eta \searrow 0$. This could indeed be proved if
we knew for example that the
function~$\hat{Q}( \,.\, - \omega_{\beta,\vec k}, \vec k )$ is integrable.
However, since~$\hat{Q}$ is only defined on the lower mass cone (see Definition~\ref{Ndef611}),
such arguments cannot be applied. Our method for avoiding this technical problem is to
work with a convergence-generating factor with compact support in momentum space. To this end,
we choose a non-negative test function~$\hat{g} \in C^\infty_0((-1,1))$ with~$\hat{g}(-\omega)=\hat{g}(\omega)$
for all~$\omega \in \R$ and~$\int_\R \hat{g}(\omega)\, d\omega =2 \pi$.
For given~$\sigma >0$ we set
\[ \hat{g}_\sigma(\omega) = \frac{1}{\sigma}\: \hat{g} \Big( \frac{\omega}{\sigma} \Big) \qquad \text{and} \qquad
g_\sigma(t) = \int_{-\infty}^\infty \frac{d\omega}{2 \pi}\: \hat{g}_\sigma(\omega)\: e^{-i \omega t} \:. \]
In the limit~$\sigma \searrow 0$, the functions~$g_\sigma(t)$ go over to the constant function one.

\begin{Lemma} \label{Nlemmacur1} Replacing~\eqref{NCurreg1} by
\beq \label{NJnew}
J = \lim_{\sigma \searrow 0} \int_{t \geq 0} d^4x \int_{t < 0} d^4y \:g_\sigma(x^0) \: g_\sigma(y^0)\:
\im \big( \Sl \psi^u(y) \,|\, Q(y,x) \,\psi^u(x) \Sr \big) \:,
\eeq
the resulting function~$J$ is of the form~\eqref{NJsum2} with~$J_{\beta, \beta}$
as given by~\eqref{NUSymm10}.
\end{Lemma}
\Proof Again rewriting~\eqref{NJnew} in momentum space and using that~$\hat{g}$ has compact support,
one sees that the resulting integrand of~$J_{\alpha, \beta}$ is well-defined for any~$\vec{k}$
for sufficiently small~$\sigma$.
In order to relate the functions~$g_\sigma$ in~\eqref{NJnew} to the
factor~$e^{-\eta x^0 + \eta y^0}$ in~\eqref{NCurreg1}, it is most convenient to work with
the Laplace transform. Thus we represent the functions~$g_\sigma$ in~\eqref{NJnew}
for~$x^0>0$ and~$y^0<0$ as
\[ g_\sigma(x^0) = \frac{1}{\sigma} \int_0^\infty h\Big( \frac{\eta}{\sigma} \Big)\, e^{-\eta x^0}\: d\eta
\qquad \text{and} \qquad
g_\sigma(y^0) = \frac{1}{\sigma} \int_0^\infty h\Big( \frac{\tilde{\eta}}{\sigma} \Big)\, e^{\tilde{\eta} y^0}\:
d\tilde{\eta} \:, \]
where~$h$ is the inverse Laplace transform of~$g$
(for basics on the Laplace transform see for example~\cite{davies}).
A straightforward computation shows that the result of Lemma~\ref{Nlemmaoffdiagonal}
remains valid with the obvious replacements.
The computation of~$J_{\beta, \beta}$, on the other hand, needs to be modified as follows. 
Formula~\eqref{NCurrentCons1} remains valid after the replacement
\[ \lim_{\eta \searrow 0} \;\cdots\; \frac{q \, \eta}{( q^2 +\eta^2)^2} \;\longrightarrow\;
 \lim_{\sigma \searrow 0} \frac{1}{\sigma^2} \int_0^\infty h\Big( \frac{\eta}{\sigma} \Big)\:d\eta
\int_0^\infty h\Big( \frac{\tilde{\eta}}{\sigma} \Big)\: d\tilde{\eta}
\;\cdots\; \frac{q \, (\eta+\tilde{\eta})}{2 ( q^2 +\eta^2) ( q^2 +\tilde{\eta}^2)} \:. \]
Substituting the Taylor expansion of~$\hat{Q}$, the first integral in~\eqref{Nexint}
is to replaced by the integral
\[ \int_0^\infty \frac{q^2 \, (\eta+\tilde{\eta})}{2 ( q^2 +\eta^2) ( q^2 +\tilde{\eta}^2)}  = \frac{\pi}{4} \]
(and similarly for the second integral in~\eqref{Nexint}).
In this way, one again obtains~\eqref{NUSymm10}, but now the remainder term vanishes in
the limit~$\sigma \searrow 0$. %
\QED

We now compute~$J_{\beta, \beta}$ more explicitly.
\begin{Lemma} \label{Nlemmacur2}
The currents~$J_{\beta, \beta}$ given by~\eqref{NUSymm10} can be written as
\beq \label{NJfinal}
J_{\beta, \beta} = - \frac{m_\beta \,c_\beta}{2} \int_{\R^3} \Sl \psi_\beta^u(x) | \gamma^0 \psi_\beta^u(x) \Sr\:
d^3x
\eeq
with the constants~$c_\beta$ as in~\eqref{NUSymm20}.
\end{Lemma}
\Proof Using~\eqref{N62f} and applying the chain rule for semi-derivatives, we obtain
\beq \label{NUSymm3}
\partial^\pm_\omega \, \hat Q \Big( - \omega_{\beta,\vec k}, \vec k \Big) 
= - 2 \omega_{\beta,\vec k} \Big(\partial^\pm_\omega  a(k_-^2) \frac{\slashed{k}_-}{| k_- | } + \partial^\pm_\omega b(k_-^2) \Big) + a(k_-^2) \, \frac{\partial}{\partial k^0} \Big( \frac{\slashed{k} }{|k|} \Big) \Big|_{k=k_-},
\eeq
where we set~$k_- = (- \omega_{\beta,\vec k} , \vec k )$ and~$|k_-| = \sqrt{k_-^2}=m_\beta$.
This formula can be further simplified when taking the expectation value with the spinor~$\chi_\beta(\vec{k})$:
In the last summand in~\eqref{NUSymm3}, we first compute the~$k$-derivative,
\[ \frac{\partial}{\partial k^0} \frac{\slashed{k} }{|k|} \Big|_{k = k_-} = \frac{\gamma^0}{m_\beta} 
- \frac{\slashed{k}_-}{|k_-|^3} \, k_-^0 \:. \]
Taking the expectation value with the spinor~$\chi_\beta(\vec{k})$ and using the Dirac equation
\[ (\slashed{k}_- - m_\beta) \chi_\beta(\vec{k})=0\:, \]
we obtain the relations
\begin{gather}
\Sl \chi_\beta(\vec{k}) | \slashed{k}_-\, \chi_\beta(\vec{k}) \Sr = m_\beta\: \Sl \chi_\beta(\vec{k}) | \chi_\beta(\vec{k}) \Sr \\
2 m_\beta \;\Sl \chi_\beta(\vec{k}) | \gamma^0 \chi_\beta(\vec{k}) \Sr
= \Sl \chi_\beta(\vec{k}) | \big\{ \slashed{k}_-, \gamma^0 \big\} \chi_\beta(\vec{k}) \Sr
= -2 \omega_{\beta,\vec k} \, \Sl \chi_\beta(\vec{k}) | \chi_\beta(\vec{k}) \Sr \:. \label{Nrel2}
\end{gather}
In this way, the last summand in~\eqref{NUSymm3} gives zero.
In the remaining first summand in~\eqref{NUSymm3}, we again employ the Dirac
equation~$(\slashed{k}_- - m_\beta) \chi_\beta(\vec{k})=0$ to obtain
\begin{align*}
\bigg(\partial^\pm_\omega  a(k_-^2) \frac{\slashed{k}_-}{| k_- | }
+ \partial^\pm_\omega b(k_-^2) \bigg) \chi_\beta(\vec{k}) =
\partial^\pm_\omega  \Big(a(k_-^2) + b(k_-^2) \Big) \chi_\beta(\vec{k}) \:.
\end{align*}
We conclude that
\beq  \label{NUSymm19}
J_{\beta, \beta} = \frac{1}{2} \, c_\beta    \int \frac{d^3k}{(2 \pi)^3} \,\omega_{\beta,\vec k}  \:
\Sl \chi_\beta(\vec k) |  \chi_\beta(\vec k ) \Sr
\eeq
with~$c_\beta$ as in~\eqref{NUSymm20}.
We finally use~\eqref{Nrel2} and apply Plancherel's theorem.
\QED

Combining Lemma~\ref{Nlemmacur1} and Lemma~\ref{Nlemmacur2} 
gives the conservation law~\eqref{Nthmcurr1}. This concludes
the proof of Theorem~\ref{Nthmcurrentmink}.

\subsection{Clarifying Remarks} \label{Nsecremark}
The following remarks explain and clarify various aspects of
the above constructions and results.
\begin{Remark}\label{Nremdiscuss}{\em{ {\Remt{(Differentiability of variations)}} 
We now explain in which sense the the differentiability assumption 
on the function~$\ell \circ \Phi$ in Theorem~\ref{Nthmcurrent} is satisfied.
First, the above computations show that, working with the specific form of~$\hat{Q}$ in the continuum limit,
the $\tau$-derivative exists and is finite.
However, this does not necessarily imply that for any UV regularization,
the corresponding local minimizers~$(\H, \F, \rho^\varepsilon)$
also satisfy the differentiability assumptions on the function~$\ell \circ \Phi$ in
Theorem~\ref{Nthmcurrent}. Indeed, thinking of a lattice regularization,
we expect that the function~$\ell \circ \Phi$ with~$\Phi$ according to~\eqref{Nvarunit}
and~\eqref{NUtaudef} will typically {\em{not}} be continuously differentiable in~$\tau$
(because in this case, $\ell$ is a sum of terms involving the Lagrangian, which is only
Lipschitz continuous). In order to bypass this technical problem, 
for a given local minimizer~$(\H, \F, \rho^\varepsilon)$
one can modify~$\Phi$ such as to obtain a variation~$\Phi^\varepsilon$ for which the
function~$\ell \circ \Phi^\varepsilon$
is continuously differentiable in~$\tau$.
For this modified variation, we have the conservation law of Theorem~\ref{Nthmcurrent}.
The strategy is to choose the~$\Phi^\varepsilon$ for every~$\varepsilon>0$
in such a way that in the limit~$\varepsilon \searrow 0$, the variations converge
in a suitable weak topology to the variation~$\Phi_\tau$ as given by~\eqref{Nvarunit}
and~\eqref{NUtaudef}. In non-technical terms, we modify~$\Phi_\tau$ by ``microscopic fluctuations''
in such a way that the functions~$\ell \circ \Phi^\varepsilon$ become differentiable in~$\tau$
for all~$\varepsilon>0$. In the limit~$\varepsilon \searrow 0$, the microscopic fluctuations
should drop out to give Theorem~\ref{Nthmcurrentmink}.

At present, this procedure cannot be carried out because, so far,
no local minimizers~$(\H, \F, \rho^\varepsilon)$ have been constructed which describe
regularized Dirac sea configurations. The difficulty is to arrange the regularization in such a way
that the EL equations are satisfied without error terms. A first step towards the construction of such ``optimal
regularizations'' is given in~\cite{reg}. %
\QEDrem }}
\end{Remark}

\begin{Remark}\label{Nremweights}{\em{ {\Remt{(Weight factors)}} 
As explained in~\cite[Section~2 and Appendix~A]{reg}, one may introduce positive
weight factors~$\rho_\beta$ into the ansatz~\eqref{NPvac},
\[ P(x,y) = \sum_{\beta=1}^3 \rho_\beta \int \frac{d^4k}{(2 \pi)^4}\: (\slashed{k}+m_\beta)\:
\delta \big(k^2-m_\beta^2 \big)\: e^{-ik(x-y)}\:. \]
The above analysis immediately extends to this situation simply by inserting suitable factors of~$\rho_\beta$
into all equations. In particular, the resulting conserved quantity~\eqref{NJfinal} becomes
\[ J_{\beta, \beta} = - \frac{\rho_\beta\, m_\beta \,c_\beta}{2} \int_{\R^3} \Sl \psi_\beta^u(x) | \gamma^0
\psi_\beta^u(x) \Sr\: d^3x \:. \]
Consequently, the conserved current in~\eqref{Nthmcurr1} is to be modified to
\[ \sum_{\beta=1}^3 \rho_\beta\, m_\beta\, c_\beta \int_{t=\text{const}} \!\!\!\!\!\!\!\!\Sl \psi_\beta^u(x)
| \gamma^0 \psi_\beta^u(x) \Sr\:d^3x\:. \]
The role of the weight factors in the interacting case will be explained in the next remark.
\QEDrem }} \end{Remark}

\begin{Remark}\label{Nreminteract}{\em{ {\Remt{(Interacting systems)}} 
We point out that for the derivation of Theorem~\ref{Nthmcurrentmink}, we
worked with the vacuum Dirac equations~\eqref{NDireqns}, so that no
interaction is present. In particular, the generations have an independent dynamics,
implying that current conservation holds separately for each generation, i.e.
\beq \int_{t=t_0} \!\!\!\!\Sl \psi_\beta^u(x) | \gamma^0 \psi_\beta^u(x) \Sr\:d^3x
= \int_{t=t_1} \!\!\!\!\Sl \psi_\beta^u(x) | \gamma^0 \psi_\beta^u(x) \Sr\:d^3x \quad
\text{for all~$\beta=1,2,3$}\:. \label{Nconssep}
\eeq
Let us now discuss the typical situation of a scattering process in which the
Dirac equations~\eqref{NDireqns} only hold asymptotically as~$t \rightarrow \pm \infty$.
In this case, choosing~$\Omega$ so large that it contains the interaction region,
one can compute the surface layer integrals again for the free Dirac equation
to obtain the conservation law~\eqref{Nthmcurr1}, where~$t_0$ lies in the past
and~$t_1$ in the future of the interaction region.
In this way, the conservation law of Theorem~\ref{Nthmcurrentmink}
immediately extends to interacting systems.

In this interacting situation, current conservation no longer holds for each
generation separately (thus~\eqref{Nconssep} is violated). Instead,
as a consequence of the Dirac dynamics, only the total charge
\beq \label{Ntotcharge}
\sum_{\beta=1}^3
\int_{t=\text{const}} \!\!\!\!\!\!\!\!\Sl \psi_\beta^u(x) | \gamma^0 \psi_\beta^u(x) \Sr\:d^3x
\eeq
is conserved.
In order for this conservation law to be compatible with~\eqref{Nthmcurr1}, we need
to impose that
\beq \label{Ncc1}
m_\alpha\, c_\alpha = m_\beta\, c_\beta \qquad \text{for all~$\alpha,\beta=1,2,3$}\:.
\eeq
This is a mathematical consistency condition which gives information on the possible
form of the distribution~$\hat{Q}(k)$ in the continuum limit (as specified in
Definition~\ref{Ndef611} above).
If weight factors are present (see Remark~\ref{Nremweights} above), this consistency condition
must be modified to
\beq \label{Ncc2}
\rho_\alpha\, m_\alpha\, c_\alpha = \rho_\beta\, m_\beta\, c_\beta \qquad \text{for all~$\alpha,\beta=1,2,3$}\:.
\eeq
The conditions~\eqref{Ncc1} and~\eqref{Ncc2} are crucial for the future project of
extending the state stability analysis in~\cite{vacstab} to systems involving neutrinos.
\QEDrem }} \end{Remark}

\begin{Remark}\label{Nremnorm}{\em{ {\Remt{(Normalization of the fermionic projector)}} 
The conservation law of Theorem~\ref{Nthmcurrentmink} has an important implication
for the normalization of the fermionic projector, as we now explain.
As worked out in detail in~\cite{norm}, there are two alternative normalization methods 
for the fermionic projector: the spatial normalization and the mass normalization.
In~\cite[Section~2.2]{norm} the advantages of the spatial normalization are discussed,
but no decisive argument in favor of one of the normalization methods is given.
Theorem~\ref{Nthmcurrentmink} decides the normalization problem in favor of the
spatial normalization. Namely, this theorem shows that the dynamics as described by the causal
action principle gives rise to a conservation law which in the continuum limit reduces
to the spatial integrals~\eqref{Nthmcurr1}. As explained in Remark~\ref{Nremweights} above,
the mathematical consistency to the Dirac dynamics implies that~\eqref{Nthmcurr1}
coincides with the conserved total charge~\eqref{Ntotcharge}.
The resulting conservation law is compatible with the spatial normalization, but contradicts
the mass normalization. We conclude that the spatial normalization of the fermionic projector
is indeed the correct normalization method which reflects the intrinsic conservation laws of the
causal fermion system.
\QEDrem }} \end{Remark}

\section{Example: Conservation of Energy-Momentum} \label{NsecexEM}
The conservation laws in Theorem~\ref{Nthmsymmgis2} also give rise to the
{\em{conservation of energy and momentum}}, as will be worked out in this section.

\subsection{Generalized Killing Symmetries and Conservation Laws}
In the classical Noether theorem, the conservation
laws of energy and momentum
are a consequence of space-time symmetries described most conveniently with
the notion of Killing fields. Therefore, one of our tasks is to extend this notion to the setting of causal
fermion systems. In preparation, we recall the procedure in the classical Noether theorem
from a specific point of view:
In the notion of a Killing field, one distinguishes the background geometry from the
additional particles and fields. The background geometry must have a symmetry
as described by the Killing equation. The additional particles and fields, however, do not
need to have any symmetries. Nevertheless, one can construct a symmetry of the whole system
by actively transporting the particles and fields along the flow lines of the Killing field.
The conservation law corresponding to this symmetry transformation gives rise to
the conservation of energy and momentum.

In a causal fermion system, there is no clear-cut distinction between the background geometry
and the particles and fields of the system, because all of these structures are inherent in the
underlying causal fermion system and mutually influence each other via the causal action principle.
Therefore, instead of working with a symmetry of the background geometry,
we shall work with the notion of an approximate symmetry. By actively transforming
those physical wave functions which do not respect the symmetry,
such an approximate symmetry again gives rise to an exact
symmetry transformation, to which our Noether-like theorems apply.

More precisely, one begins with a $C^1$-family of transformations~$(f_\tau)_{\tau \in(-\tau_{\max}, \tau_{\max})}$ of space-time,
\[ f_\tau \::\: M \rightarrow M \qquad \text{with} \qquad f_0 = \1 \:, \]
which preserve the universal measure in the sense that~$(f_\tau)_* \rho = \rho$.
This family can be regarded as the analog of the flow in space-time along a classical Killing field.
Moreover, one considers a family of unitary transformations~$(\scrU_\tau)_{\tau \in (-\tau_{\max}, \tau_{\max}) }$
on~$\H$ with the property that
\beq \label{NPropU}
\scrU_{-\tau} \,\scrU_\tau = \1 \qquad \text{for all~$\tau \in (-\tau_{\max}, \tau_{\max})$}\:, \eeq
and defines the variation
\beq \label{NEMEq3}
\Phi \::\: (-\tau_{\max}, \tau_{\max}) \times M \rightarrow \F \:,\qquad
\Phi(\tau,x) := \scrU_\tau \, x \, \scrU^{-1}_\tau \:.
\eeq
Combining these transformations should give rise to an
{\em{approximate symmetry}} of the wave evaluation operator~\eqref{Nweo}
in the sense that if we compare the transformation of the space-time point with the unitary transformation
by setting
\beq \label{NEdef}
E_\tau(u,x) := (\Psi u)\big(f_\tau(x) \big) - \big( \Psi \scrU^{-1}_\tau u \big)(x) \qquad (x \in M, u \in \H) \:,
\eeq
then the operator~$E_\tau : \H \rightarrow C^0(M, SM)$ should be so small
that the first variation is well-defined in the continuum limit (for details see Section~\ref{NCorrDiracEM} below).
There are various ways in which this smallness condition could be
formulated. We choose a simple method which is most convenient for our purposes.

\begin{Def} \label{Ndefkilling}
The transformation $(f_\tau)_{\tau \in (-\tau_{\max},\tau_{\max})}$ is called a {\textbf{Killing symmetry with finite-dimensional support}} of the causal fermion system if it is a symmetry of the universal measure that preserves the trace (see Definition~\ref{Ndefsymmnc}) and if there exists a finite-dimensional subspace $K \subset \H$ and 
a family of unitary operators~$(\scrU_\tau)_{\tau \in (-\tau_{\max},\tau_{\max})}$ with the property~\eqref{NPropU} such that
\beq \label{NCondSmall}
E_\tau(u,x) = 0 \qquad \text{for all~$u \in K^\perp$ and~$x \in M$}\:.
\eeq
\end{Def}

We now formulate a general conservation law.
\begin{Thm} \label{NThmKillingCons} Let $\rho$ be a local minimizer (see Definition~\ref{Ndeflocmin}) and $(f_\tau)_{\tau \in (-\tau_{\max},\tau_{\max})}$ be a Killing symmetry of the causal fermion system. Then the following conservation law holds:
\begin{align} \label{NKillingConserve} \begin{split}
\frac{d}{d\tau} \int_\Omega d\rho(x) \int_{M \setminus \Omega} d\rho(y) \:\Big( &\L_\kappa \big( f_\tau(x), y\big) -  \L_\kappa \big(x, f_\tau( y)\big) \\
&-\L_\kappa \big( \Phi_\tau(x), y\big) + \L_\kappa \big( x,\Phi_\tau( y) \big) \Big) \Big|_{\tau=0} = 0 \:.
\end{split}
\end{align}
\end{Thm}
\Proof Again using Lemma~\ref{NPhiSymmLag}, we know that the
variation~\eqref{NEMEq3} is a symmetry of the Lagrangian. Hence
\begin{align*}
\int_M d\rho(x) \int_\Omega d\rho(y) \, \L_\kappa \big(\Phi_\tau(x),y \big) =&  \int_M d\rho(x) \int_\Omega d\rho(y) \, \L_\kappa \big( x, \Phi_{-\tau}(y) \big) \\
=&  \int_\Omega d\rho(y) \: \ell \big(\Phi_\tau (y) \big) \:.
\end{align*}
Using this equation in Proposition~\ref{Nprpuseful2}, we obtain
\beq \label{Nphirel}
0 = \frac{d}{d \tau}  \int_\Omega d\rho(x) \, \int_{M\setminus \Omega} d\rho(y)\; \Big( \L_\kappa \big(\Phi_\tau(x),y
\big) - \L \big( x,\Phi_\tau(y) \big) \Big) \Big|_{\tau = 0} \:.
\eeq
For the transformations~$f_\tau$, on the other hand, we have the relations
\[ \int_M \L_\kappa \big(f_\tau(x) , y \big) \: d\rho(x) = \int_\F \L_\kappa(z,y)\: d\big( (f_\tau)_\ast \rho \big) (z) 
= \int_M  \L_\kappa(x,y) \:d\rho(x) \:, \]
where in the last step we used that $f_\tau$ is a symmetry of the universal measure. Since $f_\tau$ also
preserves the trace, it is a generalized integrated symmetry (see Definition~\ref{Ndefgis2}). Applying
Theorem~\ref{Nthmsymmgis2}, we obtain
\beq \label{Nfrel}
\frac{d}{d\tau} \int_\Omega d\rho(x) \int_{M \setminus \Omega} d\rho(y) \:\Big( \L_\kappa\big( f_\tau(x), y\big) -
\L_\kappa \big( x, f_\tau(y) \big) \Big) \Big|_{\tau=0} = 0 \:.
\eeq
Subtracting~\eqref{Nfrel} from~\eqref{Nphirel} gives the result.
\QED

We remark that the vector field~$w := \delta f$ is tangential to~$M$ and describes
a transformation of the space-time points. The variation~$\delta \Phi$, on the other hand,
is a vector field in~$\F$ along~$M$. It will in general not be tangential to~$M$.
The difference vector field~$v := w - \delta \Phi$ 
can be understood as an active transformation of all the objects in space-time which
do not have the space-time symmetry (similar to the parallel transport of the particles and fields along the flow lines of the Killing field in the classical Noether theorem as described above).
The variation of the integrand in~\eqref{NKillingConserve} can be rewritten as a variation
in the direction~$v$; for example,
\[ \frac{d}{d\tau} \Big( \L_\kappa \big( f_\tau(x), y\big) -\L_\kappa \big( \Phi_\tau(x), y\big) \Big) \Big|_{\tau=0}
= \delta_{v(x)} \L_\kappa(x,y) \:. \]
Expressing~$v$ in terms of the operator~$E$ in~\eqref{NEdef} and using~\eqref{NCondSmall}
will show that~$v$ is indeed so small (in a suitable sense) that the 
corresponding variation of the Lagrangian will be well-defined and finite.

\subsection{Correspondence to the Dirac Energy-Momentum Tensor} \label{NCorrDiracEM}
In order to get the connection to the conservation of energy and momentum,
as in Section~\ref{Nseccurcor} we consider the vacuum Dirac equation and the limiting case
that~$\Omega$ exhausts the region between two Cauchy surfaces~$t=t_0$ and~$t=t_1$
(see Figure~\ref{Nfignoether2}). Recall that the energy-momentum tensor of a Dirac wave function~$\psi$
is given by
\[ T_{jk} = \frac{1}{2}\:
\re  \left( \Sl  \psi |   \, \gamma_j\, i \partial_k \psi \Sr + \Sl  \psi |   \, \gamma_k\, i \partial_j \psi \Sr \right)
= - \im  \Sl  \psi |   \, \gamma_{(j}\, \partial_{k)} \psi \Sr \:. \]
We consider the situation of the vacuum Dirac sea with a finite number of holes
describing the anti-particle states~$\phi_1, \ldots, \phi_{\na}$
(for the description of particle states see again Section~\ref{Nremmicro}).
The effective energy-momentum tensor is minus the sum of the energy-momentum tensors
of all the anti-particle states. Thus for a fixed value of the generation index~$\beta$, we set
\[ (T_\beta)_{jk} = \sum_{i=1}^{\na} \im  \Sl  \phi_{i,\beta} |   \, \gamma_{(j}\, \partial_{k)} \phi_{i, \beta} \Sr \:. \]
In order to treat the generations, as in Theorem~\ref{Nthmcurrentmink} we take a linear combination
involving the non-negative constants~$c_\beta$ introduced in~\eqref{NUSymm20}.

\begin{Thm} {\Thmt{(Energy conservation)}} \label{NthmEMcons}
Let~$(\H, \F, \rho^\varepsilon)$ be local minimizers of the causal action
describing the Minkowski vacuum~\eqref{NPvac} together with particles and anti-particles.
Considering the limiting procedure explained in Figure~\ref{Nfignoether2} and taking the
continuum limit, the conservation law of Theorem~\ref{NThmKillingCons} goes over to
\[ \sum_{\beta=1}^3 m_\beta c_\beta \int_{t=t_0} d^3x \, (T_\beta)^0_0\: d^3x = 
\sum_{\beta=1}^3 m_\beta c_\beta \int_{t=t_1} d^3x \, (T_\beta)^0_0\: d^3x \:. \]
\end{Thm}

This theorem can be extended immediately to energy-momentum conservation
on general Cauchy surfaces:
\begin{Corollary} {\Thmt{(Energy-momentum conservation on Cauchy surfaces)}} \label{NcorEMcons} \\
Let~$\scrN_0, \scrN_1$ be two Cauchy surfaces in Minkowski space, where~$\scrN_1$
lies to the future of~$\scrN_0$.
Then, under the assumptions of Theorem~\ref{NthmEMcons},
the conservation law of Theorem~\ref{NThmKillingCons} goes over to the conservation law for
the energy and momentum integrals
\beq \label{NTfinal}
\sum_{\beta=1}^3 m_\beta\, c_\beta \int_{\scrN_0} (T_\beta)^j_k \,\nu_j\:d\mu_{\scrN_0}
= \sum_{\beta=1}^3 m_\beta\, c_\beta \int_{\scrN_1} (T_\beta)^j_k \,\nu_j \:d\mu_{\scrN_1} \:,
\eeq
where~$\nu$ again denotes the future-directed normal and~$k \in \{0,\ldots, 3\}$.
\end{Corollary}
\Proof We use similar arguments as in the proof of Corollary~\ref{Ncorcurrent}.
More precisely, the conservation of classical energy implies that
\[ \int_{t=t_0} d^3x \, (T_\beta)^0_0\: d^3x = \int_{\scrN} (T_\beta)^j_0 \,\nu_j\: d\mu_{\scrN} \:. \]
This gives~\eqref{NTfinal} in the case~$k=0$.
Applying a Lorentz boost, we obtain
\[ \sum_{\beta=1}^3 m_\beta\, c_\beta \int_{\scrN_0} (T_\beta)^j_k \,\nu_j\: K^k\:d\mu_{\scrN_0}
= \sum_{\beta=1}^3 m_\beta\, c_\beta \int_{\scrN_1} (T_\beta)^j_k \,\nu_j \:K^k\:d\mu_{\scrN_1} \:, \]
where~$K$ is the Killing field obtained by applying the Lorentz boost to the vector field~$\partial_t$.
This gives the result.
\QED
These results shows that the conservation laws of energy and momentum correspond to more general
conservation laws in the setting of causal fermion systems.

The remainder of this section is devoted to the proof of Theorem~\ref{NthmEMcons}.
Let~$(\H, \F, \rho^\varepsilon)$ be a regularized vacuum Dirac sea configuration
together with anti-particles (for details see~\cite[Sections~1.2 and~3.4]{cfs}).
Then, possibly after extending the particle space (see~\cite[Remark~1.2.2]{cfs}),
we can decompose the wave evaluation operator~$\Psi$ as
\beq \label{NPsicompose}
\Psi = \Psi^\text{vac} + \Delta \Psi \:,
\eeq
where~$\Psi^\text{vac}$ is the wave evaluation operator of the completely filled Dirac sea
(see~\cite[\S1.1.4]{cfs} and the operator~$\Psi(x) = e^\varepsilon_x$ in~\cite[\S1.2.4]{cfs}),
and~$\Delta \Psi$ describes the holes.
The fact that the number of anti-particles is finite implies that the operator~$\Delta \Psi$
is trivial on the orthogonal complement of a finite-dimensional subspace of~$\H$, which we denote
by~$K$,
\beq \label{NDelPsifinite}
\Delta \Psi \,u = 0 \qquad \text{for all~$u \in K^\perp \subset \H$}\:.
\eeq
We now choose~$(f_\tau)_{\tau \in \R}$ as the time translations, i.e.
\[ f_\tau \::\: \scrM \rightarrow \scrM \:,\: f_\tau(t,x_1,x_2,x_3) = (t+\tau, x_1, x_2, x_3) \]
(for the identification of~$\scrM$ with~$M:= \supp \rho$ see~\cite[Section~1.2]{cfs}).
Since the Lebesgue measure~$d^4x$ is translation invariant, it clearly is invariant
under the action of~$f_\tau$.
Constructing the universal measure as the push-forward (see~\cite[\S1.2.1]{cfs}), it follows immediately
that~$f_\tau$ is a symmetry of the universal measure.

Since~$\Psi^\text{vac}$ is composed of plane-wave solutions of the Dirac equation,
on which the time translation operator acts by multiplication with a phase,
the operator~$f_\tau$ can be represented by a unitary transformation in~$\H$.
More precisely, choosing the operator~$\scrU_\tau$ as the multiplication operator in momentum
space~$\hat{\scrU}(k) = e^{i k^0 \tau}$, we have the relation
\beq \label{NPsivacsymm}
(\Psi^\text{vac} \,u)\big(f_\tau(x) \big) = \big( \Psi^\text{vac} \,\scrU^{-1}_\tau u \big)(x) \qquad 
\text{for all~$x \in M, u \in \H$} \:.
\eeq
Using~\eqref{NPsicompose} and~\eqref{NPsivacsymm} in~\eqref{NEdef}, we conclude that
\[ E_\tau(u,x) := (\Delta \Psi u)\big(f_\tau(x) \big) - \big( \Delta \Psi \scrU^{-1}_\tau u \big)(x) \qquad 
\text{for all~$x \in M, u \in \H$} \:. \]
The assumption~\eqref{NDelPsifinite} implies that~$(f_\tau)_{\tau \in \R}$ is indeed
a Killing symmetry with finite-dimensional support (see Definition~\ref{Ndefkilling}).

In order to simplify the setting, we note that a unitary transformation~$\scrU_\tau$
was already used in Section~\ref{Nsecexcurrent} to obtain corresponding conserved currents
(see~\eqref{NUtaudef} and Theorem~\ref{Nthmcurrent}).
This means that the first variations of~$\scrU_\tau$ on the finite-dimensional subspace~$K$
give rise to a linear combination of the corresponding conserved currents.
With this in mind, we may in what follows assume that~$\scrU_\tau$ is trivial on~$K$,
\beq \label{NUtrivial}
\scrU_\tau|_K = \1_K \:.
\eeq
Modifying~$\scrU_\tau$ in this way corresponds to going over to a new conservation law,
which is obtained from the original conservation law by subtracting a
linear combination of electromagnetic currents.

For the computations, it is most convenient to work again with the kernel of the fermionic projector.
Using~\eqref{NPsicompose}, we decompose it as
\[ P(x,y) = P^\text{vac}(x,y) + \Delta P(x,y) \:, \]
where
\begin{align*}
P^\text{vac}(x,y) &= -\Psi^\text{vac}(x) \Psi^\text{vac}(y)^* \\
\Delta P(x,y) &= -\Psi^\text{vac}(x) \big(\Delta \Psi\big)(y)^* -\big(\Delta \Psi\big)(x) \Psi^\text{vac}(y)^*
-\big(\Delta \Psi\big)(x) \big(\Delta \Psi\big)(y)^*\:.
\end{align*}
Since~$\Delta \Psi$ vanishes on the complement of the finite-dimensional subspace~$K$,
the kernel~$\Delta \Psi$ is composed of a finite number of Dirac wave functions, i.e.
\[ \Delta P(x,y) = \sum_{i,j=1}^{\na} c_{ij} \:\phi_i(x) \overline{\phi_j(y)} \]
with~$\na \in \N$ and~$\overline{c_{ij}} = c_{ji}$. Diagonalizing the Hermitian matrix~$(c_{ij})$ by
a basis transformation, we can write~$\Delta P(x,y)$ as
\[ \Delta P(x,y) = \sum_{i=1}^{\na} c_i \:\phi_i(x) \overline{\phi_i(y)} \]
with real-valued coefficients~$c_i$.
Since we only consider first oder variations of the Lagrangian, by linearity we may restrict
attention to one of the summands. Thus it suffices to consider the case
\[ \Delta P(x,y) = \psi(x) \overline{\psi(y)} \:, \]
where~$\psi$ is a negative-frequency solution of the Dirac equation.

Now the first variation of the Lagrangian can be computed similar as in~\eqref{NLTrQ} to obtain
\begin{align*}
&\frac{d}{d\tau} \int_\Omega d\rho(x) \int_{M \setminus \Omega} d\rho(y) \:\Big( \L_\kappa \big( f_\tau(x), y\big) -  \L_\kappa \big( \Phi_\tau(x), y\big) \big) \Big) \Big|_{\tau=0}  \\
&= \int_\Omega d\rho(x) \int_{M \setminus \Omega} d\rho(y)
\Big(   \Tr_{S_y} \big( Q(y,x)\, \delta_{v(x)} P(x,y) \big) + \Tr_{S_x} \big( Q(x,y)\, \delta_{v(x)} P(y,x) \big)  \Big),
\end{align*}
where
\[ \delta_{v(x)} P(x,y) := \frac{d}{d\tau} \Big( P \big( f_\tau(x) ,y \big)
-P \big( \Phi_\tau(x),y \big) \Big) \Big|_{\tau = 0} \:. \]
Since~$P^\text{vac}(x,y)$ has the Killing symmetry~\eqref{NPsivacsymm},
the variation of~$P(x,y)$ simplifies to
\begin{align*}
\delta_{v(x)} P(x,y)&= \frac{d}{d\tau} \Big( \Delta P \big( f_\tau(x) ,y \big)
- \Delta P \big( \Phi_\tau(x),y \big) \Big) \Big|_{\tau = 0} \\
&= \frac{d}{d\tau} \Delta P \big( f_\tau(x) ,y \big) \big|_{\tau = 0}
= \frac{d}{d\tau} \Big( \psi \big( f_\tau(x) \big) \overline{\psi(y)} \Big) \Big|_{\tau = 0} =: (\partial_t \psi)(x) \,
\overline{\psi(y)} \:,
\end{align*}
where in the last line we used~\eqref{NUtrivial}.
Using these relations in~\eqref{NKillingConserve}, we obtain the conservation law
\begin{align*}
0 &= \int_\Omega d\rho(x) \int_{M \setminus \Omega} d\rho(y) 
\Big(   \Tr_{S_y} \big( Q(y,x)\, (\partial_t \psi)(x)\, \overline{\psi(y)} \big) + \Tr_{S_x} \big( Q(x,y)\, \psi(y) 
\,\overline{(\partial_t \psi)(x)}  \big)  \\
&\qquad \qquad -\Tr_{S_y} \big( Q(y,x)\,  \psi(x) \, \overline{(\partial_t \psi)(y)}
 \big) - \Tr_{S_x} \big( Q(x,y)\, (\partial_t \psi)(y)\, \overline{\psi(x)}  \big)  \Big) \\
&= 2\,\re \int_\Omega d\rho(x) \int_{M \setminus \Omega} d\rho(y)
\Big( \Sl \psi (y) | Q(y,x) (\partial_t \psi)(x) \Sr 
- \Sl \psi(x) | Q(x,y) (\partial_t \psi)(y) \Sr \Big) \:.
\end{align*}

Next, we consider the limiting case where~$\Omega$ exhausts the region between two
Cauchy surfaces~$t=t_0$ and~$t=t_1$ (see Figure~\ref{Nfignoether2}).
We thus obtain a conserved current~$J$ which for example at time~$t=0$
is given by
\[ J = \frac{1}{2}\: \re\,  \int_{t\leq 0} d^4x \int_{t>0} d^4y \:
\Big( \Sl \psi (y) | Q(y,x) (\partial_t \psi)(x) \Sr 
- \Sl \psi(x) | Q(x,y) (\partial_t \psi)(y) \Sr \Big) \:. \]
This equation is similar to~\eqref{Ninttask} and can be analyzed in exactly the same manner.
Indeed, regularizing the Heaviside functions and applying Plancherel, we
again obtain~\eqref{Naneqb1} and~\eqref{Naneqb2}, with the only difference
that an additional factor~$\omega_{\beta, \vec{k}}$ appears.
Thus, in analogy to~\eqref{NJsum2} and~\eqref{NUSymm19} we obtain
\[ J = -\sum_{\beta=1}^3 \frac{1}{2} \, c_\beta    \int \frac{d^3k}{(2 \pi)^3} \,\omega_{\beta,\vec k}^2  \:
\Sl \chi_\beta(\vec k) |  \chi_\beta(\vec k ) \Sr \:. \]
Applying~\eqref{Nrel2} and using again Plancherel gives the result.
This concludes the proof of Theorem~\ref{NthmEMcons}.

\section{Example: Symmetries of the Universal Measure} \label{Nsecexrho}
In this section we consider the conserved surface layer integrals corresponding to
symmetries of the universal measure (see Theorem~\ref{Nthmsymmum}
and Corollary~\ref{Ncorsymmum}). We now explain why, under the assumption
that~$\Phi_\tau$ is a bijection, these conserved surface layer integrals
can be expressed merely in terms of the volumes of the sets~$\Omega \setminus \Phi_\tau(\Omega)$
and~$\Phi_\tau(\Omega) \setminus \Omega$.
Our argument shows in particular that in the 
limiting case of Figure~\ref{Nfignoether2} when
the boundary of~$\Omega$ consists of two hypersurfaces, the conserved
surface layer integrals do not give rise to any interesting conservation laws.
Therefore, although the conservation laws
of Theorem~\ref{Nthmsymmum} and Corollary~\ref{Ncorsymmum} give non-trivial information
on the structure of a minimizing universal measure of a causal fermion system,
they do not correspond to any conservation laws in Minkowski space.

The following argument applies for example to the situation considered in Section~\ref{Nseccurcor}
that~$M$ can be identified with Minkowski space, and~$\Omega$ is the past of a Cauchy surface.
But the argument applies in a much more general setting. In particular, we do not need to
assume that~$\Omega$ is compact. We first rewrite the surface layer integral in~\eqref{Nconserve5} as
\begin{align*}
\int_\Omega& d\rho(x) \int_{M \setminus \Omega} d\rho(y) \:\Big( \L_\kappa\big( \Phi_\tau(x), y\big) -
\L_\kappa\big( x, \Phi_\tau(y) \big) \Big) \\
&= \int_\Omega d\rho(x) \int_{M \setminus \Omega} d\rho(y) \:\Big( \L_\kappa\big( \Phi_\tau(x), y\big) -
\L_\kappa(x, y) \Big) \\
&\quad + \int_\Omega d\rho(x) \int_{M \setminus \Omega} d\rho(y) \:\Big( \L_\kappa(x, y) -
\L_\kappa\big( x, \Phi_\tau(y) \big) \Big) \\
&= \bigg( \int_{\Phi_\tau(\Omega)} - \int_\Omega \bigg) \:d\rho(x) \int_{M \setminus \Omega} d\rho(y) \: \L_\kappa(x, y) \\
&\quad + \int_\Omega d\rho(x) \:\bigg( \int_{M \setminus \Omega} - \int_{\Phi_\tau(M \setminus \Omega)} \bigg)\:
d\rho(y) \:\L_\kappa(x, y) \:.
\end{align*}
Assuming that~$\Phi_\tau$ is as bijection, we can write the obtained differences of integrals as illustrated in Figure~\ref{Nfignoether3},
\begin{align*}
\bigg( \int_{\Phi_\tau(\Omega)} - \int_\Omega \bigg) \;\cdots =
\bigg( \int_{\Phi_\tau(\Omega) \setminus \Omega} - \int_{\Omega \setminus \Phi_\tau(\Omega)} \bigg) \;\cdots \\
\bigg( \int_{M \setminus \Omega} - \int_{\Phi_\tau(M \setminus \Omega)} \bigg) \;\cdots =
\bigg( \int_{\Omega \setminus \Phi_\tau(\Omega)} - \int_{\Phi_\tau(\Omega) \setminus \Omega} \bigg) \;\cdots  \: .
\end{align*}
Using that~\eqref{NLagrangeKappa} is symmetric with respect to exchange of its arguments (cf.~\eqref{NsymmL}), we thus obtain
\begin{align*}
\int_\Omega& d\rho(x) \int_{M \setminus \Omega} d\rho(y) \:\Big( \L_\kappa\big( \Phi_\tau(x), y\big) -
\L_\kappa\big( x, \Phi_\tau(y) \big) \Big) \\
&= \bigg( \int_{\Phi_\tau(\Omega) \setminus \Omega} - \int_{\Omega \setminus \Phi_\tau(\Omega)} \bigg)
\:d\rho(x) \int_M d\rho(y) \: \L_\kappa(x, y) \\
&= \bigg( \int_{\Phi_\tau(\Omega) \setminus \Omega} - \int_{\Omega \setminus \Phi_\tau(\Omega)} \bigg)
\:\ell(x)\: d\rho(x) \:,
\end{align*}
where in the last step we used~\eqref{Nelldef}. In view of~\eqref{Nseprel}, the obtained integrand
is constant. Therefore, the surface layer integral can indeed be expressed in terms of the volume of the
sets~$\Phi_\tau(\Omega) \setminus \Omega$ and~$\Omega \setminus \Phi_\tau(\Omega)$.
In particular, the surface layer integral does not capture any interesting dynamical information
of the causal fermion system.
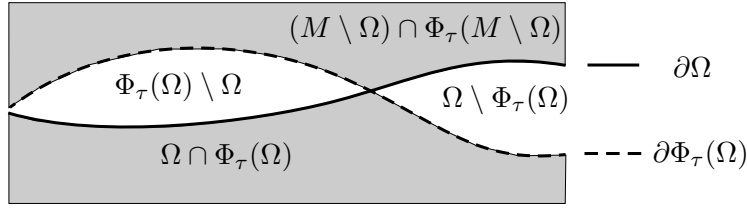
\begin{figure}\centering
\psscalebox{1.0 1.0} %
{
\begin{pspicture}(0,-1.3321118)(9.645718,1.3321118)
\definecolor{colour0}{rgb}{0.8,0.8,0.8}
\pspolygon[linecolor=black, linewidth=0.002, fillstyle=solid,fillcolor=colour0](7.3338523,0.49955472)(6.982104,0.540414)(6.46561,0.5530239)(6.0426135,0.5034688)(5.615636,0.40687802)(5.3149023,0.32343453)(4.7772465,0.16571002)(4.124508,0.4539213)(3.5894282,0.6157628)(2.9047132,0.7188196)(1.9799689,0.7152769)(1.0838828,0.51089174)(0.40659958,0.2000194)(0.014606438,-0.059563804)(0.01385916,1.3311102)(7.3338523,1.3311102)
\pspolygon[linecolor=black, linewidth=0.002, fillstyle=solid,fillcolor=colour0](0.014606438,-0.13511193)(0.35684568,-0.23256378)(0.69938534,-0.26354155)(1.0905135,-0.30791193)(1.5389885,-0.32130453)(2.1811993,-0.3126823)(2.7668514,-0.2552749)(3.4497936,-0.16217859)(4.2947426,0.015554737)(4.786874,0.15783621)(5.7142606,-0.3485638)(6.16883,-0.5659416)(6.7202015,-0.69884527)(7.3427415,-0.68933415)(7.3346066,-1.3311119)(0.014606438,-1.3311119)
\psbezier[linecolor=black, linewidth=0.04](0.005717549,-0.13111193)(1.2346064,-0.4977786)(3.2279398,-0.25555637)(4.2857175,0.019999182)(5.3434954,0.29555473)(6.0857177,0.7111103)(7.339051,0.4999992)
\rput[bl](2.0101619,-0.9066675){$\Omega \cap \Phi_\tau(\Omega)$}
\psbezier[linecolor=black, linewidth=0.04, linestyle=dashed, dash=0.17638889cm 0.10583334cm](0.02793977,-0.057949536)(1.3812732,0.96743506)(3.2146065,0.8825633)(4.46794,0.30888808)(5.721273,-0.26478714)(6.1257176,-0.8108555)(7.339051,-0.682223)
\rput[bl](3.6946065,0.77777696){$(M \setminus \Omega) \cap \Phi_\tau(M \setminus \Omega)$}
\rput[bl](1.4057175,0.01777696){$\Phi_\tau(\Omega) \setminus \Omega$}
\rput[bl](5.716829,-0.14666748){$\Omega \setminus \Phi_\tau(\Omega)$}
\psline[linecolor=black, linewidth=0.04](7.6279397,0.50444365)(8.30794,0.50444365)
\psline[linecolor=black, linewidth=0.04, linestyle=dashed, dash=0.17638889cm 0.10583334cm](7.5879397,-0.6688897)(8.334606,-0.6688897)
\rput[bl](8.485718,-0.87555635){$\partial \Phi_\tau(\Omega)$}
\rput[bl](8.756828,0.34666586){$\partial \Omega$}
\end{pspicture}
}
\caption{The surface layer integral corresponding to a symmetry of the universal measure.}
\label{Nfignoether3}
\end{figure} %

\section{Conservation Laws and Microscopic Mixing}\label{Nremmicro}

We conclude this chapter by again pointing out that the conservation laws of
Theorem~\ref{Nthmcurrent} and Theorem~\ref{NThmKillingCons} hold for causal
fermion systems without taking the continuum limit.
In particular, these conservation laws also hold for regularized Dirac sea
configurations if one analyzes the EL equations corresponding to the causal
action principle without taking the limit~$\varepsilon \searrow 0$ or in ``quantum space-times''
as outlined in Section~\ref{Csecgenst}.

In Theorem~\ref{Nthmcurrentmink} and Theorem~\ref{NthmEMcons} we restricted attention
to negative-frequency solutions of the Dirac equation. On a technical level, this was
necessary because the operator~$\hat{Q}$ is only well-defined inside the {\em{lower}}
mass cone, whereas it diverges outside the lower mass cone (see Definition~\ref{Ndef611}
and~\cite[Section~5.6]{PFP} or~\cite{reg}). In non-technical terms, this means that
introducing Dirac particles into the causal fermion system makes the causal action infinitely large.
But, as explained in detail in~\cite[Section~3]{qft} and briefly in Section~\ref{CFoundations}, the action becomes again finite if
one introduces a so-called {\em{microscopic mixing of the wave functions}}.
In other words, minimizing the causal action gives rise to the mechanism of microscopic mixing
(for more details see~\cite[\S1.5.3]{cfs}).
Microscopic mixing is also important for getting the connection to entanglement
and second-quantized bosonic fields (see~\cite{entangle, qft} and Section~\ref{CFoundations}).

If microscopic mixing is present, the conservation laws of
Theorem~\ref{Nthmcurrent} and Theorem~\ref{NThmKillingCons} again give rise to
conserved surface layer integrals. However, evaluating these surface layer integrals
in Minkowski space is more involved because a homogenization procedure over the
microstructure must be performed (in the spirit of~\cite[Section~5.1]{qft}).
Since these constructions are rather involved, we cannot give them here.
However, even without entering the detailed constructions, the following argument
shows that the conservation laws should apply to the particle states as well:

In an interacting system, a solution of the Dirac equation which at some initial time
has negative frequency may at a later time have positive frequency
(as in the usual pair production process).
The conservation law of Theorem~\ref{Nthmcurrent} implies that the surface layer
integral at the later time coincides with that at the initial time.
Using current conservation of the Dirac dynamics, we conclude that the
surface layer at the later time again coincides with the surface integral
of the Dirac current.
Using arguments of this type, one sees that, no matter
what the microscopic structure of space-time is, the conservation laws
of Theorem~\ref{Nthmcurrent} and Theorem~\ref{NThmKillingCons} should
apply similarly to positive-frequency solutions of the Dirac equation.\medskip

An explanation of the implications of the conservation laws constructed in this chapter
on the foundations of quantum theory (in particular concerning the collapse of the wave
function in the quantum mechanical measurement process)
is given in Section~\ref{IModelQT}.

\newpage \thispagestyle{empty} \  %
\chapter[Hamiltonian Formulation and Linearized Field Equations]{Hamiltonian Formulation\\ and Linearized Field Equations}\label{DissJet}

In this chapter, we give a formulation of the dynamics of causal fermion systems in terms of physical fields on space-time.

After generalizing causal variational principles to a class of lower semi-continuous Lagrangians on a smooth,
possibly non-compact manifold, the corresponding Euler-Lagrange equations are derived (Section~\ref{Jseccvpcfs}).
In Section~\ref{JSecSmooth}, we show under additional smoothness assumptions that the space of solutions of the
Euler-Lagrange equations has the structure of a symplectic Fr\'echet manifold.
The symplectic form is constructed as a surface layer integral which is shown to be invariant under
the time evolution. 
In Section~\ref{JlowerSemiCont}, the results and methods are extended to the lower semi-continuous setting. Evaluating the Euler-Lagrange equations weakly, we derive linearized field equations and the Hamiltonian time evolution.
Finally, in Section~\ref{Jseclattice}, our constructions and results are illustrated in a detailed example
on $\mathbb{R}^{1,1} \times S^1$, where a local minimizer is given by a measure supported on a two-dimensional lattice.

\section{Causal Variational Principles and Causal Fermion Systems} \label{Jseccvpcfs}
This section provides the preliminaries needed for the construction of the Hamiltonian time evolution.
After introducing causal variational principles in the non-compact setting (Section~\ref{Jsecnoncompact}),
we derive the corresponding Euler-Lagrange equations (Section~\ref{JsecEL})
and introduce the concept of local minimizers (Section~\ref{Jseclocmin}).
The connection to the theory of causal fermion systems is established in Section~\ref{Jseccfs}.

\subsection{Causal Variational Principles in the Non-Compact Setting} \label{Jsecnoncompact}
We now introduce causal variational principles in the non-compact setting
(for the simpler compact setting see Section~\ref{Nsecintroc}).
Let~$\F$ be a (possibly non-compact) smooth manifold of dimension~$m \geq 1$
and~$\rho$ a (positive) Borel measure on~$\F$ (the {\em{universal measure}}).
Moreover, we are given a non-negative function~$\L : \F \times \F \rightarrow \R^+_0$
(the {\em{Lagrangian}}) with the following properties:
\begin{itemize}[leftmargin=2em]
\item[(i)] $\L$ is symmetric: $\L(x,y) = \L(y,x)$ for all~$x,y \in \F$.\label{JCond1}
\item[(ii)] $\L$ is lower semi-continuous, i.e.\ for all sequences~$x_n \rightarrow x$ and~$y_{n'} \rightarrow y$,
\[ \L(x,y) \leq \liminf_{n,n' \rightarrow \infty} \L(x_n, y_{n'})\:. \]\label{JCond2}
\end{itemize}
If the total volume~$\rho(\F)$ is finite, the {\em{causal variational principle}} is to minimize the action
\beq \label{JSact} 
\Sact (\rho) = \int_\F d\rho(x) \int_\F d\rho(y)\: \L(x,y) 
\eeq
under variations of the measure~$\rho$, keeping the total volume~$\rho(\F)$ fixed
({\em{volume constraint}}).
If~$\rho(\F)$ is infinite, it is not obvious how to implement the volume constraint,
making it necessary to proceed as follows:
First, we make the following additional assumptions:
\begin{itemize}[leftmargin=2em]
\item[(iii)] The measure~$\rho$ is {\em{locally finite}}
(meaning that any~$x \in \F$ has an open neighborhood~$U$ with~$\rho(U)< \infty$).\label{JCond3}
\item[(iv)] The function~$\L(x,.)$ is $\rho$-integrable for all~$x \in \F$, giving
a lower semi-continuous and bounded function on~$\F$. \label{JCond4}
\end{itemize}
We remark that, since a manifold is second countable, property~(iii) implies that~$\rho$
is $\sigma$-finite.
In view of the computations later in this paper,
it is most convenient to subtract a constant~$\nu/2$ from the integral over~$\L(x,.)$
by introducing the function
\beq \label{Jelldef}
\ell(x) = \int_\F \L(x,y)\: d\rho(y) - \frac{\nu}{2} \::\: \F \rightarrow \R \quad
\text{bounded and lower semi-continuous}\:,
\eeq
where the parameter~$\nu \in \R$ will be specified below.
We let~$\tilde{\rho}$ be another Borel measure on~$\F$
which satisfies the conditions
\beq \label{Jtotvol}
\big| \tilde{\rho} - \rho \big|(\F) < \infty \qquad \text{and} \qquad
\big( \tilde{\rho} - \rho \big) (\F) = 0
\eeq
(where~$|.|$ denotes the total variation of a measure;
see~\cite[\S28]{halmosmt} or~\cite[Section~6.1]{rudin}).
Then the difference of the actions as given by
\beq \label{Jintegrals}
\begin{split}
\big( &\Sact(\tilde{\rho}) - \Sact(\rho) \big) = \int_\F d(\tilde{\rho} - \rho)(x) \int_\F d\rho(y)\: \L(x,y) \\
&\quad + \int_\F d\rho(x) \int_\F d(\tilde{\rho} - \rho)(y)\: \L(x,y) 
+ \int_\F d(\tilde{\rho} - \rho)(x) \int_\F d(\tilde{\rho} - \rho)(y)\: \L(x,y)
\end{split}
\eeq
is well-defined in view of the following lemma.\newpage
\begin{Lemma} The integrals in~\eqref{Jintegrals} are well-defined with values in~$\R \cup \{\infty\}$. Moreover,
\begin{align}\label{Jintegrals2}
\begin{split}
\big( \Sact(\tilde{\rho}) - \Sact(\rho) \big) &= 2 \int_\F \Big(\ell(x) + \frac{\nu}{2} \Big) \:d(\tilde{\rho} - \rho)(x) \\
&\quad + \int_\F d(\tilde{\rho} - \rho)(x) \int_\F d(\tilde{\rho} - \rho)(y)\: \L(x,y) \:.
\end{split}
\end{align}
\end{Lemma}
\Proof Decomposing the signed measure~$\mu=\tilde{\rho}-\rho$ into its positive and negative parts,
$\mu = \mu^+-\mu^-$ (see the Jordan decomposition in~\cite[\S29]{halmosmt}), the measures~$\mu^\pm$
are both positive measures of finite total volume and~$\mu^- \leq \rho$.
In order to show that the integrals in~\eqref{Jintegrals}
are well-defined, we need to prove that the negative contributions are finite, i.e.
\beq \label{Jfinite}
\int_\F d\mu^-(x) \int_\F d\rho\: \L(x,y) < \infty \qquad \text{and} \qquad
\int_\F d\mu^+(x) \int_\F d\mu^-(y)\: \L(x,y) < \infty
\eeq
(here we apply Tonelli's theorem and make essential use of the fact that the Lagrangian is non-negative).
The bounds~\eqref{Jfinite} follow immediately from the estimates
\begin{align*}
\int_\F & d\mu^-(x) \int_\F d\rho\: \L(x,y) = \int_\F d\mu^-(x) \:\Big(\ell(x) + \frac{\nu}{2} \Big)
\leq \Big( \sup_\F \ell + \frac{\nu}{2} \Big) \:\mu^-(\F) < \infty \\
\int_\F &d\mu^+(x) \int_\F d\mu^-(y)\: \L(x,y) \leq \int_\F d\mu^+(x) \int_\F d\rho\: \L(x,y) \\
&= \int_\F d\mu^+(x) \:\Big(\ell(x) + \frac{\nu}{2} \Big)\:
\leq \Big( \sup_\F \ell + \frac{\nu}{2} \Big) \:\mu^+(\F) < \infty \:,
\end{align*}
where we used the fact that~$\ell$ is assumed to be a bounded function on~$\F$.
\QED

\begin{Def} The measure~$\rho$ is said to be a {\textbf{minimizer}} of the causal action
if the difference~\eqref{Jintegrals2}
is non-negative for all~$\tilde{\rho}$ satisfying~\eqref{Jtotvol},
\[ \big( \Sact(\tilde{\rho}) - \Sact(\rho) \big) \geq 0 \:. \]
\end{Def}

We close this section with a remark on the existence theory.
If~$\F$ is compact, the existence of minimizers can
be shown just as in~\cite[Section~1.2]{continuum} using the Banach-Alaoglu theorem
(the fact that~$\L$ is semi-continuous implies that the weak-$*$-limit of a minimizing sequence of
measures is indeed a minimizer).
In the non-compact setting, the existence theory has not yet been developed
(for more details on this point see~\cite[\S1.1.1]{cfs}). For the purpose of the present paper,
all we need is that the causal action principle admits {\em{local}} minimizers
which satisfy the corresponding Euler-Lagrange equations.
These concepts will be introduced in Sections~\ref{JsecEL} and~\ref{Jseclocmin} below.
Moreover, in Section~\ref{Jseclattice} we will analyze an example where local minimizers
exist although~$\F$ is non-compact.

\subsection{The Euler-Lagrange Equations} \label{JsecEL}
We now derive the Euler-Lagrange (EL) equations, following the method in the
compact setting~\cite[Lemma~3.4]{support}.
\begin{Lemma} \Thmt{(Euler-Lagrange equations)} \label{JlemmaEL}
Let~$\rho$ be a minimizer of the causal action. Then
\beq \label{JEL1}
\ell|_{\supp \rho} \equiv \inf_\F \ell \:.
\eeq
\end{Lemma}
\Proof Given~$x_0 \in \supp \rho$, we choose an open neighborhood~$U$ with~$0 < \rho(U)<\infty$.
For any~$y \in \F$ we consider the family of measures~$(\tilde{\rho}_\tau)_{\tau \in [0,1)}$ given by
\[ \tilde{\rho}_\tau = \chi_{M \setminus U} \,\rho + (1-\tau)\, \chi_U \, \rho + \tau\, \rho(U)\, \delta_{y} \]
(where~$\delta_y$ is the Dirac measure supported at~$y$). Then
\beq \label{Jtilderho}
\tilde{\rho}_\tau - \rho = -\tau\, \chi_U \, \rho + \tau\, \rho(U)\, \delta_{y} 
= \tau \big( \rho(U)\, \delta_{y} - \chi_U\, \rho \big) \:,
\eeq
implying that~$\tilde{\rho}_\tau$ satisfies~\eqref{Jtotvol}. Hence
\begin{align*} 
0 &\leq \big(\Sact(\tilde{\rho}) - \Sact(\rho) \big) = 2 \tau 
\left( \rho(U)\, \Big( \ell(y) + \frac{\nu}{2} \Big)- \int_U \Big( \ell(x)+ \frac{\nu}{2} \Big)\, d\rho(x) \right) + \O \big(\tau^2 \big) \:.
\end{align*}
As a consequence, the linear term must be non-negative,
\begin{align} \label{JELInt}
\ell(y) \geq \frac{1}{\rho(U)} \int_U \ell(x)\, d\rho(x) \:.
\end{align}
Assume that~\eqref{JEL1} is false. Then there is~$x_0 \in \supp \rho$ and~$y \in \F$ such that $\ell(x_0) > \ell(y)$. 
Lower semi-continuity of $\ell$ implies that there is an open
neighborhood $U$ of $x_0$ such that $\ell(x) > \ell(y)$ for all $x \in U$, in contradiction to~\eqref{JELInt}. This gives the result.
\QED
We always choose~$\nu$ such that~$\inf_\F \ell=0$. Then the EL equations~\eqref{JEL1} become
\begin{align}\label{JELstrong}
\ell|_{\supp \rho} \equiv \inf_\F \ell = 0 \: .
\end{align}
We remark that~$\nu$ can be understood as the Lagrange multiplier describing the volume constraint;
see~\cite[\S1.4.1]{cfs}.

\subsection{Local Minimizers and Second Variations} \label{Jseclocmin}
We now introduce the concept of local minimizers of causal variational principles
and explore the connection to second variations.
We derive a convenient criterion which ensures that a measure~$\rho$ is
a local minimizer (Proposition~\ref{JLocSufficient}).
This criterion will be used in the example of Section~\ref{Jseclattice}
to prove the existence of local minimizers.

We again consider families of variations~$(\tilde{\rho}_\tau)_{\tau \in [0, \delta)}$
with~$\tilde{\rho}_0=\rho$ and assume that the measures~$\tilde{\rho}_\tau$ all satisfy
the conditions in~\eqref{Jtotvol}. Then the family of measures~$\mu_\tau$ defined by
\beq \label{Jmudef}
\mu_\tau := \tilde{\rho}_\tau - \rho \:,
\eeq
are in the Banach space~${\mathfrak{B}}(\F)$ of signed measures on~$\F$
with the norm given by the total variation.
\begin{Def} \label{Jdeflocmin}
The measure~$\rho$ is a {\textbf{local minimizer}} of the causal action if for
every family~$(\tilde{\rho}_\tau)_{\tau \in [0,\delta)}$ of Borel measures which has the property that~$\mu_\tau$
defined by~\eqref{Jmudef} is a smooth regular
curve~$\mu : [0, \delta) \rightarrow {\mathfrak{B(\F)}}$ with $\mu_\tau(\F) = 0$,
there is~$\delta_0 \in (0, \delta)$ such that
\beq \label{JSrhoin}
\big( \Sact(\tilde{\rho}_\tau) - \Sact(\rho) \big) \geq 0 \qquad \text{for all~$\tau \in [0, \delta_0)$} \:.
\eeq
\end{Def}

We first derive the implications of local minimality. To this end, we assume that~$\rho$ is a local minimizer.
Then obviously the EL equations~\eqref{JEL1} hold, because the curve~\eqref{Jtilderho}
has the properties in the above definition. We consider
variations~$(\tilde{\rho}_\tau)_{\tau \in (-\delta, \delta)}$ of the form
\beq \label{Jrhovarpsi}
\tilde{\rho}_\tau = (1 + \tau \psi)\: \rho\:,
\eeq
where~$\psi$ is a real-valued function on~$\F$.
In order to ensure that these measures are again positive for sufficiently small~$\delta>0$,
we must assume that~$\psi$ is
an essentially bounded function. Moreover, these measures satisfy the conditions in the above
definition if and only if
\[ \int_M |\psi|\: d\rho < \infty \qquad \text{and} \qquad \int_M \psi\: d\rho = 0 \:. \]
Hence we must assume that~$\psi$ is in the space
\beq \label{JDomainLrho}
{\mathscr{D}}(\L_\rho) := \Big\{ \psi \in (L^1 \cap L^\infty)(M, d\rho) \:\Big|\: \int_M \psi\: d\rho = 0 \Big\} \:.
\eeq
By interpolation, the function~$\psi$ is also a vector in the Hilbert space~$L^2(M, d\rho)$,
also denoted by~$(\H_\rho, \la .,. \ra_\rho)$.
The operator~$\L_\rho$ in the following lemma was already analyzed in~\cite[Lemma~3.5]{support} in the compact setting.
We now extend this analysis to the non-compact setting.

\begin{Lemma} For~$\psi \in (L^1 \cap L^\infty)(M, d\rho)$, the function~$\L_\rho \psi$ defined by
\[ (\L_\rho \psi)(x) = \int_M \L(x,y)\: \psi(y)\: d\rho(y) \]
is in~$L^2(M, d\rho)$, giving rise to a linear operator
\[ \L_\rho \::\: {\mathscr{D}}(\L_\rho)  \subset \H_\rho \rightarrow \H_\rho \:. \]
\end{Lemma}
\Proof We apply Tonelli's theorem to obtain
\begin{align*}
\int_M \big| \L_\rho \psi \big|^2\: d\rho
&\leq \|\psi\|_{L^\infty(M)} \int_M d\rho(x) \int_M d\rho(y)\: \L(x,y)\: |\psi(y)|  \int_M d\rho(y')\: \L(x,y') \\
&= \|\psi\|_{L^\infty(M)} \int_M d\rho(x) \int_M d\rho(y)\: \L(x,y)\: |\psi(y)| \:\Big( \ell(x) + \frac{\nu}{2} \Big) \\
&\leq \|\psi\|_{L^\infty(M)}\: \sup_{M} \Big( \ell + \frac{\nu}{2} \Big) \int_M d\rho(x) \int_M d\rho(y)\: \L(x,y) \: |\psi(y)| \\
&= \|\psi\|_{L^\infty(M)}\: \sup_{M} \Big( \ell + \frac{\nu}{2} \Big) \int_M d\rho(y)\: |\psi(y)| \int_M d\rho(x) \: \L(x,y) \\
&\leq \, \|\psi\|_{L^\infty(M)}\: \sup_{M} \Big( \ell + \frac{\nu}{2} \Big)^2\: \|\psi\|_{L^1(M)}
< \infty \:,
\end{align*}
where we used that~$\L$ is symmetric and that~$\ell$ is bounded according to our assumption~\eqref{Jelldef}.
This gives the result.
\QED

\begin{Prp} If~$\rho$ is a local minimizer, then the operator~$\L_\rho
: {\mathscr{D}}(\L_\rho) \rightarrow \H_\rho$ is positive (but not necessarily strictly positive).
\end{Prp}
\Proof Computing~\eqref{Jintegrals2} for the variation~\eqref{Jrhovarpsi} gives
\beq \begin{split}
\big( \Sact(\tilde{\rho}) - \Sact(\rho) \big) &= 2 \tau \int_\F \Big(\ell(x) + \frac{\nu}{2} \Big) \:\psi(x)\: d\rho \\
&\qquad + \tau^2 \int_\F \psi(x)\: d\rho(x) \int_\F \psi(y)\: d\rho(y)\: \L(x,y) \:.
\end{split} \label{J216}
\eeq
The first summand vanishes in view of the EL equations~\eqref{JEL1}. The second summand,
on the other hand, exists in view of the estimates
\begin{align*}
\int_\F &\psi(x)\: d\rho(x) \int_\F \psi(y)\: d\rho(y)\: \L(x,y)
\leq \|\psi\|_{L^\infty(M)}\: \int_\F \psi(x)\: d\rho(x) \int_\F  d\rho(y)\: \L(x,y) \\
& \leq \|\psi\|_{L^\infty(M)}\: \sup_{M} \Big( \ell + \frac{\nu}{2} \Big) \int_\F \psi(x)\: d\rho(x)  = \|\psi\|_{L^\infty(M)}\: \sup_{M} \Big( \ell + \frac{\nu}{2} \Big) \|\psi\|_{L^1(M)} \, .
\end{align*}
Rewriting the second summand in~\eqref{J216} as an expectation value, we obtain
\[ \big( \Sact(\tilde{\rho}) - \Sact(\rho) \big) = \tau^2 \: \la \psi, \L_\rho \psi \ra_\rho \:. \]
Applying the inequality~\eqref{JSrhoin} gives the result.
\QED

We finally give a criterion which ensures that~$\rho$ is a local minimizer.

\begin{Prp}\label{JLocSufficient} Let~$\rho$ be a Borel measure with the following properties:
\begin{itemize}[leftmargin=2em]
\item[{\textrm{(a)}}] The EL equations~\eqref{JELstrong} are satisfied and in addition
\beq \label{Jimply}
\ell(x) = 0 \quad \Longrightarrow \quad x \in \supp \rho\:.
\eeq
\item[{\textrm{(b)}}] The Lagrangian~$\L : \F \times \F \rightarrow \R^+_0$ is
a bounded function.
\item[{\textrm{(c)}}] The operator~$\L_\rho
: {\mathscr{D}}(\L_\rho) \rightarrow \H_\rho$ is strictly positive in the sense that
there is~$\varepsilon>0$ such that
\beq \label{Jspos}
\la \psi, \L_\rho \psi \ra_\rho \geq \varepsilon\: \|\psi\|_\rho^2 \qquad \text{for all~$\psi \in {\mathscr{D}}(\L_\rho)$} \:.
\eeq
\end{itemize}
Then~$\rho$ is a local minimizer.
\end{Prp} \noindent
We remark that condition~(b) could be replaced by weaker boundedness assumptions.
We do not aim for maximal generality because condition~(b) is suitable for
the applications we have in mind.

\Proof[Proof of Proposition~\ref{JLocSufficient}]
Let~$(\tilde{\rho}_\tau)_{\tau \in [0, \delta)}$ be as in Definition~\ref{Jdeflocmin}.
Since the curve~$\mu_\tau$ in Definition~\ref{Jdeflocmin} is regular, we know
that~$\dot{\tilde{\rho}}_0$ is non-zero.
Expanding~\eqref{Jintegrals2} in powers of~$\tau$, we obtain
\begin{align*}
\big( \Sact(\tilde{\rho}) - \Sact(\rho) \big) &= 2 \tau \int_\F \Big(\ell(x) + \frac{\nu}{2} \Big) \:d\dot{\tilde{\rho}}_0(x) 
+ \tau^2 \int_\F \Big(\ell(x) + \frac{\nu}{2} \Big) \:d\ddot{\tilde{\rho}}_0(x) \\
&\quad + 2 \tau^2 \int_\F d\dot{\tilde{\rho}}_0(x) \int_\F d\dot{\tilde{\rho}}_0(x)\: \L(x,y) + \O\big(\tau^3 \big)\:.
\end{align*}
Due to the volume constraint, the signed measures~$\dot{\tilde{\rho}}_0$ and~$\ddot{\tilde{\rho}}_0$
have total volume zero, so that the terms involving~$\nu$ drop out,
\beq \label{Jvary}
\begin{split}
\big( \Sact(\tilde{\rho}) - \Sact(\rho) \big) &= 2 \tau \int_\F \ell(x) \:d\dot{\tilde{\rho}}_0(x) 
+ \tau^2 \int_\F \ell(x)\:d\ddot{\tilde{\rho}}_0(x) \\
&\quad + 2 \tau^2 \int_\F d\dot{\tilde{\rho}}_0(x) \int_\F d\dot{\tilde{\rho}}_0(x)\: \L(x,y) + \O\big(\tau^3 \big)\:.
\end{split}
\eeq

We first consider the case that the measure~$\chi_{\F \setminus M} \dot{\tilde{\rho}}_0$ is non-zero.
Since the measures~$\tilde{\rho}_\tau$ are all positive,
we know that~$\chi_{\F \setminus M} \dot{\tilde{\rho}}_0$ is a positive measure.
Hence, using~\eqref{Jimply}, we conclude that
\[ \int_\F \ell(x)\:d\dot{\tilde{\rho}}_0(x) > 0 \:. \]
Hence the linear term in~\eqref{Jvary} ensures that~\eqref{JSrhoin} holds for sufficiently small~$\tau$.

It remains to consider the case that the measure~$\dot{\tilde{\rho}}_0$ is supported on~$M$.
Then the linear term in~\eqref{Jvary} vanishes because of the EL equations~\eqref{JEL1}.
Repeating the above argument with~$\dot{\tilde{\rho}}_0$ replaced by~$\ddot{\tilde{\rho}}_0$, we
find that $\chi_{\F \setminus M} \ddot{\tilde{\rho}}_0$  is a positive measure.
From this it follows that
\[ \int_\F \ell(x)\:d\ddot{\tilde{\rho}}_0(x) \geq 0 \:. \]
Therefore, in order to conclude the proof, it remains to show that
\begin{align}\label{Jshow3}
\int_M d\dot{\tilde{\rho}}_0(x) \int_M d\dot{\tilde{\rho}}_0(y)\: \L(x,y) > 0 \:.
\end{align}

We now use the following approximation argument. We choose a sequence~$\psi_n \in
{\mathscr{D}}(\L_\rho)$ such that
\beq \label{Jmeasurelim}
\psi_n\: \rho \rightarrow \dot{\tilde{\rho}}_0 \neq 0 \qquad \text{in~$\mathfrak{B(\F)}$}\:.
\eeq
Then
\begin{align*}
\int_M \psi_n(x)\: d\rho(x) \int_M \psi_n(y)\: d\rho(y)\: \L(x,y)  \rightarrow \int_M d\dot{\tilde{\rho}}_0(x) \int_M d\dot{\tilde{\rho}}_0(x)\: \L(x,y),
\end{align*}
because, setting~$\dot{\tilde{\rho}}_0 = \psi_n\: \rho + \Delta\rho$, we have
\begin{align*}
&\int_M d\dot{\tilde{\rho}}_0(x) \int_M d\dot{\tilde{\rho}}_0(x)\: \L(x,y) - \int_M \psi_n(x)\: d\rho(x) \int_M \psi_n(y)\: d\rho(y)\: \L(x,y)\\
&= \int_M d\Delta\rho(x) \int_M d\Delta\rho(y) \, \L(x,y) + 2 \int_M d\Delta\rho(x) \int_M \psi_n(y)\: d\rho(y) \, \L(x,y) \\
&\leq C \, \| \Delta \rho \|_{\mathfrak{B}(\F)}^2 + 2 \, \| \psi_n \|_{L^\infty(M)} \, \sup_{M} \Big( \ell + \frac{\nu}{2} \Big) \, \| \Delta \rho \|_{\mathfrak{B}(\F)} \rightarrow 0 \:,
\end{align*}
where $C:=\sup_{x,y \in \F}\L(x,y)$ is the pointwise bound of the Lagrangian.
Using the strict positivity~\eqref{Jspos}, we have
\begin{align} \label{Jsposuse}
\varepsilon\: \|\psi_n\|_\rho^2 \leq
\la \psi_n, \L_\rho \psi_n \ra_\H &= \int_M \psi_n(x)\: d\rho(x) \int_M \psi_n(y)\: d\rho(y)\: \L(x,y) ,
\end{align}
hence the left hand side of~\eqref{Jshow3} cannot be negative. 
Let us assume that it is zero,
\begin{align}\label{JlimNull}
\int_M d\dot{\tilde{\rho}}_0(x) \int_M d\dot{\tilde{\rho}}_0(y)\: \L(x,y) = 0 \, .
\end{align}
Using~\eqref{Jsposuse}, it follows that
\[ \|\psi_n\|_\rho^2 \rightarrow 0 \:. \]
Thus $\psi_n \rightarrow 0$ converges pointwise almost everywhere in $M$.
It follows that $\psi_n \rho \rightarrow 0$ in $\mathfrak{B}(\F)$, in contradiction to~\eqref{Jmeasurelim}.
This shows that assumption~\eqref{JlimNull} is false, concluding the proof.
\QED

\subsection{The Setting of Causal Fermion Systems} \label{Jseccfs}
We now explain how the causal action principle for causal fermion systems
(as introduced in Section~\ref{Nseccfsbasic}) can be described within the above setting.
The main difference compared to the setting in Section~\ref{Jsecnoncompact} is
that the causal action principle involves additional constraints, namely the
trace constraint and the boundedness constraint. We now explain how to
incorporate these constraints in a convenient way.
For a minimizer of the causal action, the local trace is constant on the support of the
universal measure (see~\cite[Proposition~1.4.1]{cfs}). With this in mind, we may
restrict attention to operators with fixed trace. When doing so, the trace constraint is
trivially satisfied. The boundedness constraint, on the other hand, can be incorporated by
Lagrange multiplier term. Finally, in the setting of causal fermion systems, the set~$\F$
is not necessarily a smooth manifold. In order to avoid this problem, we restrict attention
to minimizers for which all space-time points are regular (see~\cite[Definition~1.1.5]{cfs}).
Then we may restrict attention to operators which have exactly~$n$ positive and~$n$
negative eigenvalues. The resulting set of operators is a smooth manifold
(see the concept of a flag manifold in~\cite{helgason}). This leads us to the following setup:

Let~$(\H, \la .|. \ra_\H)$ be a finite-dimensional complex Hilbert space. Moreover, we are given
parameters~$n \in \N$ (the spin dimension), $c > 0$ (the constraint for the local trace)
and~$\kappa>0$ (the Lagrange multiplier of the boundedness constraint)\footnote{We
remark that the Lagrange multiplier~$\kappa$ is strictly positive because
otherwise there are no minimizers; see~\cite[Example~2.9]{continuum}
and~\cite[Exercise~1.4]{cfs}.}.
We let~$\F \subset \Lin(\H)$ be the set of all
self-adjoint operators~$F$ on~$\H$ with the following properties:
\begin{itemize}[leftmargin=2em]
\itemsep-.2em
\itemD $F$ has finite rank and (counting multiplicities) has~$n$ positive and~$n$ negative eigenvalues. \\[-0.8em]
\itemD The local trace is constant, i.e.
\[ \tr(F) = c\:. \]
\end{itemize}
On~$\F$ we consider the topology induced by the sup-norm on~$\Lin(\H)$.
For any~$x, y \in \F$, the product~$x y$ is an operator
of rank at most~$2n$. We denote its non-trivial eigenvalues counting algebraic multiplicities
by~$\lambda^{xy}_1, \ldots, \lambda^{xy}_{2n} \in \C$.
We introduce the Lagrangian by
\beq \label{JLdef}
\L(x,y) = \frac{1}{4n} \sum_{i,j=1}^{2n} \Big( \big|\lambda^{xy}_i \big| - \big|\lambda^{xy}_j \big| \Big)^2
+ \kappa\: \bigg( \sum_{i,j=1}^{2n} \big|\lambda^{xy}_i \big| \bigg)^2 \:.
\eeq
Clearly, this Lagrangian is continuous on~$\F \times \F$.
Moreover, since~$\kappa>0$, the Lagrangian is non-negative.

Therefore, we are back in the setting of
Section~\ref{Jsecnoncompact}. The EL equations in Lemma~\ref{JlemmaEL} agree with
the EL equations as derived for the causal action principle with constraints
in~\cite{lagrange} (cf. Theorem~\ref{Nthm3}).

\section{The Symplectic Form in the Smooth Setting}\label{JSecSmooth}
In order to introduce our concepts in the simplest possible setting,
in this section we assume that~$\F$ is a smooth manifold of dimension~$m \geq 1$
and that the Lagrangian~$\L \in C^\infty(\F \times \F, \R^+_0)$ is smooth.
Moreover, we let~$\rho$ be a regular Borel measure on~$\F$
which satisfies the EL equations \eqref{JELstrong} corresponding to the causal action in the sense that
the smooth function~$\ell$ defined by~\eqref{Jelldef},
\[ \ell(x) = \int_\F \L(x,y)\: d\rho(y) - \frac{\nu}{2} \in C^\infty(\F,\R_0^+) \:, \]
is minimal and vanishes on~$M$,
\beq \label{JELsmooth}
\ell|_{\supp \rho} \equiv \inf_\F \ell = 0 \:.
\eeq
The constructions in this section should be seen as a preparation
for the lower semi-continuous setting to be considered in Section~\ref{JlowerSemiCont}.
Before going on, we remark that the value of the parameter~$\nu$ can be changed arbitrarily
by rescaling the measure according to
\[ \rho \rightarrow \lambda \rho \qquad \text{with} \qquad \lambda>0 \:. \]
Therefore, without loss of generality we can keep~$\nu$ fixed when varying or perturbing the measure.

\subsection{The Weak Euler-Lagrange Equations}
Clearly, the EL equations~\eqref{JELsmooth} imply the weaker equations
\beq \label{JELweak}
\ell|_M \equiv 0 \qquad \text{and} \qquad D \ell|_M \equiv 0
\eeq
(where~$D \ell(p) : T_p \F \rightarrow \R$ is the derivative).
In order to combine these two equations in a compact form, we introduce the
smooth one-jets
\beq \label{JJdef}
\J := \big\{ \u = (a,u) \text{ with } a \in C^\infty(\F, \R) \text{ and } u \in \Gamma(\F) \big\} \:,
\eeq
where~$\Gamma(\F)$ denotes the smooth vector fields on~$\F$.
Defining the derivative in direction of a one-jet by
\beq \label{JDjet}
\nabla_{\u} \ell(x) := a(x)\, \ell(x) + \big(D_u \ell \big)(x) \:,
\eeq
we can write~\eqref{JELweak} as
\beq \label{JELweak2}
\nabla_{\u} \ell |_M \equiv 0 \qquad \text{for all~$\u \in \J$}\:.
\eeq
We refer to these equations as the {\em{weak EL equations}}.

\subsection{The Nonlinear Solution Space} \label{Jsecfam}
Our next step is to analyze families of measures which satisfy the weak EL equations.
In order to obtain these families of solutions, we want to vary a given
measure~$\rho_0$ (not necessarily a minimizer) without changing its general structure.
To this end, we multiply~$\rho_0$ by a weight function
and apply a diffeomorphism, i.e.
\beq \label{JrhoFf}
\rho = F_* \big( f \,\rho_0 \big) \:,
\eeq
where~$F : \F \rightarrow \F$ is a smooth diffeomorphism and~$f \in C^\infty(\F, \R^+)$.
We now consider a set of such measures which all satisfy the weak EL equations,
\beq
\calB \subset \left\{ \text{$\rho$ of the form~\eqref{JrhoFf}} \:\big|\: \text{the weak EL equations~\eqref{JELweak2} are satisfied} \right\} \label{JBdef}
\eeq
(for fixed~$\rho_0$). 
We make further simplifying assumptions on $\calB$.
First, we assume that~$\calB$ is a smooth {\em{Fr{\'e}chet manifold}}
(endowed with the compact-open topology on~$C^\infty(\F, \R^+_0)$
and on the diffeomorphisms; for details see Appendix~\hyperref[appclosed]{A}).
Then for~$\rho \in \calB$, a tangent vector~$\v \in T_\rho \calB$,
being an infinitesimal variation of the measures in~\eqref{JrhoFf},
consists of a function~$b$ (describing the infinitesimal change of the weight) and
a vector field~$v$ (being the infinitesimal generator of the diffeomorphism).
Having chosen the Fr{\'e}chet topology such that~$b$ and~$v$ are smooth, we obtain a jet~$\v = (b,v) \in \J$.
Hence the tangent space can be identified with a subspace of the one-jets,
\[ \T_\rho \calB \subset \J \:. \]
A second assumption is needed in order to ensure that we can exchange integration with differentiation in the proof of Lemma~\ref{Jlemmalin} below.
To this end, we assume that for every smooth curve~$(\tilde{\rho}_\tau)_{\tau \in (-\delta, \delta)}$
in~$\calB$, the corresponding functions~$(f_\tau, F_\tau)$ have 
the properties that the derivatives
\begin{align}\label{JBassumption2}
 \frac{d}{d\tau} \L\big(F_\tau(x), F_\tau(y) \big)\: f_\tau(y) \qquad \textrm{and} \qquad  \frac{d}{d\tau} D_1 \L\big(F_\tau(x), F_\tau(y) \big)\: f_\tau(y)
\end{align}
are bounded uniformly in $\tau$ for every $x,y \in M$ and are~$\rho$-integrable in $y$ for every $x \in M$
(just as in~\eqref{JIntreqLinFieldEq}, the subscripts of~$D_1$ and~$D_2$ denote the
partial derivatives acting on the first respectively second argument of the Lagrangian).

\begin{Lemma} \label{Jlemmalin}
For any~$\u \in \J$ and $\v \in T_\rho \calB$,
\beq \label{Jeqlin}
\nabla_\u \nabla_\v \ell(x) + \int_M \nabla_{1,\u} \nabla_{2,\v}\L(x,y)\: d\rho(y) = 0 \qquad \text{for all~$x \in M$}\:.
\eeq
\end{Lemma} \noindent
We refer to~\eqref{Jeqlin} as the {\em{linearized field equations}}
(see also the explanation in the introduction before~\eqref{JIntreqLinFieldEq}).

Before giving the proof of this lemma, we point out that in~\eqref{Jeqlin}, the order of differentiation is irrelevant.
This is obvious for the term~$\nabla_{1,\u} \nabla_{2,\v}\L(x,y)$ because the derivatives act on
different variables. In the first term, it follows from the computation
\[  \nabla_\u \nabla_\v \ell(x) - \nabla_\v \nabla_\u \ell(x) = (D_u a)(x)\, \ell(x) - (D_v b)(x)\, \ell(x)
+ D_{[u,v]} \ell(x) = 0 \:, \]
where in the last step we used the weak EL equations~\eqref{JELweak}.
\Proof[Proof of Lemma~\ref{Jlemmalin}]
Given~$\v=(b,v) \in T_\rho \calB$, we let~$(\tilde \rho_\tau)_{\tau \in (-\delta,\delta)}$ be a smooth curve in~$\calB$
with~$\tilde{\rho}_0=\rho$ and~$\dot{\tilde{\rho}}_0 = \v$. 
As shown in Lemma~\ref{Jcurve} below, there are~$F_\tau$ and~$f_\tau$ such that
\beq \label{Jparamtilrho}
\tilde{\rho}_\tau = (F_\tau)_* \big( f_\tau \,\rho \big) \:,
\eeq
and therefore
\beq \label{Jparamtilrho2}
\dot{\tilde{\rho}}_0 = \frac{d}{d\tau} \Big((F_\tau)_* \big( f_\tau \,\rho \big) \Big) \Big|_{\tau=0}
\qquad \text{with} \qquad \dot{f}_0 = b,\;\;\dot{F}_0 = v \:.
\eeq
Setting~$M_\tau = \supp \tilde{\rho}_\tau$ and using that~$M_\tau = \overline{F_\tau(M)}$,
the weak EL equations~\eqref{JELweak} can be written as
\beq \label{JELweak3}
\ell_\tau \big(F_\tau(x) \big) \equiv 0 \quad \text{and} \quad D \ell_\tau \big( F_\tau(x) \big) \equiv 0
\qquad \text{for all~$x \in M$}\:,
\eeq
where
\begin{align}\label{JellTauSmooth}
 \ell_\tau(z) := \int_\F \L(z,y)\: d\tilde{\rho}_\tau(y) - \frac{\nu}{2} \;\in\; C^\infty(\F, \R)\:.
\end{align}
Differentiating the first equation in~\eqref{JELweak3} with respect to~$\tau$, we obtain
\begin{align*}
0 &= \frac{d}{d\tau} \ell_\tau \big(F_\tau(x) \big) \big|_{\tau=0}
= \frac{d}{d\tau} \int_\F \L\big(F_\tau(x), y \big)\: d \Big( \big(F_\tau \big)_* \big( f_\tau \,\rho \big) \Big) (y)
\Big|_{\tau=0} \\
&= \frac{d}{d\tau} \int_M \L\big(F_\tau(x), F_\tau(y) \big)\: f_\tau(y)\: d\rho (y) \Big|_{\tau=0}\\
&= D_v \ell(x) + \int_M \L(x,y)\: b(y)\: d\rho(y) + \int_\F D_{2,v}\L(x,y)\: d\rho(y) \:.
\end{align*}
In the last step, we exchanged integration with differentiation. This is justified
by our assumption~\eqref{JBassumption2}, which 
ensures that the integrand of the second line is $L^1(\F,d\rho)$ for every $\tau$, is differentiable in $\tau$ for every $x,y \in M$ and is dominated by a $L^1(\F,\rho)$-function uniformly in $\tau$. 
Using the notation~\eqref{JDjet}, we can write this as
\beq \label{JjEL1}
D_v \ell(x) + \int_M \nabla_{2,\v}\L(x,y)\: d\rho(y) = 0 \:.
\eeq
Differentiating the second equation in~\eqref{JELweak3}, a similar computation gives
for any vector field~$u$
\beq \label{JjEL2}
D_v D_u \ell(x) + \int_M D_{1,u} \nabla_{2,\v}\L(x,y)\: d\rho(y) = 0 \:.
\eeq
Multiplying~\eqref{JjEL1} by~$a(x)$ and adding~\eqref{JjEL2}, we obtain
\begin{align*}
0 &= a(x) \:D_v \ell(x) + D_v D_u \ell(x) + \int_M \nabla_{1,\u} \nabla_{2,\v}\L(x,y)\: d\rho(y) \\
&= D_v \nabla_\u \ell(x) - (D_v a)(x)\, \ell(x)+ \int_M \nabla_{1,\u} \nabla_{2,\v}\L(x,y)\: d\rho(y) \\
&= \nabla_v \nabla_\u \ell(x) - b(x)\, \nabla_\u \ell(x)
- (D_v a)(x)\, \ell(x)+ \int_M \nabla_{1,\u} \nabla_{2,\v}\L(x,y)\: d\rho(y) \:.
\end{align*}
Using the weak EL equations~\eqref{JELweak2}, the second and third summands vanish, giving the result.
\QED

\begin{Lemma}\label{Jcurve}
Let $\rho \in \calB$ and $(\tilde \rho_\tau)_{\tau \in (-\delta,\delta)}$ be a curve such that $\tilde \rho_0 = \rho$.
Then there is a family of smooth diffeomorphisms~$F_\tau : \F \rightarrow \F$ and functions~$f_\tau \in C^\infty(\F, \R^+)$ such that
\begin{align}
\label{JCurveCurve}
\tilde{\rho}_\tau = (F_\tau)_* \big( f_\tau \,\rho \big) \: .
\end{align}
If the curve~$(\tilde \rho_\tau)_{\tau \in (-\delta,\delta)}$ is smooth,
then both~$F_\tau$ and $f_\tau$ are smooth in $\tau$.
\end{Lemma}
\Proof Let $\rho_0$ be the measure in the definition of $\calB$,~\eqref{JBdef}. Then there are $G : \F \rightarrow \F$ and~$g \in C^\infty(\F, \R^+)$ such that
\[ \rho = G_* \big( g \,\rho_0 \big) \:.\]
Thus
\[ \rho_0 = \frac{1}{ g} \:\Big( \big(G^{-1} \big)_* \rho \Big) = ( G^{-1})_* \bigg( \frac{1}{ g \circ G^{-1} } \:\rho \bigg) \:. \]
Similarly, there are mappings~$\tilde G: \F \rightarrow \F$ and~$\tilde g\in C^\infty(\F, \R^+)$ such that
\begin{align*}
\tilde \rho_\tau &=  \tilde G_* \big( \tilde g \,\rho_0 \big) 
= \tilde G_* \bigg( \tilde g \:\big( G^{-1} \big)_* \Big( \frac{1}{ g \circ G^{-1} } \:\rho \Big)  \bigg) \\
&= (\tilde G \circ G^{-1})_* \bigg(  \frac{\tilde g \circ G^{-1} }{ g \circ G^{-1} }\: \rho \bigg) \:.
\end{align*}
This gives the desired functions~$F_\tau$ and~$f_\tau$ for fixed~$\tau$.
The claim about smoothness follows from the topology and the differential structure of $\calB$ as described Appendix~\hyperref[appclosed]{A} (a curve~\eqref{JCurveCurve} is smooth if and only if~$F_\tau$ and $f_\tau$ are smooth
in $\tau$).
\QED
We remark that the strict positivity of the weight function~$f$ in~\eqref{JrhoFf} is needed
because in the last proof we divided by these weight functions.

\subsection{The Symplectic Form and Hamiltonian Time Evolution}\label{JSecSympForm}
For any~$\u, \v \in T_\rho \calB$ and~$x,y \in M$, we set
\[ \sigma_{\u, \v}(x,y) := \nabla_{1,\u} \nabla_{2,\v} \L(x,y) - \nabla_{1,\v} \nabla_{2,\u} \L(x,y) \:. \]
For any compact~$\Omega \subset \F$, we introduce the surface layer integral
\beq \label{JOSI}
\sigma_\Omega \::\: T_\rho \calB \times T_\rho \calB \rightarrow \R\:,
\qquad \sigma_\Omega(\u, \v) = \int_\Omega d\rho(x) \int_{M \setminus \Omega} d\rho(y)\:
\sigma_{\u, \v}(x,y) \:.
\eeq

We are now in the position to specify and prove the theorem mentioned in Section~\ref{JIntr} of the introduction.
\begin{Thm} \label{JthmOSI}
Let~$\F$ be a smooth manifold of dimension~$m \geq 1$, and~$\L \in C^\infty(\F \times \F, \R^+_0)$ 
be a smooth Lagrangian. Moreover, let~$\calB$ be a Fr{\'e}chet manifold of measures
of the form~\eqref{JBdef}.
Then for any compact~$\Omega \subset \F$, the surface layer integral~\eqref{JOSI}
vanishes for all~$\u, \v \in T_\rho \calB$.
\end{Thm}
\Proof Anti-symmetrizing~\eqref{Jeqlin} in~$\u$ and~$\v$
and using that~$\nabla_{[\u, \v]} \ell = 0$, we obtain
\[ \int_M \sigma_{\u, \v}(x,y) \:d\rho(y) = 0 \qquad \text{for all~$x \in M$} \:. \]
We integrate this equation over~$\Omega$,
\begin{align}
0 &= \int_{\Omega} d\rho(x) \int_M  d\rho(y)\: \sigma_{\u, \v}(x,y) \notag \\
&= \int_{\Omega} d\rho(x) \int_\Omega d\rho(y)\: \sigma_{\u, \v}(x,y)
+ \int_{\Omega} d\rho(x) \int_{M \setminus \Omega} d\rho(y)\: \sigma_{\u, \v}(x,y) \:. \label{Jconserve}
\end{align}
Since the Lagrangian is symmetric in its two arguments, the
function~$\sigma_{\u, \v}$ is obviously anti-symmetric, i.e.\ $\sigma_{\u, \v}(x,y) = -\sigma_{\u, \v}(y,x)$.
Therefore, the first summand in~\eqref{Jconserve} vanishes. This gives the result.
\QED

At this point, we want to use the construction explained after~\eqref{JIntrOSIN}
in the introduction to obtain a
conserved symplectic form~$\sigma$. In the present smooth setting, this construction can be
made precise as follows.
Let us assume that~$M$ is a smooth manifold being a topological product
\[ M = \R \times N \]
with a (possibly non-compact) smooth manifold~$N$. (This is for example the case if $M$ is
a globally hyperbolic  Lorentzian manifold, where $\R$ corresponds to the time of a foliation
and $N$ corresponds to space.) For any~$t \in \R$, the
set~$N_t := \{t\} \times N$ is a hypersurface in~$M$, and it can be realized as a boundary,
\[ N_t = \partial \Omega_{N_t} \qquad \text{with} \qquad \Omega_{N_t} := (-\infty, 0) \times N \:. \]
Next, let us assume that the jets~$\v = (b,v) \in T_\rho \calB$
have suitable decay properties on $N_t$ (at spatial infinity) which ensure that the surface layer integrals~\eqref{JOSI}
exist for~$\Omega=\Omega_{N_t}$ and every~$t \in \R$.
Under these assumptions, Theorem~\ref{JthmOSI} implies that the bilinear form~$\sigma_{\Omega_{N_t}}$
is well-defined and does not depend on~$t$. This makes it possible to introduce the mapping
\beq \label{Jsigmaform}
\sigma \::\: T_\rho \calB \times T_\rho \calB \rightarrow \R \:, \qquad (\u, \v) \mapsto \sigma_{\Omega_{N_t}}(\u,\v)
\eeq
(where~$t \in \R$ is arbitrary). Due to the anti-symmetry, we can regard~$\sigma$ as a two-form
on~$\calB$. The next lemma shows that~$\sigma$ endows~$\calB$ with
the structure of a presymplectic Fr{\'e}chet manifold.
\begin{Lemma} \label{Jlemmaclosed}
The bilinear form~$\sigma$ is closed.
\end{Lemma}
\Proof Inspired by classical field theory (see for example~\cite[\S2.3]{deligne+freed}),
our strategy is to write~$\sigma$ locally as the exterior derivative of a one-form~$\gamma$.
Then the claim follows immediately from the fact that~$d^2=0$.

We let~$\tilde{\rho}$ be a measure in a neighborhood of~$\rho \in \calB$.
By definition of~$\calB$ we can represent~$\tilde{\rho}$ as
\beq \label{J362}
\tilde{\rho} = F_* \big( f \,\rho \big) \in \calB \:.
\eeq
We next define~$\gamma : T_{\tilde{\rho}} \calB \rightarrow \R$ by
\beq
\gamma(\u) = \int_{\Omega_{N_t}} d\rho \int_{M \setminus \Omega_{N_t}} \!\! d\rho\:
f(x)\: \nabla_{2,\u} \L\big(F(x), F(y)\big)\: f(y) \:. \label{JgammaUdef} 
\eeq
Computing the outer derivative with the formula
\[ (d\gamma)(\u,\v) = \u \gamma(\v) - \v \gamma(\u) - \gamma([\u,\v]) \:, \]
one finds that~$\sigma =  d\gamma$ (for details see Appendix~\hyperref[appclosed]{A}). This concludes the proof.
\QED

In order to obtain a symplectic structure on~$\calB$, the presymplectic form~$\sigma$ must be non-degenerate.
We do not see a general reason why this should be the case.
Therefore, we proceed as follows. Given~$\rho \in \calB$,
an abstract method to obtain a non-degenerate form is to mod out the kernel of~$\sigma$
defined by
\[ \text{ker} \,\sigma = \{ \v \in T_\rho \calB \;\big|\; \sigma(\u, \v) = 0 \text{ for all~$\u \in T_\rho \calB$} \big\} \:. \]
In most applications, it is useful to choose concrete representatives
of the vectors of this quotient space. To this end, one chooses a maximal
subspace~$\J^\text{symp}$ of~$T_\rho \calB$ on which~$\sigma$ is non-degenerate
(the existence of such a subspace is guaranteed by Zorn's lemma).
Then the restriction
\[ \sigma \::\: \J^\text{symp} \times \J^\text{symp} \rightarrow \R \]
is non-degenerate. The specific choice of~$\J^\text{symp}$ depends on the application.

To summarize, the above constructions gave us a presymplectic form~$\sigma$
on~$T_\rho \calB$ which is given as a surface layer integral~\eqref{JIntrOSI}
for~$\Omega=\Omega_{N_t}$. This presymplectic form is independent of~$t$.
In other words, the time evolution as specified by the linearized field equations~\eqref{Jeqlin}
preserves the symplectic form and is thus a symplectomorphism.
This is what we mean by {\em{Hamiltonian time evolution}}.

\section{The Lower Semi-Continuous Setting}\label{JlowerSemiCont}
We now return to the lower semi-continuous setting of Section~\ref{Jseccvpcfs}.
We assume that~$\rho$ satisfies the EL equations of the causal action
(see~\eqref{Jelldef} and~\eqref{JELstrong}). Thus we assume that the function~$\ell : \F \rightarrow \R$ defined by
\beq
\ell(x) := \int_\F \L(x,y)\: d\rho(y) - \frac{\nu}{2} \quad \text{is bounded and lower semi-continuous,} \label{Jelldeflip}
\eeq
and that it is minimal on the support of~$\rho$,
\beq
\ell|_{\supp \rho} \equiv \inf_\F \ell = 0 \label{JELstronglip}
\eeq
(here~$\nu>0$ is again the Lagrange multiplier describing the volume constraint;
see~\cite[\S1.4.1]{cfs}). We again introduce space-time as the support of the
universal measure,
\beq\label{JMlip}
M := \supp \rho\:. 
\eeq

\subsection{The Weak Euler-Lagrange Equations}\label{JWeakElLSC}
Since the function~$\ell$ as defined in~\eqref{Jelldeflip} is only lower semi-continuous,
the derivative in~\eqref{JELweak} in general does not exist. But for lower semi-continuous
functions, it is a reasonable assumption that the semi-derivatives exist, but may
take the value $+\infty$. This leads us to the following additional assumptions:
\begin{itemize}[leftmargin=2.5em]
\item[(v)] $\L$ has directional semi-derivatives in~$\R \cup \{\infty \}$: For any~$x,y \in \F$, $v \in T_x \F$ and
any curve $\gamma \in C^1((-1,1), \F)$ with~$\gamma(0)=x$ and~$\gamma'(0)=v$, the following
generalized semi-derivative exists
\beq \label{JDvL}
D^+_{1,v} \L(x,y) := \lim_{\tau \searrow 0} \frac{1}{\tau} \Big( \L \big( \gamma(\tau), y \big)
- \L \big( \gamma(0), y \big) \Big) \;\in\; \R \cup \{\infty\}
\eeq
and is independent of the choice of~$\gamma$.\label{JCond5}
\item[(vi)] For any~$x \in M$ and~$v \in T_x\F$, both sides of the 
following equation exist and are equal,\label{JCond6}
\[ D^+_{v} \ell(x) = \int_M D^+_{1,v} \L(x,y)\: d\rho(y) \:. \]
\end{itemize}
Under these assumptions, the EL equations~\eqref{JELstronglip} can again be tested weakly
with smooth jets. We define the jet space~$\J$ as in~\eqref{JJdef},
\beq \label{JJdeflip}
\J := \big\{ \u = (a,u) \text{ with } a \in C^\infty(\F, \R) \text{ and } u \in C^\infty(\F, T\F) \big\} \:.
\eeq
We remark that it would suffice to define the mappings~$f$ and~$F$ on an open neighborhood of~$M$. But,
keeping in mind that every such mapping can be extended smoothly to all of~$\F$, there is no
loss in generality to assume that~$f$ and~$F$ are defined on all of~$\F$.
When testing, only the restriction of the jets to~$M$ is of relevance. We thus define the jet space
\beq \label{JJM}
\J|_M := \big\{ \u = (a,u) \text{ with } a \in C^\infty(M, \R) \text{ and } u \in C^\infty(M, T\F) \big\} \:,
\eeq
where smooth functions and sections on~$M$ are defined as those functions (respectively sections)
which have a smooth extension to~$\F$.
Since in general only the semi-derivatives exist, in contrast to~\eqref{JELweak2}
the weak EL equations read
\beq \label{JELweaklip}
\nabla^+_{\u} \ell(x) \geq 0 \qquad \text{for all~$x \in M$ and~$\u \in \J|_M$}\:,
\eeq
where, similar to~\eqref{JDjet}, $\nabla^+_\u$ is defined as
\begin{align}\label{JDjetSemi}
\nabla^+_{\u} \ell(x) := a(x)\, \ell(x) + \big(D^+_u \ell \big)(x) \:.
\end{align}

We introduce~$\Jdiff$ as the subspace of jets on~$M$ such that~$\ell$ is differentiable in
the direction of the vector field, i.e.
\beq\label{JJDiffDef}
\Jdiff := \{ \u \in \J|_M \text{ with } \nabla^+_{\u} \ell = - \nabla^+_{-\u} \ell \} \;\subset\; \J|_M\:.
\eeq
Note that the last equation does not impose a condition for the scalar component of the jet, so that
\begin{align}\begin{split}
\Jdiff &= C^\infty(M, \R) \oplus \Gdiff \qquad \text{where} \\
\Gdiff \,&\!:= \{ u \in C^\infty(M, T\F) \text{ with } D^+_{u} \ell = - D^+_{-u} \ell \} \:.
\end{split} \label{JJDiffLip}
\end{align}
Thus for jets in~$\Jdiff$, the directional derivatives exist, so that~$\nabla^+_u \ell = \nabla_u \ell$.
Then~\eqref{JELweaklip} implies that
\beq\label{JWElDiff}
\nabla_{\u} \ell(x) = 0 \qquad \text{for all~$x \in M$ and~$\u \in \Jdiff$}\:. 
\eeq
As explained in the introduction after~\eqref{JIntreqLinFieldEq}, in physical applications
it suffices to use only part of the information contained in these equations.
To this end, we choose a linear subspace
\beq\label{JGammatest}
\Jtest = \Ctest(M, \R) \oplus \Gtest \;\subset\; \Jdiff
\eeq
and consider the {\em{weak EL equations}}
\beq \label{JELtest}
\nabla_{\u} \ell|_M = 0 \qquad \text{for all~$\u \in \Jtest$}\:.
\eeq
The choice of~$\Jtest$ depends on the specific application and is of no relevance for
the remainder of this section.

\subsection{Families of Solutions and Linearized Solutions}\label{JSmoothVar}
We now consider families of solutions of the weak EL equations.
To this end, we let~$(\tilde{\rho}_\tau)_{\tau \in (-\delta, \delta)}$ be a family of measures,
which similar to~\eqref{JrhoFf} and~\eqref{J362} we assume to be of the form
\beq \label{Jrhotau}
\tilde{\rho}_\tau = (F_\tau)_* \big( f_\tau \, \rho \big) \:,
\eeq
where~$f$ and~$F$ are smooth,
\[ f \in C^\infty\big((-\delta, \delta) \times \F \rightarrow \R^+ \big) \qquad \text{and} \qquad
F \in C^\infty\big((-\delta, \delta) \times \F \rightarrow \F \big) \:, \]
and have the properties~$f_0(x)=1$ and~$F_0(x) = x$ for all~$x \in M$.
Then the support of~$\tilde{\rho}_\tau$ is given by
\[ M_\tau := \supp \tilde{\rho}_\tau = \overline{F_\tau(M)} \:. \]
In order to formulate the weak EL equations~\eqref{JELtest} for~$\tau \neq 0$, there
is the complication that the jets in~$\Jtest$ are defined only on~$M$, whereas the
weak EL equations must be evaluated on~$M_\tau$. Therefore, we must introduce
a jet space~$\Jtest_\tau$ on~$M_\tau$. We choose~$\Jtest_\tau$ as the push-forward of~$\Jtest$
under~$F_\tau$, with an additional scalar component formed of the directional derivative of~$f_\tau$
(the reason for this choice will become clear in Lemma~\ref{Jlemmalinlip}).
More precisely,
\beq \label{JJtesttau}
\Jtest_\tau := \Big\{ \Big( (F_\tau)_* \big( a + D_u \log f_\tau \big), \:(F_\tau)_* u \Big) \text{ with } \u = (a,u) \in \Jtest \Big\} \:,
\eeq
where the push-forward is defined by
\beq \begin{split}
(F_\tau)_* a \:&:\: M_\tau \rightarrow \R^+ \:,\qquad \big( (F_\tau)_* a \big)(F_\tau(x)) = a(x) \\
(F_\tau)_* u \:&:\: M_\tau \rightarrow T\F \:,\qquad \big( (F_\tau)_* u \big)(F_\tau(x)) = DF_\tau|_x \,u(x)
\end{split} \label{Jpushforward}
\eeq
(equivalently, the last relation can be written as~$((F_\tau)_* u)|_{F_\tau(x)} \eta = u|_x (\eta \circ F_\tau)$
for any test function~$\eta$ defined in a neighborhood of~$F_\tau(x)$).

We point out that the push-forward of~$\Jdiff$ is in general not the same as the differentiable
jets corresponding to the measure~$\tilde{\rho}_\tau$, as is illustrated in the following example.

\begin{Example} (the causal variational principle on the sphere) \label{Jexsphere} {\em{
The causal variational principle on the sphere is obtained from the setting of causal
fermion systems by taking a mathematical simplification of a special case.
It was introduced in~\cite[Section~1]{continuum} (cf.~\cite[Examples~1.5, 1.6 and~2.8]{continuum})
and analyzed in more detail in~\cite[Section~5]{support}.
We choose~$\F = S^2$ and let~$\D \in C^\infty(\F \times \F, \R)$ be the smooth function
\[ \D(x,y) := 
2 \tau^2\: (1+ \langle x,y \rangle) \left( 2 - \tau^2 \:(1 - \langle x,y \rangle) \right) , \]
where $\tau \geq 1$ is a parameter of the model and $\langle .,. \rangle$ is the scalar product on~$\R^3$.
Obviously, the function~$\D$ depends only on the angle~$\vartheta \in [0,\pi]$ between the
points~$x, y \in S^2$ (defined by~$\cos \vartheta = \langle x, y \rangle$).
We here choose~$\tau = \sqrt{2}$, so that
\[ \D = \D(\vartheta) = 8 \: (1+ \cos \vartheta ) \cos \vartheta\:. \]
The function~$\D$ has a maximum at~$\vartheta=0$
and changes signs at~$\vartheta_{\max}:= \frac{\pi}{2}$; more precisely
\[ \D|_{[0, \vartheta_{\max})} > 0 \:, \quad \D(\vartheta_{\max})=0\:, \quad \D|_{(\vartheta_{\max}, \pi]} \leq 0 \:. \]
We define the Lipschitz-continuous {\em{Lagrangian}}~$\L$ by
\begin{align} \label{JLform}
\L = \max(0, \D) \in C^{0,1}(\F \times \F, \R^+_0) \:.
\end{align}
Hence $\L(\vartheta)$ is positive if and only if $0 \leq \vartheta < \vartheta_{\max}$.
Furthermore, $\L$ is not differentiable at $\vartheta = \vartheta_{\max}$ since the semi-derivatives $\partial_\vartheta^+ \L(\vartheta)$ and $\partial_\vartheta^- \L(\vartheta)$ do not agree at this point.

It is shown in~\cite{support} that, for our choice of $\tau$, a minimizer of the causal variational
principle~\eqref{JSact} is given by a normalized counting measure supported on an octahedron.
Thus, denoting the set of unit vectors in $\R^3$ by~$\mathbb B:= \{e_1,e_2,e_3\}$,
the measure
\begin{align*}
\rho = \frac{1}{6} \sum_{x \in  \pm \mathbb B } \delta_x
\end{align*}
is a minimizer (where~$\delta_x$ denotes the Dirac measure supported at~$x \in S^2$).
Note that for all distinct points~$x, y \in \supp \rho$, the angle~$\vartheta$
is either~$\frac{\pi}{2}$ or~$\pi$, implying that~$\L(x,y)=0$. As a consequence,
\begin{align*}
\ell(x) = \D(0) \qquad \textrm{ for all } x \in M \, .
\end{align*}

In order to determine $\Jdiff$ as defined in~\eqref{JJDiffLip}, given any~$x \in M$ and a
non-zero vector~$u \in T_x\F$,
we let~$\gamma: (-\delta,\delta) \rightarrow \F$ be a smooth curve with~$\gamma(0) = x$
and~$\dot \gamma(0) = u$. Qualitatively speaking, the function $\ell(\gamma(\tau))$
has a ``cusp-like minimum'' at~$\tau=0$ because for $\tau > 0$, there is at least one point $y \in M$ which contributes to 
$\ell(\gamma(\tau))$ whereas for $\tau < 0$, the same is true for a different point $\tilde y$.
This can be made precise as follows.
There is at least one point~$y \in M$ with~$\vartheta_{x,y}=\vartheta_{\max}$
and~$\partial_\tau \vartheta_{\gamma(\tau), y}<0$ at $\tau=0$. This point
contributes to~$\ell$ for positive~$\tau$, i.e.
\[ \ell(\gamma(\tau)) \geq \D(\vartheta_{\gamma(\tau), x}) + \D(\vartheta_{\gamma(\tau), y})
\qquad \text{if~$\tau\geq0$}  \]
and thus
\begin{align*}
D^+_u \ell(x) &= \partial_\tau^+ \ell \big( \gamma(\tau) \big) \big|_{\tau=0} \geq
\partial_\tau \big( \D(\vartheta_{\gamma(\tau), x}) + \D(\vartheta_{\gamma(\tau), y}) \big) \big|_{\tau=0} \\
&=  \D'(0)\, \partial_\tau\vartheta_{\gamma(\tau), x}|_{\tau=0} + \D'(\vartheta_{\max})
\: \partial_\tau \vartheta_{\gamma(\tau), y} > 0 \:.
\end{align*}
Here, in the second step we used that $\D$ is differentiable, and hence the one-sided derivatives agree with the derivative, and in the last step we used that~$\D'(0)=0$ and~$\D'(\vartheta_{\max})<0$.
Likewise, there is a point~$\tilde{y} \in M$ with~$\vartheta_{x,\tilde{y}}=\vartheta_{\max}$
and~$\partial_\tau \vartheta_{\gamma(\tau), \tilde{y}}>0$. This point contributes to~$\ell$ for negative~$\tau$,
implying that
\begin{align*}
D_{-u} \ell(x) &= -\partial_\tau^- \ell \big( \gamma(\tau) \big) \big|_{\tau=0} \geq
-\partial_\tau \big( \D(\vartheta_{\gamma(\tau), x}) - \D(\vartheta_{\gamma(\tau), \tilde{y}}) \big) \big|_{\tau=0} \\
&= -\D'(0)\, \partial_\tau\vartheta_{\gamma(\tau), x}|_{\tau=0} - \D'(\vartheta_{\max})
\: \partial_\tau \vartheta_{\gamma(\tau), \tilde{y}} > 0 \:.
\end{align*}
Hence~$D^+_u \ell(x) \neq -D^+_{-u}\ell(x)$, so that the directional derivative~$D_u \ell(x)$
does not exist. We conclude that~$\Gdiff = \{0\}$ and thus
\beq \label{JJdiffex}
\Jdiff = C^\infty(M) \oplus \{0\} \:.
\eeq

We next define a family of measures~$\tilde{\rho}_\tau$ as in~\eqref{Jrhotau}.
To this end, we choose trivial weight functions~$f_\tau \equiv 1$, and choose
a family of diffeomorphism which change the angles between the points of~$M$ in the sense that
\[ \vartheta_{F_\tau(x), F_\tau(y)} \neq \vartheta_{\max} \qquad \text{for all~$x,y \in M$ and all~$\tau \neq 0$}\:. \]
Then for any~$\tau \neq 0$, the Lagrangian is smooth on a neighborhood of~$M_\tau \times M_\tau$, so that
\[ \Jdiff(\tilde{\rho}_\tau) = C^\infty(M_\tau) \oplus C^\infty(M_\tau, T\F) \:. \]
On the other hand, using the formula in~\eqref{JJtesttau} to define the push-forward of~$\Jdiff$
as given by~\eqref{JJdiffex}, we obtain the jet space~$C^\infty(M_\tau) \oplus \{0\}$.
Hence the push-forward of~$\Jdiff$ does not coincide with the 
differentiable jets of the measure~$\tilde{\rho}_\tau$.
}} \QEDrem
\end{Example}

In the next lemma we bring the weak EL equations for families of solutions into a convenient form.
To this end, as in~\eqref{JellTauSmooth}, we define
\beq\label{Jelltau}
\ell_\tau(z) = \int_\F \L(z,y)\: d{\rho}_\tau(y) - \frac{\nu}{2} \:.
\eeq

\begin{Lemma} \label{Jlemmalinlip}
Assume that~$(\tilde{\rho}_\tau)_{\tau \in (-\delta, \delta)}$ is a family of
solutions of the weak EL equations~\eqref{JELtest} in the sense that
\beq\label{JweakNull}
\nabla_{\u(\tau)} \ell_\tau(z) = 0 \qquad \text{for all~$z \in M_\tau$ and~$\u(\tau) \in \Jtest_\tau$}\:.
\eeq
Then for any~$\u \in \Jtest$ and all~$\tau \in (-\delta, \delta)$,
\beq \label{Jfinal}
0 = \nabla_{\u} \bigg( \int_M f_\tau(x) \:\L\big(F_\tau(x), F_\tau(y) \big)\: f_\tau(y)\: d\rho(y) 
-\frac{\nu}{2} \: f_\tau(x) \bigg) \:.
\eeq
\end{Lemma} \noindent
\Proof Writing the jet~$\u(\tau)$ as in~\eqref{JJtesttau} as
\beq\label{Jpushfdef}
\u(\tau) = \Big( (F_\tau)_* \big( a + D_u \log f_\tau \big), \:(F_\tau)_* u \Big) 
\eeq
with~$\u \in \Jtest$ and using the definition of the
push-forward~\eqref{Jpushforward}, the weak EL equations~\eqref{JweakNull} yield
\beq \label{Jprelim}
0 = \nabla_{\tilde{\u}} \bigg( \int_M \L\big(F_\tau(x), F_\tau(y) \big)\: f_\tau(y)\: d\rho(y) - \frac{\nu}{2} \bigg)
\eeq
with~$\tilde{\u}=(a+(D_u \log f_\tau), \,u)$, valid for all~$\tau \in (-\delta, \delta)$ and $x \in M$.
Multiplying by~$f_\tau(x)$, we obtain
\begin{align*}
0 &= f_\tau(x) \:\nabla_{\tilde{\u}} \bigg( \int_M \L\big(F_\tau(x), F_\tau(y) \big)\: f_\tau(y)\: d\rho(y) -\frac{\nu}{2} \bigg) \\
&= f_\tau(x) \Big( a(x) + (D_u \log f_\tau)(x) + D_u \Big)
\bigg(  \int_M \L\big(F_\tau(x), F_\tau(y) \big)\: f_\tau(y)\: d\rho(y) -\frac{\nu}{2} \bigg) \\
&= \big( a(x) + D_u \big) \: f_\tau(x) \bigg(  \int_M \L\big(F_\tau(x), F_\tau(y) \big)\: f_\tau(y)\: d\rho(y) 
-\frac{\nu}{2} \bigg) \\
&= \nabla_\u \bigg( \int_M f_\tau(x) \:\L\big(F_\tau(x), F_\tau(y) \big)\: f_\tau(y)\: d\rho(y) -\frac{\nu}{2} \: f_\tau(x) \bigg)\:,
\end{align*}
making it possible to write~\eqref{Jprelim} ``more symmetrically'' in the form~\eqref{Jfinal}.
\QED

Differentiating~\eqref{Jfinal} naively with respect to~$\tau$, we obtain
the {\em{linearized field equations}}
\beq \label{Jeqlinlip}
\la \u, \Delta \v \ra|_M = 0 \qquad  \textrm{ for all } \u \in \Jtest
\eeq
with
\beq \label{Jeqlinlip2}
\la \u, \Delta \v \ra(x) := \nabla_{\u} \bigg( \int_M \big( \nabla_{1, \v} + \nabla_{2, \v} \big) \L(x,y)\: d\rho(y) - \nabla_\v \:\frac{\nu}{2} \bigg) \:,
\eeq
where~$\v$ is the jet~$\v = (\dot{f}_0, \dot{F}_0)$.
In order to see the connection to the linearized field equations in the smooth setting~\eqref{Jeqlin},
we note that a formal computation using~\eqref{Jelldeflip} gives
\[ \la \u, \Delta \v \ra(x) = \nabla_\u \nabla_\v \ell(x) + \int_M \nabla_{1, \u} \nabla_{2, \v} \L(x,y)\: d\rho(y)\:. \]
In the present lower semi-continuous setting, it is not clear if the derivatives exist, nor if
the derivatives may be interchanged with the integrals. It turns out to be preferable to
work with~\eqref{Jeqlinlip2}. In order to make sense of this expression, we
need to impose conditions on~$\v$. This leads us to the following definition:

\begin{Def} \label{Jdeflin} A jet~$\v \in \J$ 
is referred to as a {\textbf{solution of the linearized field equations}}
(or for brevity a {\textbf{linearized solution}}) if it has the following properties:
\begin{itemize}[leftmargin=3em]
\item[\textrm{(l1)}] For all~$y \in M$ and all~$x$ in an open neighborhood of~$M$, the 
following combination of derivatives exists,
\beq \label{Jderex1}
\big( \nabla_{1, \v} + \nabla_{2, \v} \big) \L(x,y) \;\in\; \R \:.
\eeq
Here the combination of directional derivatives is defined by
\[ \big( D_{1, v} + D_{2, v} \big) \L(x,y) := \frac{d}{d\tau} 
\L\big( F_\tau(x), F_\tau(y) \big) \big|_{\tau=0} \:,\]
where~$F_\tau$ is the flow of the vector field~$v$.
\item[\textrm{(l2)}] Integrating the expression~\eqref{Jderex1} over~$y$ with respect to the measure~$\rho$,
the resulting function (defined on an open neighborhood of~$M$) is differentiable in the
direction of every jet~$\u \in \Jtest$ and satisfies the linearized field equations~\eqref{Jeqlinlip}.
\end{itemize}
The vector space of all linearized solutions is denoted by~$\Jlin \subset \J$.
\end{Def}

In order to illustrate the significance of the condition~\eqref{Jderex1}, we now give an
an example where the derivative in~\eqref{Jderex1} does not exist.

\begin{Example} (the causal variational principle on the sphere continued) \label{Jexspherecont} {\em{
We return to the causal variational principle on the sphere as considered in Example~\ref{Jexsphere}.
We first want to give an example where for $\L$ as given in~\eqref{JLform},
the condition~\eqref{Jderex1} is violated.
To this end, we choose two points~$x, y \in S^2$ such that~$\vartheta_{x, y} = \vartheta_{\max}$.
Furthermore, we choose a vector field~$v$ which vanishes in a neighborhood of~$y$. Then
\beq \label{JLFF}
\L(F_\tau(x),F_\tau(y)) = \L(F_\tau(x),y)
\eeq
(where~$F_\tau$ is again the flow of the vector field~$v$).
Next, we choose~$v(x)$ to be nonzero, tangential to the great circle joining~$x$ and~$y$
and pointing in the direction of smaller geodesic distance to $y$. Then
\[ D^+_{1, -v} \L(x,y) = 0 \qquad \text{but} \qquad
D^+_{1, v} \L(x,y) = -\D'(\vartheta_{\max}) > 0\:. \]
Hence~\eqref{JLFF} is not differentiable at~$\tau=0$.
We conclude that the derivative in~\eqref{Jderex1} does not exist.

Despite the just-explained difficulty to satisfy the condition~\eqref{Jderex1},
the space~$\Jlin$ contains non-trivial vector fields. For example, if~$(F_\tau)_{\tau \in (-\delta, \delta)}$
is a smooth family of isometries of the sphere (for example rotations around a fixed axis), then the
corresponding jet~$\v = (0, v)$ with $v = \partial_\tau F_\tau|_{\tau=0}$ is in~$\Jlin$.
This follows because
\begin{align*}
\frac{d}{d\tau} \L\big( F_\tau(x), F_\tau(y) \big) \big|_{\tau=0} = 0 \, ,
\end{align*}
i.e.\ the derivatives in~\eqref{Jderex1} exist and are zero. It follows that condition (l2),
including the linearized field equations, are satisfied as well.
This shows that~$\v \in \Jlin$.
We conclude that~$\Jlin$ is a vector space of dimension at least three.

In view of~\eqref{JJdiffex}, this example also illustrates that~$\Jlin$ is in general not a subspace of~$\Jdiff$.
}} \QEDrem
\end{Example}

\subsection{The Symplectic Form and Hamiltonian Time Evolution}\label{JlscSymp}
Following the construction in Section~\ref{JSecSympForm}, we want
to anti-symmetrize the linearized field equations~\eqref{Jeqlinlip} in the jets~$\u$ and~$\v$.
To this end, we now consider~$\u, \v \in \Jtest \cap \Jlin$.
The conditions in Definition~\ref{Jdeflin} ensure that the linearized field equations~\eqref{Jeqlinlip} are well-defined.
For the construction of the symplectic form, we need additional technical assumptions.
In order to make minimal assumptions, we work with semi-derivatives only. The existence of right semi-derivatives in $\R \cup  \{\infty\}$ is
guaranteed by condition~(v) on page~\pageref{JDvL}. We define the left semi-derivative as
\beq \label{JLeftSemiDeriv}
D^-_{1,v} \L(x,y) := \lim_{\tau \nearrow 0} \frac{1}{\tau} \Big( \L \big( \gamma(\tau), y \big)
- \L \big( \gamma(0), y \big) \Big) \;\in\; \R \cup \{\infty\} \:,
\eeq
so that~$D^-_{1,v} = - D^+_{1,-v}$. Similar to~\eqref{JDjetSemi}, we define
\[ \nabla^-_{\u} \ell(x) := a(x)\, \ell(x) + \big(D^-_u \ell \big)(x) \:. \]
Before stating the additional assumptions, we explain why they are needed. First, we must
ensure that the individual terms in~\eqref{Jeqlinlip} exist and that we may exchange the differentiation with integration.
Second, when taking second derivatives, we must take into account that the jets are also differentiated.
We use the notation
\beq \label{JnotationJetsDiff}
\nabla^s_{\u(x)} \nabla^{s'}_{1, \v} \L(x,y)
\eeq
to indicate that the $\u$-derivative also acts on~$\v$. 
With this notation, we can state the additional technical assumptions as follows:
\begin{itemize}[leftmargin=3em]
\item[\textrm{(s1)}] The first and second semi-derivatives of the Lagrangian
in the direction of~$\Jtest \cap \Jlin$ exist in~$\R$. \label{JCondS1}
Moreover, for all~$x$ and~$y$ in a neighborhood of~$M$, the symmetrized first semi-derivatives of the Lagrangian
\[ \big( \nabla^+_{1,\u} + \nabla^-_{1,\u} \big) \L(x,y) \]
are linear in~$\u \in \Jtest \cap \Jlin$.
\item[\textrm{(s2)}] The second semi-derivatives can be interchanged with the $M$-integration, i.e.\
for all~$\u, \v \in \Jtest$ and~$s, s' \in \{\pm\}$,
\begin{align*}
\int_M \nabla^s_{\u(x)} \nabla^{s'}_{1,\v} \L(x,y) \:d\rho(y) &= 
\nabla^s_{\u} \int_M \nabla^{s'}_{1,\v}  \L(x,y) \:d\rho(y)
\\
\int_M \nabla^s_{1,\u} \nabla^{s'}_{2,\v} \L(x,y) \:d\rho(y) &= \nabla^s_{\u} \int_M \nabla^{s'}_{2,\v} \L(x,y) \:d\rho(y) \:.
\end{align*}
\item[\textrm{(s3)}] For any~$\u, \v \in \Jtest \cap \Jlin$, the 
commutator~$[\u, \v]$ is in~$\Jtest$.
\end{itemize}

\begin{Thm} \label{JthmOSIlip}
Under the above assumptions~{\em{(l1), (l2)}} and~{\em{(s1)--(s3)}}, 
for any compact~$\Omega \subset \F$, the surface layer integral
\beq \label{JOSIlip}
\sigma^{s,s'}_\Omega(\u, \v) = \int_\Omega d\rho(x) \int_{M \setminus \Omega} d\rho(y)\:
\sigma^{s,s'}_{\u, \v}(x,y) 
\eeq
with
\beq \label{JOSIIntegrand}
\sigma^{s,s'}_{\u, \v}(x,y) := \nabla^s_{1,\u} \nabla^{s'}_{2,\v} \L(x,y) - \nabla^{s'}_{1,\v} \nabla^s_{2,\u} \L(x,y)
\eeq
vanishes for all~$\u, \v \in \Jtest \cap \Jlin$ and all~$s,s' \in \{\pm\}$.
\end{Thm}
\Proof Our starting point are the linearized field equations~\eqref{Jeqlinlip}.
Condition (s1) ensures that we can treat the terms of~\eqref{Jeqlinlip2} independently if we
take semi-derivatives. First, 
for $\u, \v \in \Jlin \cap \Jtest$
we consider the term
\[ \nabla^s_{\u} \int_M \nabla^{s'}_{1, \v} \L(x,y)\: d\rho(y)  =  \int_M \nabla^s_{\u(x)} \nabla^{s'}_{1, \v} \L(x,y)\:
d\rho(y) \:. \]
We now exchange $\u$ and $\v$ as well as $s$ and $s'$ and take the difference. Using the relation
\[ \nabla_{1,\v}^{s} = ss' \, \nabla^{s'}_{1, ss' \v} \:, \]
we obtain
\begin{align*}
&\int_M \big( \nabla^s_{\u(x)} \nabla^{s'}_{1, \v} -\nabla^{s'}_{\v(x)} \nabla^{s}_{1, \u} \big)\L(x,y)\: d\rho(y) =
s s' \,  \int_M  \nabla^s_{1,ss' [\u,\v]} \L(x,y)\: d\rho(y)\\
&\;\;= s s' \,  \nabla^s_{ss' [\u,\v]} \Big( \ell(x) + \frac{\nu}{2} \Big)  =   \nabla_{[\u,\v]} \Big( \ell(x) + \frac{\nu}{2} \Big)  = \nabla_{[\u,\v]}  \frac{\nu}{2}
\end{align*}
Here, in the second step we used the assumption~(vi). In the third step, we applied assumptions~(s3)
and~\eqref{JGammatest}. In the last step, we used the weak EL equations.

Using this result, anti-symmetrizing~\eqref{Jeqlinlip2} (again by exchanging~$\u$ and $\v$ as well as $s$ and $s'$)
and using~\eqref{Jeqlinlip}, we obtain the equations
\[ \int_M \sigma^{s,s'}_{\u, \v}(x,y) d\rho(y) = 0 \:. \]
Integrating over~$\Omega$ gives
\beq \label{Jterm1.2}
\int_\Omega d\rho(x) \int_M d\rho(y)\: \sigma^{s,s'}_{\u, \v}(x,y) = 0 \:.
\eeq
Moreover, it is obvious by the anti-symmetry of~\eqref{JOSIIntegrand} that
\[ \int_\Omega d\rho(x) \int_\Omega d\rho(y)\: \sigma^{s,s'}_{\u, \v}(x,y) = 0 \:. \]
Subtracting this equation from~\eqref{Jterm1.2} gives the result.
\QED

Theorem~\ref{JthmOSIlip} again makes it possible to introduce a Hamiltonian time evolution,
just as explained in the introduction and in Section~\ref{JSecSympForm}.
The only difference compared to Section~\ref{JSecSympForm}
is that, due to the semi-derivatives in~\eqref{JOSIIntegrand}, the surface layer integral~\eqref{JOSIlip}
is in general not linear in the jets~$\u$ and~$\v$. In order to obtain a bilinear form,
we modify~\eqref{Jsigmaform} as follows.
\begin{Prp}\label{JlscSympBilin} The mapping~$\sigma$ defined by
\beq \label{JlscSympBilinEq}
\sigma \::\: (\Jtest \cap \Jlin) \times (\Jtest \cap \Jlin) \rightarrow \R \:, \qquad (\u, \v) \mapsto
\frac{1}{4} \sum_{s,s'=\pm} \sigma^{s,s'}_{\Omega_{N_t}}( \u, \v)
\eeq
is bilinear.
\end{Prp}
\Proof We have
\begin{align*}
\sigma(\u,\v) = \int_{\Omega_{N_t}} d\rho(x) \int_{M \setminus \Omega_{N_t}} d\rho(y)\:  \sigma_{\u, \v}(x,y) 
\end{align*}
with
\begin{align*}
&\sigma_{\u, \v}(x,y) =  \frac{1}{4} \sum_{s,s'=\pm} \sigma^{s,s'}_{\u, \v}(x,y) \\
&\;= \frac{1}{4} \: \Big( \nabla_{1,\u}^+ + \nabla_{1,\u}^- \Big) \Big( \nabla_{2,\v}^+  + \nabla_{2,\v}^- \Big) \L(x,y)
- \frac{1}{4} \: \Big( \nabla_{1,\v}^+ + \nabla_{1,\v}^- \Big) \Big( \nabla_{2,\u}^+ + \nabla_{2,\u}^- \Big) \L(x,y) \,.
\end{align*}
The assumption (s1) and the symmetry of $\L(x,y)$ (condition (i) on page~\pageref{JCond1}) imply that $\big( \nabla^+_{i,\u} + \nabla^-_{i,\u} \big) \L(x,y)$
is linear in $\u \in \Jtest \cap \Jlin$ for $i=1,2$. Hence~$\sigma_{\u, \v}(x,y)$ is indeed linear in~$\u$
and~$\v$.
\QED
Since the other considerations at the end of Section~\ref{JSecSympForm} apply without changes,
we do not repeat them here.

\section{Example: A Lattice System in~$\R^{1,1} \times S^1$} \label{Jseclattice}
We now illustrate the previous constructions in a detailed example
on two-dimensio\-nal Minkowski space~$\R^{1,1}$. This example is inspired by physical field theories
whose fields take values in~$S^1$, but it differs substantially form any such theory. 
Its main purpose is to allow for a test of the above constructions, it is not supposed 
to represent an actual model of a physical application of causal fermion systems.
This is the case mainly because we choose a special Lagrangian $\L$ instead of working with~\eqref{JLdef} or a special
case thereof.
Our choice of $\L$ has the advantage that the minimizer is discrete,
making the system suitable for a numerical analysis.

\subsection{The Lagrangian}
Let~$(\R^{1,1}, \la .,. \ra)$ be two-dimensional Minkowski space. Thus, denoting the
space-time points by~$\underline{x} = (x^0, x^1)$ and~$\underline{y}$, the inner product
takes the form
\[ \la \underline{x}, \underline{y} \ra = x^0 y^0 - x^1 y^1 \:. \]
Moreover, let~$\F$ be the set
\[ \F = \R^{1,1} \times S^1\:. \]
We denote points in~$x \in \F$ by $x=(\underline{x}, x^\varphi)$
with~$\underline{x} \in \R^{1,1}$ and~$x^\varphi \in [-\pi, \pi)$.
Next, we let~$A$ be the square
\[ A = (-1,1)^2 \subset \R^{1,1} \:. \]
Moreover, given~$\varepsilon \in (0,\frac{1}{4})$,
we let~$I$ be the the following subset of the interior of the light cones,
\[ I = \big\{ \underline{x} \in \R^{1,1} \, \big| \, \la \underline{x}, \underline{x} \ra > 0
\textrm{ and } |x^0| < 1 +\varepsilon \big\} \:. \]
Furthermore, we let~$f : \R^{1,1} \rightarrow \R$ be the function
\[ f(\underline{x}) = \chi_{B_\varepsilon(0,1)}(\underline{x}) + \chi_{B_\varepsilon(0,-1)}(\underline{x})
- \chi_{B_\varepsilon(1,0)}(\underline{x}) - \chi_{B_\varepsilon(-1,0)}(\underline{x}) \]
(where~$\chi$ is the characteristic function and~$B_\varepsilon$ denotes the open Euclidean ball 
of radius~$\varepsilon$ in~$\R^2 \simeq \R^{1,1}$). Finally, we let~$V : S^1 \rightarrow \R$ be the function
\[ V(\varphi) = 1 - \cos \varphi \:. \]
Given parameters~$\delta > 0$, $\lambda_I \geq 2$ and~$\lambda_A \geq 2 \lambda_I+\varepsilon$, the
Lagrangian~$\L$ is defined by
\beq
\begin{split}
\L(x,y) &= \lambda_A \ \chi_{A}( \underline x - \underline y) + \lambda_I \ \chi_{I}( \underline x - \underline y) +
V(x^\varphi -y^\varphi) \: f(\underline x - \underline y) \\
&\quad + \delta \, \chi_{B_\varepsilon(0,0)}(\underline{x}-\underline{y})\: V(x^\varphi -y^\varphi )^2 \:.
\end{split} \label{JDefL}
\eeq

\begin{Lemma}
The function~$\L(x,y)$ is non-negative and satisfies the conditions~\textrm{(i)} and \textrm{(ii)} on page~\pageref{JCond1}.
\end{Lemma}
\Proof The only negative contributions to~$\L(x,y)$ arise in the
term~$V(x^\varphi -y^\varphi) \: f(\underline x - \underline y)$ if $\underline x - \underline y \in B_\varepsilon(1,0) \cup B_\varepsilon(-1,0)$. For $x$ and $y$ with this property, we have
\begin{align*}
V(x^\varphi -y^\varphi) \: f(\underline x - \underline y) \geq -2 \, \quad \textrm{and} \quad
\lambda_I \: \chi_{I}( \underline x - \underline y) \geq 2 \,
\end{align*}
because $V(S^1) \subset [0,2] \subset \R$,~$\lambda_I \geq 2$ and $ B_\varepsilon(1,0) \cup B_\varepsilon(-1,0) \subset I$. We conclude that~$\L(x,y) \geq 0$.

Condition~(i) is satisfied because the sets $A$ and $I$ are point-symmetric around $\underline x = 0$, $f$ is a sum of characteristic functions of sets which are mutually point-symmetric and $V(\varphi) =  V(-\varphi)$.
Condition~(ii) is satisfied because $V$ is continuous and
because characteristic functions of open sets are lower semi-continuous.
\QED

\subsection{A Local Minimizer}
We next introduce a universal measure~$\rho$ supported on the unit
lattice~$\Gamma := \Z^2 \subset \R^{1,1}$
and show that for any~$\delta > 0$, it is a local minimizer of the causal action in the sense of Definition~\ref{Jdeflocmin}.
\begin{Lemma} The measure~$\rho$ given by
\beq
\label{JDefRho}
\rho =  \sum_{\underline x \in \Gamma} \: \delta_{(\underline x,0)}
\eeq
satisfies the conditions~\textrm{(iii)} and~\textrm{(iv)} (on page~\pageref{JCond3}) as well as~\textrm{(v)} and~\textrm{(vi)} (on page~\pageref{JCond5}).
\end{Lemma}
\Proof Condition~(iii) is satisfied because for every $x \in \F$, a neighborhood $U$ with $\rho(U) < \infty$
is given for example by~$U = B_\varepsilon(\underline x)  \times  (x^\varphi -\varepsilon,x^\varphi +\varepsilon) \subset \R^{1,1} \times S^1$.
Condition~(iv) is satisfied because the function~$\L(x, .)$ is bounded and has compact support (which implies $\rho$-integrability and boundedness of $\ell$), and because $\ell$ is a finite sum of lower semi-continuous functions.

Condition~(v) is satisfied because the generalized derivative of characteristic functions exists.
Condition~(vi) holds because $\ell(x)$ is a finite sum of terms of the form $\L(x,y)$. Hence differentiation and
integration may be interchanged.
\QED

Clearly, the support of the above measure is given by
\beq \label{JMlattice}
M:= \supp \rho = \Gamma \times \{0\} \subset \F \:.
\eeq

\begin{Lemma}\label{JExStrongEL} The measure~\eqref{JDefRho}
satisfies the EL equations~\eqref{JELstrong} if the parameter~$\nu$ in~\eqref{Jelldef} is chosen as
\beq \label{Jnuval}
\nu = 2 \lambda_A + 4 \lambda_I \:.
\eeq
If~$\delta>0$, the implication~\eqref{Jimply} holds.
\end{Lemma}
\Proof If $x \in M$, a direct computation using~\eqref{JDefL} and~\eqref{Jelldef} shows that
\[ \ell(x) =  \lambda_A + 2 \lambda_I - \frac{\nu}{2} = 0\:. \]
Conversely, if~$x \notin M$, then either~$\underline{x} \notin \Gamma$ or~$x^\varphi \neq 0$.
In the first case, the characteristic function~$\chi_A(\underline{x} - . )$ equals one on at least two lattice points,
implying that~$\ell(x) \geq  2 \lambda_A + 2 \lambda_I - \frac{\nu}{2}=\lambda_A>0$.
In the remaining case~$\underline{x} \in \Gamma$ and~$x^\varphi \neq 0$,
the term~$\delta \: V(x^\varphi -0 )^2$ is non-negative, and it is strictly positive
if~$\delta>0$. This concludes the proof.
\QED
This lemma shows that for $\delta >0$, condition~(a) in Proposition~\ref{JLocSufficient} is satisfied. The following Lemma shows that condition~(c) holds as well:
\begin{Lemma}\label{JExSpos} Choosing~$\lambda_A \geq 2 \lambda_I + \varepsilon$,
the inequality
\[ \la \psi, \L_\rho \psi \ra_\rho \geq \varepsilon\: \|\psi\|_\rho^2 \qquad
\text{holds for all~$\psi \in {\mathscr{D}}(\L_\rho)$} \:. \]
\end{Lemma}
\begin{proof}
For $\psi \in {\mathscr{D}}(\L_\rho)$ as defined by~\eqref{JDomainLrho} and $e_t := (1,0) \in \R^{1,1}$, we have
\[ \big( \L_\rho \psi \big)(\underline x) = \lambda_A \psi(\underline x) + \lambda_I \big( \psi \big(\underline x + e_t \big) + \psi \big (\underline x - e_t \big) \big) \]
and hence
\beq
\la \psi, \L_{\rho} \: \psi \ra_\rho = \lambda_A \| \psi \|_\rho^2  + \lambda_I \sum_{\underline x \in \Gamma}
 \overline{\psi \big(\underline x \big)}
\Big( \psi \big(\underline x + e_t \big) + \psi \big (\underline x - e_t \big) \Big) \:. \label{JLrhoEstimate}
\eeq
Applying Young's inequality
\[ \Big| \overline{\psi \big(\underline x \big)} \:\psi \big(\underline x + e_t \big) \Big| \leq \frac{1}{2} \Big(
\big| \psi \big(\underline x + e_t \big) \big|^2 + \big| \psi \big(\underline x \big) \big|^2 \Big) \:, \]
the second term of~\eqref{JLrhoEstimate} can be estimated by
\[ \lambda_I \:\sum_{\underline x \in \Gamma}
 \overline{\psi \big(\underline x \big)}
\Big( \psi \big(\underline x + e_t \big) + \psi \big (\underline x - e_t \big) \Big) \geq -2 \lambda_I
\| \psi \|_\rho ^2 \:. \]
Hence~$\la \psi, \L_{\rho} \: \psi \ra_\rho \geq (\lambda_A-2 \lambda_I) \| \psi \|_\rho ^2$,
giving the result.
\end{proof}
\begin{Corollary}%
The measure $\rho$ is a local minimizer of the causal action.
\end{Corollary}
\begin{proof} Lemma~\ref{JExStrongEL} shows that condition~(a) in Proposition~\ref{JLocSufficient} holds.
Condition~(b) follows because~$\L(x,y)$ as defined by~\eqref{JDefL} is bounded on $\F \times \F$.
Lemma~\ref{JExSpos} yields condition~(c).
\end{proof}

\subsection{The Jet Spaces}
We next determine the jet spaces. Clearly, in our setting of a discrete lattice~\eqref{JMlattice}, every function on~$M$
can be extended smoothly to a neighborhood of~$M$. Thus the jet space~\eqref{JJM} can be written as
\[ \J|_M = \big\{ \u = (a,u) \text{ with } a : M \rightarrow \R \text{ and } u : M \rightarrow T\F \big\} \:. \]
When extending these jets to~$\F$, for convenience we always
choose an extension~$\u : \F \rightarrow T\F$ which is locally constant on~$M$.
We denote the vector component by $u=(u^0, u^1, u^\varphi)$.
In order to determine the differentiable jets, we recall from the the proof of Lemma~\ref{JExStrongEL} that
\[ \left\{ \begin{array}{ll}
\ell(\underline x, x^\varphi) =  \delta \, V(x^\varphi)^2 & \textrm{if }\underline{x} \in \Gamma \\[0.3em]
\ell(\underline x, x^\varphi) \geq \lambda_A + \delta \, V(x^\varphi)^2 &  \textrm{if } \underline{x} \notin \Gamma \:.
\end{array} \right. \]
Hence the differentiable jets~\eqref{JJDiffLip} are given by
\[ \Jdiff = \big\{ \u = (a,u) \text{ with } a : M \rightarrow \R \text{ and } u=(0,0,u^\varphi) : M \rightarrow T\F \big\} \:. \]
We choose $\Jtest = \Jdiff$. 

\begin{Prp}\label{JExLin} The linearized solutions $\Jlin$ of Definition~\ref{Jdeflin} consist of all
jets $\v = (b,v) \in \J$ with the following properties:
\begin{itemize}
\item[{\textrm{(A)}}] The scalar component~$b : M \rightarrow \R$ satisfies the equation
\beq \label{JExScalComp}
\lambda_A \: b(\underline x,0) + \lambda_I\, \big( b(\underline x + e_t,0) + b(\underline x - e_t,0)
\big) = 0 \:.
\eeq
\item[{\textrm{(B)}}] The vector component~$v : M \rightarrow T\F$ consists of a constant
vector~$\underline{v} \in \R^{1,1}$ and a function~$v^\varphi : M \rightarrow \R$, i.e.
\[ v(x) = \big(\underline{v}, v^\varphi(x) \big) \:, \]
where the function~$v^\varphi$ satisfies the discrete wave equation on $\Gamma$,
\beq \label{JExDiscrWaveEq}
\sum_{\underline y \in \Gamma}  f(\underline x - \underline y) \, v^\varphi\!(\underline y, 0) =  0 \, .
\eeq
\end{itemize}
\end{Prp}
\Proof For ease in notation, we identify~$M = \Gamma \times \{0\}$ with the lattice~$\Gamma$.
Our first aim is to show that that a jet~$\v \in \J$ satisfies condition~(l1) in Definition~\ref{Jdeflin}
if and only if it is of the form
\beq \label{Jvform}
\v(\underline{x}) = \big(b(\underline{x}), (\underline{v}, v^\varphi(\underline{x})) \big)
\eeq
with a constant vector~$\underline{v} \in \R^{1,1}$ and mappings~$b, v^\varphi : \Gamma \rightarrow \R$.

Since~(l1) does not pose a condition on the scalar component, it suffices to
consider the vector component~$v=(\underline{v}, v^\varphi)$.
Moreover, using that the Lagrangian is smooth in the variables~$x^\varphi$ and~$y^\varphi$,
it suffices to consider the component~$\underline{v}$.
If~$\underline{v}$ is a constant vector in Minkowski space, its flow does not change the
difference vector~$\underline{x}-\underline{y}$ in the Lagrangian~\eqref{JDefL}.
Therefore, the combination of directional derivatives in~\eqref{Jderex1} exists and vanishes,
implying that the jets of the form~\eqref{Jvform} satisfy the condition~(l1).

The following argument shows that for every jet satisfying~(l1) the component~$\underline{v}$
is indeed constant: Assume conversely that~$\underline{v} : \Gamma \rightarrow \R^{1,1}$ is not
constant. Then there are neighboring points~$\underline{x}, \underline{y} \in \Gamma$
with~$\underline{v}(\underline{x}) \neq \underline{v}(\underline{y})$.
In order for the combination of directional derivatives in~\eqref{Jderex1} to exist,
the function~$\L(F_\tau(x), F_\tau(y))$ must be differentiable in~$\tau$.
This implies that the characteristic functions in~\eqref{JDefL} must be continuous at~$\tau=0$.
We first evaluate this condition if~$\underline{x}$ and~$\underline{y}$ are diagonal neighbors
(i.e.\ $\underline{y} = \underline{x} + (\pm1, \pm1)$). In this case, for the characteristic
function~$\chi_I(F_\tau(x)-F_\tau(y))$ to be continuous, the
vector~$\underline{v}(\underline{x})-\underline{v}(\underline{y})$ must be
collinear to~$\underline{x}-\underline{y}$. But then the continuity of the
characteristic function~$\chi_A((F_\tau(x)-F_\tau(y))$ implies
that~$\underline{v}(\underline{x})=\underline{v}(\underline{y})$. This a contradiction.
We conclude that~$\underline{v}$ is constant on the even and odd sublattices of~$\Gamma$.
In the remaining case that~$\underline{v}$ describes a constant translation of the even sublattice relative to the
odd sublattice, we can choose neighboring lattice points~$\underline{x}, \underline{y} \in \Gamma$
on the odd and even sublattice, respectively, such that the vector~$\underline{v}(\underline{x})
-\underline{v}(\underline{y})$ is non-zero and is not tangential to the discontinuity of the characteristic
function~$\chi_A$. This implies that the function~$\chi_A(F_\tau(x)-F_\tau(y))$ is not continuous at~$\tau=0$,
which is again a contradiction.

We conclude that condition (l1) in Definition~\ref{Jdeflin} is satisfied precisely by all jets of the form~\eqref{Jvform}.
For such jets, the first part of condition (l2) is satisfied because $\L(x,.)$ has bounded support and is smooth in $x^\varphi$. It remains to evaluate the linearized field equations~\eqref{Jeqlinlip}. 
For the constant component $\underline{v} \in \R^{1,1}$,~\eqref{Jderex1} vanishes. Therefore
we can set $\underline{v} = 0$ in the remainder of this proof.
Since all our jets are locally constant,  for $\u = (a,u) \in \Jtest$ and $x, y \in M$ we have
\begin{align*}
\nabla_{\u(x)} \nabla_{1, \v} \L(x,y) &= \Big( a(\underline x) + u^\varphi\!( \underline x) \, 
\frac{\partial}{\partial x^\varphi} \Big) \Big( b(\underline x) + 
v^\varphi\!(\underline x) \, \frac{\partial}{\partial x^\varphi}  \Big) \L(x,y)  \\
&\stackrel{(\ast)}{=} \Big( a(\underline x) \,b(\underline x) 
+ u^\varphi\!( \underline x) \, v^\varphi\!(\underline x) \,\frac{\partial^2}{\partial {x^\varphi}^2}   \Big) \L(x,y) \\[0.2em]
&= a(\underline x) \,b(\underline x) \, \L(x,y) 
- u^\varphi\!( \underline x) \,v^\varphi\!(\underline x) \, f(\underline x - \underline y) \,,
\end{align*}
where in $(\ast)$ we used that the term involving the first derivative vanishes since $V(\varphi)$ is minimal
at $\varphi = 0$.  Similarly,
\begin{align} \label{JExPart}
&\nabla_{1, \u} \nabla_{2, \v} \, \L(x,y) =  a(\underline x) \,b(\underline y) \,\L(x,y) + u^\varphi\!( \underline x) \,v^\varphi\!(\underline y)  \, f(\underline x - \underline y) \:.
\end{align}
Hence, for any~$x \in M$, the linearized field equation~\eqref{Jeqlinlip} can be written as
\beq \label{JExWeakEL}
\begin{split}
\big( \lambda_ A + 2 \lambda_I \big) \,a(\underline{x})\, b(\underline{x}) &+
\lambda_A \,a(\underline{x}) \,b(\underline{x}) + \lambda_I \big( a(\underline{x}) \,b(\underline{x}+e_t)
+ a(\underline{x}) \,b(\underline{x}-e_t) \big) \\
&+  u^\varphi\!( \underline x) \sum_{\underline y \in \Gamma} v^\varphi\!(\underline y)  \, f(\underline x - \underline y)     - a(\underline x) \,b(\underline x) \:\frac{\nu}{2} \:,
\end{split}
\eeq
where we used the fact that~$\sum_{\underline y \in \Gamma} f(\underline x - \underline y) =0$. Evaluating~\eqref{JExWeakEL} for $\u = (a,0) \in \Jtest$ and using~\eqref{Jnuval} gives~\eqref{JExScalComp}.
Similarly, evaluating~\eqref{JExWeakEL} for $\u = (0,u) \in \Jtest$ yields~\eqref{JExDiscrWaveEq}. 
\QED

\subsection{The Symplectic Form}
We now construct the symplectic form. In preparation, we 
need to verify all technical assumptions.
\begin{Lemma} The conditions \emph{(s1)} to \emph{(s3)} on page~\pageref{JCondS1} are satisfied.
\end{Lemma}

\begin{proof}
Condition (s1) follows because jets in $\Jlin \cap \Jtest$ only have a $v^\varphi$-component and $\L(x,y)$
is smooth in $x^\varphi$ and $y^\varphi$.
Condition (s2) follows from definition~\eqref{JDefRho} and the fact that the support of~$\L(x,.)$
is finite, so that the $\rho$-integration reduces to a finite sum.
Condition (s3)  is satisfied because for $\u, \v \in \Jtest \cap \Jlin$, the commutator vanishes
due to our choice of locally constant jets.
\end{proof}

Next, in order to find an explicit expression for the symplectic form~\eqref{JOSIlip}, we choose a constant time slice
\begin{align*}
N_t &= \big\{ \underline x \in \Gamma \,\big|\,  x^0 = t\big\} \quad \text{with} \quad t \in \Z\:.
\end{align*}
We let~$\Omega_{N_t}$ be the past of~$N_t$, i.e.
\[  \Omega_{N_t} = \big\{ \underline x \in \Gamma \,\big|\, x^0 \leq t \big\} \:. \]
\begin{Prp}\label{JExSympPrp}
If we choose~$\Jtest$ as the jets with spacelike compact support,
\[ \Jtest = \big\{ \u \in \Jdiff \, \big| \, \text{$\supp \u|_{N_t} $ is a finite set for all~$t \in \Z$} \big\} \:, \]
the symplectic form~\eqref{JOSIlip} is given by
\begin{align}\label{JExSympl}
\begin{split}
\sigma_{\Omega_{N_t}}(\u, \v) &= \lambda_I  \sum_{\underline x \in N_t} \, \big( a(\underline x) \:b(
\underline{x} + e_t) - a(\underline{x} + e_t) \:b(\underline{x}) \big) \\
& \quad \, \, + \sum_{\underline x \in N_t} \big( \,u^\varphi\!( \underline x+e_t) \,v^\varphi\!(\underline x)   - u^\varphi\!( \underline x) \,v^\varphi\!(\underline{x}+e_t) \big) \:,
\end{split}
\end{align}
where again $e_t := (1,0) \in \R^{1,1}$.
\end{Prp} 
\begin{Remark}\em Note that the second sum~\eqref{JExSympl} is the usual symplectic form associated to
a discrete version of the wave equation on $\R^{1,1}$. Namely, after adding the terms $v^\varphi\!( \underline x) \,u^\varphi\!(\underline x )-v^\varphi\!( \underline x) \,u^\varphi\!(\underline x ) = 0$ 
to the each summand (and similarly for the scalar component)
and taking a suitable limit $e_t \rightarrow 0$, we obtain
\begin{align*}
\sigma_{\Omega_{N_t}}(\u, \v) = \sum_{\underline x \in N_t} \big( \dot u^\varphi\!(\underline x ) \,v^\varphi\!( \underline x)  - u^\varphi\!( \underline x) \,
\dot v^\varphi\!(\underline x ) \big) + \sum_{\underline x \in N_t}  \lambda_I \big( a(\underline x)\, \dot b(\underline x)
- \dot a(\underline x) \,b(\underline x)     \big) \, ,
\end{align*}
where the dot denotes the discrete $t$-derivative.
\QEDrem
\end{Remark}
 
\begin{proof}[Proof of Proposition~\ref{JExSympPrp}]
The proof of Proposition~\ref{JExLin} still goes through if~$\Jtest$ is restricted to jets with 
spatially compact support. Given $\u, \v \in \Jtest \cap \Jlin$, by applying~\eqref{JExPart} and Proposition~\ref{JExLin}, we can compute~\eqref{JOSIIntegrand} to obtain
(for any choice of~$s,s' \in \{\pm\}$)
\begin{align*}
&\sigma_{\u, \v}(x,y) = \nabla_{1,\u} \nabla_{2,\v} \L(x,y) - \nabla_{1,\v} \nabla_{2,\u} \L(x,y) \\
& = a(\underline x) \,b(\underline y) \,\L(x,y) - b(\underline x) \,a(\underline y) \,\L(x,y)
+ u^\varphi\!( \underline x) \,v^\varphi\!(\underline y)  \, f(\underline x - \underline y) - v^\varphi\!( \underline x)
\,u^\varphi\!(\underline y)  \, f(\underline x - \underline y)
\end{align*}
(where we again identified~$M= \Gamma \times \{0\}$ with~$\Gamma$).
Using that
\begin{align*}
 \sigma_{\Omega_{N_t}}(\u, \v) &= \sum_{\underline x \in \Omega_{N_t}} \, \sum_{\underline y \in M \setminus \Omega_{N_t}}  
 \sigma_{\u, \v}\big((\underline x,0) ,(\underline y,0) \big) \:,
\end{align*}
we obtain~\eqref{JExSympl}.
 \end{proof}

\section*{Appendix A: The Fr{\'e}chet Manifold Structure of~$\calB$}
\phantomsection
\addcontentsline{toc}{section}{Appendix A: The Fr{\'e}chet Manifold Structure of~$\calB$}\label{appclosed}

As in Section~\ref{JSecSmooth} we assume that~$\F$ is a smooth manifold of dimension~$m \geq 1$.
We want to endow the set~$\calB$ with the structure of a Fr{\'e}chet manifold.
The first step is to specify the topology on the set~$\calB$ in~\eqref{JBdef}.
We choose the {\em{compact-open topology}} defined as follows.
First, parametrizing the measures according to~\eqref{JrhoFf}
by a pair~$(f, F) \in C^\infty(\F, \R) \times C^\infty(\F, \F)$, we can identify~$\calB$ with a subset of
the space of such pairs,
\[ \calB \subset C^\infty(\F, \R) \times C^\infty(\F, \F) \:. \]
Our task is to endow the sets~$C^\infty(\F, \R)$ and~$C^\infty(\F, \F)$
with the structure of a Fr{\'e}chet manifold. Once this has been accomplished,
the Fr{\'e}chet structure on~$\calB$ can be introduced simply by
assuming that~$\calB$ is a Fr{\'e}chet submanifold of the product manifold~$C^\infty(\F, \R) \times C^\infty(\F, \F)$.

Being a vector space, the space~$C^\infty(\F, \R)$ can be endowed even with the structure
of a Fr{\'e}chet space. To this end, on~$\F$ we choose an at most countable
atlas~$(x_\lambda, U_\lambda)_{\lambda \in \Lambda}$ (with an index set~$\Lambda \subset \N$)
whose charts $x_\lambda : U_\lambda \rightarrow \R^m$
are defined on relative compact subsets~$U_\lambda \subset \F$.
We then consider the Fr{\'e}chet topology induced by the~$C^k$-norms in these charts, i.e.\
\[ \| f \|_{k,\lambda} := \big\| f \circ x_\lambda^{-1} \big\|_{C^k(x_\lambda(U_\lambda))} \:. \]
The resulting topology is metrizable. It is induced for example by the distance function
\[ d(f,g) := \sum_{k=0}^\infty \;\sum_{\lambda \in \Lambda} \:2^{-k-\lambda}\, \arctan  \| f -g \|_{k,\lambda} \:. \]

In order to endow~$C^\infty(\F, \F)$ with the structure of a Fr{\'e}chet manifold,
we work locally in a neighborhood of a point~$F \in C^\infty(\F, \F)$.
First, we refine the previous atlas~$(x_\lambda, U_\lambda)$
in such a way that the domains~$U_\lambda$ are all
convex geodesic neighborhoods with respect to a chosen Riemannian metric~$g$ on~$\F$.
Moreover, we further refine this atlas to a new atlas~$(y_\gamma, V_\gamma)_{\gamma \in \Gamma}$
in such a way that~$F$ maps the domains of the new charts to domains of the old charts, meaning
that for every~$\gamma \in \Gamma$ there is~$\lambda(\gamma) \in \Lambda$ such that
\[ F(V_\gamma) \subset U_{\lambda(\gamma)} \]
(for example, one can choose the domains of the new charts as open subsets of the
sets~$U_\lambda \cap F^{-1}(U_\nu)$ for~$\lambda, \nu \in \Lambda$ and introduce the
charts as the restrictions of~$x_\lambda$ to the new domains).
We restrict attention to mappings~$G$ which are so close to~$F$ that they
map to the same charts, i.e.
\beq \label{JGcond}
G(V_\gamma) \subset U_{\lambda(\gamma)} \qquad \text{for all~$\gamma \in \Gamma$}\:.
\eeq
For such mappings, we can define the $C^k$-norms by
\[ \|G - F\|_{k, \gamma} = \big\| x_{\lambda(\gamma)} \circ G \circ y_\gamma^{-1} 
- x_{\lambda(\gamma)} \circ F \circ y_\gamma^{-1} \big\|_{C^k(y_\gamma(V_\gamma))} \:. \]
The resulting Fr{\'e}chet topology is again metrizable, as becomes obvious for example by setting
\begin{align*}
d(F,G) &= \left\{ \begin{array}{cl} 4 & \text{if~\eqref{JGcond} is violated} \\
\displaystyle 
\sum_{k=0}^\infty \;\sum_{\gamma \in \Gamma} \:2^{-k-\lambda}\, \arctan \| F-G \|_{k,\gamma}
& \text{if~\eqref{JGcond} holds}\:.
\end{array} \right.
\end{align*}
It remains to construct a local chart around~$F$. To this end, it suffices to consider mappings~$G$
which satisfy~\eqref{JGcond}. Then, since the domains of the charts~$U_\lambda$
are all geodesically convex, for any~$x \in \F$ there is a unique vector~$v(x) \in T_{F(x)} \F$
with the property that~$G(x) = \exp_{F(x)} v(x)$. In this way, the mapping~$G$ can be
described uniquely by a vector field~$v \in C^\infty(F(\F), T\F)$ on~$\F$ along~$F(\F)$.
The mapping~$G \rightarrow v$ is the desired chart, taking values in the linear space~$C^\infty(F(\F), T\F)$.

For clarity, we finally explain what the tangent vectors of~$\calB$ are, how these tangent vectors
act on functions, and how these derivatives are related to the derivative~$\nabla_{\u}$
as defined in~\eqref{JDjet}. These elementary facts are also needed for the
computation of the exterior derivative~$d\gamma$ in the proof of Lemma~\ref{Jlemmaclosed}.
Given~$\v =(b,v) \in T_\rho {\mathcal{B}}$, we let~$\tilde{\rho}_\tau$ be a smooth curve in~$\calB$
with~$\tilde{\rho}_\tau |_{\tau = 0}=\rho$ and~$\dot{\rho}_\tau |_{\tau =0} = \v$.
We again write the measures~$\tilde{\rho}_\tau$ in the form~\eqref{Jparamtilrho}
(see Lemma~\ref{Jcurve}), so that~\eqref{Jparamtilrho2} holds.
Then the directional derivative of a smooth function~$\phi$ on~$\calB$ is defined as usual by
\[ \v \phi = \frac{d}{d\tau} \phi\big(\tilde{\rho}_\tau \big) \big|_{\tau=0} \:. \]
In particular, the derivative of~$\gamma(\u)$ as defined in~\eqref{JgammaUdef} is given by
\begin{align*}
\v \gamma(\u)\big|_{\tilde{\rho}} &= \frac{d}{d\tau}  \int_{\Omega_{N_t}} d\rho \int_{M \setminus \Omega_{N_t}}
d\rho\:
f_\tau(x)\: \nabla_{2,\u} \L\big(F_\tau(x), F_\tau(y)\big)\: f_\tau(y) \Big|_{\tau=0} \\
&= \int_{\Omega_{N_t}} d\rho \int_{M \setminus \Omega_{N_t}} d\rho\:  
\big( \nabla_{\v(x)} + \nabla_{\v(y)} \big) \nabla_{2,\u} \L(x, y) \:.
\end{align*}
We point out that here the derivative~$\nabla_{\v(y)}$ also acts on the jet~$\u$
in the derivative~$\nabla_{2,\u}$. The commutator of such products of derivatives
can be computed with the help of the following lemma.

\begin{Lemma}%
For $\u, \v \in T_\rho\calB$, we have
\beq \label{Jcommrel}
\nabla_{[\u, \v]} =  \big[ \nabla_\u, \nabla_\v \big] \:.
\eeq
\end{Lemma}
\Proof
Again denoting~ $\u=(a,u)$ and~$\v=(b,v)$, for any smooth function~$\eta$ on~$\F$ we have
\begin{align}\begin{split}\label{Jcomm}
&\nabla_\u \nabla_\v \eta(x) = \big( a(x) + D_u \big) \big( b(x) + D_v \big) \eta(x) \qquad \textrm{and}\\
&\big[\nabla_\u, \nabla_\v \big] \eta(x) = D_{[u,v]} \eta(x) + (D_u b)(x)\, \eta(x) - (D_v a)(x)\, \eta(x)\:.
\end{split}\end{align}
In order to compute the commutator~$[\u, \v]$, we consider
diffeomorphisms~$\Phi_\tau, \tilde{\Phi}_s :  \calB \rightarrow \calB$
along the vector fields~$\u$ and~$\v$, i.e.
\[ \Phi_0 = \text{id},\quad \partial_\tau \Phi_\tau = \u \circ \Phi_\tau \qquad \text{and} \qquad
\tilde{\Phi}_0 = \text{id},\quad \partial_s \tilde{\Phi}_s = \v \circ \tilde{\Phi}_s\:. \]
Then
\begin{align*}
\int_\F \eta(x)\, d \big(\Phi_\tau \tilde{\Phi}_s \rho \big) (x) &=
\int_\F f_\tau(x)\, \eta\big(F_\tau(x) \big) \,d \big(\tilde{\Phi}_s \rho \big)(x) \\
&= \int_\F \tilde{f}_s(x) \: f_\tau\big(\tilde{F}_s(x) \big)\, \eta\Big(F_\tau \big(\tilde{F}_s(x) \big) \Big) \,d \rho(x) 
\end{align*}
Differentiating with respect to~$s$ and~$\tau$ at~$\tau=s=0$ gives
\begin{align*}
\int_\F &\eta(x)\, d \big(\u \v \rho \big) (x) =
\frac{d^2}{d\tau ds} \int_\F \eta(x)\, d \big(\Phi_\tau \tilde{\Phi}_s \rho \big) (x) \bigg|_{s=\tau=0} \\
&= \frac{d^2}{d\tau ds} \int_\F \tilde{f}_s(x) \: f_\tau\big(\tilde{F}_s(x) \big)\, \eta\Big(F_\tau \big(\tilde{F}_s(x) \big) \Big) \,d \rho(x) \bigg|_{s=\tau=0} \\
&= \int_\F \Big( a(x) \,b(x) + (D_v a)(x) \Big)\, \eta(x)\, d\rho(x) \\
&\quad + \int_\F \Big( b(x)\: (D_u\eta)(x) + a(x)\: (D_v\eta)(x) \Big) d\rho(x) \\
&\quad + \int_\F D_v D_u \eta(x) \,  d\rho(x) \:.
\end{align*}
Likewise, exchanging the two diffeomorphism gives the vector~$\v \u \rho$. Hence
\begin{align}\label{JCommutator}
 \int_\F \eta(x)\, d \big([\u, \v] \rho \big) (x) = \int_\F \Big( D_{[v, u]} \eta
+ (D_v a)\: \eta - (D_u b)\: \eta \Big) d\rho(x) \:,
\end{align}
This shows that 
\[ [\u, \v] = \big( D_u b - D_v a \, , [u, v] \big)\,. \]
Comparing with~\eqref{Jcomm} and the definition of $\nabla_{\u}$ in~\eqref{JDjet} gives~\eqref{Jcommrel}.\medskip

We finally note for clarity that
the minus sign in~\eqref{JCommutator} arises because jets $\u, \v$ act on functions on~$\calB$,
whereas the derivatives $\nabla_\u$ and~$\nabla_\v$
act on functions on~$M$. When rewriting compositions of jets $\u \v$ as compositions of derivatives
on~$M$, the order of the composition is interchanged to~$\nabla_v \nabla_u$. 
\QED

\chapter[Stochastic and Non-Linear Correction Terms]{Stochastic and Non-Linear\\ Correction Terms}\label{DissStoch}

In Chapter~\ref{DissJet}, we have established a formalism which allows to relate causal variational principles,
the core analytic principle in the theory of causal fermion systems,
to fields on space-time. These fields are jets $\u,\v,\w \in \J$ as defined in~\eqref{JJdeflip} and~\eqref{JJM}. 
We have found that the jets satisfy the linearized field equations~\eqref{Jeqlinlip}, which in turn generate an evolution which is invariant w.r.t to an adaptation of a symplectic form
to this setting (cf. Theorem~\ref{JthmOSIlip}). We referred to this evolution as \emph{Hamiltonian time evolution}. The linearized field equations are conjectured to give rise to the fundamental equations of contemporary physics in the continuum limit.

The crucial step in the derivation of the linearized field equations is to test the Euler-Lagrange equations~\eqref{JELstrong} weakly in directions of \emph{differentiable jets} $\Jdiff$. Here, ``differentiable'' means that the directional derivative of the the function $\ell(x)$, as defined in~\eqref{JDjetSemi}, exists (cf.~\eqref{JJDiffDef} and~\eqref{JJDiffLip}). However, since the Lagrangian of the theory of causal fermion systems, given by~\eqref{NLagrange}, is only Lipschitz continuous,
there may in general only be comparably few differentiable directions. This motivates the question of
what happens to the linearized field equations if one tests in non-differentiable directions and which information is contained in the non-differentiable directions.

In this chapter, we show that extending the space of test-jets to include non-differentiable directions leads to the appearance of a stochastic term. Furthermore, we show 
that terms which are non-linear in the fields arise. As explained in Section~\ref{IModelQT}, this supports the conjecture that in the continuum limit, the theory of causal fermion systems gives rise to a dynamical collapse theory. Furthermore, the correction terms might ultimately relate to experimental examination.

This chapter is organized as follows. In Section~\ref{SPrelim}, we give preliminaries for this chapter, fix notation and introduce two technical assumptions. Section~\ref{SNonDiffFieldEq} is devoted to the derivation of the non-differentiable linearized field equations. A new term arises which has a stochastic nature. In Section~\ref{StochTermVan}, we show that a natural assumption can be made which implies that this term
vanishes macroscopically. Section~\ref{SConDiff} contains the connection to the differentiable case of Chapter~\ref{DissJet}. In Section~\ref{SympStoch}
we study the implications for the Hamiltonian time evolution. In Section~\ref{SO2} we prove the full non-differentiable field equations to second order, which include contributions which are non-linear in the fields (Theorem~\ref{SFull2O}). We conclude in Section~\ref{SNoether} by extending a Noether-like theorem of Chapter~\ref{DissNoether} to the setting of this chapter.

\section{Preliminaries}\label{SPrelim}

As in Chapter~\ref{DissJet}, for a given choice of \emph{minimizing} measure $\rho$, we consider families of measures $(\rho_\tau)_{\tau \in (-\delta, \delta)}$ of the form
\beq \label{Srhotau}
\rho_\tau = (F_\tau)_* \big( f_\tau \, \rho \big) \:,
\eeq
where~$f$ and~$F$ are smooth,
\[ f \in C^\infty\big((-\delta, \delta) \times \F \rightarrow \R^+ \big) \qquad \text{and} \qquad
F \in C^\infty\big((-\delta, \delta) \times \F \rightarrow \F \big) \:, \]
with~$f_0(x)=1$ and where $F$ is a flow. (Details about flows are reviewed in Section~\ref{SO2Prelim}. Here,
it suffices to note that $F_0(x) = x$ and that $F_\tau := F(\tau, .)$ is a diffeomorphism onto an open subset of $\F$.)
As before, infinitesimally, the family~\eqref{Srhotau} gives rise to a jet $\v = (\dot f_0, v) \in \J$ where
$\dot f_0 = \frac{d}{d\tau} f_\tau |_{\tau=0}$ and where $v$ is the infinitesimal generator of $F$. We refer to $\v$ as the \emph{generator}
of the family~\eqref{Srhotau} and use dots to indicate also higher $\tau$-derivatives (such as $\ddot f_0$).
For the sake of brevity, we introduce the following abbreviations. \Later{Do we need to assume that $F_\tau$ is a global diffeo for small $\tau$?}
\begin{Notation}\label{Sgeneralnotation}\em
We define the notions of `integrability', `almost everywhere' and `measure zero' to refer to the given measure $\rho$ unless specified otherwise,\\[-1.8em]
\begin{itemize}
\itemsep-.2em
\itemD `integrability' $:=$ $\rho$-integrability,
\itemD `almost everywhere' $:=$ $\rho$-almost everywhere,
\itemD `measure zero' $:=$ measure zero with respect to $\rho$,\\[-1.8em]
\end{itemize}
and use the following shorthand notation for the right and left semi-derivatives,
\begin{align}\begin{split} \label{Snotation0}
\frac{d}{ds}^{\! +}_{|_0} \! h(s) &= \frac{d}{ds}^{\! +} \! h(s) \big|_{s=0} = \lim_{s \searrow 0} \frac{h(s) - h(0)}{s} \\
\frac{d}{ds}^{\! -}_{|_0} \! h(s) &= \frac{d}{ds}^{\! -} \! h(s) \big|_{s=0} = \lim_{s \nearrow 0} \frac{h(s) - h(0)}{s} = - \frac{d}{ds}^{\! +}_{|_0} \! h(-s) \: ,
\end{split}\end{align}
where $h$ is any function whose semi-derivatives exist.
For any $x \in \F$ and $v \in T_x \F$, we define directional semi-derivatives as
\beq \label{SDefDSDeriv}
D^\pm_{v} \: f(x) := \frac{d}{d\tau}^{\! \pm}_{|_0} f \big( \gamma(\tau) \big) \:,
\eeq
where $\gamma \in C^1((-1,1), \F)$ is any curve with~$\gamma(0)=x$ and~$\gamma'(0)=v$,
and we again indicate on which argument of a function $f(x,y)$ a semi-derivative acts by subscripts,
$D^\pm_{i,v} \: f(x,y)$ for $i=1,2$. Furthermore, \emph{semi-derivatives in direction of a jet} $\v=(b,v) \in \J$ are defined as
\beq\label{SDjetSemi}
\nabla^\pm_\v := b(x) + D^\pm_v 
\eeq
and we denote \emph{symmetric (directional) semi-derivatives} as
\beq \label{SsymDD}
\frac{\tilde d}{d s}_{|_0} = \frac{1}{2} \Big( \: \frac{d}{d s}^{\!+}_{|_0} +  \frac{d}{d s}^{\!-}_{|_0} \Big) \quad \textrm{ and } \quad \tilde D_v := \frac{1}{2} \big( D^+_{v} - D^+_{-v} \big) = \frac{1}{2} \big( D^+_{v} + D^-_{v} \big) \: ,
\eeq
as well as
\beq\label{SnotationSymmSem}
\widetilde \nabla_\v := \frac{1}{2} \: \big( \nabla^+_\v - \nabla^+_{-\v} \big) = \frac{1}{2} \: \big( \nabla^+_\v + \nabla^-_\v \big) \: .
\eeq
This notation is compatible with the commutator of two jets $\v,\w \in \J$, defined as
\beq \label{SDefCommutator}
\nabla^+_{[\w,\v]} := \nabla^+_\w \nabla^+_\v - \nabla^+_\v \nabla^+_\w \: ,
\eeq
because
\begin{align*}\begin{split}
&\wnabla_\w \: \wnabla_\v - \wnabla_\v \: \wnabla_\w\\
&\: = \frac{1}{4} \: \big( \nabla^+_\w - \nabla^+_{-\w} \big) \big( \nabla^+_\v - \nabla^+_{-\v} \big) - \frac{1}{4} \: \big( \nabla^+_\v - \nabla^+_{-\v} \big) \big( \nabla^+_\w - \nabla^+_{-\w} \big) \\
& \: = \frac{1}{4} \: \big( \nabla^+_{[\w,\v]}  - \nabla^+_{[\w,-\v]} - \nabla^+_{[-\w,\v]} + \nabla^+_{[-\w,-\v]} \big) \\
& \: = \frac{1}{2} \: \big( \nabla^+_{[\w,\v]}  - \nabla^+_{-[\w,\v]} \big) = \wnabla_{[\w,\v]} \: ,
\end{split}\end{align*}
and analogously for $\tilde D_{[w,v]}$.

In the case of repeated derivatives which act on the same argument, such as $\nabla_{i,\v} \nabla_{i,\w}$,
we \emph{always define the first derivative to act on the second derivative and on the corresponding argument of the function} (the $i$th argument in this case).
E.g., with respect to the commutator~\eqref{SDefCommutator} of two jets $\v=(b,v), \ \w =(c,w) \in \J$,
this implies
\begin{align*}\begin{split}
\nabla^+_{[\w,\v]} &= \big( c(x) + D^+_w \big) \: \big( b(x) +  D^+_v \big) - \big( b(x) +  D^+_v \big) \: \big( c(x) +  D^+_w \big) \\
&= (D_w b)  - (D_v c)  +  D^+_{[w,v]} \: ,
\end{split}\end{align*}
or written in terms of the notation~\eqref{SnotationSymmSem},
\begin{align}\begin{split}\label{SnablaCommutator}
\wnabla_{[\w,\v]} &= \big( c(x) + \tilde D_w \big) \: \big( b(x) + \tilde D_v \big) - \big( b(x) + \tilde D_v \big) \: \big( c(x) + \tilde D_w \big) \\
&= (D_w b)  - (D_v c)  + \tilde D_{[w,v]} \: .
\end{split}\end{align}
We do not use the notation~\eqref{JnotationJetsDiff} in this chapter.

Finally, in the context of the family~\eqref{Srhotau}, for every $\Omega \in \Sigma(\F)$, where $\Sigma(\F)$ denotes the Borel sigma algebra of $\F$, we use the notation
\[
\Omega_\tau := F_\tau(\Omega) \: .
\]

Throughout this chapter, we use $\J$ as defined in~\eqref{JJdeflip}. I.e., we use jets which
are defined on $\F$ rather than on $M$ (where $M$ is defined in~\eqref{JMlip}).
\QEDrem \end{Notation}

We work in the non-compact setting introduced in Section~\ref{Jsecnoncompact}. In particular, we make (and refer to) Assumptions~(i) and~(ii) on page~\pageref{JCond1}, as well as Assumptions~(iii) and~(iv). Assumptions~(v) and~(vi) on page~\pageref{JCond5} have to be strengthened in the present context to ensure well-definedness. One of the main differences to Chapter~\ref{DissJet} is that in the latter, second semi-derivatives of $\L(x,y)$ are only required to exist in directions of differentiable jets (cf. Assumption~(s1) on page~\pageref{JCondS1}), whereas here we demand that the second semi-derivatives exist for all jets.
\Later{Sagen, dass das Lipschitz-Kontinutität erzwingt.}

To state the following assumptions, we define a function $h(x,y)$ to be\\[-1.3em]
\begin{center}
integrable in $x$ and$_{\textrm{c}}$/or $y$ 
\end{center}
if it is integrable in $x$ or $y$ over every $\Omega \in \Sigma(\F)$ and integrable in $x$ and $y$ for every
$\Omega \times \Omega' \in \Sigma(\F \times \F)$ whenever either $\Omega$ or $\Omega'$ is compact.
(Note that integrability always refers to the chosen minimizer $\rho$, cf. Notation~\ref{Sgeneralnotation}. The subscript `c'
stands for ``compact'', indicating that one of two integrations needs to have compact domain.)

\begin{Assumption}\label{SA12}\em For any $x,y \in \F$, $v \in \Gamma(T\F)$, $ F,G \in C^\infty\big((-\delta, \delta) \times \F \rightarrow \F \big)$, $f,g \in C^\infty\big((-\delta, \delta) \times \F \rightarrow \R^+ \big)$, $\tau \in (-\delta,\delta)$ and $
h_\tau(x,y) :=  f_\tau(x) \: \L(F_\tau(x),G_\tau(y)) \: g_\tau(y) $
we assume that \\[-1.5em]
\begin{enumerate}
\item[\textrm{(v)}] $h_\tau(x,y)$ is integrable in $x$ and$_{\textrm{c}}$/or $y$ and the semi-derivatives 
\beq\label{SCond8Eq}
\frac{d}{d\tau}^{\!+}_{|_0} h_\tau(x,y)  \quad \textrm{ and } \quad D^+_{1,v} \frac{d}{d\tau}^{\!+}_{|_0} h_\tau(x,y)
\eeq
as well as
\beq\label{SCond8Eq2}
\frac{d^{\,2}}{d\tau^2}^{\!+}_{|_0} h_\tau(x,y)  \quad \textrm{ and } \quad D^+_{1,v} \frac{d^{\,2}}{d\tau^2}^{\!+}_{|_0} h_\tau(x,y)
\eeq
exist and are integrable in $x$ and$_{\textrm{c}}$/or $y$ as well. Furthermore, we assume that $\frac{d}{d\tau}^{\!+}_{|_{\tilde \tau}} \L(F_\tau(x),G_\tau(y)) = \frac{d}{d\tau}^{\!+}_{|_{\tilde \tau}} \L(F_\tau(x),G_{\tilde \tau}(y)) + \frac{d}{d\tau}^{\!+}_{|_{\tilde \tau}} \L(F_{\tilde \tau}(x),G_\tau(y))$ for all $\tilde \tau \in (-\delta,\delta)$.
\item[\textrm{(vi)}] the semi-derivatives in~\eqref{SCond8Eq} and~\eqref{SCond8Eq2} can be exchanged with integration over $x$ and$_{\textrm{c}}$/or $y$.
\Later{In dem Abschnitt mit dem Jet-Beispiel Annahmen geeignet abschwächen. In diesem Abschnitt auch die Annahmen am Beispiel erläutern.}%
\end{enumerate} 
\end{Assumption}

We give some direct consequences of the assumptions. First, note that~\eqref{SCond8Eq}
implies that for any $v \in T_x\F$, $D^+_{1,v} \: \L(x,y)$ and $D^+_{2,v} \: \L(x,y)$ exist in $\R$.
Furthermore, choosing $f=1=g$, $G(x) = x$ and $F = \Phi$, where $\Phi$ is the flow of $v$, as well as $\Omega = M$,
Assumption~(vi) implies that
\[
\int_M D^+_{1,v} \L(x,y) \: d\rho(y) =  D^+_{1,v} \int_M \L(x,y) \: d\rho(y) =  D^+_{1,v} \ell(x)
\]
for any $x \in \F$, where $\ell$ is defined in~\eqref{Jelldeflip}. Therefore, Assumptions~(v) and~(vi) of Chapter~\ref{DissJet} (page~\pageref{JCond5}) are satisfied. In the following, by~(v) and~(vi) we always refer to Assumptions~\ref{SA12}. For completeness, we mention that Assumptions~\ref{SA12} also imply the first condition of Assumption~(s1)
as well as Assumption~(s2) on page~\pageref{JCondS1}.

\noindent Finally, note that since $\rho$ is a minimizer (cf.~\eqref{Srhotau}), we have 
\beq\label{SELgk}
D^+_v \ell(x) \geq 0 \qquad \textrm{ but } \qquad D^-_v \ell(x) \leq 0
\eeq
for any $x \in M$ and $v \in T_x\F$ (a consequence of the notation~\eqref{Snotation0}), and equally so for
$\nabla^+_\u\ell(x)$ and $\nabla^-_\u\ell(x)$ with $\u \in \J$ (cf.~\eqref{JELweaklip}).

\Later{Remark einfügen, mit einem "Setting" in dem die Annahmen erfüllt sind, also, z.B. L und Ableitung und zweite Ableitung dominiert, jets bounded, etc.}

\section{Non-Differentiable Linearized Field Equations}\label{SNonDiffFieldEq}

In this section, we derive field equations for the non-differentiable case. The next proposition gives the result. Note that $\ell_\tau$ is defined in~\eqref{Jelltau}.

\begin{Prp}\label{SmodFieldEqA}
Let $\v \in \J$ be the generator of a family~\eqref{Srhotau} of minimzers and $\w \in \J$ arbitrary. Then
\beq \label{SlinFEqA}
\frac{1}{2} \: \big( \nabla^+_{\w} - \nabla^+_{- \w} \big)  \: \int_M d\rho(y) \: 
\big( \nabla^+_{1,\v} + \nabla^+_{2,\v} \big) \: \L \big( x, y \big)  - \nabla^+_\w \nabla^+_\v \: \frac{\nu}{2}\: = \chi_{\w,\v}(x)
\eeq
for every $x \in M$, with \Later{Man kann~\eqref{SchiA} direkt aus~\eqref{Sinhom} durch anwenden der Symmetrie des Lagr. bekommen und umgekehrt.}
\beq\label{SchiA}
\chi_{\w,\v}(x) =   \frac{1}{2} \: \frac{d}{d s}^{\!+}_{|_0}  \frac{d}{d \tau}^{\!+}_{|_0}  \: 
f_\tau(x) \: \Big( \ell_\tau \big (F_\tau(\Phi_s (x)) \big)   -  \ell_\tau \big( F_\tau(\Phi_{-s}(x))\Big) \: ,
\eeq
where $\Phi$ is the flow of the vectorial component $w$ of $\w$.
\end{Prp}

\noindent We refer to~\eqref{SlinFEqA} as the \textbf{non-differentiable linearized field equations} and denote the set of all $\v$ which satisfy~\eqref{SlinFEqA}
for all~$\w \in \J$ by $\Jfield$. (``Non-differentiable'' indicates that the assumption of $\w$ being a differentiable jet is not necessary.) We give an interpretation of this result after its proof.

\Proof Choose $\v = (b,v)$, $\w=(c,w)$, $f$ and $F$ as in~\eqref{Srhotau} and $\Phi$ as in the proposition. For every $x \in \F$, $\tau \in (-\delta,\delta)$ and $s \in (-s_{\max},s_{\max})$ we have  by definition of~$\ell_\tau$
\begin{align}\begin{split}\label{SchiA1}
& \int_{M_\tau} d\rho_\tau(y) \: f_\tau(x)
\: \Big( \L \big( F_\tau(\Phi_s(x)),y \big) - \L \big( F_\tau(\Phi_{-s}(x)), y \big) \Big) \\
&\qquad =  f_\tau(x) \: \Big(  \ell_\tau \big(F_\tau(\Phi_s(x)) \big) + \frac{\nu}{2}   -  \ell_\tau \big(F_\tau(\Phi_{-s}(x))\big) - \frac{\nu}{2} \Big) \\
&\qquad =  f_\tau(x) \: \Big(  \ell_\tau \big(F_\tau(\Phi_s(x)) \big)   -  \ell_\tau \big( F_\tau(\Phi_{-s}(x)) \big) \Big) \:,
\end{split}\end{align}
where the existence of the terms follows from Assumption~(v).
By definition of the directional semi-derivatives~\eqref{SDefDSDeriv},
\begin{align*}
\frac{d}{d \tau}^{\!+}_{|_0} \L \big(F_\tau(\Phi_s (x)), y \big) = D^+_{1,v|_{\Phi_s(x)}} \L(\Phi_s(x) , y) \: ,
\end{align*}
where we have added the subscript $\Phi_s(x) $ to $v$ merely to highlight that the derivative acts at the point $\Phi_s(x) \in \F$,
\[
D^+_{1,v} \L(\Phi_s(x) , y) = D^+_{1,v|_{\Phi_s(x)}} \L(\Phi_s(x) , y) \:.
\]
Therefore,
\begin{align}\begin{split}\label{SchiA5}
& \frac{d}{d s}^{\!+}_{|_0} \frac{d}{d \tau}^{\!+}_{|_0} \int_M  f_\tau(x) \:  \L \big(F_\tau(\Phi_s (x)),  F_\tau(y) \big)  \: f_\tau(y) \: d\rho(y) \\
& = \frac{d}{d s}^{\!+}_{|_0} \int_M \frac{d}{d \tau}^{\!+}_{|_0}  f_\tau(x) \:  \L \big(F_\tau(\Phi_s (x)),  F_\tau(y) \big)  \: f_\tau(y) \: d\rho(y) \\
& = \frac{d}{d s}^{\!+}_{|_0} \int_M \big( b(x) + D^+_{1,v|_{\Phi_s(x)}} + \nabla^+_{2,\v} \big) \: \L \big( \Phi_s (x), y \big) \: d\rho(y) \\
&= D^+_w \int_M  \big( b(x) + D^+_{1,v} + \nabla^+_{2,\v} \big) \: \L \big( x, y \big)  \: d\rho(y) - (D_w) b(x) \int_M \L (x,y) \: d\rho(y) \\
&= D^+_w \int_M \big( \nabla^+_{1,\v} + \nabla^+_{2,\v} \big) \: \L \big( x, y \big) \: d\rho(y) - (D_w) b(x) \int_M \L (x,y)  \: d\rho(y) \:,
\end{split}\end{align}
where we have exchanged the differentiation with integration using Assumption~(vi) and where the existence follows again from Assumption~(v).
For $x \in M$, this implies that
\begin{align*}
&\frac{d}{d s}^{\!+}_{|_0}  \frac{d}{d \tau}^{\!+}_{|_0}  \int_{M_\tau} d\rho_\tau(y)  \: f_\tau(x) \: \Big( \L \big(F_\tau(\Phi_s (x)),y \big) - \L \big(F_\tau(\Phi_{-s} (x)),y \big)  \Big) \\
&\ \ = \frac{d}{d s}^{\!+}_{|_0}  \frac{d}{d \tau}^{\!+}_{|_0}    \int_{M} d\rho(y) \: f_\tau(x) \: \Big( \L \big(F_\tau(\Phi_s (x)), F_\tau(y) \big) - \L \big(F_\tau(\Phi_{-s} (x)), F_\tau(y) \big) \Big) \: f_\tau(y) \\
&\stackrel{\eqref{SchiA5}}{=} \big( D^+_w - D^+_{-w} \big) \int_{M} d\rho(y) \: \big( \nabla^+_{1,\v} + \nabla^+_{2,\v} \big) \: \L \big( x, y \big) - 2 \: (D_w b)(x) \int_{M} d\rho(y)  \: \L \big(x, y \big) \\
&\ \ = \big( D^+_w - D^+_{-w} \big) \int_{M} d\rho(y) \: \big( \nabla^+_{1,\v} + \nabla^+_{2,\v} \big) \: \L \big( x, y \big)  -  2\: (D_w b)(x) \: \frac{\nu}{2} \:,
\end{align*}
where we have used that $-D^+_{-w} b=D^-_w b = D_wb$ since $b$ is differentiable, and where in the last step we have used the definition of $\ell$ in~\eqref{Jelldef} and $\ell|_M \equiv 0$ (which follows from $\rho$ being a minimizer).
Therefore,~\eqref{SchiA1} gives
\begin{align}\begin{split} \label{SchiA2}
&\frac{1}{2} \: \big( D^+_w  - D^+_{-w} \big) \int_{M} d\rho(y) \: \big( \nabla^+_{1,\v} + \nabla^+_{2,\v} \big) \: \L \big( x, y \big) 
-  (D_w b)(x) \: \frac{\nu}{2} \\
&\qquad =  \frac{1}{2} \: \frac{d}{d s}^{\!+}_{|_0}  \frac{d}{d \tau}^{\!+}_{|_0}  \: 
f_\tau(x) \: \Big( \ell_\tau \big (F_\tau(\Phi_s (x)) \big)   -  \ell_\tau \big( F_\tau(\Phi_{-s}(x))\Big) \: ,
\end{split}\end{align}
where the extistence of the right hand side follows from the existence of the left hand side.
Furthermore, we have
\begin{align}\begin{split} \label{SchiA3}
& c(x) \int_M d\rho(y) \: \big( \nabla^+_{1,\v} + \nabla^+_{2,\v} \big) \: \L (x, y) \\
& \qquad = c(x) \int_M d\rho(y) \: \frac{d}{d\tau}^{\!+}_{|_0} f_\tau(x) \: \L\big( F_\tau(x), F_\tau(y) \big) \: f_\tau(y) \\
&\qquad = c(x) \: \frac{d}{d\tau}^{\!+}_{|_0}   f_\tau(x) \: \Big( \ell_\tau \big(F_\tau(x) \big) + \frac{\nu}{2} \Big) \: ,
\end{split} \end{align}
where existence follows again from Assumption~(v). In the second step we have used Assumption~(vi) and
the definition of $\ell_\tau$ in~\eqref{Jelltau}.
Since $\v$ generates a family of minimizers, for every $x \in M$, we have \Later{Term "generates a family" formal einführen.}
\[ \ell_\tau(F_\tau(x)) \equiv 0 \qquad \text{ on } (-\delta,\delta) \:, \]
which implies
\[ f_\tau(x) \: \ell_\tau(F_\tau(x)) \equiv 0 \qquad \text{ on } (-\delta,\delta)\] 
and therefore 
\beq\label{SchiA7}
\frac{d}{d\tau}^{\!+}_{|_0}  f_\tau(x) \: \ell_\tau(F_\tau(x)) = 0 \: .
\eeq
Hence~\eqref{SchiA3} gives
\beq \label{SchiA4}
c(x) \int_M d\rho(y) \: \big( \nabla^+_{1,\v} + \nabla^+_{2,\v} \big) \: \L (x, y) - c(x) \: b(x) \: \frac{\nu}{2} = 0 \: 
\eeq
for every $x \in M$, where we have used that $b = \dot f_0$ due to our designation of $\v$ as $(b,v)$.
Adding~\eqref{SchiA4} and~\eqref{SchiA2} gives the result.
\QED

To interpret Proposition~\ref{SmodFieldEqA}, note that the left hand side of~\eqref{SlinFEqA} is equal to the left hand side of the linearized field equations~\eqref{Jeqlinlip}
if $\int_M d\rho(y) \big( \nabla_{1,\v} + \nabla_{2,\v} \big) \L \big( x, y \big)$ exists and is differentiable.
Therefore,~\eqref{SlinFEqA} constitutes a generalization of the linearized field equations~\eqref{Jeqlinlip} to include testing with
jets $\w \in \J$ which need not be elements of $\Jdiff$. We investigate the precise relation to the differentiable case in Section~\ref{SConDiff} below.

This brings us to the interpretation of the term $\chi_{\w,\v}(x)$. Expanding it as
\begin{align}\begin{split}\label{SdiffCase}
\chi_{\w,\v}(x) & =\frac{1}{2} \: b(x) \: \frac{d}{d s}^{\!+}_{|_0}  \Big( \ell \big (\Phi_s (x)\big)   -  \ell \big(\Phi_{-s}(x)\big) \Big) \\
& + \frac{1}{2} \: \frac{d}{d s}^{\!+}_{|_0}  \frac{d}{d \tau}^{\!+}_{|_0}   \Big( \ell_\tau \big (F_\tau(\Phi_s (x))\big)   -  \ell_\tau \big(F_\tau(\Phi_{-s}(x))\big) \Big) \:,
\end{split}\end{align}
we see that $\chi_{\w,\v}(x)$ is related to the appearance of symmetric directional semi-derivatives $\tilde D_w\ell$ as defined in~\eqref{SsymDD}.
Due to~\eqref{SELgk}, for $x \in M$, $\tilde D_w\ell(x)$ in general has an arbitrary sign, depending on the exact numerical values of $D^+_w\ell(x)$ and $D^-_w\ell(x)$ at the space-time point $x$. Since $\ell$ is in general only Lipschitz-continuous, $D^+_w\ell(x)$ and $D^-_w\ell(x)$ need not even be continuous in $x$. Therefore, we are led to interpret $\chi_{\w,\v}(x)$ as a \emph{stochastic term}, fluctuating with varying sign as $x$ changes. 
(The connection to random variables, which is implicit in this name, is intentional: If the microscopic structure of space-time, and hence the point-wise behaviour of $\ell$, is not known, the the term effectively appears as a random contribution and can be modelled as a random variable.) Due to this interpretation, in the following, we refer to $\chi_{\w,\v}(x)$ as ``the stochastic term''.

The mathematical side of this interpretation is supported by the next example. In the following sections, we study the behaviour of the stochastic term in more detail.

\begin{Example}\label{SExStoch}\em  Due to the EL equations~\eqref{JELstrong}, for fixed $x \in M$, the function $\ell \big(\Phi_s(x)\big)$ has a cusp-like behaviour as $s$ varies (cf. also~\eqref{SELgk}).
In a first approximation, this behaviour is similar to the behaviour of a function $l: \R \rightarrow \R$ given by
\begin{align*} 
l(x) = \begin{cases} 
\ \ \alpha \: x & x \geq 0\\
- \beta \: x & x < 0 \:,
\end{cases}
\end{align*}
where $\alpha, \beta \in \R$, $\alpha, \beta >0$ and $x \in \R$. We use this function to illustrate the behaviour of the stochastic term~\eqref{SchiA}.
Denote by $w$ a unit-vector field on $\R$
so that we have $D^+_w f = \partial^+f$ for any $f: \R \rightarrow \R$, where $\partial^+$ denotes the right semi-derivative in~\eqref{Snotation0}.
This gives
\begin{align*}
D^+_w l (x) =  \begin{cases} 
\ \ \alpha  & x > 0\\
\ \ \alpha & x = 0\\
- \beta & x < 0 \:,
\end{cases} \qquad 
D^+_{-w} l (x) =  \begin{cases} 
 -\alpha  & x > 0\\
\ \ \beta & x = 0\\
\ \ \beta & x < 0 
\end{cases}
\end{align*}
and
\begin{align*}
D^-_{w} l (x) =  \begin{cases} 
\ \ \alpha  & x > 0\\
-  \beta & x = 0\\
- \beta & x < 0 \:.
\end{cases}
\end{align*}
This shows that 
\[ \big(D^+_w - D^+_{-w} \big) l(0) = \alpha - \beta \:, \]
which changes sign depending on the particular values of $\alpha$ and $\beta$, indicating that the first line of~\eqref{SdiffCase} has varying sign depending on the behaviour of the function $\ell \big (\Phi_s (x)\big)$ at the space-time point $x \in M$.

We can include a variation~\eqref{Srhotau} of minimizers in this model by promoting $\alpha$ and $\beta$ to functions of a parameter $\tau$,
\begin{align*} 
l_\tau(x) = \begin{cases} 
\ \ \alpha_\tau \: x & x \geq 0\\
- \beta_\tau \: x & x < 0 \:.
\end{cases}
\end{align*}
Assuming that the semi-derivatives $\dot \alpha^+_0 := \frac{d}{d\tau}^{\!+}_{|_0} \alpha_\tau$ and $\dot \beta^+_0 := \frac{d}{d\tau}^{\!+}_{|_0} \beta_\tau$ exist, we have
\begin{align*} 
\frac{d}{d\tau}^{\!+}_{|_0} l_\tau(x) = \begin{cases} 
\ \ \dot \alpha^+_0 \: x & x \geq 0\\
- \dot \beta^+_0 \: x & x < 0 \:,
\end{cases}
\end{align*}
giving
\[ \big(D^+_w - D^+_{-w}\big) \frac{d}{d\tau}^{\!+}_{|_0} l_\tau(0)  = \dot \alpha^+_0  - \dot \beta^+_0 \:. \]
In general, based on the theory of causal fermion systems, there is no reason to postulate a particular sign of $\dot \alpha^+_0$ or $\dot \beta^+_0$.
Hence also the second line in~\eqref{SdiffCase} does not have a fixed sign. (We come back to this point at the end of
Section~\ref{StochTermVan}.)
\end{Example}

\section{Stochastic Term Vanishes Macroscopically}\label{StochTermVan}

In the last section, we have seen that a stochastic term appears in the non-differentiable linearized field equations.
In this section, we show that a typical physical assumption implies that the stochastic term vanishes when averaged over macroscopic regions. Even so this assumption is not
enforced by the first principles of the theory, it seems to be supported a posteriori by the fact that most contemporary experiments seem to be aligned with the predictions of contemporary physics within experimental errors.

Before giving the assumption, we define a minimal notion of ``macroscopicity'' in the context of causal fermion systems. The idea is that for a given physical situation which is described by the theory, the data required by the following definition is given.

\begin{Def}\label{SDefMac}\textbf{\em(Macroscopicity)} We assume that the macroscopic regions of a minimizer $\rho$ of the causal variational principle  are
described by a subset $\mathscr M \subset \Sigma(\F)$.
\Chd{$\Sigma(\F)$ ist nach Gleichung~\eqref{SnablaCommutator} definiert.}%
Furthermore we assume that for a family~\eqref{Srhotau} of minimizers, the macroscopic regions
\chd{are invariant with respect to the diffeomorphisms $F_\tau$.}
(I.e., for every $\tilde \Omega \in \mathscr M$  and every $\tau \in (-\delta,\delta)$, we have
$ F_\tau( \tilde \Omega) \in \mathscr M_\tau$, where $\mathscr M_\tau \subset \Sigma(\F)$ denotes the macroscopic region associated with $\rho_\tau$.)
\Later{Was zu Diffeo onto subset sagen?}%
\Later{Brauchen wir $M \setminus \tilde \Omega \in \mathscr M$ for all $\tilde \Omega \in \mathscr M$?}%
\end{Def}

The next definition captures the physical intuition that when averaged over a macroscopic region, fluctuations form the different space-time points might cancel each other, or ``average out''. Here, this means that the 
sign-varying terms $\frac{1}{2} (D^+_w \ell  - D^+_{-w} \ell)(x)$ at different space-time points $x$ cancel each other's contribution 
to an integration over a macroscopic region $\tilde \Omega$.

\begin{Def} \label{SmacSymm0} \textbf{\em(Symmetric derivatives vanish macroscopically)}
We assume that for every minimizer~$\rho$ of the causal variational principle, every $\Phi \in C^\infty(\F \times (-\delta,\delta) \rightarrow \F)$ and every compact macroscopic
region $\tilde \Omega \in \mathscr M$,
\beq\label{SmacDiffMin0}
\frac{1}{2} \: \frac{d}{d s}^{\!+}_{|_0} \int_{\tilde \Omega} \Big( \ell \big(\Phi_s(x)\big) - \ell \big(\Phi_{-s}(x)\big)  \Big) \, d\rho(x) = 0 \:. 
\eeq
\end{Def}

Note that the restriction to compact $\tilde \Omega$ arises because Assumption~(v) guarantees existence of the left hand side
of~\eqref{SmacDiffMin0} only if $\tilde \Omega$ is compact.
\Later{ Was zu nicht-kompakten makroskopischen Regionen sagen. Exhaustion by compact sets und ggf. nutzen, dass Borel-Maß regulär ist.}

\begin{Remark}\em \Remt{(Symmetric derivatives vanish approximately)}
One might wonder why the symmetric derivatives should vanish exactly, as compared to vanishing approximately, when integrated over a macroscopic region.
Indeed, this is not crucial, the right hand side of~\eqref{SmacDiffMin0} could be replaced by~$\varepsilon$ for any~$\varepsilon \in \R^+_0$, 
or by $\rho(\tilde \Omega)^{-1}$, thus requiring 
that the symmetric derivatives only vanish approximately. The following proposition would still hold, with the right hand side of~\eqref{SchiOmega} replaced by~$\varepsilon$
or~$\rho(\tilde \Omega)^{-1}$, respectively.
\Later{Wie ist das bei den anderen Resultaten in den weiteren Sections?}%
\Later{Sagen dass $\rho(\tilde \Omega)^{-1}$ endlich ist, weil $\tilde \Omega$ kompakt ist und wir mit lokal endlichen Maßen arbeiten (Annahme (iii) oder so). Dann folgt schon dass jedes kompakte Gebiet ein endliches Maß hat, nach Definition vn `kompakt'.}%
\QEDrem
\end{Remark}

\begin{Remark}\em \Remt{(Relation to macroscopic differentiability)}
We note that even though~\eqref{SmacDiffMin0} implies that the derivative
$\frac{d}{d s}_{|_0} \int_{\tilde \Omega} \: ( \ell (\Phi_s(x)) - \ell  (\Phi_{-s}(x)) ) \: d\rho(x) $
exists and vanishes, this does not imply that the individual terms 
$ \frac{d}{d s}_{|_0} \int_{\tilde \Omega} \:  \ell (\Phi_s(x)) \: d\rho(x)$
exist. Therefore, Definition~\ref{SmacSymm0} is considerably weaker than the requirement of minimizers to be macroscopically differentiable
in the sense that $ \frac{d}{d s}_{|_0} \int_{\tilde \Omega} \:  \ell (\Phi_s(x)) \: d\rho(x)$ exists.
\QEDrem
\end{Remark}

The next proposition shows that the assumption of Definition~\ref{SmacSymm0} is sufficient to conclude that the stochastic term vanishes macroscopically. 

\begin{Prp}\label{SstochTermVanishes} Assume that symmetric derivatives vanish macroscopically and let $\v \in \J$ be the generator
of a family of minimizers~\eqref{Srhotau}. Then for every $\w \in \J$ and every compact $\tilde \Omega \in \mathscr M$, 
\beq\label{SchiOmega}
\int_{\tilde \Omega}  \chi_{\w,\v}(x) \: d\rho(x)  = 0 \: .
\eeq
\end{Prp}

\Proof For every $u \in \Gamma(T\F)$, Definition~\ref{SmacSymm0} and Assumption~(vi) imply
\beq\label{SmacDiffMin0infini}
\frac{1}{2} \int_{\tilde \Omega} \big( D^+_u - D^+_{-u} \big) \ell(x) \: d\rho(x) = 0 \:. 
\eeq
Denote $\w=(c,w)$ and let $\Phi$ be the flow of $w$. Since $(\rho_\tau)_{\tau \in (-\delta, \delta)}$, as defined
in~\eqref{Srhotau}, is a family of minimizers by assumption and since $\tilde \Omega$ compact implies that $\tilde \Omega_\tau$ is compact,~\eqref{SmacDiffMin0}, evaluated for $\rho_\tau$, reads
\beq \label{SmacSymmFam2}
\frac{1}{2} \: \frac{d}{d s}^{\!+}_{|_0} \int_{\tilde \Omega_\tau} \Big( \ell_\tau \big(\Phi_s(x)\big) - \ell_\tau \big(\Phi_{-s}(x)\big)  \Big) \, d\rho_\tau(x) = 0 \qquad \text{for all }\tau \in (-\delta,\delta) \: ,
\eeq
which in turn implies
\beq \label{SmacSymmFam3}
\frac{1}{2} \: \frac{d}{d \tau}^{\!+}_{|_0} \frac{d}{d s}^{\!+}_{|_0} \int_{\tilde \Omega_\tau} \Big( \ell_\tau \big(\Phi_s(x)\big)
 - \ell_\tau \big(\Phi_{-s}(x)\big)  \Big) \, d\rho_\tau(x) = 0 \: .
\eeq
The existence of this integral is guaranteed by Assumption~(v). We have
\begin{align}\begin{split}\label{SympMac}
 &  \: \frac{d}{d \tau}^{\!+}_{|_0} \frac{d}{d s}^+_{|_0} \int_{\tilde \Omega} f_\tau(x) \:  \ell_\tau \big(\Phi_s(F_\tau(x))\big) \: d\rho(x) \\
& =  \frac{d}{d \tau}^{\!+}_{|_0}  \int_{\tilde \Omega} f_\tau(x) \: D^+_{w|_{F_\tau(x)}}  \ell_\tau \big( F_\tau(x) \big) \: d\rho(x) \\
&= \int_{\tilde \Omega} \big( b(x) \:  D^{+}_w +  D^{+}_v  D^{+}_w \big) \ell(x) \: d\rho(x) + \int_{\tilde \Omega} \frac{d}{d\tau}^{\!+}_{|_0}  D^{+}_w \ell_\tau(x) \: d\rho(x)\\
& =  \int_{\tilde \Omega} \big( b(x) \: D^{+}_w +   D^{+}_w  D^{+}_v \big) \ell(x) \: d\rho(x) \\
& \quad  + \int_{\tilde \Omega} \frac{d}{d\tau}^{\!+}_{|_0}  D^{+}_w \ell_\tau(x) \: d\rho(x) + \int_{\tilde \Omega} D^{+}_{[v,w]}\ell(x) \: d\rho(x) \\
& =  \frac{d}{d s}^{\!+}_{|_0}  \int_{\tilde \Omega} \big( b(x) \: \ell(\Phi_s(x)) +  D^{+}_{v|_{\Phi_s(x)}} \ell(\Phi_s(x) \big) \: d\rho(x)  \\
& \quad  + \int_{\tilde \Omega} \frac{d}{d\tau}^{\!+}_{|_0} \frac{d}{ds}^{\!+}_{|_0}  \ell_\tau(\Phi_s(x)) \: d\rho(x) + \int_{\tilde \Omega} D^+_{[v,w]}\ell(x) \: d\rho(x) \\
& = \frac{d}{d s}^{\!+}_{|_0}  \frac{d}{d \tau}^{\!+}_{|_0} \int_{\tilde \Omega} 
f_\tau(x) \:  \ell_\tau \big (F_\tau(\Phi_s (x)) \big) \: d\rho(x) + \int_{\tilde \Omega} D^+_{[v,w]}\ell(x) \: d\rho(x) \: ,
\end{split}\end{align}
where we have used Assumption~(vi) to exchange differentiation and integration, and where in the last step we have used
\beq \label{SympMac2}
 \frac{d}{d\tau}^{\!+}_{|_0} \frac{d}{ds}^{\!+}_{|_0}  \ell_\tau(\Phi_s(x))  = \frac{d}{ds}^{\!+}_{|_0} \frac{d}{d\tau}^{\!+}_{|_0} \ell_\tau(\Phi_s(x)) \: ,
\eeq
which holds because the derivatives act on different variables (apparent upon expanding $\ell_\tau$).
Therefore,~\eqref{SmacSymmFam3} implies
\[
\int_{\tilde \Omega} \chi_{\w,\v}(x) \: d\rho(x) = \frac{1}{2} \: \int_{\tilde \Omega} \big( D^+_{[w,v]} - D^+_{-[w,v]} \big) \ell(x) \: d\rho(x) \: .
\]
Using~\eqref{SmacDiffMin0infini} for $u = [w,v]$ gives the result.
\QED

We conclude this section with the remark that Proposition~\ref{SstochTermVanishes} supports the interpretation of $\chi_{\w,\v}$ as a stochastic term.
To this end, recall that Example~\ref{SExStoch} shows that the first line in~\eqref{SdiffCase} has varying sign, but only implies that the second line
in~\eqref{SdiffCase} does not have a fixed sign (since $\dot \alpha^+_0$ and $\dot \beta^+_0$ do not both have the same sign in general).
The previous proposition shows that if the symmetric derivatives~\eqref{SmacDiffMin0} vanish macroscopically, also
the second line in \eqref{SdiffCase} vanishes macroscopically. This implies that it also needs to have varying sign in general.

Furthermore, we remark that the assumption in Definition~\ref{SmacSymm0} is also supported by Proposition~\ref{SConservLawMacroscopic} in Section~\ref{SNoether}, where we show that
if symmetric derivatives vanish macroscopically, the conservation laws of Chapter~\ref{DissNoether} also apply in the non-differentiable case considered here.

\section{Connection to the Differentiable Case}\label{SConDiff}

In this section, we study the connection to the differentiable case established in Chapter~\ref{DissJet}.
We give two different sets of assumptions which imply that the stochastic term~\eqref{SchiA} vanishes pointwise.
In our first approach, we adapt the Definition~\eqref{JJDiffDef} of differentiable jets to a family~\eqref{Srhotau}. In the second approach,
we work with~\eqref{JJtesttau} to extend differentiable jets $\Jdiff$ to $\tau \in (-\delta,\delta)$ using the push-forward~\eqref{Jpushforward}.

\begin{Def}\label{SdiffFjet} A jet $\w \in \J$ is \textbf{differentiable with respect to a family} $(\rho_\tau)_{\tau \in (-\delta,\delta)}$ of measures if
for all $\tau \in (-\delta,\delta)$ and all $x \in M$ we have
\[
(\nabla^+_{\w}\ell_\tau)(F_\tau(x)) =  (- \nabla^+_{-\w} \ell_\tau)(F_\tau(x)) \:.
\]
\end{Def}

\begin{Prp}\label{Sdiffchi0} Let $\w \in \J$ be differentiable with respect to a family of minimizers of the form~\eqref{Srhotau} whose generator is $\v$. If $[\w,\v] \in \Jdiff$,
for every $x \in M$, we have
\[
\chi_{\w,\v}(x) = 0 \:.
\]
\end{Prp}

\Proof Denote $\w = (c,w)$ and let $\Phi$ be the flow of $w$.
Since $\w$ is differentiable with respect to the family~\eqref{Srhotau}, for every $x \in M$ and all $\tau \in (-\delta,\delta)$, $(D_w \ell_\tau)(F_\tau(x))$ exists.
Since the family consist of minimizers, we have 
\beq\label{SDiffCase3}
(D_w \ell_\tau)(F_\tau(x)) = 0
\eeq
and therefore
\[
\frac{1}{2} \: \frac{d}{ds}^{\!+}_{|_0} \Big( \ell_\tau \big (\Phi_s (F_\tau(x))\big)   -  \ell_\tau \big(\Phi_{-s}(F_\tau(x))\big)  \Big)
=  \frac{d}{ds}_{|_0} \ell_\tau \big (\Phi_s (F_\tau(x))\big) = 0
\]
for all $\tau \in (-\delta,\delta)$.
This implies
\beq\label{SDiffCase}
\frac{1}{2} \: \frac{d}{d\tau}^{\!+}_{|_0} \frac{d}{ds}^{\!+}_{|_0} f_\tau(x) \Big( \ell_\tau \big (\Phi_s (F_\tau(x))\big)   -  \ell_\tau \big(\Phi_{-s}(F_\tau(x))\big)  \Big) =  0 \: .
\eeq
Arguing as in~\eqref{SympMac}, we have
\begin{align*}\begin{split}
 &  \: \frac{d}{d \tau}^{\!+}_{|_0} \frac{d}{d s}^+_{|_0}  f_\tau(x) \:  \ell_\tau \big(\Phi_s(F_\tau(x))\big)  
 =  \frac{d}{d \tau}^{\!+}_{|_0}  f_\tau(x) \: D^+_{w|_{F_\tau(x)}}  \ell_\tau \big( F_\tau(x) \big) \\
&=  \big( b(x) \:  D^{+}_w +  D^{+}_v  D^{+}_w \big) \ell(x) +  \frac{d}{d\tau}^{\!+}_{|_0}  D^{+}_w \ell_\tau(x) \\
& =  \big( b(x) \: D^{+}_w +   D^{+}_w  D^{+}_v \big) \ell(x) + \frac{d}{d\tau}^{\!+}_{|_0}  D^{+}_w \ell_\tau(x) +  D^{+}_{[v,w]}\ell(x) \\
& =  \frac{d}{d s}^{\!+}_{|_0} \big( b(x) \: \ell(\Phi_s(x)) +  D^{+}_{v|_{\Phi_s(x)}} \ell(\Phi_s(x) \big)    + \frac{d}{d\tau}^{\!+}_{|_0} \frac{d}{ds}^{\!+}_{|_0}  \ell_\tau(\Phi_s(x))+  D^+_{[v,w]}\ell(x) \\
& = \frac{d}{d s}^{\!+}_{|_0}  \frac{d}{d \tau}^{\!+}_{|_0} 
f_\tau(x) \:  \ell_\tau \big (F_\tau(\Phi_s (x)) \big) + D^+_{[v,w]}\ell(x) \: ,
\end{split}\end{align*}
where in the last step we have again used~\eqref{SympMac2}.
Therefore,~\eqref{SDiffCase} implies
\beq \label{SDiffCase5}
\chi_{\w,\v}(x) = \frac{1}{2} \big( D^+_{[w,v]}\ell(x) - D^+_{-[w,v]}\ell(x) \big) \:.
\eeq
(We remark that this equation can also be obtained by noting that~\eqref{SDiffCase3} implies
\begin{align}\begin{split}\label{SDiffCase4}
0 &= \frac{d}{d\tau}_{|_0}  \frac{d}{ds}_{|_0} \ell_\tau \big (\Phi_s (F_\tau(x))\big) = \frac{d}{d\tau}_{|_0}  D_{w|_{F_\tau(x)}} \ell_\tau \big (F_\tau(x)\big) \\
&= D_v D_w \ell(x) +  \frac{d}{d\tau}_{|_0}  D_{w} \ell_\tau \big (x\big) \: .
\end{split}\end{align}
Using~\eqref{SchiA2} and the definition of $\ell_\tau$, for the last term we have
\begin{align*}
\frac{d}{d\tau}_{|_0}  D_{w} \ell_\tau \big (x\big) &= \frac{d}{d\tau}_{|_0}  D_{w} \, \Big( \int_M \L \big(x,F_\tau(y) \big)  \: f_\tau(y) \: d\rho(y) - \frac{\nu}{2} \: \Big) \\
& = D_{w} \int_M \nabla_{2,\v} \: \L \big(x,y \big)  \: d\rho(y) \\
& = - \frac{1}{2} \: \big( D^+_w  - D^+_{-w} \big) \int_{M} d\rho(y) \: \nabla^+_{1,\v} \: \L \big( x, y \big) 
+  (D_w b)(x) \: \frac{\nu}{2} + \chi_{\w,\v}(x) \\
& = - \frac{1}{2} \: \big( D^+_w  - D^+_{-w} \big) \nabla^+_{1,\v}  \Big( \ell(x) + \frac{\nu}{2} \Big)
+  (D_w b)(x) \: \frac{\nu}{2} + \chi_{\w,\v} (x)\\
& \stackrel{(\star)}{=}  - \frac{1}{2} \: \big( D^+_w  - D^+_{-w} \big) D^+_{v}  \ell (x) + \chi_{\w,\v} (x)
\end{align*}
where in $(\star)$ we have used $D_w b(x) \ell(x) = 0$ by~\eqref{SDiffCase3} and since $\rho$ is a minimizer.
Thus~\eqref{SDiffCase4} implies
\begin{align*}
0 &= \frac{1}{2} \: D^+_v \big( D^+_w  - D^+_{-w} \big) \ell(x) - \frac{1}{2} \: \big( D^+_w  - D^+_{-w} \big) D^+_{v}  \ell(x)  + \chi_{\w,\v} (x) \\
& = \frac{1}{2} \big( D^+_{[v,w]} - D^+_{-[v,w]} \big) \ell(x)  + \chi_{\w,\v} (x) \: ,
\end{align*}
giving~\eqref{SDiffCase5}.)

Since $[\w,\v] \in \Jdiff$, $D_{[w,v]} \ell(x)$ exists, and since $\rho$ is a minimizer, this implies that
\[ \frac{1}{2} \big( D^+_{[w,v]}\ell(x) - D^+_{-[w,v]}\ell(x) \big)  = D_{[w,v]} \ell(x) = 0 \]
for every $x \in M$. Thus~\eqref{SDiffCase5} gives the result. \QED

This concludes the first approach to show that the stochastic term vanishes pointwise. 
For the second approach, working with the push-forward~\eqref{JJtesttau}, the following definition is necessary.
It allows us to exchange the order of differentiation in the proof of Proposition~\ref{SDiffCase_2}.

\begin{Def} \label{Ssuitdiff} A transformation~$\Phi \in  C^\infty\big((-s_{\max}, s_{\max}) \times \F \rightarrow \F \big) $ is
\Chd{Ich würde den Namen ``sufficiently differentiable'' in der Diss lassen, und dann in der Publikation ändern, falls das iO ist?}%
{\textbf{sufficiently differentiable}} with respect to a family~\eqref{Srhotau} if the
partial derivatives
\[ \partial_\tau (\ell \circ F \circ \Phi ), \ \partial_s (\ell \circ F \circ \Phi) \text{ and } \: \partial_\tau \partial_s (\ell \circ F \circ \Phi) \]
exist in $(-\delta,\delta) \times (-s_{\max}, s_{\max}) $ and if $\partial_\tau \partial_s (\ell \circ F \circ \Phi)$
is continuous in $s$ and $\tau$ at the point $(s,\tau) = (0,0)$.
\Later{Bei Gelegenheit mal überlegen, ob wir vielleicht ein \textcolor{lightred}{``besseres''
Theorem zum vertauschen von \emph{Semi-}Ableitungen} geben können. - Also mit weniger Annahmen.}%
\Later{Bemerkung: Falls wir im Beweis von Prp.~\ref{SmodFieldEqA} die $s$ und $\tau$-Ableitungen vertauschen würden, bräuchten
wir diese Annahmen die ganze Zeit nicht. Allerdings sähen dann die Feldgleichungen anders aus.}
\end{Def}

\begin{Prp}\label{SDiffCase_2} Let $\w=(c,w)  \in \Jdiff$. Let $\v \in \J$ be the generator of a family~\eqref{Srhotau} of 
minimizers which satisfies \beq\label{SweakNull}
\nabla_{\w(\tau)} \ell_\tau(z) = 0 \qquad \text{for all~$z \in M_\tau$} \: ,
\eeq
where $\w(\tau)$ is the push-forward~\eqref{Jpushfdef} of $\w$.
Assume that the flow $\Phi$ of $w$ is sufficiently differentiable
with respect to the family~\eqref{Srhotau}.
Then
\[
\chi_{\w,\v}(x) = 0 \:.
\]
\end{Prp}

\Proof
From $\w \in \Jdiff$, it follows that $\tilde{\w}:=(c+(D_w \log f_\tau), \,w) \in \Jdiff$. Hence the proof of Lemma~\ref{Jlemmalinlip} 
applies and we have
\begin{align*}
0 &= \frac{d}{d\tau}_{|_0} \nabla_{\w} \bigg( \int_M f_\tau(x) \:\L\big(F_\tau(x), F_\tau(y) \big)\: f_\tau(y)\: d\rho(y) 
-\frac{\nu}{2} \: f_\tau(x) \bigg) \\
&= \frac{d}{d\tau}_{|_0} \frac{d}{ds}_{|_0} \bigg( \int_M f_\tau(\Phi_s(x)) \:\L\big(F_\tau(\Phi_s(x)), F_\tau(y) \big)\: f_\tau(y)\: d\rho(y) 
-\frac{\nu}{2} \: f_\tau(\Phi_s(x)) \bigg) \\
& \quad +  c(x) \bigg( \int_M \frac{d}{d\tau}_{|_0} f_\tau(x) \:\L\big(F_\tau(x), F_\tau(y) \big)\: f_\tau(y)\: d\rho(y) 
-\frac{\nu}{2} \: b(x) \bigg) 
\end{align*}
The last line vanishes by~\eqref{SchiA4}. 
Hence we are left with
\begin{align}\begin{split}\label{SDiffCase_21}
0 &=   \frac{d}{d\tau}_{|_0} \frac{d}{ds}_{|_0} \bigg( \int_M f_\tau(\Phi_s(x)) \:\L\big(F_\tau(\Phi_s(x)), F_\tau(y) \big)\: f_\tau(y)\: d\rho(y) 
-  \frac{\nu}{2}  \: f_\tau(\Phi_s(x)) \bigg)  \\
&=   \frac{d}{d\tau}_{|_0} \frac{d}{ds}_{|_0} \bigg( f_\tau(\Phi_s(x)) \: \ell_\tau\big(F_\tau(\Phi_s(x)) \big) + \frac{\nu}{2}  \: f_\tau(\Phi_s(x))
-  \frac{\nu}{2}  \: f_\tau(\Phi_s(x)) \bigg)  \\
&=   \frac{d}{d\tau}_{|_0} \frac{d}{ds}_{|_0} f_\tau(x) \: \ell_\tau\big(F_\tau(\Phi_s(x)) \big) 
+ \frac{d}{d\tau}_{|_0} \frac{d}{ds}_{|_0} f_\tau(\Phi_s(x)) \: \ell_\tau\big(F_\tau(x) \big) \\
&=   \frac{d}{d\tau}_{|_0} \frac{d}{ds}_{|_0} f_\tau(x) \: \ell_\tau\big(F_\tau(\Phi_s(x)) \big) \: ,
\end{split}\end{align}
where in the first step we have used the definition of $\ell_\tau$ and in the last step we have argued similarly as in~\eqref{SchiA7} to conclude that
\[
\frac{d}{d\tau}_{|_0} \frac{d}{ds}_{|_0} f_\tau(\Phi_s(x)) \: \ell_\tau\big(F_\tau(x) \big)  = 0
\]
(using that~\eqref{Srhotau} consists of minimizers  by assumption).
The conditions in Definition~\ref{Ssuitdiff} imply that the assumptions of Theorem~9.41 of~\cite{RudinAna} are satisfied so that we can exchange the $s$- and $\tau$-derivative
in~\eqref{SDiffCase_21},
\begin{align*}
0 &=  \frac{d}{ds}_{|_0} \frac{d}{d\tau}_{|_0}  f_\tau(x) \: \ell_\tau\big(F_\tau(\Phi_s(x)) \big) = \frac{1}{2} \Big( \frac{d}{ds}^{\!+}_{|_0}  + \frac{d}{ds}^{\!-}_{|_0} \Big) \frac{d}{d\tau}_{|_0}  f_\tau(x) \: \ell_\tau\big(F_\tau(\Phi_s(x)) \big) \\
& = \frac{1}{2} \: \frac{d}{d s}^{\!+}_{|_0}  \frac{d}{d \tau}^{\!+}_{|_0}  \: 
f_\tau(x) \: \Big( \ell_\tau \big (F_\tau(\Phi_s (x)) \big)   -  \ell_\tau \big( F_\tau(\Phi_{-s}(x))\Big) = \chi_{\w,\v}(x)  \: ,
\end{align*}
giving the result.
\QED

\begin{Remark}\em \Remt{(Interpretation of Propositions~\ref{Sdiffchi0} and~\ref{SDiffCase_2})}
The above results show that the non-differentiable
linearized field equations~\eqref{SlinFEqA} indeed generalize the linearized field equations~\eqref{Jeqlinlip}
in the sense that either of the assumptions of\\[-1.8em]
\begin{itemize}
\itemsep-.2em
\itemD $\w \in \J$ being differentiable with respect to the family~\eqref{Srhotau} generated by $\v$ and $[\w,\v] \in \Jdiff$ (Proposition~\ref{Sdiffchi0}), or
\itemD the flow of the vector component of $\w$ being sufficiently differentiable and the push-forward $\w(\tau)$ being a differentiable jet for every $\tau \in (-\delta,\delta)$
(Proposition~\ref{SDiffCase_2})\\[-1.8em]
\end{itemize} 
imply that the stochastic term vanishes. The fact that it does not suffice to demand $\w \in \Jdiff$ to obtain the linearized field equations
from the non-differentiable linearized field equations arises because in order to derive~\eqref{Jeqlinlip} from~\eqref{Jfinal}
one needs to exchange the order of differentiation as in the proof of Proposition~\ref{SDiffCase_2}. This is implicit in Definition~\ref{Jdeflin}. Indeed, Definition~\ref{Jdeflin}
\Chd{``it'' $\rightarrow$ `` Definition~\ref{Jdeflin}''}%
does imply that for every solution of the linearized field equations we have $\chi_{\w,\v}(x) = 0 $ for all $x \in M$ if $\w \in \Jdiff$.

We remark that Example~\ref{Jexsphere} shows that the push-forward $\Jdiff_\tau$, defined as in~\eqref{JJtesttau}, and Definition~\ref{SdiffFjet} do not agree in general.
Thus the two approaches considered in this section are distinct. 
Since the push-forward $\w(\tau)$ of a differentiable jet $\w$ is not in general differentiable with respect to the member $\rho_\tau$ of the family~\eqref{Srhotau}, 
Definition~\ref{SdiffFjet} and Proposition~\ref{Sdiffchi0} could be considered to be more adequate than the Assumption~\eqref{SweakNull} in Proposition~\ref{SDiffCase_2}.
\QEDrem
\end{Remark}

\section{The Symplectic Form and Hamiltonian Time Evolution}\label{SympStoch}

In this section, we study the symplectic form introduced in Section~\ref{JlscSymp}, in the context of the non-differentiable linearized field equations.
Thus for any compact $\Omega \in \Sigma(\F)$, we consider the mapping
\beq \label{SOSIlipNonDiff}
\sigma_\Omega(\u, \v) = \int_\Omega d\rho(x) \int_{M \setminus \Omega} d\rho(y)\:
\sigma_{\u, \v}(x,y) 
\eeq
with
\begin{align}\begin{split} \label{SOSIIntegrandNonDiff}
&\sigma_{\u, \v}(x,y) := \frac{1}{4} \sum_{s,s'=\pm} \sigma^{s,s'}_{\u, \v}(x,y) \qquad \textrm{ and}  \\
&\sigma^{s,s'}_{\u, \v}(x,y) = \nabla^s_{1,\u} \nabla^{s'}_{2,\v} \L(x,y) - \nabla^{s'}_{1,\v} \nabla^s_{2,\u} \L(x,y) 
\end{split}\end{align}
as defined in~\eqref{JOSIlip} and~\eqref{JlscSympBilinEq}. (Recall that the name ``symplectic form''
for this mapping is justified by Section~\ref{JSecSmooth}, where we show that if $\L$ and $\ell$ are smooth,~\eqref{SOSIlipNonDiff} indeed is a symplectic form.)

In Chapter~\ref{DissJet}, we have found that, given assumptions to ensure the existence of certain terms, for
all $\u,\v$ which are differentiable jets and solutions of the linearized field equations,
this symplectic form vanishes,
\beq\label{SOSIlip0}
\sigma_\Omega(\u,\v) = 0,
\eeq
(Theorem~\ref{JthmOSIlip}) and is bilinear (Proposition~\ref{JlscSympBilin}).
As explained in Section~\ref{JIntr} of the introduction and in Section~\ref{JSecSympForm}, 
taking the limit indicated in Figure~\ref{Jfigjet1} and assuming suitable decay properties of jets,~\eqref{SOSIlip0} gives rise to a surface layer integral which is conserved under the evolution of the linearized
field equations, the \emph{Hamiltonian time evolution}.

The next theorem specifies how the stochastic term~\eqref{SchiA} relates to the Hamiltonian time evolution. 

\begin{Thm}\label{SthmOSIlipNonDiff}
Let $\w$ and $\v$ be solutions of the non-differentiable linearized field equations~\eqref{SlinFEqA}.
Then for any compact~$\Omega \in \Sigma(\F)$, the symplectic form~\eqref{SOSIlipNonDiff}
satisfies
\begin{align}\label{SsympNonDiff}
\sigma_{ \Omega}(\w, \v) &= \int_{ \Omega} d\rho(x) \: \widetilde \chi_{\w,\v}(x)  - \int_{ \Omega} d\rho(x) \:  \wnabla_{[\w,\v]} \: \ell(x)  \: ,
\end{align}
where
\beq\label{SantiSym5}
\widetilde \chi_{\w,\v}(x) = \frac{1}{2} \big( \chi_{\w,\v}(x) -  \chi_{\w,-\v}(x) - \chi_{\v,\w}(x) + \chi_{\v,-\w}(x) \big) \: .
\eeq
\end{Thm}

Note that for $\wnabla_{[\w,\v]} $, we have used Notation~\eqref{SnotationSymmSem}. We interpret this result in Remark~\ref{SHamiltInt}. 
An alternative form of the right hand side of~\eqref{SsympNonDiff} is given in Proposition~\ref{SAltFormSymp}.
\Proof
Anti-symmetrizing~\eqref{SlinFEqA} in $\v$ and $-\v$ yields
\beq \label{SantiSym1}
\wnabla_\w  \: \int_M d\rho(y) \: 
\big( \wnabla_{1,\v} + \wnabla_{2,\v} \big) \: \L \big( x, y \big)  - \wnabla_\w \wnabla_\v \: \frac{\nu}{2}\: =  \frac{1}{2} \big( \chi_{\w,\v}(x) -  \chi_{\w,-\v}(x) \big)
\eeq
Define
\beq\label{SantiSym2}
\sigma_{\w,\v}(x) := \int_M d\rho(y) \: \big( \wnabla_{1,\w} \: \wnabla_{2,\v}  - \wnabla_{1,\v} \: \wnabla_{2,\w} \big) \: \L(x,y) \: .
\eeq
Anti-symmetrizing~\eqref{SantiSym1} in $\w$ and $\v$ gives
\beq\label{SantiSym6}
\wnabla_{[\w,\v]} \: \big( \ell(x) + \frac{\nu}{2} \big) + \sigma_{\w,\v}(x)  - \wnabla_{[\w,\v]} \frac{\nu}{2} = \wnabla_{[\w,\v]} \: \ell(x)  + \sigma_{\w,\v}(x) = \widetilde \chi_{\w,\v}(x) \: ,
\eeq
where we have used Assumption~(vi) to exchange differentiation and integration.
We now integrate~\eqref{SantiSym6} over a compact $ \Omega \subset \F$.
Assumption~(v) implies that the integrals exist.\Later{Noch etwas ausführen.}
Hence we have 
\beq\label{SantiSym4}
\int_{ \Omega} d\rho(x) \: \sigma_{\w,\v}(x)  =  \int_{ \Omega} d\rho(x) \: \widetilde \chi_{\w,\v}(x)  - \int_{ \Omega} d\rho(x) \:  \wnabla_{[\w,\v]} \: \ell(x)  \:.
\eeq
Since the integrand of~\eqref{SantiSym2} is anti-symmetric with respect to $x$ and $y$, the integration over $\Omega \times \Omega$ vanishes,
\beq\label{SantiSym3}
\int_\Omega d\rho(x) \: \int_\Omega d\rho(y) \: \big( \wnabla_{1,\w} \: \wnabla_{2,\v}  - \wnabla_{1,\v} \: \wnabla_{2,\w} \big) \: \L(x,y) = 0 \:.
\eeq
From~\eqref{SOSIIntegrandNonDiff} we have
\begin{align*}
&\sigma_{\w, \v}(x,y)
= \frac{1}{4} \: \big( \nabla_{1,\w}^+ + \nabla_{1,\w}^- \big) \big( \nabla_{2,\v}^+  + \nabla_{2,\v}^- \big) \L(x,y) \\
& \; - \frac{1}{4} \big( \nabla_{1,\v}^+ + \nabla_{1,\v}^- \big) \big( \nabla_{2,\w}^+ + \nabla_{2,\w}^- \big) \L(x,y) = \big( \wnabla_{1,\w} \: \wnabla_{2,\v} - \wnabla_{1,\v} \: \wnabla_{2,\w} \big)  \L(x,y) \: .
\end{align*}
Together with~\eqref{SantiSym3}, this implies
\[
\sigma_{ \Omega}(\w, \v) = \int_{ \Omega} d\rho(x) \:  \sigma_{\w,\v}(x) \: .
\]
Therefore,~\eqref{SantiSym4} gives~\eqref{SsympNonDiff}.
\QED

Thus, in general, the evolution according to the non-differentiable linearized field equations does not conserve the surface layer integrals~\eqref{JIntrOSIN}. 
However, they are conserved if Definition~\ref{SmacSymm0} holds, as shown by the next proposition. Note that since Definition~\ref{SmacSymm0} makes a statement
about minimizers, in the next proposition we need to assume that the jets $\w$ and $\v$ are generators of families~\eqref{Srhotau} of minimizers.
\Later{Könnte man anpassen.}

\begin{Prp}\label{SthmOSIlipNonDiffMac}
If symmetric derivatives vanish macroscopically 
and $\w, \v$ are generators of families of minimizers of the form~\eqref{Srhotau} (and hence solutions of the non-differentiable linearized field equations),
for every compact $\tilde \Omega \in \mathscr M$, we have
\Chd{$\Omega \subset \Sigma(\F) \rightarrow \tilde \Omega \in \mathscr M$}%
\Chd{$\sigma_{\Omega} \rightarrow \sigma_{\tilde \Omega}$}%
\[
\sigma_{\tilde \Omega}(\w, \v) = 0 \: .
\]
\end{Prp}

\Proof
We again use that for a family~\eqref{Srhotau} of minimizers,
Definition~\ref{SmacSymm0} implies~\eqref{SmacSymmFam2}.
Proceeding as in~\eqref{SmacSymmFam3} and~\eqref{SympMac}, we obtain
\begin{align*}
&\int_{\tilde \Omega} \widetilde \chi_{\w,\v}(x) \: d\rho(x) 
=  - \frac{1}{4} \: \int_{\tilde \Omega} \Big( \big( D^+_{[v,w]} - D^+_{[v,-w]} \big) - \big( D^+_{[-v,w]} - D^+_{[-v,-w]}  \big) \\
&- \big( D^+_{[w,v]} - D^+_{[w,-v]} \big) + \big( D^+_{[-w,v]} - D^+_{[-w,-v]} \big) 
\Big) \: \ell(x) \: d\rho(x)  \\
&=  - \frac{1}{2} \: \int_{\tilde \Omega} \Big( D^+_{[v,w]} - D^+_{-[v,w]} -  D^+_{[w,v]} + D^+_{-[w,v]}  
\Big) \: \ell(x) \: d\rho(x)  \\
&= - \int_{\tilde \Omega} \Big( D^+_{[v,w]} - D^+_{-[v,w]} 
\Big) \: \ell(x) \: d\rho(x)  = - 2 \int_{\tilde \Omega} \tilde D_{[v,w]}  \ell(x) \: d\rho(x)  \\
& \stackrel{(\star)}{=}  - 2 \int_{\tilde \Omega} \tilde \nabla_{[\v,\w]}  \ell(x) \: d\rho(x)  = 2 \int_{\tilde \Omega} \tilde \nabla_{[\w,\v]}  \ell(x) \: d\rho(x)  \: ,
\end{align*}
where in $(\star)$ we have used~\eqref{SnablaCommutator} and $\ell|_M = 0$. In the last step we have used the definition of $\nabla^+_{[\w,\v]}$,~\eqref{SDefCommutator}.
Hence~\eqref{SsympNonDiff} implies
\[
\sigma_{\tilde \Omega}(\w, \v) = \int_{\tilde \Omega} \tilde \nabla_{[\w,\v]}  \ell(x) \: d\rho(x)  \: .
\]
This term vanishes by~\eqref{SmacDiffMin0infini},
\[
\int_{\tilde \Omega} \tilde \nabla_{[\w,\v]}  \ell(x) \: d\rho(x) = \int_{\tilde \Omega} \tilde D_{[w,v]}  \ell(x) \: d\rho(x) = 0 \: ,
\]
where we have again used~\eqref{SnablaCommutator} and $\ell|_M \equiv 0$.
\QED

\begin{Prp}\label{SAltFormSymp} If $ \v=(b,v), \: \w=(c,w)$ and $\Omega$ are as in Theorem~\ref{SthmOSIlipNonDiff}, we have
\begin{align}\begin{split}\label{SsympNonDiff2}
\sigma_{ \Omega}(\w, \v)&=  \int_{ \Omega} d\rho(x) \: \big(  b(x) \: \wnabla_\w \ell (x)  - c(x) \: \wnabla_\v \ell (x) \big) \\
& \ \ \  +  \int_{ \Omega} d\rho(x)  \int_M d\rho(y) \: \Big( \tilde D_{1,w}  \wnabla_{2,\v} \: \L(x,y) - \tilde D_{1,v}  \wnabla_{2,\w} \: \L(x,y) \Big) \:.
\end{split}
\end{align}
\end{Prp}

\Proof
To show~\eqref{SsympNonDiff2}, expand $\chi_{\w,\v}(x)$ similarly to~\eqref{SdiffCase}
to obtain
\begin{align}\label{SdiffCase2-1}
\chi_{\w,\v}(x) & =\frac{1}{2} \: b(x) \: \frac{d}{d s}^{\!+}_{|_0}  \Big( \ell \big (\Phi_s (x)\big)   -  \ell \big(\Phi_{-s}(x)\big) \Big) \\ \label{SdiffCase2}
& + \frac{1}{2} \: \frac{d}{d s}^{\!+}_{|_0}  \frac{d}{d \tau}^{\!+}_{|_0}   \Big( \ell \big (F_\tau(\Phi_s (x))\big)   -  \ell \big(F_\tau(\Phi_{-s}(x))\big) \Big) \\
& + \frac{1}{2} \: \frac{d}{d s}^{\!+}_{|_0}  \frac{d}{d \tau}^{\!+}_{|_0}   \Big( \ell_\tau \big (\Phi_s (x)\big)   -  \ell_\tau \big(\Phi_{-s}(x)\big) \Big) \:. \label{SdiffCase2+1}
\end{align}
The second term~\eqref{SdiffCase2} yields a contribution to~\eqref{SantiSym5} of
\[
\tilde D_{[w, v]} \ell(x) \:.
\]
Using~\eqref{SnablaCommutator}, we see that the $\tilde D_{[w,v]}$-terms in~\eqref{SantiSym4} cancel. Furthermore, for $x \in M$ we have
\[
(\tilde D_w b) \ell(x)  =  (\tilde D_v c)  \ell(x) \: = 0 \: .
\]
Hence the only contributions to the right hand side of~\eqref{SantiSym4} come from~\eqref{SdiffCase2-1} and~\eqref{SdiffCase2+1}.
In~\eqref{SantiSym5}, they add up to
\[
b(x) \: \tilde D_w \ell (x)  - c(x) \: \tilde D_v \ell (x) + \tilde D_w \int_M \wnabla_{2,\v} \: \L(x,y) d\rho(y)  - \tilde D_v \int_M \wnabla_{2,\w} \: \L(x,y) d\rho(y) \:.
\]
We add $b(x) \: c(x) \: \ell (x)  - c(x) \: b(x) \: \ell(x) $ to this equation. This gives
\[
b(x) \: \wnabla_\w \ell (x)  - c(x) \: \wnabla_\v \ell (x) + \tilde D_w \int_M \wnabla_{2,\v} \: \L(x,y) d\rho(y)  - \tilde D_v \int_M \wnabla_{2,\w} \: \L(x,y) d\rho(y) \:.
\]
Integrating over $\Omega$ yields~\eqref{SsympNonDiff2}.
\QED

\begin{Remark}\label{SHamiltInt}\em \Remt{(Interpretation of Theorem~\ref{SthmOSIlipNonDiff})}
To interpret Theorem~\ref{SthmOSIlipNonDiff}, suppose that space-time $M$ carries a suitable notion of Cauchy surfaces $(\Sigma_t)_{t \in \R}$.
Denote the restriction of $\J$ to $\Sigma_t$ by $\J_{\Sigma_t}$. We define an operator $\mathcal U_{t_2,t_1}: \J_{\Sigma_{t_1}} \rightarrow \J_{\Sigma_{t_2}}$
by $\mathcal U_{t_2,t_1} \: \u_1 := \u|_{\Sigma_{t_2}}$, where $\u \in \Jfield$ is such that $\u|_{\Sigma_{t_1} }= \u_1$. Put differently,
$\mathcal U_{t_2,t_1}$ represents the time evolution, as specified by the non-differentiable linearized field equations~\eqref{SlinFEqA}.
(This operator might not be unique.) 
As explained in the context of~\eqref{JIntrOSIN} in the introduction, we can introduce a symplectic form on $\Sigma_t$ by
choosing~$\Omega = \Omega_t$ in~\eqref{SOSIlipNonDiff},  where $\Omega_t$ is the past of the Cauchy surface $\Sigma_t$.
If suitable decay conditions hold for the solutions of the non-differentiable linearized field equations as one approaches infinity on the Cauchy surfaces
(``spatial infinity''), this surface layer integral exists and Theorem~\ref{SthmOSIlipNonDiff} and Proposition~\ref{SthmOSIlipNonDiffMac} apply.

Theorem~\ref{SthmOSIlipNonDiff} implies that
\[
\sigma_{\Sigma_{t_2}}(\mathcal U_{t_2,t_1}  \u_1  ,  \mathcal U_{t_2,t_1} \v_1) \neq \sigma_{\Sigma_{t_1}} (\u_1  ,  \v_1) 
\]
for all $t_1, t_2$ and $\u_1,\v_1 \in \J_{\Sigma_{t_1}}$ in general. I.e., the the time evolution operator $\mathcal U_{t_2,t_1}$
does not constitute a  symplectomorphism with respect to the symplectic form~\eqref{SOSIlipNonDiff} in general.
If however, Definition~\ref{SmacSymm0} holds, the situation is different whenever $t_1$ and $t_2$ are such that the time-strip between them is 
a macroscopic region $\tilde \Omega \in \mathscr M$ as in Definition~\eqref{SDefMac}. Then Proposition~\eqref{SthmOSIlipNonDiffMac}
implies that 
\[
\sigma_{\Sigma_{t_2}}(\mathcal U_{t_2,t_1}  \u_1  ,  \mathcal U_{t_2,t_1} \v_1) = \sigma_{\Sigma_{t_1}} (\u_1  ,  \v_1) 
\]
for all $\u_1,\v_1 \in \J_{\Sigma_{t_1}}$.

As mentioned in the context of Definition~\ref{SDefMac}, the specification of what is `macroscopic'
depends on the application. (In fact, in principle it is possible to prove or disprove~\eqref{SmacDiffMin0} by investigating properties
of general minimizers of the causal action principle. Hence one might be able to determine from first principles which regions $\tilde \Omega$
\Chd{``to turn Definition ... around'' gelöscht (Denglish)}%
are `macroscopic' in the sense that~\eqref{SmacDiffMin0} is satisfied.)
Based on contemporary physics, it is reasonable to suppose that a necessary requirement for a region to be macroscopic is that its diameter is \chd{many} orders
\Chd{``several'' $\rightarrow$ ``many''}%
of magnitude larger than the Planck length. In this case, Proposition~\ref{SthmOSIlipNonDiffMac} would imply that for times which are separated by an interval which is \chd{many} orders of magnitude larger than the Planck time, the evolution is in fact a symplectomorphism, and differences from the Hamiltonian time evolution disappear.

To understand what this means, let us compare the time evolution operator $ \mathcal U_{t_2,t_1}$ to the time evolution operator $U(t_2,t_1)$ in non-relativistic quantum theory  and the property that $\mathcal U_{t_2,t_1}$ is a symplectomorphism to the property that $U(t_2,t_1)$ is unitary. Then the above explanations
would compare to a situation where microscopically, the unitary evolution is broken by a stochastic term, opening the doors for ``dynamical reduction'' to occur  (cf. Section~\ref{IModelQT}). Macroscopically, however, the stochastic contributions cancel, giving rise to a unitary time evolution on time-scales which are \chd{many} orders of magnitude larger than the Planck time. \Later{Hier evt. wirklich ein $\Lambda$ einführen, wie im auskommentierten Teil, um die stochastische Korrektur zu fassen. (Ggf. sogar gleich für quadratische Korrekturen ebenso machen.) Vergleiche mit~\cite{TumulkaQFTCollapse}.}
\QEDrem
\end{Remark}

\section{Second Order Terms}\label{SO2}

In this section, we study higher order corrections to the non-differentiable linearized field equations. For the connection to quantum 
theory (see Section~\ref{IModelQT}), it is sufficient to consider second order terms, and for simplicity and clarity we restrict the following analysis to this order.

A systematic expansion of the weak Euler-Lagrange equations in the differentiable case can be found in~\cite{perturb}. Since the ansatz in~\cite{perturb} differs from the
following ansatz, we include some preliminary and motivating remarks. A short comparison of the two approaches can be found im Remark~\ref{ScPerturb}. 

\subsection{Preliminaries} \label{SO2Prelim}

Recall that integral curves of a vector field $v \in \Gamma(T\F)$ are defined as maps $\gamma: \R \supset  J \rightarrow \F$ such that
$\gamma(t)' = v|_{\gamma(t)}$ for all $t \in J$, where $J$ is an open interval. A \emph{flow domain} on $\F$ is an open subset $\mathcal D \subset \R \times \F$
such that for every $x \in \F$, the set $\mathcal D_x := \{ t \in \R | (t,x) \in \mathcal D \}$ is an open interval containing $0$ (cf. e.g.~\cite[Chapter 12]{LeeSmooth}).
A \emph{flow} on $\F$ is a smooth map $F : \mathcal D \rightarrow \F$ that satisfies
\begin{align*}
&F_0(x) = x  & &\textrm{ for all } x \in \F \\
&F_{\tau'}(F_{\tau}(x)) = F_{\tau' + \tau}(x)  & &\textrm{ for all } \tau \in \mathcal D_x \textrm{ and } \tau' \in \mathcal D_{F_\tau(x)} \:,
\end{align*}
where as before we have used the notation $F_\tau(x) = F(\tau,x)$. A flow is maximal if it cannot be extended to a larger flow domain.
It follows that the \emph{infinitesimal generator} $v$ of a flow $F$ defined
by
\[
v|_x = (F|_x)'(0) \:,
\]
where $F|_x(\tau) = F(\tau,x)$ is the orbit of $x$ under $F$, is a smooth vector field on $\F$. 
Crucial for the following ansatz is that the reverse is also true. Even so the infinitesimal generator is defined to be the derivative of the flow
at $\tau = 0$, is uniquely determines the flow $F$ in the following sense.

\noindent \emph{Theorem}: (\cite[Theorem 12.9]{LeeSmooth})
Let $v$ be a smooth vector field on a smooth manifold~$\F$. Then there is a unique maximal flow whose infinitesimal generator is $v$.

\begin{Remark}\label{ScPerturb}\em \Remt{(Comparison to~\cite{perturb})} 
The crucial difference between this approach and the approach in~\cite{perturb} is that in~\cite{perturb}, one assumes that a smooth vector field $\tilde v$ in a neighborhood of $M$ is given. Since an extension to $\F$ remains arbitrary, $\tilde v$ does not specify a unique flow on $\F$ and further information is necessary to determine it.
Here, on the other hand, we are working with a
smooth vector field $v$ on $\F$, the freedom in extending $\tilde v$ to $\F$ is implicit in the choice of family~\eqref{Srhotau}.
Which approach is suitable depends on the application.
\QEDrem
\end{Remark}

\subsection{Second Order Terms}

Let $h \in C^\infty(\F,\R)$ and let $v \in \Gamma(\F)$ be the infinitesimal generator of a flow $F$. The above definitions imply that
\beq\label{Sflow1}
v(h)\big|_{F_{\tilde \tau}(x)} = 
\frac{d}{d\tau} h(F_\tau(x)) \big|_{\tau = \tilde \tau} \:,
\eeq
where we have used that $F|_x(\tau)$ is an integral curve of $v$ by~\cite[Lemma 12.7(d)]{LeeSmooth}.
Since $h$ is smooth and $F$ is smooth (by definition of a flow), it follows that the right hand side of~\eqref{Sflow1}
is smooth in $\tau$ as well. Thus, we can take a further $\tau$-derivative. This gives
\begin{align}\begin{split}\label{Sflow2}
& \frac{d^2}{d\tau^2} h(F_\tau(x)) \big|_{\tau = \bar \tau} = \frac{d}{d\tilde \tau} \Big( \frac{d}{d\tau} h(F_\tau(x)) \big|_{\tau = \tilde \tau} \Big)  \Big|_{\tilde \tau = \bar \tau}\\
& \qquad = \frac{d}{d\tilde \tau} \Big(  v(h)\big|_{F_{\tilde \tau}(x)}  \Big)  \Big|_{\tilde \tau = \bar \tau} = v(v(h))\big|_{F_{\bar \tau}(x)} \: .
\end{split}\end{align}
Clearly, this object is not a vector field. \Later{Evt. Verbindung mit $2$-jets herstellen.}
To align~\eqref{Sflow2} with the usual notation in this chapter, we write
\[
\frac{d^2}{d\tau^2} h(F_\tau(x)) \big|_{\tau = \bar \tau} =  D_{1,v} D_{1,v} h(F_{\bar \tau}(x)) \: ,
\]
where as usual the first derivative also acts on the argument of the second.

Note that whereas a vector field $v \in \Gamma(T\F)$ determines its flow $F_\tau(x) $ for all $\tau \in \mathcal D_x$,
this is not the case for the scalar component of the transformation~\eqref{Srhotau}. The first derivative $\dot f_0$ does of course not allow
to deduce the full function $f_\tau$. Therefore, higher $\tau$-derivatives of $f_\tau$ appear in what follows.
In order to present the result in an elegant way, for a given family~\eqref{Srhotau}, we define the \emph{directional quadratic semi-derivative} for jets as
\beq\label{Sflow4}
\nabla^{+}_{\v,\v} := \ddot f_0 + 2 \dot f_0 D^+_v + D^+_v D^+_v \: ,
\eeq
and respectively $\nabla^{+}_{i,\v,\v}$ to indicate that the directional quadratic semi-derivative acts on the $i$th argument of a function.

\begin{Lemma}\label{Sprepfull1}
For any family~\eqref{Srhotau} of measures with generator $\v = (b,v)$, any $x \in \F$ and any flow $\Phi$, we have
\begin{align*}
& \frac{1}{2} \: \frac{d}{ds}^{\!+}_{|_0}   \frac{d^{\, 2}}{d\tau^2}^{\!+}_{|_0}  \int_{M_\tau} d\rho_\tau(y) \: f_\tau(x) \:
\Big( \L \big( F_\tau(\Phi_s(x)),y \big)  - \L \big( F_\tau(\Phi_{-s}(x)),y \big) \Big) \\
&\qquad = \tilde D_{w}  \int_M d\rho(y) \: \Big( \nabla^{+}_{1,\v,\v} + \nabla^{+}_{2,\v,\v} +  2 \, \nabla^+_{1,\v} \nabla^+_{2,\v}  \Big) \: 
\L(x,y) \\
& \qquad  -  \int_M d\rho(y) \: \Big((D_w \ddot f_0)(x) + 2 \, (D_w b) (x) \big( D^+_{1,v}  + D^+_{2,v} + b(y) \big) \Big) \: \L(x,y) \: ,
\end{align*}
where $w$ is the generator of $\Phi$.
\end{Lemma}

\Proof
We compute the second $\tau$-semi-derivative of the left hand side of~\eqref{SchiA1}. We have
\begin{align*}\begin{split}
& \frac{d^{\, 2}}{d\tau^{2}}^{\!+}_{|_0}  \int_{M_\tau} d\rho_\tau(y) \: f_\tau(x) \: \L \big( F_\tau(\Phi_s(x)),y \big)   = \int_M d\rho(y) \: \mathcal O(x,y)
\end{split}\end{align*}
where
\begin{align*}
&\mathcal O(x,y) = \Big( \ddot f_0(x) + D^+_{1,v} D^+_{1,v} + D^+_{2,v} D^+_{2,v} + \ddot f_0(y) \Big) \: \L\big(\Phi_s(x),y \big) \\
&+ 2 \: \Big( b(x) D^+_{1,v} + b(x) D^+_{2,v} + b(x) b(y) + D^+_{1,v} D^+_{2,v} + D^+_{1,v} b(y) + b(y) D^+_{2,v} \Big) \: \L\big(\Phi_s(x),y \big) \\
& = \Big( \nabla^{+}_{1,\v,\v} + \nabla^{+}_{2,\v,\v}	+  2 \, \nabla^+_{1,\v} \nabla^+_{2,\v}  \Big) \: \L\big(\Phi_s(x),y \big) \: .
\end{align*}
Here, Assumption~(v) guarantees existence and we have exchanged the derivatives with integration using Assumption~(vi).
Note that the derivatives which act on the first argument of $ \L $ are evaluated at the space-time-point $\Phi_s(x)$.
Therefore, taking the $s$-semi-derivative at $s=0$ gives
\begin{align*}\begin{split}
& \frac{d}{ds}^{\!+}_{|_0}   \frac{d^{\, 2}}{d\tau^2}^{\!+}_{|_0}  \int_{M_\tau} d\rho_\tau(y) \: f_\tau(x) \: \L \big( F_\tau(\Phi_s(x)),y \big)   = \int_M d\rho(y) \: \mathcal O'(x,y)
\end{split}\end{align*}
 with
\begin{align*}
\mathcal O'(x,y) &= D^+_{1,w} \: \Big(\nabla^{+}_{1,\v,\v}+\nabla^{+}_{2,\v,\v}+  2 \, \nabla^+_{1,\v} \nabla^+_{2,\v}  \Big) \: 
\L(x,y)\\
&  - \Big( (D_w \ddot f_0)(x) + 2 \, (D_w b) (x) \big( D^+_{1,v}  + D^+_{2,v} + b(y) \big) \Big) \: \L(x,y).
\end{align*}
Here, the conditions on~\eqref{SCond8Eq2} in Assumption~(v) guarantee existence. Anti-symmetrizing in $w$ gives the result.
\QED

\begin{Lemma}\label{Sprepfull2}
For a family~\eqref{Srhotau} of minimizers, $c \in C^\infty(\F,\R)$ and $x \in M$, we have
\begin{align*}
c(x) \int_M d\rho(x) \: \Big( \nabla^{+}_{1,\v,\v} + \nabla^{+}_{2,\v,\v} +  2 \, \nabla^+_{1,\v} \nabla^+_{2,\v}  \Big) \: \L (x,y) 
- c(x) \: \ddot f_0(x) \: \frac{\nu}{2} = 0 \: .
\end{align*}
\end{Lemma}

\Proof
We consider again the equation
\begin{align*}\begin{split}
& \int_{M_\tau} d\rho_\tau(y) \: f_\tau(x) \: \L \big( F_\tau(x),y \big)  \\
&\qquad =  f_\tau(x) \: \Big(  \ell_\tau \big(F_\tau(x) \big) + \frac{\nu}{2} \, \Big) \: .
\end{split}\end{align*}
Taking the $ \frac{d^{\, 2}}{d\tau^2}^{\!+}_{|_0}$ semi-derivative of the left hand side yields as above
\begin{align*}
\int_M d\rho(x) \: \Big( \nabla^{+}_{1,\v,\v} + \nabla^{+}_{2,\v,\v}	+  2 \, \nabla^+_{1,\v} \nabla^+_{2,\v}  \Big) \: \L (x,y) \: .
\end{align*}
Concerning the right hand side, we can argue as in~\eqref{SchiA7} to conclude
\[
\frac{d^{\, 2}}{d\tau^2}^{\!+}_{|_0}  f_\tau(x) \:   \ell_\tau \big(F_\tau(x) \big)   = 0 
\]
for all $x \in M$. Hence we are left with
\[
\frac{d^{\, 2}}{d\tau^2}^{\!+}_{|_0}  f_\tau(x) \: \Big(    \ell_\tau \big(F_\tau(x) \big) + \frac{\nu}{2} \, \Big) = \ddot f_0(x) \: \frac{\nu}{2} \: .
\]
Multiplying by $c(x)$ gives the result.
\QED

\begin{Thm} \Thmt{(Full non-differentiable field equations to second order)} \label{SFull2O}
\Chd{Ich denke ``full'' müsste hier grammatikalisch schon okay sein.
``Fully non-differentiable field equations'' würde ich eher verstehen als
vollkommen nicht-differenzierbar, und das soll hier ja nicht ausgesagt werden.}%
For every family~\eqref{Srhotau} of minimizers with generator $\v$, any $x \in M$ and any $\w \in \J$, we have
\begin{align*}\begin{split}
&\wnabla_\w  \int_M d\rho(y) \: \big( \nabla^+_{1,\v} + \nabla^+_{2,\v} \big) \: \L \big( x, y \big)  - \wnabla_\w \nabla^+_\v \: \frac{\nu}{2}\: \\
& + \wnabla_\w \int_M d\rho(y) \: \Big( \nabla^+_{1,\v,\v} + \nabla^+_{2,\v,\v} +  2 \, \nabla^+_{1,\v} \nabla^+_{2,\v}  \Big) \: 
\L(x,y) - \wnabla_\w \:  \nabla^+_{\v,\v} \: \frac{\nu}{2} \\
& \quad = \chi_{\w,\v}(x) + \chi^{(2)}_{\w,\v}(x) \: ,
\end{split}\end{align*}
where~$\chi_{\w,\v}(x)$ is as in~\eqref{SchiA} and
\[
\chi^{(2)}_{\w,\v}(x) =   \frac{1}{2} \: \frac{d}{d s}^{\!+}_{|_0}  \frac{d^{\, 2}}{d \tau^2}^{\!+}_{|_0}  \: 
f_\tau(x) \: \Big( \ell_\tau \big (F_\tau(\Phi_s (x)) \big)   -  \ell_\tau \big( F_\tau(\Phi_{-s}(x))\Big) \: .
\]
\end{Thm}

\Proof
Consider~\eqref{SchiA1}. Terms which arise from first derivatives in~$\tau$ are given by Proposition~\ref{SmodFieldEqA}. The terms which arise from second derivatives in $\tau$
are obtained by adding Lemmas~\ref{Sprepfull1} and~\ref{Sprepfull2}.
Furthermore, using that $\rho$ is a minimizer, for every $x \in M$, we have
\begin{align*}
& \int_M d\rho(y) \Big( (D_w \ddot f_0)(x) + 2 (D_w b) (x) \big( D^+_{1,v}  + D^+_{2,v} + b(y) \big) \Big) \: \L(x,y) + c(x) \ddot f_0(x) \: \frac{\nu}{2}  \\
& = ( \nabla_\w \ddot f_0)(x) \: \frac{\nu}{2} +  2 (D_w b) (x) \Big( D^+_{1,v}  \ell(x)   + \int_M d\rho(y) \: \nabla^+_{2,\v} \L(x,y) \Big) \\
& \stackrel{(\star)}{=} ( \nabla_\w \ddot f_0)(x) \: \frac{\nu}{2} +  2 (D_w b) (x) \Big( \nabla^+_{1,\v}  \ell(x)   + \int_M d\rho(y) \: \nabla^+_{2,\v} \L(x,y) \Big)  
\end{align*}
where in $(\star)$ we have added the term $2 (D_w b) (x) \: b(x) \: \ell(x)$, which vanishes for every $x \in M$. For the last term in brackets, we have
\begin{align*}
 \nabla^+_{1,\v}  \ell(x)   + \int_M d\rho(y) \: \nabla^+_{2,\v} \L(x,y) = \int_M d\rho(y) \: \big(  \nabla^+_{1,\v} + \nabla^+_{2,\v} \big) \: \L(x,y)  -  \nabla^+_\v \frac{\nu}{2} \:,
\end{align*}
which vanishes by~\eqref{SchiA4} since $\nabla^+_\v \frac{\nu}{2} = b(x) \:  \frac{\nu}{2} $.
Finally, using~\eqref{Sflow4} the remaining term $ ( \nabla_\w \ddot f_0)(x) \: \frac{\nu}{2}$ can be written as
\[
\nabla_\w \ddot f_0 (x) \: \frac{\nu}{2}  = \nabla_\w  \: \nabla^+_{\v,\v} \: \frac{\nu}{2} = \wnabla_\w  \: \nabla^+_{\v,\v} \: \frac{\nu}{2} \: ,
\]
giving the result. \QED

\noindent For an interpretation of Theorem~\ref{SFull2O} we refer to Section~\ref{IModelQT}.

\Later{Symplektische form mit second order term?}
\Later{Verschwindet $\chi_{\w,\v}(x)$ stochastisch falls man ... annimmt?}

\section{Noether-Like Theorems in the Non-Differentiable Setting}\label{SNoether}

We conclude this chapter by extending Theorem~\ref{Nthmsymmlag} to our setting of jets which are not differentiable.
Note that in Chapter~\ref{DissNoether}, we have worked with a slightly different definition of $\ell$ compared to~\eqref{Jelldeflip} (cf.~\eqref{Nldef}).

\begin{Thm}\label{SNoetherNonDiff}
Let $\Phi_\tau$ be a symmetry of the Lagrangian as in Definition~\ref{Ndefsymmlagr}. Then for any measure $\rho$ and any compact~$\Omega \in \Sigma(M)$,
we have
\beq
\label{SconservationLagrEqNd}
\frac{1}{2} \: \frac{d}{ds}_{|_0} \int_\Omega d\rho(x) \int_{M \setminus \Omega} d\rho(y)\:
\Big( \L \big( \Phi_s(x),y \big) - \L \big( \Phi_{-s}(x), y \big) \Big) = 
\int_\Omega d\rho(x) \: \tilde D_w \ell(x)   \:.
\eeq
\end{Thm}

\Proof We proceed exactly as in the proof of Theorem~\ref{Nthmsymmlag}.
We multiply~\eqref{Nsymmlagr} by a bounded measurable function~$f$ on~$M$ with compact support 
and integrate. For all $s \in (-s_{\max},s_{\max})$, this gives
\begin{align*}
0 &= \iint_{M \times M} f(x)\, f(y)\: \Big(
\L \big( x, \Phi_s(y) \big) - \L \big( \Phi_{-s}(x), y \big) \Big)\, d\rho(x)\, d\rho(y) \\
&=\iint_{M \times M} f(x)\, f(y)\: \Big(
\L \big( \Phi_s(x),y \big) - \L \big( \Phi_{-s}(x), y \big) \big) \Big) \, d\rho(x)\, d\rho(y)\:,
\end{align*}
where in the last step we used the symmetry in the arguments of the Lagrangian (Assumption~(i) on page~\pageref{JCond1}).
Existence of this expression is given by Assumption~(v). We replace~$f(y)$ by~$1 - (1-f(y))$, multiply out and use the definition of~$\ell$, \eqref{Jelldeflip}.
We thus obtain
\begin{align*}
0 &=\int_M f(x)\: 
\Big(  \ell \big( \Phi_s(x)\big) + \frac{\nu}{2}  - \ell\big( \Phi_{-s}(x)) - \frac{\nu}{2} \Big) \,d\rho(x) \\
&\quad -\iint_{M \times M} f(x)\, \big(1-f(y) \big)\:
\Big( \L \big( \Phi_s(x),y \big) - \L \big( \Phi_{-s}(x), y \big) \Big) \, d\rho(x)\, d\rho(y)\:.
\end{align*}
Choosing~$f$ as the characteristic function of~$\Omega$, we obtain the identity
\[ \begin{split}
\int_\Omega &d\rho(x) \: \int_{M \setminus \Omega} d\rho(y)\:
\Big( \L \big( \Phi_s(x), y \big) - \L \big(\Phi_{-s}(x),  y \big) \Big) \\
&= \int_\Omega \Big( \ell \big( \Phi_s(x) \big)  - \ell \big( \Phi_{-s}(x) \big) \Big)\: d\rho(x) \:.
\end{split}%
\]
Since this equation holds for all $s \in (-s_{\max},s_{\max})$, its $s$-derivative exists. However, using Assumption~(vi), on the right hand side
we can only exchange the semi-derivatives with integration. Doing so gives the result.
\QED

The previous theorem shows that in general, the non-differentiability of $\L$ and $\ell$ leads to a break-down of the conservation law~\eqref{NconservationLagrEq}
and thus, by Sections~\ref{Nsecexcurrent} and~\ref{NsecexEM}, to a break-down of current conservation and conservation of energy-momentum. (Recall that we had to include the a differentiability assumption in Theorem~\ref{Nthmcurrent}.) However, the Euler-Lagrange equations imply that the right hand side of~\eqref{SconservationLagrEqNd} has a varying sign (cf.~\eqref{SELgk}),
making Definition~\ref{SmacSymm0} applicable. The next proposition shows that, given this assumption, the conservation laws hold macroscopically.

\begin{Prp}\label{SConservLawMacroscopic}
If symmetric derivatives vanish macroscopically, for every minimizer~$\rho$ of the causal variational principle and every compact macroscopic $\tilde \Omega \in \mathscr M$,
we have
\[
\frac{d}{ds}_{|_0} \int_{\tilde \Omega} d\rho(x) \int_{M \setminus \tilde \Omega} d\rho(y)\:
\Big( \L \big( \Phi_s(x),y \big) - \L \big( \Phi_{-s}(x), y \big) \Big) = 0 \: .
\]
\end{Prp}

\Proof
Definition~\ref{SmacSymm0} implies 
\[ \int_{\tilde \Omega} d\rho(x) \: \tilde D_w \ell(x)  = 0 \:, \]
where we have exchanged differentiation and integration with Assumption~(vi).
Hence Theorem~\ref{SNoetherNonDiff} gives the result.
\QED

\Later{
Explain why macroscopic regions do not contain a limit of neighbourhoods of a point which converge towards that point.
And hence can't use Lebesgue differentiation theorem to conclude that $\chi_{\w,\v} (x) = 0$ in some setting from a.e.?}

\chapter{Outlook}\label{DissOutlook}

We conclude this thesis by giving a list of
research projects which could be carried out based on the results which have been established in this thesis.\\

\noindent \emph{\large Connected to the Continuum Limit:}
\begin{itemize}[leftmargin=1.5em]
\itemsep.3em
\itemD Most importantly, it is necessary to evaluate the jet-formalism, the linearized field equations and the correction terms
of Theorem~\ref{SFull2O} in the continuum limit.
If successful, this will provide correction terms to fundamental equations of contemporary physics.
\itemD Connected to this is that our formalism might be a suitable starting point for getting the connection
to the canonical formulation of quantum field theory in Fock spaces.
Working out physical applications in this context might ultimately contribute to the goal of making experimental predictions in the form of corrections to
measurable physical quantities. 
\end{itemize}\medskip

\noindent \emph{\large Internal to the Theory:}
\begin{itemize}[leftmargin=1.5em]
\itemsep.3em
\itemD One could derive higher order corrections to the linearized field equations. In particular, it seems promising to check whether
higher order field equations, as well as higher order stochastic terms, can be controlled by Assumption~\ref{SDefMac} or similar
assumptions, and how the higher order terms relate to the Hamiltonian time evolution. In case it would turn out that these higher order corrections
cannot be suppressed by reasonable assumptions, anticipating an analysis in the continuum limit, this might hint towards correction terms to
the fundamental equations of contemporary physics which cannot be ``argued away''.
\itemD In the same direction, it would be very interesting to study whether Definition~\ref{SDefMac} can be proven or disproven for minimizers
of the causal action. Even if this turns out to be impossible at the time due to computational limitations, the investigation of minimizing examples might 
give hints for one or the other case. To see why research in this direction might be relevant, suppose that one were able to prove that
Assumption~\ref{SDefMac} cannot hold for minimizers of the causal action. This would imply that the major argument for the negligibility of the stochastic term
on macroscopic scales would break away, again paving the path to new experimental predictions.
Put differently, investigations into the validity of Assumption~\ref{SDefMac} might yield lower bounds for the correction terms established in Chapter~\ref{DissStoch}.
\itemD It would be promising to establish a direct connection between Noether-like theorems and the dynamics of jets. More concretely, one might ask
in which way the existence of symmetries of the Lagrangian or of symmetries of the universal measure is related to conservation laws for solutions of the non-differentiable linearized field equations.
\itemD One could follow up on the example constructed in Section~\ref{Jseclattice}.
E.g., one could modify the Lagrangian~\eqref{JDefL} to make it compatible with the assumptions of Chapter~\ref{DissStoch}.
This would allow to study the stochastic and non-linear correction terms in an example. 
Furthermore, one could compare this calculation to numerical analysis. 
\end{itemize}\medskip

\noindent \emph{\large Concerning Foundations of Quantum Theory:}\\[-.5em]

As explained in Section~\ref{IModelQT}, the results of this thesis suggest a particular resolution of the quantum mechanical measurement problem
based on the theory of causal fermion systems. 
Of course, the exact investigation of the form of the stochastic term in the continuum limit (cf. above) is one of the most promising research projects with
respect to foundations of quantum theory, as it might give rise to a relativistic and interacting dynamical collapse model. 
Furthermore, the following questions seem promising to study from the perspective of causal fermion systems:
\begin{itemize}[leftmargin=1.5em]
\itemsep.3em
\itemD Can the Born rule be derived from the mathematical structure of the theory of causal fermion systems? Or can it be shown to emerge when minimizing the
causal action principle? Is there higher order interference?
\itemD Is there any explanation, based on the theory of causal fermion systems, of why quantum theory determines the tensor product to be the right structure to compose systems?
\itemD Which stochastic processes are compatible with the stochastic term derived in this thesis?
\end{itemize}

\clearpage\phantomsection
\addcontentsline{toc}{chapter}{Bibliography}
\newcommand{\etalchar}[1]{$^{#1}$}

\end{document}